\newif\ifConferenceVersion
\newif\ifLLNCSVersion
\newif\ifPublisherVersion
\newif\ifFullVersion
\newif\ifAnonymous
\newif\ifDraftNotes

\ConferenceVersionfalse
\Anonymousfalse
\DraftNotesfalse

\ifdefined\forceeprint
	\ConferenceVersionfalse
\fi
\ifdefined\forceconference
	\ConferenceVersiontrue
\fi
\ifdefined\forceauthor
	\Anonymousfalse
\fi
\ifdefined\forceanonymous
	\Anonymoustrue
\fi
\ifdefined\forcenotesoff
	\DraftNotesfalse
\fi
\ifdefined\forcenoteson
	\DraftNotestrue
\fi

\ifConferenceVersion
	\LLNCSVersiontrue
	\FullVersionfalse
\else
	\LLNCSVersionfalse
	\FullVersiontrue
\fi

\ifConferenceVersion
	\documentclass[english,runningheads,envcountsect]{llncs}
\else
	\documentclass[11pt,letterpaper,english,twoside]{article}
\fi

\NewExpandableDocumentCommand{\papertitle}{}{%
    Simple Extremely Lossy Functions from Small-Exponent Hashing%
}

\usepackage{iftex}
\ifPDFTeX
    \usepackage[T1]{fontenc}
    \usepackage[utf8]{inputenc}
\fi

\usepackage[english]{babel}
\usepackage{bm}
\usepackage{paralist}
\usepackage{codelines}
\ifLLNCSVersion\else
    \usepackage[
        framemethod = TikZ,
        ntheorem,
        skipabove   = 10pt,
        skipbelow   = 10pt
    ]{mdframed}
\fi
\usepackage{setspace}
\usepackage{xspace}

\ifPublisherVersion
    \let\LoadPaperPackages\LoadPublisherPackages
\else
    \ifLLNCSVersion
        \let\LoadPaperPackages\LoadLLNCSPackages
    \else
        \let\LoadPaperPackages\LoadArticlePackages
    \fi
\fi
\LoadPaperPackages

\usepackage[most]{tcolorbox}
\usepackage{tikz}

\ifFullVersion
    \usepackage{cases}
    \usepackage{dirtytalk}
    \usepackage{epigraph}
    \usepackage{fancyhdr}
    \usepackage[bottom]{footmisc}
    \usepackage{footnotehyper}
    \usepackage{orcidlink}
    \usepackage[nottoc]{tocbibind}
\fi

\ifDraftNotes
    \usepackage{lineno}
\fi

\ifLLNCSVersion
    
\fi

\ifDraftNotes
    
    \AtBeginDocument{\linenumbers}
\fi

\hypersetup{%
    colorlinks = true,
    allcolors  = PaleTurquoise4,
    citecolor  = DarkSlateGray4,
    linkcolor  = SteelBlue4,
    urlcolor   = SeaGreen4,
    pdfborder  = {0 0 0},
    pdftitle   = {\papertitle}
}

\ifPublisherVersion\else
\fi

\allowdisplaybreaks

\ifPublisherVersion
    \numberwithin{equation}{section}
\else
    \ifLLNCSVersion
        \numberwithin{equation}{section}
    \fi
\fi

\ifPublisherVersion
    \makeatletter
    \@ifundefined{theorem}{%
        \theoremstyle{plain}
        \newtheorem{theorem}{Theorem}[section]
        \theoremstyle{definition}
        \newtheorem{definition}[theorem]{Definition}
    }{}
    \makeatother
\fi

\ifLLNCSVersion
    \NewDeferredThm{deferredtheorem}{theorem}
    \NewDeferredThm{deferredlemma}{lemma}
    \NewDeferredThm{deferredcorollary}{corollary}
    \NewDeferredThmAlias{deferredtheoremsketch}{theorem}
    \NewDeferredThmAlias{deferredlemmasketch}{lemma}
\else
    \NewDeferredThmAlias{deferredtheorem}{theorem}
    \NewDeferredThmAlias{deferredlemma}{lemma}
    \NewDeferredThmAlias{deferredcorollary}{corollary}
    \NewDeferredThmAlias{deferredtheoremsketch}{theorem}
    \NewDeferredThmAlias{deferredlemmasketch}{lemma}
\fi

\ifLLNCSVersion
    \NewDeferredThm{theoremE}{theorem}
    \NewDeferredThm{lemmaE}{lemma}
    \NewDeferredThm{corollaryE}{corollary}
    \NewDeferredThmAlias{theoremEsketch}{theorem}
    \NewDeferredThmAlias{lemmaEsketch}{lemma}
\else
    \NewDeferredThmAlias{theoremE}{theorem}
    \NewDeferredThmAlias{lemmaE}{lemma}
    \NewDeferredThmAlias{corollaryE}{corollary}
    \NewDeferredThmAlias{theoremEsketch}{theorem}
    \NewDeferredThmAlias{lemmaEsketch}{lemma}
\fi

\ifPublisherVersion
\else\ifLLNCSVersion
    \else
        \docufont{%
            definitionhdr   = smallcaps,
            highlighthdr    = smallcaps,
            remarkhdr       = smallcaps,
            theoremhdr      = smallcaps,
            definitionbody  = normal,
            highlightbody   = normal,
            remarkbody      = normal,
            theorembody     = normal,
        }
    \fi\fi

\ifLLNCSVersion
    \zcRefTypeSetup{corollary}{name-sg={corollary},name-pl={corollaries},Name-sg={Corollary},Name-pl={Corollaries}}
    \zcRefTypeSetup{definition}{name-sg={definition},name-pl={definitions},Name-sg={Definition},Name-pl={Definitions}}
    \zcRefTypeSetup{example}{name-sg={example},name-pl={examples},Name-sg={Example},Name-pl={Examples}}
    \zcRefTypeSetup{fact}{name-sg={fact},name-pl={facts},Name-sg={Fact},Name-pl={Facts}}
    \zcRefTypeSetup{lemma}{name-sg={lemma},name-pl={lemmas},Name-sg={Lemma},Name-pl={Lemmas}}
    \zcRefTypeSetup{openproblem}{name-sg={open problem},name-pl={open problems},Name-sg={Open Problem},Name-pl={Open Problems}}
    \zcRefTypeSetup{proposition}{name-sg={proposition},name-pl={propositions},Name-sg={Proposition},Name-pl={Propositions}}
    \ProvideDocumentCommand{\labelcref}{O{}m}{\zcref[noname,#1]{#2}}
\else
    \AtBeginEnvironment{corollary}{\crefalias{theorem}{corollary}}
    \AtBeginEnvironment{definition}{\crefalias{theorem}{definition}}
    \AtBeginEnvironment{example}{\crefalias{theorem}{example}}
    \AtBeginEnvironment{fact}{\crefalias{theorem}{fact}}
    \AtBeginEnvironment{lemma}{\crefalias{theorem}{lemma}}
    \AtBeginEnvironment{openproblem}{\crefalias{theorem}{openproblem}}
    \AtBeginEnvironment{proposition}{\crefalias{theorem}{proposition}}
    \def\PaperAliasTheoremCounter#1{\AtBeginEnvironment{#1}{\crefalias{theorem}{#1}}}
    \forcsvlist{\PaperAliasTheoremCounter}{assertion,assumption,axiom,claim,conclusion,conjecture,criterion,folklore,hypothesis,informaltheorem,metatheorem,observation,postulate,property,application,construction,convention,experiment,notation,problem,protocol,result,solution,step,commentary,discussion,exercise,motivation,notationabuse,note,question,remark,syntax,guideline,recall,summary,tip,warning}

    \def\PaperSmallCapsCrefNameParse#1/#2/#3\PaperSmallCapsCrefNameEnd{%
        \crefname{#1}{\textsc{#2}}{\textsc{#3}}%
        \Crefname{#1}{\textsc{#2}}{\textsc{#3}}%
    }
    \def\PaperSmallCapsCrefName#1{%
        \PaperSmallCapsCrefNameParse#1\PaperSmallCapsCrefNameEnd
    }
    \forcsvlist{\PaperSmallCapsCrefName}{%
        application/Application/Applications,%
        assertion/Assertion/Assertions,%
        assumption/Assumption/Assumptions,%
        axiom/Axiom/Axioms,%
        claim/Claim/Claims,%
        commentary/Commentary/Commentaries,%
        conclusion/Conclusion/Conclusions,%
        conjecture/Conjecture/Conjectures,%
        construction/Construction/Constructions,%
        convention/Convention/Conventions,%
        corollary/Corollary/Corollaries,%
        criterion/Criterion/Criteria,%
        definition/Definition/Definitions,%
        discussion/Discussion/Discussions,%
        example/Example/Examples,%
        exercise/Exercise/Exercises,%
        experiment/Experiment/Experiments,%
        fact/Fact/Facts,%
        folklore/Folklore/Folklore,%
        guideline/Guideline/Guidelines,%
        hypothesis/Hypothesis/Hypotheses,%
        informaltheorem/Informal Theorem/Informal Theorems,%
        lemma/Lemma/Lemmas,%
        metatheorem/Meta Theorem/Meta Theorems,%
        motivation/Motivation/Motivations,%
        notation/Notation/Notations,%
        notationabuse/Notation Abuse/Notation Abuses,%
        note/Note/Notes,%
        observation/Observation/Observations,%
        openproblem/Open Problem/Open Problems,%
        postulate/Postulate/Postulates,%
        problem/Problem/Problems,%
        proof/Proof/Proofs,%
        property/Property/Properties,%
        proposition/Proposition/Propositions,%
        protocol/Protocol/Protocols,%
        question/Question/Questions,%
        recall/Recall/Recalls,%
        remark/Remark/Remarks,%
        result/Result/Results,%
        solution/Solution/Solutions,%
        step/Step/Steps,%
        summary/Summary/Summaries,%
        syntax/Syntax/Syntaxes,%
        theorem/Theorem/Theorems,%
        tip/Tip/Tips,%
        warning/Warning/Warnings,%
        figure/Figure/Figures,%
        table/Table/Tables,%
        algorithm/Algorithm/Algorithms,%
        appendix/Appendix/Appendices,%
        equation/Equation/Equations,%
        enumi/Item/Items%
    }

    \AtBeginDocument{%
        \crefformat{section}{#2\textsection#1#3}%
        \Crefformat{section}{#2\textsection#1#3}%
        \crefformat{subsection}{#2\textsection#1#3}%
        \Crefformat{subsection}{#2\textsection#1#3}%
    }
\fi

\ifLLNCSVersion\else
    \ProvideDocumentCommand{\zcref}{O{}m}{\cref{#2}}
\fi

\NewDocumentEnvironment{parts}{}{%
    \begin{enumerate}[label=(\alph*),ref=(\alph*)]
}{%
    \end{enumerate}
}

\ProvideDocumentCommand{\SetQEDSymbol}{m}{\renewcommand{\qedsymbol}{#1}}

\ifdefined\thmendsymb
    \RenewDocumentCommand{\thmendsymb}{}{%
        \ensuremath{\color{SteelBlue4}\scriptstyle\lrcorner}%
    }
\fi

\ifLuaTeX
    \ifdefined
        \AtBeginDocument{%
            \SetQEDSymbol{\textcolor{SteelBlue4}{\ensuremath{\scriptstyle\mdlgblksquare}}}%
        }
    \else
        \AtBeginDocument{%
            \SetQEDSymbol{\textcolor{SteelBlue4}{\ensuremath{\scriptstyle\mdlgblksquare}}}%
        }
    \fi
\else
    \AtBeginDocument{%
        \SetQEDSymbol{\textcolor{SteelBlue4}{\ensuremath{\scriptstyle\blacksquare}}}%
    }
\fi

\ProvideDocumentCommand{\parhead}{m}{%
    \par\smallskip\noindent\textbf{#1.}\enspace
}
\ProvideDocumentCommand{\subparhead}{m}{%
    \par\smallskip\noindent\textit{#1.}\enspace
}

\NewDocumentCommand{\clauselabel}{m}{\textbf{#1}\enspace}

\NewDocumentCommand{\Email}{m}{\texttt{\href{mailto:#1}{#1}}}

\NewDocumentEnvironment{aidisclosure}{}{%
    \addcontentsline{toc}{section}{AI Disclosure and Provenance}
    \parhead{AI Disclosure and Provenance}%
    \ignorespaces
}{}

\RenewDocumentEnvironment{acknowledgements}{}{%
    \addcontentsline{toc}{section}{Acknowledgements}
    \parhead{Acknowledgements}%
    \ignorespaces
}{}

\RenewDocumentEnvironment{abstract}{}
{\addcontentsline{toc}{section}{Abstract}%
    \begin{quotation}
    \small\parhead{Abstract.}\ignorespaces}
{\end{quotation}}

\NewDocumentCommand{\NewAuthorComment}{m m m}{%
    \ifDraftNotes
        \expandafter\newcommand\csname #1\endcsname[1]{%
            \textcolor{#3}{\textbf{#2:} ##1}%
        }%
    \else
        \expandafter\newcommand\csname #1\endcsname[1]{}%
    \fi
}

\ProvideDocumentCommand{\newguymarker}{m m m}{\NewAuthorComment{#1}{#2}{#3}}

\ProvideDocumentCommand{\paperheadtitle}{}{\papertitle}

\ifFullVersion

    \let\paperoriginalepigraph\epigraph
    \RenewDocumentCommand{\epigraph}{mm}{%
        \paperoriginalepigraph{\textsf{#1}}{\textsc{#2}}%
    }

    \fancypagestyle{plain}{%
        \fancyhf{}%
        \fancyfoot[C]{\normalfont\scriptsize\thepage}%
    }
\fi

\ifFullVersion
    \titleformat{name=\section,numberless}
    {\normalfont\normalsize\scshape\filcenter}{}{0pt}{}
\fi

\NewDocumentCommand{\PrintPaperContents}{}{%
    \clearpage
    {\setstretch{1}\setcounter{tocdepth}{2}\tableofcontents}%
    \clearpage
}

\NewDocumentCommand{\PrintPaperBibliography}{}{%
    \begingroup
    \clearpage
    \setstretch{1}
    \ifLLNCSVersion
        \bibliographystyle{splncs04}%
    \else
        \bibliographystyle{alphaurl}%
    \fi
    \bibliography{bib/abbrev0,bib/crypto,bib/cstheory,bib/master}%
    \clearpage
    \endgroup
}

\NewDocumentCommand{\BeginPaperAppendices}{}{%
    \clearpage
    \appendix
    \renewcommand*{\theHsection}{appendix.\Alph{section}}%
    \ifLLNCSVersion\else
        \appendixpage
    \fi
}

\NewDocumentCommand{\PrintDeferredProofAppendix}{}{%
    \IfDeferredProofsActive{%
        \section{Deferred Proofs}\label{app:deferred-proofs}%
        \printDeferredProofs
    }{}%
}

\makeatletter
\let\paper@flsetnum\@flsetnum
\def\@flsetnum#1{\ifx#1\@colnum\else\paper@flsetnum#1\fi}
\makeatother

\NewDocumentEnvironment{publicationtable}{O{\normalsize}}{%
    \begingroup
    #1%
    \setlength{\tabcolsep}{6pt}%
}{%
    \endgroup
}

\def\tablealtrow{\rowcolor{SteelBlue4!5}}
\NewDocumentCommand{\tablegroup}{m}{\textsc{#1}}
\NewDocumentCommand{\tablehead}{m}{\textsc{#1}}

\NewDocumentEnvironment{proofbox}{O{} m m}{%
    \begin{figure}[!htbp]
    \begin{mdframed}[
        linecolor         = black,
        backgroundcolor   = white,
        roundcorner       = 2pt,
        leftmargin        = 0pt,
        rightmargin       = 0pt,
        innertopmargin    = 10pt,
        innerbottommargin = 10pt,
        innerleftmargin   = 10pt,
        innerrightmargin  = 10pt,
        #1
    ]
}{%
    \end{mdframed}
    \caption{#2}\label{#3}%
    \end{figure}
}

\NewDocumentEnvironment{openproblemframe}{O{}}{%
    \begin{tcolorbox}[
        breakable,
        sharp corners,
        colframe = SteelBlue4,
        colback  = white,
        boxrule  = 0.5pt,
        #1
    ]
}{%
    \end{tcolorbox}
}

\usetikzlibrary{
    arrows.meta,
    backgrounds,
    calc,
    decorations.pathreplacing,
    fit,
    positioning,
    shapes.geometric
}

\Crefname{clcodei}{Line}{Lines}

\AtBeginEnvironment{codelines}{\footnotesize}

\newlength{\ProofBoxSidePadding}
\newlength{\ProofBoxRuleWidth}
\newlength{\CodeFigureWidth}
\newlength{\CodeLineIndentWidth}
\NewDocumentCommand{\SetCodeFigureWidth}{m}{%
    \settowidth{\CodeFigureWidth}{#1}%
    \addtolength{\CodeFigureWidth}{\codelineslinenrwidth}%
    \addtolength{\CodeFigureWidth}{\codelineslinenrsep}%
    \addtolength{\CodeFigureWidth}{2\ProofBoxSidePadding}%
    \addtolength{\CodeFigureWidth}{2\ProofBoxRuleWidth}%
    \ifdim\CodeFigureWidth>\linewidth
        \setlength{\CodeFigureWidth}{\linewidth}%
    \fi
}

\NewDocumentCommand{\SetIndentedCodeFigureWidth}{O{1}m}{%
    \SetCodeFigureWidth{\hspace*{\dimexpr#1\CodeLineIndentWidth\relax}#2}%
}

\ProvideDocumentCommand{\mathsc}{m}{\ifmmode\mbox{\normalfont\scshape #1}\else{\normalfont\scshape #1}\fi}

\NewDocumentCommand{\KK}{}{\mathbb{K}}

\RenewDocumentCommand{\xor}{}{\mathbin{\oplus}}
\NewDocumentCommand{\hybchange}{m}{\textcolor{IndianRed1}{#1}}

\NewDocumentCommand{\mbar}{}{\;\middle|\;}

\ifLuaTeX\else
    \DeclareSymbolFont{PaperSymbols}{OMS}{cmsy}{m}{n}
    \let\emptyset\relax
    \DeclareMathSymbol{\emptyset}{\mathord}{PaperSymbols}{"3B}
    \let\emptyset\varnothing
\fi

\ifLuaTeX
    \ifdefined
        \RenewDocumentCommand{\mathbf}{m}{\symbf{#1}}

    \else
        
    \fi
\else
    
\fi

\usepackage{placeins}

\hypersetup{%
    pdfkeywords = {extremely lossy functions, trapdoor hashing, small-exponent discrete logarithm, Damgard-Jurik, generic group model}
}
\ifAnonymous\else
    \hypersetup{pdfauthor = {Damiano Abram, Agni Datta, Archisman Dutta, Lawrence Roy}}
\fi

\ifDraftNotes
    \ifdefined\PrelimText
        \NewCommandCopy{\PaperPrelimText}{\PrelimText}
        \RenewDocumentCommand{\PrelimText}{}{\smash{\PaperPrelimText}}
    \fi
\fi

\ifConferenceVersion
    \def\distpairsep{\\}
\else
    \def\distpairsep{&}
\fi

\ifConferenceVersion
    \def\telfhybwidth{0.49\linewidth}
    \def\telfhybbreak{\par\medskip}
\else
    \def\telfhybwidth{0.32\linewidth}
    \def\telfhybbreak{\hfill}
\fi

\NewDocumentCommand{\prelimparhead}{mm}{\phantomsection\label{#2}\parhead{#1}}

\NewDocumentCommand{\paperpart}{mm}{%
    \ifConferenceVersion
        \prelimparhead{#1}{#2}%
    \else
        \subsection{#1}\label{#2}%
    \fi
}

\RenewDocumentCommand{\codelinesrightalign}{m}{\hfill{\scriptsize\textcolor{gray}{\textit{$\ast$\,#1\,$\ast$}}}}

\AddToHook{env/proofbox/begin}{\codelinesresetlinenr}

\colorlet{tableaccent}{SteelBlue4}
\colorlet{tableheadbackground}{SteelBlue4!14}
\colorlet{tablegroupbackground}{SteelBlue4!9}
\colorlet{tablealtbackground}{SteelBlue4!5}

\ifConferenceVersion\else
\fi

\counterwithin{figure}{section}
\counterwithin{table}{section}

\NewDocumentCommand{\ConferenceProofSketch}{m +m m m}{%
    \ifConferenceVersion
        \par\addvspace{\medskipamount}%
        \noindent{\itshape Proof Sketch.}\enspace\ignorespaces#2\qed
        \par\addvspace{\medskipamount}%
    \fi
}

\ExplSyntaxOn
\msg_new:nnn { paper } { unknown-relocated }
    { No~relocated~material~is~stored~under~'#1';~rerun~LaTeX. }
\tl_new:N \l__paper_relocated_tl
\cs_new_protected:Npn \__paper_relocated_store:nn #1#2
    {
        \tl_set:Nn \l__paper_relocated_tl {#2}
        \regex_replace_once:nnN { (?: \s | \c{par} )* \Z } { } \l__paper_relocated_tl
        \tl_trim_spaces:N \l__paper_relocated_tl
        \tl_put_right:Nn \l__paper_relocated_tl { ~ }
        \tl_gclear_new:c { g__paper_relocated_#1_tl }
        \tl_gset_eq:cN { g__paper_relocated_#1_tl } \l__paper_relocated_tl
    }
\NewDocumentCommand{\RelocatedAuxBody}{m +m}{ \__paper_relocated_store:nn {#1} {#2} }
\NewDocumentEnvironment{relocated}{m +b}
    {
        \legacy_if:nTF { ConferenceVersion }
            {
                \__paper_relocated_store:nn {#1} {#2}
                \exp_args:Nc \iow_now:Nx { @auxout }
                    { \token_to_str:N \RelocatedAuxBody {#1} { \exp_not:n {#2} } }
            }
            { \tl_trim_spaces_apply:nN {#2} \use:n }
    }
    {}
\NewDocumentCommand{\PlaceRelocated}{m}
    {
        \tl_if_exist:cTF { g__paper_relocated_#1_tl }
            { \tl_use:c { g__paper_relocated_#1_tl } }
            { \msg_warning:nnn { paper } { unknown-relocated } {#1} }
    }
\ExplSyntaxOff

\newguymarker{damiano}{Damiano}{Maroon3}
\newguymarker{agni}{Agni}{Blue4}
\newguymarker{arch}{Archisman}{DarkOrange3}
\newguymarker{lance}{Lawrence}{DarkOliveGreen4}

\RenewDocumentCommand{\advA}{}{\mathcal{A}}
\RenewDocumentCommand{\advB}{}{\mathcal{B}}
\RenewDocumentCommand{\advC}{}{\mathcal{C}}
\RenewDocumentCommand{\advD}{}{\mathcal{D}}
\RenewDocumentCommand{\advV}{}{\mathcal{V}}
\RenewDocumentCommand{\Distribution}{}{\mathfrak{D}}

\NewDocumentCommand{\bndexp}{}{\mathfrak{B}}
\NewDocumentCommand{\uniformexp}{}{\mathfrak{U}}

\DeclareMathOperator*{\indist}{\equiv}

\NewDocumentCommand{\DSEDL}{}{\schemestyle{DSEDL}}
\NewDocumentCommand{\vDSEDL}{}{\schemestyle{vDSEDL}}
\NewDocumentCommand{\SEDL}{}{\schemestyle{SEDL}}
\NewDocumentCommand{\ggdsedl}{}{\mathsf{ggm\text{-}DSEDL}}
\NewDocumentCommand{\aiggdsedl}{}{\schemestyle{ai\text{-}ggm\text{-}DSEDL}}
\NewDocumentCommand{\AIGG}{}{\schemestyle{AI\text{-}GG}}
\NewDocumentCommand{\BFGG}{}{\schemestyle{BF\text{-}GG}}
\NewDocumentCommand{\numforms}{}{\mathfrak{n}}
\ProvideDocumentCommand{\DCR}{}{\schemestyle{DCR}}
\ProvideDocumentCommand{\CCA}{}{\schemestyle{CCA}}

\NewDocumentCommand{\DDSEDL}{}{\Distribution^{\DSEDL}}
\NewDocumentCommand{\DvDSEDL}{}{\Distribution^{\vDSEDL}}
\NewDocumentCommand{\Doff}{}{\Distribution^{\var{off}}}
\NewDocumentCommand{\AdvGG}{}{\Adv^{\ggdsedl}}

\ProvideDocumentCommand{\CGen}{}{\mathsf{CGen}}
\NewDocumentCommand{\SampG}{}{\mathsf{Samp}_{\GG}}
\ProvideDocumentCommand{\pp}{}{\mathsf{pp}}
\ProvideDocumentCommand{\DJGen}{}{\mathsf{DJGen}}
\ProvideDocumentCommand{\SafeGen}{}{\mathsf{SafeGen}}
\ProvideDocumentCommand{\SafePrimes}{}{\mathcal{P}}
\ProvideDocumentCommand{\ord}{}{\mathsf{ord}}
\ProvideDocumentCommand{\DDLog}{}{\mathsf{DDLog}}
\ProvideDocumentCommand{\expl}{}{L_{\mathsf{exp}}}
\ProvideDocumentCommand{\Expo}{}{\mathsf{exp}}
\ProvideDocumentCommand{\Logg}{}{\mathsf{log}}

\NewDocumentCommand{\ZN}{}{\ZZ/N\ZZ}
\NewDocumentCommand{\ZNsqx}{}{(\ZZ/N^2\ZZ)^\times}
\NewDocumentCommand{\easypart}{}{1+N\ZZ/N^2\ZZ}
\NewDocumentCommand{\DamgardJurik}{}{Damg{\aa}rd--Jurik\xspace}
\NewDocumentCommand{\BSGS}{}{baby-step--giant-step\xspace}
\ProvideDocumentCommand{\ELTDH}{}{\schemestyle{ELTDH}}

\NewDocumentCommand{\ELTDHSetupLos}{}{\ELTDH.\funct{Setup}}
\NewDocumentCommand{\ELTDHHash}{}{\ELTDH.\funct{Hash}}
\NewDocumentCommand{\ELTDHGen}{}{\ELTDH.\funct{Gen}}
\NewDocumentCommand{\ELTDHEnc}{}{\ELTDH.\funct{Enc}}
\NewDocumentCommand{\ELTDHDec}{}{\ELTDH.\funct{Dec}}

\NewDocumentCommand{\Rfull}{}{R_{\var{full}}}
\NewDocumentCommand{\ffull}{}{f_{\var{full}}}
\NewDocumentCommand{\Efull}{}{\Event_{\var{full}}}
\NewDocumentCommand{\hkbot}{}{\hk_\bot}
\NewDocumentCommand{\Lina}{}{L_{\bvec a}}
\NewDocumentCommand{\Qhash}{}{Q_{\var{hash}}}
\NewDocumentCommand{\rcap}{}{\overline r}
\NewDocumentCommand{\binw}{}{\binencode_w}

\NewDocumentCommand{\TELF}{}{\schemestyle{TELF}}
\NewDocumentCommand{\TELFEval}{}{\TELF.\funct{Eval}}
\NewDocumentCommand{\TELFGenInj}{}{\TELF.\funct{GenInj}}
\NewDocumentCommand{\TELFGenLos}{}{\TELF.\funct{GenLos}}
\NewDocumentCommand{\TELFInv}{}{\TELF.\funct{Inv}}

\NewDocumentCommand{\ABOTELF}{}{\schemestyle{ABO\text{-}TELF}}
\NewDocumentCommand{\ABO}{}{\schemestyle{ABO}}
\NewDocumentCommand{\ABOGenInj}{}{\ABOTELF.\funct{GenInj}}
\NewDocumentCommand{\ABOGenLos}{}{\ABOTELF.\funct{GenLos}}
\NewDocumentCommand{\ABOEval}{}{\ABOTELF.\funct{Eval}}
\NewDocumentCommand{\ABOInv}{}{\ABOTELF.\funct{Inv}}

\NewDocumentCommand{\ABSTELF}{}{\schemestyle{ABS\text{-}TELF}}
\NewDocumentCommand{\ABS}{}{\schemestyle{ABS}}
\NewDocumentCommand{\ABSGenInj}{}{\ABSTELF.\funct{GenInj}}
\NewDocumentCommand{\ABSGenLos}{}{\ABSTELF.\funct{GenLos}}
\NewDocumentCommand{\ABSEval}{}{\ABSTELF.\funct{Eval}}
\NewDocumentCommand{\ABSInv}{}{\ABSTELF.\funct{Inv}}

\NewDocumentCommand{\Yfk}{}{\calY_\fk}

\NewDocumentCommand{\THF}{}{\schemestyle{THF}}

\NewDocumentCommand{\ELFGen}{}{\ELF.\funct{Gen}}
\NewDocumentCommand{\FuncClass}{}{\mathcal{C}}

\RenewDocumentCommand{\ELFEval}{}{\ELF.\funct{Eval}}

\RenewDocumentCommand{\fk}{}{\mathsf{fk}}
\RenewDocumentCommand{\ek}{}{\mathsf{ek}}
\RenewDocumentCommand{\hk}{}{\mathsf{hk}}
\RenewDocumentCommand{\td}{}{\mathsf{td}}

\NewDocumentCommand{\mat}{m}{\mathbf{#1}}

\NewDocumentCommand{\bvec}{m}{\mathbf{#1}}
\ProvideDocumentCommand{\symbf}{m}{\boldsymbol{#1}}
\NewDocumentCommand{\grplen}{}{\ell_{\mathsf{grp}}}
\NewDocumentCommand{\groupsec}{}{\lambda_{\mathsf{grp}}}

\NewDocumentCommand{\ind}{m}{\mathbb{1}\!\left[#1\right]}
\NewDocumentCommand{\binencode}{}{\var{bin}}
\NewDocumentCommand{\IdMap}{}{\var{id}}
\NewDocumentCommand{\transpose}{}{\mathsf{T}}

\NewDocumentCommand{\clog}{m}{\ceil{\log_2#1}}

\DeclareMathOperator{\Supp}{\mathsf{supp}}
\DeclareMathOperator{\rep}{\mathsf{rep}}
\DeclareMathOperator{\Enumerate}{\mathsf{Enumerate}}
\DeclareMathOperator{\Deduplicate}{\mathsf{Deduplicate}}

\NewDocumentCommand{\GN}{}{\GG_N}
\NewDocumentCommand{\piq}{}{\pi_q}
\NewDocumentCommand{\nuG}{}{\nu_{\GG}}
\NewDocumentCommand{\Nbot}{}{N_\bot}
\NewDocumentCommand{\NTelf}{}{N_{\var{T}}}
\NewDocumentCommand{\dlam}{}{d_\secparam}
\NewDocumentCommand{\nlam}{}{n_\secparam}
\NewDocumentCommand{\Rlam}{}{R_\secparam}
\NewDocumentCommand{\wlam}{}{w_\secparam}
\NewDocumentCommand{\Tai}{}{T_{\var{ai}}}
\NewDocumentCommand{\ceval}{}{c_{\var{eval}}}

\NewDocumentCommand{\Qbar}{}{\overline Q}
\NewDocumentCommand{\Qvec}{}{Q_{\var{vec}}}
\NewDocumentCommand{\Qplain}{}{Q_{\var{plain}}}
\NewDocumentCommand{\dhat}{}{\widehat\delta}
\NewDocumentCommand{\phat}{}{\widehat p}
\NewDocumentCommand{\that}{}{\widehat t}
\NewDocumentCommand{\deltastar}{}{\delta^\star}
\NewDocumentCommand{\pstar}{}{p^\star}

\NewDocumentCommand{\Sstar}{}{\calS^\star}
\NewDocumentCommand{\Bstar}{}{B^\star}
\NewDocumentCommand{\bstar}{}{b^\star}
\NewDocumentCommand{\betastar}{}{\beta^\star}
\NewDocumentCommand{\sstar}{}{s^\star}
\NewDocumentCommand{\jstar}{}{j^\star}
\NewDocumentCommand{\Ptau}{}{P_\tau}

\NewDocumentCommand{\Phihk}{}{\Phi_{\hk}}
\NewDocumentCommand{\ffk}{}{f_\fk}
\NewDocumentCommand{\fbot}{}{f_\bot}
\NewDocumentCommand{\fstar}{}{f_\star}
\NewDocumentCommand{\finj}{}{f_{\var{inj}}}
\NewDocumentCommand{\fDJ}{}{f_{\var{DJ}}}
\NewDocumentCommand{\Flin}{}{\FuncClass^{\var{lin}}}
\NewDocumentCommand{\Oxi}{}{\calO_\xi}
\NewDocumentCommand{\Oinv}{}{\calO_{\var{inv}}}
\NewDocumentCommand{\Hfp}{}{\Hybrid^{\var{fp}}}
\NewDocumentCommand{\DSEDLgamma}{}{\DSEDL_\gamma}
\NewDocumentCommand{\SEDLgamma}{}{\SEDL_\gamma}

\NewDocumentCommand{\Ein}{}{\Event_{\var{in}}}
\NewDocumentCommand{\Eout}{}{\Event_{\var{out}}}
\NewDocumentCommand{\Ecoll}{}{\Event_{\var{coll}}}
\NewDocumentCommand{\Edistinct}{}{\Event_{\var{distinct}}}
\NewDocumentCommand{\Einj}{}{\Event_{\var{inj}}}
\NewDocumentCommand{\Elos}{}{\Event_{\var{los}}}
\NewDocumentCommand{\nulos}{}{\nu_{\var{los}}}
\NewDocumentCommand{\epsinj}{}{\varepsilon_{\var{inj}}}
\NewDocumentCommand{\epslos}{}{\varepsilon_{\var{los}}}
\NewDocumentCommand{\epsunused}{}{\varepsilon_{\var{unused}}}
\NewDocumentCommand{\deltapar}{}{\delta_{\var{par}}}
\NewDocumentCommand{\pbot}{}{p^\bot}

\NewDocumentCommand{\DDCR}{}{\Distribution^{\DCR}}
\NewDocumentCommand{\Dinput}{}{\Distribution^{\var{input}}}

\NewDocumentCommand{\Dmode}{}{\Distribution^{\var{mode}}}

\NewDocumentCommand{\kLin}{}{\schemestyle{k}-\schemestyle{Lin}}
\NewDocumentCommand{\ZNsq}{}{\ZZ/N^2\ZZ}

\NewDocumentCommand{\Nout}{}{N_{\var{out}}}

\NewDocumentCommand{\shat}{}{\hat{s}}
\NewDocumentCommand{\qhat}{}{\hat{q}}
\NewDocumentCommand{\dstar}{}{d^\star}

\NewDocumentCommand{\Deltahk}{}{\Delta_{\hk,\ek,\td}}
\NewDocumentCommand{\Fhk}{}{F_{\hk,\ek}}
\NewDocumentCommand{\Jtdc}{}{J_{\td,c}}
\NewDocumentCommand{\explbot}{}{{\expl}_{\bot}}
\NewDocumentCommand{\Dinpriv}{}{\Distribution^{\var{in}}}
\NewDocumentCommand{\Dfnpriv}{}{\Distribution^{\var{fn}}}
\NewDocumentCommand{\Dref}{}{\Distribution^{\var{ref}}}
\NewDocumentCommand{\pcond}{}{\widetilde p}
\NewDocumentCommand{\replen}{m}{\algostyle{len}\left(#1\right)}

\NewDocumentCommand{\Qtelf}{}{Q_{\TELF}}
\NewDocumentCommand{\ELFGenAI}{}{\ELFGen^{\star}}
\NewDocumentCommand{\JtdZero}{}{J_{\td}}

\NewDocumentCommand{\nuinj}{}{\nu_{\var{inj}}}
\NewDocumentCommand{\Qhashbar}{}{\overline{\Qhash}}
\NewDocumentCommand{\ABSInvBundle}{m}{\ABSInv_{#1}}

\ifcsname propositionE\endcsname\else
    \ifLLNCSVersion
        \NewDeferredThm{propositionE}{proposition}
    \else
        \NewDeferredThmAlias{propositionE}{proposition}
    \fi
\fi

\ifConferenceVersion
	\title{\papertitle}
	\ifAnonymous
		\titlerunning{Simple Extremely Lossy Functions from Small-Exponent Hashing}
		\authorrunning{Anonymous Authors}
		\author{}
		\institute{}
	\else
		\titlerunning{Simple Extremely Lossy Functions from Small-Exponent Hashing}
		\authorrunning{D. Abram, A. Datta, A. Dutta, and L. Roy}
		\author{Damiano Abram\inst{1} \and
			Agni Datta\inst{2} \and
			Archisman Dutta\inst{3} \and
			Lawrence Roy\inst{4}}
		\institute{The University of Edinburgh, UK\\
			\Email{damianoabram@gmail.com}
			\and
			The University of Edinburgh, UK\\
			\Email{agnidatta.org@gmail.com}
			\and
			Aarhus University, Denmark\\
			\Email{archisman@cs.au.dk}
			\and
			IBM Zurich Research, Switzerland\\
			\Email{ldr709@gmail.com}}
	\fi
\else
	\title{\papertitle}
	\ifAnonymous
		\author{Anonymous Authors}
	\else
		\author[1]{Damiano Abram}
		\author[1]{Agni Datta}
		\author[2]{Archisman Dutta}
		\author[3]{Lawrence Roy}
		\affil[1]{The University of Edinburgh, UK\par
			\Email{abram.damiano@protonmail.com}\par
			\Email{agnidatta.org@gmail.com}}
		\affil[2]{Aarhus University, Denmark\par
			\Email{archismandutta@proton.me}}
		\affil[3]{IBM Zurich Research, Switzerland\par
			\Email{ldr709@gmail.com}}
	\fi
	\date{\today}
\fi

\begin{document}

\maketitle

\begin{abstract}
Extremely Lossy Functions (\ELF{}s) are a standard model primitive that captures many useful properties of random oracles (Zhandry, Crypto 2016). While there are many variations of \ELF{}s with additional properties, every construction (excluding obfuscation) has followed essentially the same template of bootstrapping from a sequence of \ELF{}s secure only against fixed-size adversaries, and every construction was based on only the exponential hardness of $\DDH$ (or \kLin{}), an assumption that is only reasonable over elliptic curves.

We introduce and construct Extremely Lossy Trapdoor Hashing (\ELTDH{}), a stronger notion that implies all known variations of \ELF{}s. Our construction achieves \ELTDH{} in one go, without bootstrapping from schemes secure for only fixed-size adversaries, which makes it simpler and more efficient than existing \ELF{}s. We assume exponential security of the small-exponent discrete logarithm, together with polynomial security of decisional composite residuosity (\DCR{}). Exponential security is only required in the size of the secret exponent, not the size of the group, so the assumption plausibly holds for multiplication modulo $N^2$ (and for many other cryptographic groups), despite the subexponential-time discrete logarithm attacks from index calculus. Our results diversify the assumptions underlying \ELF{}s, while also giving a simpler construction.
\end{abstract}

\ifConferenceVersion\else
	{\PrintPaperContents}
\fi

\section{Introduction}\label{sec:introduction}

Cryptographic hash functions are used for a wide variety of different purposes, many of which do not follow from standard properties of collision resistance or pseudorandomness. Security is often analyzed in the random oracle model, a very powerful model in which the hash function behaves as a statistically uniform function, which the adversary can only ever see a negligible piece of through its queries. However, real hash functions are described by small circuits, and so are not random oracles, making any result proven in the random oracle model potentially vacuous when instantiated with real cryptographic hash functions. To bridge this gap, many strong security properties for standard model hash functions have been defined, including correlation intractability~\cite{STOC:CanGolHal98}, correlation robustness~\cite{C:IKNP03}, correlated input security~\cite{TCC:GoyONeRao11}, universal computational extraction~\cite{C:BelHoaKee13}, and extreme lossiness~\cite{C:Zhandry16}.

This last property captures how a random oracle can only be queried on a polynomial number of positions, allowing a reduction to efficiently iterate over the effective range of the oracle. An extremely lossy function (\ELF{}) is a hash function that can be set up in two indistinguishable modes: an injective mode, and a lossy mode where the range of the \ELF{} becomes polynomial in size. Applications of \ELF{}s include auxiliary input point obfuscation (a model of password hashing), leakage resilience, and output intractable hash functions~\cite{C:Zhandry16}. Several variants of the \ELF{} definition have also been proposed, with applications to deterministic encryption~\cite{EC:Zhandry19a}, correlated-input password hashing~\cite{EC:Zhandry19a}, and anamorphic-resistant encryption~\cite{EC:ABGMM26}.

We address four limitations in the \ELF{} literature.
\begin{enumerate}[itemsep=0.5ex,topsep=0.5ex]
\item There is only one construction (excluding obfuscation) of \ELF{}s and its variants, based on the exponential hardness of $\DDH$ (or \kLin{}).
\item Due to subexponential-time index calculus attacks on most cryptographic groups, this construction can only be plausibly instantiated using elliptic curve groups.
\item This construction builds a sequence of \ELF{}s targeting different fixed adversary runtimes, and then uses a combiner to build an \ELF{} for all polynomial adversaries. This structure is complicated, inefficient, and makes it harder to expose additional structure or homomorphism in the \ELF{}.
\item Each variant of \ELF{}s requires redoing this entire subexponential $\DDH$ and combiner pipeline, creating a mess of many similar constructions.
\end{enumerate}

In this work, we define and construct Extremely Lossy Trapdoor Hashing (\ELTDH{}), a primitive that combines \ELF{}s with the Trapdoor Hashing (\TDH{}) of \cite{C:DGIMMO19}. We show how this notion implies \ELF{}s as well as all known variants of \ELF{}s; given the wide application of \TDH{}s, we expect that new applications of \ELTDH{}s can be found as well.

Our construction follows a significantly different structure from \cite{C:Zhandry16}. We assume exponential security of small-exponent discrete logarithm, together with polynomial security of decisional composite residuosity (\DCR{}), diversifying the set of assumptions from which \ELF{}s can be built. For small-exponent discrete logarithm, we only need security to be exponential in the size of the secret exponent, not the size of the group, so the assumption plausibly holds for any cryptographic group, even ones where index calculus gives a subexponential-time attack on discrete logarithm. Making the group large enough foils index calculus attacks, and in this case there are no known algorithms for small-exponent discrete logarithm that do better than the generic algorithm of \BSGS, which takes exponential time. The security in this assumption also can easily be tuned without changing out the group, simply by changing the size of the small exponent; this tunability lets us avoid using a combiner.

\subsection{Our Results: Extremely Lossy Trapdoor Hashing}
Extremely Lossy Trapdoor Hashing can be viewed as standard trapdoor hashing with optimal compression. We recall the interaction pattern of trapdoor hashing in \cref{fig:eltdh-interaction}. At the beginning of a trapdoor hashing execution, the evaluator holds no information about the output $f(x)$; however, if we reveal the hash $h$ and the hasher output share $e$, the value of $f(x)$ can be derived with (almost) perfect accuracy. This means that, even given $\ek$, the pair $(h, e)$ carries at least as much information as the output $f(x)$.

\begin{figure}[!hbt]
\centering

\begin{tikzpicture}[
    x=1cm,
    y=1cm,
    font=\small,
    >=Latex,
    msgLR/.style={
        -{Latex[length=1.7mm,width=1.05mm]},
        line width=.5pt,
        draw=SlateBlue4
    },
    msgRL/.style={
        -{Latex[length=1.7mm,width=1.05mm]},
        line width=.5pt,
        draw=HotPink4
    },
    hair/.style={
        line width=.4pt
    },
    note/.style={
        font=\scriptsize,
        inner sep=1pt,
        fill=white
    },
    role/.style={
        font=\small,
        align=center
    },
    local/.style={
        font=\small,
        align=left
    },
    sub/.style={
        font=\scriptsize,
        align=center
    }
]

\draw[hair] (-5.45,0.95) rectangle (6,-6);

\node at (0,0.62)
    {\textcolor{Chartreuse4}{$\hk$}};

\draw[hair] (-5.00,0.30) -- (5.55,0.30);

\begin{scope}[shift={(-3.45,-0.23)},scale=.76,color={SlateBlue4}]
    \draw[hair] (0,0) circle (.18);
    \draw[hair]
        (-.34,-.56)
        arc[start angle=180,end angle=0,x radius=.34,y radius=.30];
\end{scope}

\node[role] at (-3.45,-0.86)
    {\textcolor{SlateBlue4}{\textsc{Hasher}}};

\node[sub,text width=3.8cm] at (-3.45,-1.17)
    {$x,\rho$};

\begin{scope}[shift={(3.45,-0.23)},scale=.76,color={HotPink4}]
    \draw[hair] (0,0) circle (.18);
    \draw[hair]
        (-.34,-.56)
        arc[start angle=180,end angle=0,x radius=.34,y radius=.30];
\end{scope}

\node[role] at (3.45,-0.86)
    {\textcolor{HotPink4}{\textsc{Evaluator}}};

\node[sub,text width=3.8cm] at (3.45,-1.17)
    {$f$};

\node[local,anchor=west] at (-4.85,-1.85)
    {$
        h \gets \Hash(\hk,x;\rho)
    $};

\draw[msgLR]
    (-2.28,-2.85) -- (2.28,-2.85)
    node[midway,above=2pt,fill=white,inner sep=1pt]
        {$h$};

\node[local,anchor=west] at (2,-1.85)
    {$
        (\ek,\td) \sample \Gen(\hk,f)
    $};

\draw[msgRL]
    (2.28,-3.05) -- (-2.28,-3.05)
    node[midway,below=2pt,fill=white,inner sep=1pt]
        {$\ek$};

\node[local,anchor=west] at (-4.85,-4.05)
    {$
        e \gets \Enc(\ek,x;\rho)
    $};

\node[local,anchor=west] at (2.3,-4.05)
    {$
        e' \gets \Dec(\td,h)
    $};

\draw[hair] (-5,-5) -- (5.45,-5);

\node[local,anchor=west] at (-1,-5.50)
    {\textcolor{DodgerBlue4}{$f(x)=e-e'$}};

\end{tikzpicture}

\caption{Interaction in trapdoor hashing. \textcolor{Chartreuse4}{Green} denotes the public hash key $\hk$, \textcolor{SlateBlue4}{purple} the hasher with input $x$ and random tape $\rho$, \textcolor{HotPink4}{pink} the evaluator with predicate $f$, and \textcolor{DodgerBlue4}{blue} the predicate value recovered at the end of the interaction. The hasher computes $h\gets \Hash(\hk,x;\rho)$ and sends $h$. The evaluator runs $(\ek,\td)\sample \Gen(\hk,f)$, sends the public encoding key $\ek$, and keeps the trapdoor $\td$. Reusing the same tape $\rho$, the hasher computes $e\gets \Enc(\ek,x;\rho)$. The evaluator then derives $e'\gets \Dec(\td,h)$, so that $f(x)=e-e'$.}
\label{fig:eltdh-interaction}
\end{figure}

For security reasons, however, the entropy of $(h, e)$ conditioned on $\ek$ must be strictly greater than that of $f(x)$. If that were not the case, the triple $(h, e, \ek)$ would be uniquely determined by $f(x)$. The hasher could exploit this to distinguish between the encoding key for a function $f_0$ and that for a function $f_1$. Specifically, the hasher could pick two inputs $x$ and $\hat{x}$ such that $f_0(x)= f_0(\hat{x})$ but $f_1(x)\ne f_1(\hat{x})$. Then, it can be derived
\begin{align*}
    h&\gets \Hash(\hk, x)\qquad\qquad & e&\gets\Enc(\ek, x)\\
    \hat{h}&\gets \Hash(\hk, \hat{x}) \qquad\qquad& \hat{e}&\gets\Enc(\ek, \hat{x})
\end{align*}
and check whether $(h, e)=(\hat{h}, \hat{e})$. In such a case, the hasher can be confident that the evaluator encoded the function $f_0$.

This attack can be generalized to the setting in which, for every $\ek$, there are at most $r$ pairs $(h, e)$ that encode $f(x)$: a malicious hasher can repeat the attack on independently sampled inputs and obtain a collision with probability at the birthday scale, namely $\Omega(k^2/r)$ after $k$ evaluations, up to the probability of repeated inputs. Thus, when $r$ is polynomially bounded, polynomially many evaluations suffice to distinguish the two cases. To overcome this issue, in standard trapdoor hashing constructions, the value of $r$ grows superpolynomially in the security parameter $\secparam$. This ensures security against all polynomial-time attackers at the same time.

In this work, we take a more fine-grained approach, following the blueprint of \cite{C:Zhandry16}. We allow the setup procedure of trapdoor hashing to take as input the targeted parameter $r$. If we pick a superpolynomial $r$, the construction will be secure against any polynomial-time attacker; however, when we pick a polynomial $r$, the construction will still provide weaker but nevertheless useful security guarantees. Specifically, for all $\PPT$ attackers of bounded complexity, the hash key $\hk$ corresponding to a polynomial choice of $r$ will be almost indistinguishable from a $\hk$ for a superpolynomial $r$. More formally, for every fixed $\PPT$ adversary $\advA$ and every targeted inverse polynomial function $\delta(\secparam)$, there exists a sufficiently large polynomial $r(\secparam)$ such that $\advA$ cannot distinguish between $\hk\sample\Setup(1^\secparam, r)$ and $\hk\sample\Setup(1^\secparam, R)$ with advantage greater than $\delta(\secparam)$, where $R$ is the superpolynomial parameter of the standard scheme.

To understand why this property is useful, consider a cryptographic construction that relies on trapdoor hashing and that has a security loss proportional to the number $r$ of pairs $(h, e)$ that encode the same output $f(x)$. In standard trapdoor hashing, the parameter $r$ is superpolynomial, and therefore the security proof will incur a superpolynomial security loss. Using extremely lossy trapdoor hashing, we can avoid this. In the real world, we instantiate the construction using a superpolynomial $r$, e.g.\ $2^\secparam$, as for standard trapdoor hashing. Now, suppose that there exists a $\PPT$ attacker $\advA$ that breaks the security of this unspecified cryptographic construction with some inverse polynomial advantage $\varepsilon(\secparam)$. By the properties of the extremely lossy trapdoor hashing, there exists a polynomial choice for $r$ that makes the hash key $\hk$ at most $\varepsilon/2$-indistinguishable from the original one. In other words, even if we instantiate the construction with this weaker key $\hk$, the adversary $\advA$ would still break security with advantage at least $\varepsilon-\varepsilon/2=\varepsilon/2$. At this point, we can continue the security proof for the construction as before; however, given that $r$ is polynomial, we will incur only a polynomial security loss.

\subsection{Our Results: Applications}

\paragraph{Extremely Lossy Functions.} It is easy to see that extremely lossy trapdoor hashing implies extremely lossy functions (\ELF{}s)~\cite{C:Zhandry16} following the connection between trapdoor hashing and lossy cryptography \cite{STOC:PeiWat08} highlighted by D\"ottling et al.~\cite{C:DGIMMO19}. The \ELF{} description corresponds to the pair $(\hk, \ek)$, whereas the evaluation on some value $x$ corresponds to the pair $(h, e)$ where $h=\Hash(\hk, x)$ and $e=\Enc(\ek, x)$. When the \ELF{} is in injective mode, we generate $\hk$ at the full image parameter, and we let $\ek$ encode the identity function $\IdMap$. Notice that the trapdoor $\td$ allows us to invert $(h, e)$ to recover $x=e-\Dec(\td, h)$.\footnote{This trapdoor lets us turn it into a trapdoor \ELF{} as well, since it allows efficient inversion.} For the lossy mode, instead, we generate $\hk$ using a polynomial value $r$ and we let $\ek$ encode the zero function $f\equiv 0$.

However, extremely lossy trapdoor hashing is stronger than extremely lossy functions. First of all, due to their potential injectivity, \ELF{}s must have a large output. Moreover, they typically allow us to extract information about the input from the output when they are set in injective mode~\cite{C:Zhandry16}. On the other hand, by introducing small modifications to the above blueprint of \ELF{}s from extremely lossy trapdoor hashing, we can easily build a stronger variant of the primitive that achieves compression and allows the extraction of partial information about the input even when the \ELF{} is in lossy mode. This more closely resembles the properties of a random oracle, which \ELF{}s attempt to emulate~\cite{C:Zhandry16}. We additionally show that the image of the public map is at most the image parameter times the range of the predicate, so a predicate of small range leaves the map lossy while we can still recover its value with the trapdoor.

\paragraph{A Plain \ELF{} and Exact Images.} In fact, we can drop the trapdoor altogether and get something simpler: an ordinary \ELF{} over any cyclic group sampler that satisfies tunable $\DSEDL$. It is published as $n$ bases and evaluated to a single group element. Indeed, composite residuosity does not play any part here. Both \ELF{}s here have exactly enumerable images, as do the selected branches of the branch variants. This means that, from a lossy key and $r$ alone, we can use a prefix dynamic program to list the image in time $\poly(r,\log m)$, which lets an application sample from it uniformly.

\paragraph{Branch Variants.} Applications involving multiple correlated \ELF{} evaluations (such as correlated input secure one-way functions and \CCA-secure deterministic public key encryption~\cite{EC:Zhandry19a}) require more than a single function from a trapdoor \ELF{}. We therefore add branch functionality, where a designated set of branches is lossy, and every other branch stays invertible under an inversion oracle~\cite{EC:Zhandry19a}. The same hash gives this directly. We publish two independent bundles of encoding rows, so that one bundle still answers the inversion oracle while a function-privacy hybrid rewrites the other, as selective security requires. The all-but-one primitive of~\cite{EC:Zhandry19a} is the case of a single selected branch.

\ifConferenceVersion\else
\paragraph{Anamorphic-Resistant Encryption.} Without its trapdoor, our trapdoor \ELF{} is always injective in the injective mode, as the \ELF{} definition of~\cite{EC:ABGMM26} requires, so the anamorphic-resistant encryption of~\cite{EC:ABGMM26} from \ELF{}s and fully unique non-interactive zero-knowledge arguments holds under our assumptions. Its second construction, which avoids indistinguishability obfuscation, uses a partial trapdoor of a variant of the trapdoor \ELF{} of~\cite{EC:Zhandry19a}, and whether our construction admits such a partial trapdoor is open.
\fi

\subsection{Our Construction}\label{sec:construction}

In this work, we present an extremely lossy trapdoor hashing scheme that supports the evaluation of any linear function.\footnote{We expect the same technique to reach functions described by multivariate polynomials whose monomials come from a fixed set $\Omega$ of polynomial size, although we prove nothing about that class here. The recipe would ask one party to hash the evaluation of all the monomials in $\Omega$ on their input $x$, and the other party to provide the encoding key corresponding to the coefficients of the multivariate polynomial describing $f$.} Our construction rests on tunable $\DSEDL$ over a cyclic group $\GG=\langle g\rangle$, a decisional form of the small-exponent discrete logarithm problem. At a fixed interval $[\gamma]$, it asks whether $g^s$, for $s$ sampled from $[\gamma]$, can be distinguished from a uniform group element. A fixed polynomial interval is searchable by \BSGS in $\bigO(\sqrt\gamma)$ group operations, so the assumption rests on the order of the quantifiers, with the tunable form fixing the observer first and only then choosing an interval large enough to defeat that one observer. For every polynomial size bound $t$ and every inverse polynomial $\delta$, there is a polynomial $Q$ such that every adversary $\advA$ of size $t$ separates a uniform element of $\GG$ from $g^s$ with advantage at most $\delta$, for every interval bound $\gamma\ge Q(\secparam)$. Zhandry argues that \ELF{}s call for an exponentially hard assumption too~\cite{C:Zhandry16}.

\paragraph{How Small the Image Can Be.} The tuning is forced on us. Evaluation is public, so anyone can run the function on $k$ random inputs and watch for a repeat, and an image of $r$ points produces one at the birthday scale $k^2/r$. Any \ELF{} that hides its lossy mode from every observer of size $k$ up to advantage $\varepsilon$ therefore needs $r=\Omega(k^2/\varepsilon)$, which is why the image parameter has to be chosen after the observer and never before. We obtain the matching upper bound in the generic group model, where a distinguisher making $T$ queries gains $\bigO(T^2/\gamma)$ in a group whose order has large prime factors, and \BSGS attains it.

\paragraph{Recalling $\DDLog$-based Trapdoor Hashing.} To better understand our construction, it is worth recalling the group-based trapdoor hashing schemes of~\cite{C:DGIMMO19,C:RoySin21}. These build on the distributed discrete logarithm framework of~\cite{C:ADOS22}, which can be instantiated over \DamgardJurik groups or class groups (which have a transparent setup). The framework is defined over a group $\bbH$ that can be decomposed as the direct product of two cyclic groups $\bbD=\langle \hat g\rangle$ and $\GG=\langle g\rangle$, where the uniform distribution over $\bbH$ is computationally indistinguishable from the uniform distribution over $\GG$. The groups also satisfy the property that discrete logarithms are easy in $\bbD$ and hard in $\GG$. Moreover, an efficient procedure called distributed discrete logarithm converts group elements $g_0,g_1$ with $g_1/g_0=\hat g^\sigma$ into integers $s_0=\DDLog(g_0)$ and $s_1=\DDLog(g_1)$ satisfying $s_1-s_0=\sigma \bmod N$, where $N\cqq\ord(\hat g)$.

In the construction, the hash key $\hk$ consists of $d$ random group elements $g_1, \dots, g_d$, where $g_i\sample\GG$. To hash an integer string $\bvec{x}\cqq(x_1, \dots, x_d)\in(\ZN)^d$ we simply compute $h=g_1^{x_1}\cdot \dots \cdot g_d^{x_d}$.

To encode the linear function $f_{\bvec{y}}(\bvec{x})\cqq\langle \bvec{x}, \bvec{y}\rangle$, where $\bvec{y}=(y_1, \dots, y_d)\in(\ZN)^d$, we sample a random $\sigma$ over a sufficiently large interval and we compute \begin{align*}K_1&=\hat g^{y_1}\cdot g_1^\sigma,& &\dots & K_d&=\hat g^{y_d}\cdot g_d^\sigma.\end{align*}
The encoding key $\ek$ consists of the tuple $(K_1, \dots, K_d)$, whereas the trapdoor is $\td=\sigma$. Observe that $\ek$ hides $\bvec{y}$, since a uniform element of $\GG$ is indistinguishable from a uniform element of $\bbH$.

Finally, we recall how the parties can derive a secret-sharing of $f_{\bvec{y}}(\bvec{x})=\langle \bvec{x}, \bvec{y}\rangle$ from $\ek$ and $h$. The hasher computes
\begin{align*}h_1&= K_1^{x_1}\cdot \dots \cdot K_d^{x_d}\\
&= \hat g^{x_1\cdot y_1}\cdot g_1^{x_1\cdot \sigma}\cdot \dots \cdot \hat g^{x_d\cdot y_d}\cdot g_d^{x_d\cdot \sigma}\\
&=\hat g^{\langle \bvec{x}, \bvec{y}\rangle}\cdot h^\sigma.\end{align*}
On the other hand, the evaluator can compute $h_0=h^\sigma$ using the fact that $\td=\sigma$. At this point, the parties hold group elements $h_1$ and $h_0$ such that $h_1/h_0=\hat g^{\langle \bvec{x},\bvec{y}\rangle}$. By locally applying $\DDLog$ to their group elements, the parties derive a secret sharing of the result $\langle \bvec{x},\bvec{y}\rangle$.

\paragraph{Our Extremely Lossy Trapdoor Hashing Scheme.} Our construction is essentially the same as the construction of \cite{C:DGIMMO19,C:RoySin21} with two main modifications. First of all, instead of directly hashing $\bvec{x}$, we hash the bit decomposition $\mat{G}^{-1}(\bvec{x})$, the bits of the coordinates of $\bvec{x}$ laid out in one vector (the $\hk$ needs to become $\clog{N}$ times bigger). On the other hand, instead of encoding $f_{\bvec{y}}$, we ask the evaluator to encode $f_{\bvec{y}\otimes \bvec{g}}$ where $\bvec{g}=(1,2, \dots, 2^{\clog{N}-1})$. Notice that correctness is preserved as \[f_{\bvec{y}\otimes \bvec{g}}(\mat{G}^{-1}(\bvec{x}))=\langle \mat{G}^{-1}(\bvec{x}), \bvec{y}\otimes \bvec{g}\rangle=\langle \bvec{x}, \bvec{y}\rangle.\]

The second adjustment is that, instead of sampling the elements $g_1,\dots,g_d$ in the hash key $\hk$ uniformly at random, we sample $s_i$ over a small set $[\gamma]$ of polynomial size and set $g_i\gets g^{s_i}$. Tunable $\DSEDL$ immediately makes the new hash key indistinguishable from the old one for every adversary running in time well below $\sqrt\gamma$, up to a small inverse polynomial advantage.

Now, fix $\hk$ and $\ek$ and let us count how many pairs $(h, e)$ encode a given output $\langle \bvec{x}, \bvec{y}\rangle$. First of all, observe that for every $h$, there is at most one $e$ that is consistent with $\langle \bvec{x}, \bvec{y}\rangle$ (assuming perfect correctness). Indeed, this must be $\langle \bvec{x}, \bvec{y}\rangle+\Dec(\td, h)$. Therefore, the question boils down to computing the size of the image of $\Hash(\hk, \cdot)$. Notice that every hash $h$ is of the form $g^z$ where $z=\langle \mat{G}^{-1}(\bvec{x}), \bvec{s}\rangle$. Since $\mat{G}^{-1}(\bvec{x})$ is a bit vector of length $d\clog{N}$ and every $s_i$ lies in $[\gamma]$, we conclude that $z\in[0,d\gamma\clog{N}]$. Notice that this is an interval of polynomial size. In other words, we have proven that there are at most polynomially many pairs $(h, e)$ that encode a given output $\langle \bvec{x}, \bvec{y}\rangle$. The rest of the scheme survives the change. The two shares still differ by exactly the predicate value, so the construction is perfectly correct and rate one, and the digest still hides the input statistically.

\paragraph{From the Hash to a Trapdoor \ELF{}.} A single residue is already enough for a whole \ELF{} input. For the domain $[m]$, we send $u$ to the residue $x=u-1$ and publish the pair $(h,e)$, the digest beside the encoder's share. Injective generation puts the identity predicate in the encoding key at a superpolynomial $r$, which makes $e-\Dec(\td,h)=x$, so the trapdoor inverts every output. Turning to the lossy mode, we instead take the zero predicate at a polynomial $r$; the share $e$ is then already fixed by $h$, and the public map inherits the image of the hash, of size at most $r$. The bit decomposition never leaves the hash.

\paragraph{Branch Arithmetic.} We encode a branch $b\in\bits^a$ by the integer $\beta(b)$ and associate a selected set $\Sstar$ with the polynomial
\[
P_{\Sstar}(T)=\prod_{\bstar\in\Sstar}\bigl(\beta(\bstar)-T\bigr)
=(-1)^sT^s+\sum_{i=0}^{s-1}c_iT^i.
\]
Each coefficient $c_i$ becomes one scalar \ELTDH{} predicate $x\mapsto c_i x$. On branch $b$, weighting the decoded rows by $1,\beta(b),\ldots,\beta(b)^{s-1}$ reconstructs $P_{\Sstar}(\beta(b))x$. For selective security, we duplicate these coefficient rows into two bundles: while one bundle changes row by row under function privacy, the other retains the trapdoors needed to answer inversion queries.

\subsection{Comparison with Prior Work}\label{sec:comparison}
\Cref{tab:elf-quantitative-comparison} sets the constructions side by side at domain size $m$, writing $n=\clog{m}$ for the input length, $a$ for the branch length, and $\sstar=\abs{\Sstar}$ for the number of selected branches, with \ABO{} and \ABS{} the all-but-one and all-but-some variants. Group, field, exponent, residue, hash, and group-operation units appear as $\var{G}$, $\var{F}$, $\var{E}$, $\var{R}$, $\var{H}$, and $\var{Op}$. Our rows count elements beside the common group parameters and integer metadata, with $w=2\ell(\groupsec(\cdot))$ the padded bit width of one residue, for the group-security function $\groupsec$ and the prime bit length $\ell$, evaluated at $n$ for the trapdoor \ELF{}, at $n+a$ for \ABO{}, and at $n+a\sstar$ for \ABS{}. One group element of the plain construction takes $\grplen(\secparam)$ bits at $\secparam=\ceil{cn}$, for the representation-length bound $\grplen$ of its group sampler, so its element count is linear while its bit length grows faster.

The trapdoor rows follow from the construction of~\cite{EC:Zhandry19a} rather than from a bound stated there. That construction composes $t=\ceil{\sqrt{\log\log m}}$ layers of the matrix-in-the-exponent function of~\cite{PKC:FGKRS10} at $r_i=2^{2^i}$, and the bounded-adversary output bound $\log N\le C\log M\log r$ leaves the last layer with $n^{3/2+o(1)}$ input bits, so its square matrix dominates the key at $n^{3+o(1)}$ elements. A branch length $a$ multiplies each layer by the $2a+1$ matrices of $A_b=B+\sum_iA_{i,b_i}$.

\begin{table}[!hbt]
	\centering
	\begin{publicationtable}[\small]
		\begin{adjustbox}{max width=\textwidth,center}
			\begin{tabular}{@{}lcccc@{}}
				\toprule
				\tablehead{Construction}                       & \tablehead{Public}                          & \tablehead{Trapdoor}         & \tablehead{Output}           & \tablehead{Evaluation}                      \\
				\midrule
				\multicolumn{5}{@{}l@{}}{\tablegroup{Prior Work}}                                                                                                                                                        \\
				\ELF{}~\cite{C:Zhandry16}                 & \(\Theta(n^2)\var{G}+\Theta(n)\var{H}\) & --                           & \(3n\) bits                  & \(\Theta(n^2)\var{Op}+\Theta(n)\) hashes \\
				\tablealtrow \TELF{}~\cite{EC:Zhandry19a} & \(n^{3+o(1)}\var{G}\)                     & \(n^{3+o(1)}\var{F}\)      & \(n^{3/2+o(1)}\var{G}\)    & \(n^{3+o(1)}\var{Op}\)                   \\
				\ABO{}~\cite{EC:Zhandry19a}               & \(an^{3+o(1)}\var{G}\)                    & \(an^{3+o(1)}\var{F}\)     & \(n^{3/2+o(1)}\var{G}\)    & \(an^{3+o(1)}\var{Op}\)                  \\
				\multicolumn{5}{@{}l@{}}{\tablegroup{This Work}}                                                                                                                                                         \\
				Plain \ELF{}                                      & \(n\var{G}\)                              & --                           & \(1\var{G}\)               & \(\bigO(n)\var{Op}\)                     \\
				\tablealtrow Trapdoor \ELF{}                      & \((2w+2)\var{G}\)                        & \(1\var{E}\)               & \(1\var{G}+1\var{R}\)  & \(\bigO(n)\var{Op}\)                     \\
				Selective \ABO{}                                  & \(3(w+1)\var{G}\)                         & \(2\var{E}\)               & \(1\var{G}+2\var{R}\)  & \(\bigO(n)\var{Op}\)                     \\
				\tablealtrow Selective \ABS{}                     & \((2\sstar+1)(w+1)\var{G}\)              & \(2\sstar\var{E}\)        & \(1\var{G}+2\var{R}\)  & \(\bigO(\sstar n)\var{Op}\)             \\
				\bottomrule
			\end{tabular}
		\end{adjustbox}
	\end{publicationtable}
	\caption{Asymptotic efficiency comparison of extremely lossy functions and their trapdoor variants.}
	\label{tab:elf-quantitative-comparison}
\end{table}

\parhead{Outline} We state our results in \cref{sec:introduction} and give the technical overview in \cref{sec:construction}. In \cref{sec:preliminaries} we fix the syntax and algebraic facts that our constructions use, and in \cref{sec:limits-small-exponent} we take up the image-size barrier, the tunable \DSEDL{} assumption, and the evidence from the generic group model. We construct \ELTDH{} in \cref{sec:hashing}, compile it into trapdoor and plain \ELF{}s in \cref{sec:elf}, and build the branch variants in \cref{sec:branch-telf}, with polynomially many selected branches and \ABO{} as the singleton case. \ifConferenceVersion We give the formal preliminary definitions and the full group-sampling derivations in \cref{app:supplementary-preliminaries}, the collision bounds and the complete generic-group analysis in \cref{app:small-exponent-details}, and additional instantiations and exact enumeration algorithms in \cref{app:instantiations-enumeration}. \IfDeferredProofsActive{We defer every remaining proof to \cref{app:deferred-proofs}. }{}\fi We compare with prior work at length in \cref{sec:related-work}.

\ifConferenceVersion\else
\subsection{Related Work}\label{sec:related-work}
\fi
\begin{relocated}{related-work}
\parhead{Lossy and Extremely Lossy Functions} For lossy trapdoor functions~\cite{STOC:PeiWat08} and their many refinements, namely compact and correlation-secure variants~\cite{ACNS:BoyWat10,PKC:FGKRS10}, slight and identity-based lossiness~\cite{PKC:MolYil10,EC:BKPW12,EC:Wee12}, lossy codes~\cite{EC:DotMul13}, lossy encryption~\cite{AC:HLOV11,AC:HemOst13}, all-but-many branches~\cite{EC:Hofheinz12,C:BoyLi17,C:LSSS17}, adversary-dependent and targeted lossiness~\cite{C:YYHK16,C:QuaWatWic21}, and code-based constructions~\cite{C:DaoJai24}, the size of the lossy image is fixed before the adversary is known, and lossiness does not amplify generically~\cite{TCC:PieRosSeg12}. For extremely lossy functions~\cite{C:Zhandry16}, in contrast, a polynomial image is chosen after the observer, which is exactly what random-oracle-style reductions need, and Zhandry~\cite{EC:Zhandry19a} uses their trapdoor and branch variants for deterministic encryption and correlated-input security. Outside obfuscation~\cite{FOCS:IlaLom25}, every construction is obtained by bootstrapping from a sequence of exponentially $\DDH$-hard instances, while direct algebraic constructions are consistent with the black-box separations of~\cite{TCC:FisRoh23}. We give such a direct construction, in which we replace the sequence of security levels by one bounded subset-sum space and obtain the trapdoor, all-but-one, and all-but-some variants from one hash.

\parhead{Trapdoor Hashing and Succinct Computation} Trapdoor hashing~\cite{C:DGIMMO19} followed work on short openings at hidden positions~\cite{AC:OPWW15,ITCS:HubWic15}, and from its \DCR{} instance, Roy and Singh~\cite{C:RoySin21} obtain rate-one trapdoor hashing and large-message \HSS{}. It shares the goal of succinct computation with laconic oblivious transfer, function evaluation, and conditional disclosure~\cite{C:CDGGMP17,FOCS:QuaWeeWic18,FOCS:DGGM19,AC:BDHL25,STOC:AbrMalRoy25}, with somewhere-extractable hashing~\cite{FOCS:DGKV22,PKC:BDSZ24}, and with enhanced, vector, and sublinear trapdoor hashes~\cite{EC:BCDJMS25,EC:BitGar25,EC:CouHegPu25,TCC:AbrMalRoy25}. Its group-based instances are based on distributed discrete logarithms from homomorphic and function secret sharing~\cite{EC:GilIsh14,EC:BoyGilIsh15,CCS:BoyGilIsh16,C:BoyGilIsh16,CCS:BCGIO17,ITCS:BGILT18,AC:LaiMalSch18,C:DinKelKle18,EC:BoyKohSch19,EC:OrlSchYak21,PKC:IshLaiMal21,C:DIJL23,PKC:AttCapKoh23}. In all of these instances, the hash bases are uniform, and we change only their distribution, so that the digest depends on one short subset sum and trapdoor hashing acquires a tunable, extremely small image. With a counting argument, we show that independently bounded rows cannot do the same.

\parhead{Short Exponents and Generic Evidence} Short exponents were studied for efficiency and pseudorandomness~\cite{EC:vOoWie96,C:PatSun98,PKC:KosKur04} and as the $c$-decisional short discrete logarithm problem~\cite{EC:AFLT12,JC:AFLT16}, always over superpolynomial intervals, while constrained discrete logarithms depend on the shape of the allowed set~\cite{EPRINT:MirMitNis06} and preprocessing can help decisional problems~\cite{EC:CorKog18}. Our tunable assumption instead needs polynomial intervals chosen after the observer, so we analyze it directly. Against the classical generic lower bounds~\cite{Nechaev94,EC:Shoup97,IMA:Maurer05} and square-root searches~\cite{Sha1971,Pol1978,EC:vOoWie96,JC:Pollard00}, we prove an $\bigO(T^2/\gamma+T^2/\piq)$ bound for $T$ generic operations, tight even for one challenge, and extend it to auxiliary input by presampling~\cite{C:CorDodGuo18}, while the results of~\cite{EC:CorHenWu26} in structured generic models do not yet include the decisional claim. The interval must grow quadratically in the observer's work, the same scale that public collision testing forces on every \ELF{}.

\end{relocated}%

\FloatBarrier

\section{Preliminaries}\label{sec:preliminaries}
\ifConferenceVersion
\prelimparhead{Trapdoor Hashing}{subsec:trapdoor-hashing}\phantomsection\label{subsec:notation}%
We use the trapdoor hashing of~\cite{C:DGIMMO19}, which consists of setup, hashing, predicate-dependent generation of an encoding key and a trapdoor, encoding with the same random tape as hashing, and deterministic decoding of a digest with the trapdoor into the two candidate encodings indexed by the possible predicate values. Correctness requires the encoding to equal the candidate indexed by $f(x)$, input privacy requires the digests of any two inputs of equal length to be indistinguishable, and function privacy requires the encoding keys of any two predicates to be indistinguishable~\cite{C:DGIMMO19}.

\prelimparhead{Extremely Lossy Functions}{subsec:elf-syntax}\phantomsection\label{subsec:telf-syntax}%
An extremely lossy function~\cite{C:Zhandry16} has computationally indistinguishable injective and lossy modes, and the bound on its lossy image can be any sufficiently large polynomial chosen after the running time and target advantage of the adversary. Following~\cite{C:Zhandry16}, the image of a trapdoor extremely lossy function is weakly or strongly efficiently enumerable when an algorithm that runs in time polynomial in $r$ and $\log m$ outputs, from a lossy description with image bound $r$, a superset of its image or its image, respectively, except with negligible probability.

\PlaceRelocated{telf-definition}

\prelimparhead{Cyclic Group Samplers and \DamgardJurik Groups}{subsec:cryptographic-groups}\phantomsection\label{subsec:dj-groups}%
We state our assumptions over a cyclic group sampler, an algorithm that outputs the parameters of a finite cyclic group with a public generator, efficient group operations, and efficient uniform sampling, whose order need not be computable from the parameters, and that outputs $\bot$ with negligible probability. Under the safe-prime setup convention, the \DamgardJurik sampler~\cite{PKC:DamJur01,C:RoySin21} is such a sampler: for a product $N=pq$ of two safe primes, the unit group $\ZNsqx$ contains the cyclic subgroup $\GN$ of $2N$-th powers and the easy subgroup $\easypart$, and $\GN\cap\easypart=\{1\}$. In our construction, the hidden shared exponents lie in $\GN$ and the messages lie in $\easypart$, to which $\Expo(x)\cqq1+Nx$ maps $x\in\ZN$, and with $\DDLog_N$, we remove the shared multiplicative shift.
\else
We fix our notation and recall the definitions of trapdoor hashing~\cite{C:DGIMMO19} and of extremely lossy functions~\cite{C:Zhandry16}, which we construct in this paper. We also define the cyclic group samplers over which we state our assumptions, and we recall the \DamgardJurik group, which is such a sampler under the safe-prime setup convention, and the composite-residuosity assumption.

\fi
\begin{relocated}{notation}
\subsection{Notation and Conventions}\label{app:notation-conventions}

We fix the encodings of integers and residues, our probability notation, and the quantitative indistinguishability relation in which we state every mode-indistinguishability bound.

\subparhead{Sets and encodings} We denote the security parameter by $\secparam$ and write $\bits$ for the binary alphabet. For integers $a\le b$, we write $[a..b]\cqq\{a,a+1,\dots,b\}$, and for every positive integer $n$, we write $[n]\cqq[1..n]$. For $u\in[2^n]$, we write $x(u)=(x_1(u),\dots,x_n(u))\in\bits^n$ for the unique bit vector such that
\[
	u-1=\sum_{j=1}^n2^{j-1}x_j(u).
\]
For an integer $M\ge1$ and $v\in\ZZ/M\ZZ$, we write $\rep(v)\in[0..M-1]$ for the least nonnegative representative of $v$. For every integer $w\ge\ceil{\log_2M}$, we write $\binw(v)=(v_0,\dots,v_{w-1})\in\bits^w$ for the unique padded binary expansion of $v$, that is, the unique bit vector such that
\[
	\rep(v)=\sum_{k=0}^{w-1}2^kv_k.
\]
For a residue vector $\bvec x=(x_1,\dots,x_d)$, we apply $\binw$ coordinatewise and write $x_{j,k}$ for the $k$th bit of $x_j$. We write vectors and matrices in boldface, $\bvec{v}_j$ or $\bvec{v}[j]$ for a coordinate of a vector, $\mat{A}_{i,j}$ for an entry of a matrix, and $\mat{A}^{\transpose}$ for the transpose.

\subparhead{Algorithms and probability} For a randomized algorithm $\advA$, $y\sample\advA(x)$ denotes the output of $\advA$ on input $x$ and independent uniform coins, and $y\gets\advA(x;\rho)$ denotes its output on random tape $\rho$. For a deterministic algorithm $\advA$, we write $y\gets\advA(x)$, and for a finite set $S$, $y\sample S$ denotes a uniform sample from $S$. The notation $\rho\sample\bits^*$ denotes an infinite tape of independent uniform bits, and a polynomial-time algorithm depends only on a polynomially bounded prefix of it. For $k\in\NN$, we write $\rho\vert_k\in[0..2^k-1]$ for the integer whose binary representation is the first $k$ bits of $\rho$. For a proposition $P$, we write $\ind{P}\in\bits$ for the indicator that equals one if and only if $P$ holds. Unless we say otherwise, we allow every $\PPT$ adversary to be nonuniform.

\subparhead{Asymptotics} We write $\poly(n)$ for an unspecified function bounded above by a polynomial in $n$. Unless we say otherwise, a polynomial in an asymptotic quantifier is positive on all sufficiently large inputs. When we use such a polynomial as an integer parameter or as an endpoint of an integer interval, we replace it by an integer-valued polynomial upper bound, and this replacement does not change any asymptotic statement.

\subparhead{Distributions} For ensembles $X=(X_k)_{k\in\NN}$ and $Y=(Y_k)_{k\in\NN}$ over a common domain, we write $X\equivdist Y$, $X\sindist Y$, and $X\cindist Y$ when $X$ and $Y$ are identically distributed, statistically indistinguishable, and computationally indistinguishable against nonuniform $\PPT$ distinguishers, respectively. For random variables $X_k$ and $Y_k$ at a fixed parameter $k$, we write $X_k\indist_{s,\alpha}Y_k$ when every probabilistic circuit of size at most $s(k)$ has advantage at most $\alpha(k)$ in distinguishing $X_k$ from $Y_k$. The first subscript is a size bound and the second an advantage bound, so a target advantage $1/\delta(k)$ is written as the subscript $1/\delta$, and the relation is required at every parameter $k$ that the surrounding quantifiers fix.

In every proof of mode indistinguishability, we fix one threshold for all circuits of a given polynomial size, and such a threshold exists by the following consequence of nonuniform indistinguishability.

\begin{lemmaE}[Threshold for Polynomial-Size Circuits]\label{lem:uniform-quantitative-threshold}
	Let $X=(X_k)_{k\in\NN}$ and $Y=(Y_k)_{k\in\NN}$ be ensembles with $X\cindist Y$. For every polynomial $p$ and every positive polynomial $\delta$, there exists $k_0\in\NN$ such that $X_k\indist_{p,1/\delta}Y_k$ for every $k\ge k_0$.
\end{lemmaE}
\begin{proofE}
	We combine one violating circuit for each bad parameter into one nonuniform family, whose advantage then exceeds $1/\delta(k)$ infinitely often. We suppose, toward a contradiction, that the conclusion fails for some $p$ and $\delta$. Then for infinitely many $k$, there is a probabilistic circuit $A_k$ of size at most $p(k)$ such that
	\[
		\abs{\Prob[A_k(X_k)=1]-\Prob[A_k(Y_k)=1]}>1/\delta(k).
	\]
	For each such $k$, we choose one violating circuit and complement its output if necessary, so that
	\[
		\Prob[A_k(X_k)=1]-\Prob[A_k(Y_k)=1]>1/\delta(k),
	\]
	and for every other $k$, we let $A_k$ output $0$. The family $A=(A_k)_{k\in\NN}$ is nonuniform and has polynomial size, and its distinguishing advantage exceeds $1/\delta(k)$ for infinitely many $k$. Since this advantage is not negligible, $X\cindist Y$ is false, a contradiction.
\end{proofE}

\end{relocated}%
\ifConferenceVersion\else
\subsection{Trapdoor Hashing}\label{app:tdh-definitions}

In trapdoor hashing~\cite{C:DGIMMO19}, a short digest of a long input $x$ and a trapdoor for a predicate $f$ determine a short encoding of $f(x)$. We require that the encoding computed from $x$ equals the candidate that the trapdoor and the digest select for $f(x)$, that the digests of any two inputs of equal length are indistinguishable, and that the encoding keys of any two predicates are indistinguishable. When a scheme has auxiliary public parameters shared by all algorithms, we give them to every algorithm and adversary and suppress them from the syntax.

We use the syntax of~\cite{C:DGIMMO19} with three changes: the predicate class can depend on the hash key, we leave the lengths of digests and encodings open, and we require the decoding algorithm to be deterministic.

\begin{definition}[Trapdoor Hashing]\label{def:trapdoor-hashing}
	Let $\FuncClass$ be a predicate class indexed by $\secparam,n\in\NN$ and by a hash key, such that every $f\in\FuncClass_{\secparam,n,\hk}$ maps $\bits^n$ to $\bits$. A trapdoor-hashing scheme for $\FuncClass$ is a tuple $\THF=(\TDHSetup,\allowbreak\TDHHash,\allowbreak\TDHGen,\allowbreak\TDHEnc,\allowbreak\TDHDec)$ of the following five probabilistic polynomial-time algorithms.
	\begin{itemize}
		\item $\hk\sample\TDHSetup(1^\secparam,1^n)$. On input the security parameter and an input length $n$, the setup algorithm samples a hash key $\hk$ for inputs in $\bits^n$.
		\item $h\gets\TDHHash(\hk,x;\rho)$. On input a hash key $\hk$, an input $x\in\bits^n$, and a random tape $\rho$, the hashing algorithm deterministically outputs a digest $h$.
		\item $(\ek,\td)\sample\TDHGen(\hk,f)$. On input a hash key $\hk$ and a predicate $f\in\FuncClass_{\secparam,n,\hk}$, the generation algorithm samples a public encoding key $\ek$ and a trapdoor $\td$.
		\item $t\gets\TDHEnc(\ek,x;\rho)$. On input an encoding key $\ek$, an input $x\in\bits^n$, and a random tape $\rho$, the encoding algorithm deterministically outputs an encoding $t$.
		\item $(t_0,t_1)\gets\TDHDec(\td,h)$. On input a trapdoor $\td$ and a digest $h$, the deterministic decoding algorithm outputs the two candidate encodings $t_0$ and $t_1$, indexed by the two possible values of $f(x)$.
	\end{itemize}
\end{definition}

We state the correctness of~\cite{C:DGIMMO19} as one failure bound $\delta$, for every hash key in the support of setup, on the event that the encoding equals the candidate indexed by $f(x)$ and the two candidates differ. For a negligible $\delta$, this implies the full correctness of~\cite{C:DGIMMO19}, and in the lossy constructions, we add $\delta$ to the other failure events by a union bound.

\begin{definition}[Correctness]\label{def:tdh-correctness}
	Let $\delta\colon\NN\to[0,1]$. A trapdoor-hashing scheme $\THF$ for $\FuncClass$ is $\delta$-correct if for every $\secparam,n\in\NN$, every $x\in\bits^n$, every $\hk\in\Supp(\TDHSetup(1^\secparam,1^n))$, and every $f\in\FuncClass_{\secparam,n,\hk}$,
	\[
		\Prob\left[
			t=t_{f(x)} \land t_0 \neq t_1
			\mbar
			\begin{array}{l}
				\rho\sample\bits^*,             \\
				h\gets\TDHHash(\hk,x;\rho),     \\
				(\ek,\td)\sample\TDHGen(\hk,f), \\
				t\gets\TDHEnc(\ek,x;\rho),      \\
				(t_0,t_1)\gets\TDHDec(\td,h)
			\end{array}
			\right]\ge 1-\delta(\secparam).
	\]
	We call $\THF$ fully correct if it is $\varepsilon$-correct for some negligible function $\varepsilon$, and perfectly correct if it is $0$-correct.
\end{definition}

Input privacy is the notion of~\cite{C:DGIMMO19}, with nonuniform sequences of input pairs, to which we add a perfect variant.

\begin{definition}[Input Privacy]\label{def:tdh-input-privacy}
	Let $x^0=(x^0_\secparam)_{\secparam\in\NN}$ and $x^1=(x^1_\secparam)_{\secparam\in\NN}$ be sequences of polynomial length with $\nlam\cqq\abs{x^0_\secparam}=\abs{x^1_\secparam}$ for every $\secparam\in\NN$. For $b\in\bits$, define the ensemble $\Dinput_b=(\Dinput_{b,\secparam})_{\secparam\in\NN}$ by
	\[
		\Dinput_{b,\secparam}\cqq
		\left\{(\hk_\secparam,h^b_\secparam)\ \middle|\
		\begin{array}{l}
			\hk_\secparam\sample\TDHSetup(1^\secparam,1^{\nlam}), \\[-1pt]
			\rho^b_\secparam\sample\bits^*,                             \\[-1pt]
			h^b_\secparam\gets\TDHHash(\hk_\secparam,x^b_\secparam;\rho^b_\secparam)
		\end{array}\right\}.
	\]
	A trapdoor-hashing scheme $\THF$ is input private if $\Dinput_0\cindist\Dinput_1$ for every such pair of sequences. It has statistical input privacy if $\Dinput_0\sindist\Dinput_1$, and perfect input privacy if $\Dinput_0\equivdist\Dinput_1$, for every such pair.
\end{definition}

Function privacy is the analogue for trapdoor hashing of mode indistinguishability for lossy trapdoor functions~\cite{STOC:PeiWat08}. In the function privacy of~\cite{C:DGIMMO19}, the two predicate sequences are fixed in advance. Since our predicate class depends on the hash key, we let a nonuniform selector choose the two predicates after it receives the hash key, and we give its state $\var{st}$ to the distinguisher.

\begin{definition}[Function Privacy]\label{def:tdh-function-privacy}
	Let $n\colon\NN\to\NN$ be polynomially bounded, and write $\nlam\cqq n(\secparam)$. A selector is a nonuniform $\PPT$ algorithm $\advS$ that on every input $(1^\secparam,\hk)$ outputs a triple $(f^0,f^1,\var{st})$ with $f^0,f^1\in\FuncClass_{\secparam,\nlam,\hk}$. A trapdoor-hashing scheme $\THF$ for $\FuncClass$ is function private if for every selector $\advS$,
	\begin{align*}
		           & \left\{(\hk,\var{st},\ek)\ \middle|\
		\begin{array}{l}
			\hk\sample\TDHSetup(1^\secparam,1^{\nlam}),      \\[-1pt]
			(f^0,f^1,\var{st})\sample\advS(1^\secparam,\hk), \\[-1pt]
			(\ek,\td)\sample\TDHGen(\hk,f^0)
		\end{array}\right\}_{\secparam\in\NN} \\
		\cindist{} & \left\{(\hk,\var{st},\ek)\ \middle|\
		\begin{array}{l}
			\hk\sample\TDHSetup(1^\secparam,1^{\nlam}),      \\[-1pt]
			(f^0,f^1,\var{st})\sample\advS(1^\secparam,\hk), \\[-1pt]
			(\ek,\td)\sample\TDHGen(\hk,f^1)
		\end{array}\right\}_{\secparam\in\NN}.
	\end{align*}
\end{definition}
\fi
\begin{relocated}{elf-definition}
\subsection{Extremely Lossy Functions}\label{app:elf-syntax}

\cite{STOC:PeiWat08} defines lossy trapdoor functions, whose injective and lossy modes are computationally indistinguishable, and constructs chosen-ciphertext-secure encryption from them, and~\cite{C:BolFehONe08} constructs deterministic encryption from them. \cite{C:Zhandry16} defines extremely lossy functions, in which the bound $r$ on the lossy image can be any sufficiently large polynomial chosen after the running time and target advantage of the adversary.

We use the definition of~\cite{C:Zhandry16} with three changes. We write $m$ and $\Nout(m)$ for the domain size and the bound on the range, which~\cite{C:Zhandry16} calls $M$ and $N$, because $N$ denotes a modulus in this paper, and we require $\log\Nout(m)=\poly(\log m)$ as~\cite{EC:Zhandry19a} does when it recalls this definition. We give evaluation as a separate deterministic algorithm. We state mode indistinguishability for probabilistic circuits of size $t$, where~\cite{C:Zhandry16} quantifies over adversaries whose running time is at most a polynomial $p$.

\begin{definition}[Extremely Lossy Function]\label{def:extremely-lossy-function}
	For every $m\ge2$, write $n_m\cqq\clog{m}$. An extremely lossy function $\ELF=(\ELFGen,\ELFEval)$ consists of a function $\Nout$ with $\log\Nout(m)=\poly(n_m)$ and the following two algorithms.
	\begin{itemize}
		\item $f\sample\ELFGen(m,r)$. On input integers $m\ge2$ and $r\in[m]$, the probabilistic polynomial-time generation algorithm samples a description $f$ of a function $f\colon[m]\to[\Nout(m)]$.
		\item $y\gets\ELFEval(f,u)$. On input a description $f$ and $u\in[m]$, the deterministic evaluation algorithm outputs $y=f(u)$ in time $\poly(n_m)$.
	\end{itemize}
	We identify a description $f$ with the function $u\mapsto\ELFEval(f,u)$. The pair $\ELF$ satisfies the following three properties.
	\begin{itemize}
		\item \textbf{Injective Mode.} If $f\sample\ELFGen(m,m)$, then $f$ is injective except with probability negligible in $n_m$.
		\item \textbf{Lossy Mode.} For every $r\in[m]$, if $f\sample\ELFGen(m,r)$, then $\abs{\Image(f)}\le r$ except with probability negligible in $n_m$.
		\item \textbf{Mode Indistinguishability.} For every polynomial $t$ and every positive polynomial $\delta$, there exists a polynomial $Q$ such that for every sufficiently large $m$ and every $r\in[Q(n_m)..m]$,
		      \[
			      \left\{
			      \begin{array}{@{}c|l@{}}
				      f & f\sample\ELFGen(m,m)
			      \end{array}
			      \right\}
			      \indist_{t,1/\delta}
			      \left\{
			      \begin{array}{@{}c|l@{}}
				      f & f\sample\ELFGen(m,r)
			      \end{array}
			      \right\},
		      \]
		      where $t$ and $\delta$ are evaluated at $n_m$.
	\end{itemize}
\end{definition}

\end{relocated}%
\begin{relocated}{group-samplers}
\subsection{Cyclic Group Samplers}\label{app:cyclic-group-samplers}

We state our bounded-exponent assumptions and our trapdoor hash over a finite cyclic group with a public generator, efficient group operations, and an efficient uniform sampler, whose order need not be computable from the public parameters. Since the order can be hidden, we define the sampler without an exponent space.

\begin{definition}[Cyclic Group Sampler]\label{def:group-sampler}
	A cyclic group sampler is a randomized polynomial-time algorithm $\CGen(\cdot)$ that on input $1^\secparam$ outputs either a distinguished failure symbol $\bot$ or parameters $\pp$, such that the probability that $\CGen(1^\secparam)$ outputs $\bot$ is negligible in $\secparam$, and such that there exist a positive integer-valued polynomial $\grplen$ and a negligible function $\nuG$ for which every $\secparam\in\NN$ and every $\pp\in\Supp(\CGen(1^\secparam))\setminus\{\bot\}$ satisfy the following four conditions.
	\begin{itemize}
		\item \textbf{Group.} The parameters $\pp$ determine a finite cyclic group $\GG=\GG(\pp)$ of order $\ord(\pp)\cqq\abs{\GG}\ge2^\secparam$ and a generator $g=g(\pp)$ of $\GG$.
		\item \textbf{Representation.} The parameters, the generator, and every group element have canonical representations of length at most $\grplen(\secparam)$.
		\item \textbf{Operations.} The group operation, inversion, and exponentiation by an integer of polynomially bounded bit length take time $\poly(\secparam)$.
		\item \textbf{Sampling.} A randomized algorithm $\SampG(\pp)$ runs in time $\poly(\secparam)$ and outputs an element of $\GG$ or $\bot$, with
		      \[
			      \Prob[\SampG(\pp)=\bot]\le\nuG(\secparam),
		      \]
		      and its output conditioned on $\SampG(\pp)\ne\bot$ is uniform on $\GG$.
	\end{itemize}
\end{definition}

Membership in $\GG$ and the order $\ord(\pp)$ need not be efficiently computable from $\pp$. In our experiments, we only sample from $\GG$ and compute in it, and in our inversion algorithms, we decide membership only in an ambient group that contains $\GG$.

Every quantifier over $\Supp(\CGen(1^\secparam))$, or over the support of an algorithm that calls $\CGen$, ranges over outputs other than $\bot$. An algorithm that receives $\bot$ from its sampler or as input outputs $\bot$, unless it explicitly outputs a declared fallback object, and a generator with such an explicit fallback branch has the fallback object in its support. We write $G\sample\GG(\pp)$ for a call to $\SampG(\pp)$ and identify a call that does not output $\bot$ with the uniform distribution on $\GG$, and an experiment in which such a call outputs $\bot$ outputs $\bot$. By a union bound, the probability that one of polynomially many such calls outputs $\bot$ is negligible, and in each quantitative statement with a target advantage $1/\delta$, we enlarge the threshold by this probability.

\end{relocated}%
\begin{relocated}{dj-groups}
\subsection{\DamgardJurik Groups}\label{app:dj-groups}\label{app:dj-sampling-details}\label{app:dj-structure-details}

Modulo the square of an RSA modulus $N$, the unit group is the direct product of a subgroup with easy discrete logarithms and a subgroup in which computing them is assumed hard. \cite{EC:Paillier99} bases Paillier encryption on this decomposition, and~\cite{PKC:DamJur01} generalizes it to higher powers of $N$. In our construction, the encoded predicate value lies in the easy subgroup, and the shared exponent, which the trapdoor removes, lies in the hard subgroup. We follow~\cite{C:RoySin21}, which computes distributed discrete logarithms without error in this setting, and we use only the modulus $N^2$, whose message space is $\ZN$. For a function $\ell\colon\NN\to\NN$, we write
\[
	\SafePrimes_\ell(\secparam)\cqq\left\{p\in\Primes\ \middle|\ (p-1)/2\in\Primes,\ 2^{\ell(\secparam)-1}<p<2^{\ell(\secparam)}\right\},
\]
and for $p,q\in\SafePrimes_\ell(\secparam)$, we write $p'\cqq(p-1)/2$, $q'\cqq(q-1)/2$, and $\varphi\cqq(p-1)(q-1)$ for Euler's totient of $N\cqq pq$.

We use the standard setup convention that distinct safe primes of a prescribed bit length can be sampled in probabilistic polynomial time with negligible failure probability. Its unconditional asymptotic validity is open, because the density, and even the infinitude, of safe primes is not known~\cite{VS2013}. We use this convention only to generate the public parameters, and each of our computational assumptions is an assumption about the resulting distribution of moduli. \cite{C:RoySin21} also uses safe-prime moduli in its trapdoor-hashing construction, and we realize the convention by a bounded search with $W(\secparam)$ trials for a polynomial $W$.

For our sampler, we take the safe-prime key generation of~\cite{C:RoySin21} and the \DamgardJurik setting of~\cite{PKC:DamJur01} at the modulus $N^2$, with three changes: we publish a generator $g$ of the subgroup $\GN$ of $2N$-th powers, which we obtain by rejection sampling with the factorization of $N$, we output the ordered pair $(N,g)$, and we output $\bot$ when a search exhausts its trials. Since we reject every element that does not generate $\GN$, the published element always generates $\GN$.

\begin{definition}[\DamgardJurik Sampler]\label{def:dj-group}\label{def:dj-group-full}
	Let $\ell$ be an efficiently computable positive integer-valued polynomial with $\ell(\secparam)\ge\secparam+2$ for every $\secparam\in\NN$, and let $W$ be a positive integer-valued polynomial such that each search in the first step exhausts its $W(\secparam)$ trials with probability negligible in $\secparam$. The \DamgardJurik sampler $\DJGen(1^\secparam)$ is the following randomized algorithm.
	\begin{itemize}
		\item \textbf{Safe primes.} It runs $(p,q)\sample\SafeGen_{\ell,W}(1^\secparam)$, which performs two independent searches, each drawing at most $W(\secparam)$ independent uniform odd integers from $\left(2^{\ell(\secparam)-1},2^{\ell(\secparam)}\right)$ and keeping the first one that lies in $\SafePrimes_\ell(\secparam)$. If either search exhausts its trials or $p=q$, then $\SafeGen_{\ell,W}$ outputs $\bot$, and so does $\DJGen$.
		\item \textbf{Group.} It sets $N\cqq pq$ and $\GN\cqq\left\{z^{2N}\mid z\in\ZNsqx\right\}$.
		\item \textbf{Generator.} With the factorization of $N$, it repeats the steps $z\sample\ZNsqx$ and $g\gets z^{2N}$, for at most $\secparam$ independent trials, until $g^{p'}\ne1$ and $g^{q'}\ne1$, and it outputs $\bot$ if all $\secparam$ trials are rejected.
		\item \textbf{Output.} Otherwise, it outputs $\pp\cqq(N,g)$ and discards the factorization.
	\end{itemize}
\end{definition}

We use safe primes as in~\cite{C:RoySin21}, so that $\GN$ is cyclic and $\varphi/4=p'q'$ is a product of two primes. \cite{C:DGIMMO19} assumes composite residuosity for a uniform Blum integer and uses only the subgroup decomposition. Under the setup convention, $\SafeGen_{\ell,W}$ outputs $\bot$ with negligible probability, and conditioned on success, its output is uniform among ordered pairs of distinct elements of $\SafePrimes_\ell(\secparam)$. An inverse-polynomial lower bound on the density of $\SafePrimes_\ell(\secparam)$ among the odd integers of the interval suffices to realize the setup convention.

We need the easy subgroup at every modulus $N\ge2$, including moduli that are not products of two safe primes, so we state it without the subgroup decomposition. \cite{C:RoySin21} uses the same isomorphism modulo $N^{s+1}$ for every $s\ge1$, and we need only $s=1$.

\begin{lemmaE}[Easy Part at an Arbitrary Modulus]\label{lem:easy-part-arbitrary-modulus}
	Let $N\ge2$ be an integer. The map $\Expo(x)\cqq1+Nx$ is a group isomorphism from $(\ZN)^+$ onto the subgroup $\easypart$ of $\ZNsqx$, which has order $N$, and its inverse $\Logg$ satisfies $\Logg(1+Nx)=x$. Both maps are computable in time $\poly(\log N)$.
\end{lemmaE}

\begin{proofE}
	In $\ZNsq$, we have
	\[
		(1+Nx)(1+Ny)=1+N(x+y)+N^2xy=1+N(x+y),
	\]
	so $\Expo$ is a homomorphism from $(\ZN)^+$. It is injective, since $1+Nx=1+Ny$ implies $N(x-y)\equiv0\pmod{N^2}$ and hence $x=y$ in $\ZN$. Its image is the set of the $N$ residues congruent to $1$ modulo $N$, and each of them is coprime to $N$, hence a unit modulo $N^2$. We compute $\Expo$ with one multiplication and $\Logg$ with one division by $N$, in time $\poly(\log N)$.
\end{proofE}

For a successfully sampled modulus, we need the order of $\ZNsqx$, the cyclicity and order of $\GN$, uniform sampling from $\GN$ with the public modulus, and the trivial intersection of $\GN$ with the easy subgroup. \cite{C:BraGol10} shows that $\ZNsqx$ is the direct product of $\easypart$ and the subgroup of $N$-th residues, of coprime orders $N$ and $\varphi$, for primes $p$ and $q$ of equal length. We prove the corresponding facts for the subgroup $\GN$ of $2N$-th powers and safe primes, and we prove that $\DJGen$ is a cyclic group sampler.

\begin{lemmaE}[Structure of the \DamgardJurik Group]\label{lem:dj-structure}
	Let $\ell$ be the polynomial of \cref{def:dj-group}, let $p,q\in\SafePrimes_\ell(\secparam)$ with $p\ne q$, and let $N\cqq pq$ and $M\cqq p'q'$. Then the following hold.
	\begin{parts}
		\item\label{lem:dj-structure:a} $\abs{\ZNsqx}=N\varphi$.
		\item\label{lem:dj-structure:b} $\GN$ is cyclic of order $M=\varphi/4>2^{2\ell(\secparam)-4}\ge2^{2\secparam}$, and the map $z\mapsto z^{2N}$ sends the uniform distribution on $\ZNsqx$ to the uniform distribution on $\GN$.
		\item\label{lem:dj-structure:c} $\GN\cap\easypart=\{1\}$, and the map $\GN\times\ZN\to\ZNsqx$, $(c,x)\mapsto c\Expo(x)$, is injective.
		\item\label{lem:dj-structure:d} Under the setup convention of \cref{def:dj-group}, $\DJGen$ is a cyclic group sampler in the sense of \cref{def:group-sampler}.
	\end{parts}
\end{lemmaE}

\begin{proofE}
	We compute the image of $z\mapsto z^{2N}$ in each factor of the Chinese-remainder decomposition of $\ZNsqx$. For \cref{lem:dj-structure:a} of \cref{lem:dj-structure}, by the Chinese remainder theorem, $\ZNsqx\cong(\ZZ/p^2\ZZ)^\times\times(\ZZ/q^2\ZZ)^\times$, and the two factors are cyclic of orders $p(p-1)$ and $q(q-1)$. Hence, $\abs{\ZNsqx}=pq(p-1)(q-1)=N\varphi$.

	For \cref{lem:dj-structure:b} of \cref{lem:dj-structure}, we fix a cyclic group $C$ of order $p(p-1)=2pp'$. The image of $c\mapsto c^{2N}$ on $C$ is the unique subgroup of order $2pp'/\gcd(2N,2pp')$, and $\gcd(2N,2pp')=2p\gcd(q,p')$. Every element of $\SafePrimes_\ell(\secparam)$ has bit length $\ell(\secparam)$, so $p'=(p-1)/2<2^{\ell(\secparam)-1}<q$. Hence, $q$ and $p'$ are distinct primes, $\gcd(2N,2pp')=2p$, and the image has order $p'$. In the same way, $q'<2^{\ell(\secparam)-1}<p$, and the image in the second factor has order $q'$. Since $p'$ and $q'$ are distinct primes, the product of the two images is cyclic of order $p'q'=M$, and $M=\varphi/4$ because $\varphi=4p'q'$. Since $p$ and $q$ are odd, distinct, and larger than $2^{\ell(\secparam)-1}$, we have $p',q'\ge2^{\ell(\secparam)-2}$ with at least one strict inequality, so $M>2^{2\ell(\secparam)-4}\ge2^{2\secparam}$ because $\ell(\secparam)\ge\secparam+2$. By the definition of $\GN$, the map $z\mapsto z^{2N}$ is a homomorphism from $\ZNsqx$ onto $\GN$. All its preimages are cosets of its kernel and have the same size, so it sends the uniform distribution on $\ZNsqx$ to the uniform distribution on $\GN$.

	For \cref{lem:dj-structure:c} of \cref{lem:dj-structure}, the orders $M=p'q'$ and $N=pq$ are coprime, because $p',q'<2^{\ell(\secparam)-1}<p,q$ and all four numbers are prime. By \cref{lem:easy-part-arbitrary-modulus}, $\easypart$ is a subgroup of order $N$. An element of $\GN\cap\easypart$ generates a subgroup whose order divides both $M$ and $N$, hence equals $1$. If $c\Expo(x)=c'\Expo(x')$ with $c,c'\in\GN$ and $x,x'\in\ZN$, then $c'^{-1}c=\Expo(x'-x)\in\GN\cap\easypart=\{1\}$, so $c=c'$, and $x=x'$ because $\Expo$ is injective.

	For \cref{lem:dj-structure:d} of \cref{lem:dj-structure}, we first analyze the generator step of \cref{def:dj-group}. Since $\GN$ is cyclic of order $p'q'$, its elements $g$ with $g^{p'}=1$ or $g^{q'}=1$ form its subgroups of orders $p'$ and $q'$, so $g\in\GN$ generates $\GN$ if and only if $g^{p'}\ne1$ and $g^{q'}\ne1$. Hence, every accepted $g$ generates $\GN$. By \cref{lem:dj-structure:b} of \cref{lem:dj-structure}, each trial samples $g$ uniformly from $\GN$, and $\GN$ has $(p'-1)(q'-1)$ generators, so each trial is accepted with probability
	\[
		\left(1-\frac{1}{p'}\right)\left(1-\frac{1}{q'}\right)\ge\left(1-\frac{1}{2}\right)\left(1-\frac{1}{3}\right)=\frac{1}{3},
	\]
	where we use that $p'$ and $q'$ are distinct primes. The $\secparam$ trials use independent samples $z$, so they are independent, and all of them are rejected with probability at most $(2/3)^\secparam$.

	We verify the conditions of \cref{def:group-sampler} for an output $(N,g)\ne\bot$ of $\DJGen(1^\secparam)$. The group $\GN$ is cyclic, $g$ generates it, and since $p,q>2^{\ell(\secparam)-1}$ and $\ell(\secparam)\ge\secparam+2$, its order satisfies $M=p'q'\ge2^{2\ell(\secparam)-4}\ge2^{2\secparam}\ge2^\secparam$. The sampler $\DJGen(1^\secparam)$ outputs $\bot$ only if $\SafeGen_{\ell,W}$ does, which happens with negligible probability under the setup convention, or if all generator trials are rejected, which happens with probability at most $(2/3)^\secparam$. Every element of $\ZNsq$ has a canonical representative in $[0..N^2-1]$ of length at most $4\ell(\secparam)$ bits, and the group operation, inversion, and exponentiation by an integer of polynomially bounded bit length take time $\poly(\log N)=\poly(\secparam)$.

	For the sampling condition, we define $\SampG(N,g)$ as follows: it draws $z\sample\ZNsq$, outputs $\bot$ if $\gcd(z,N)\ne1$, and otherwise outputs $z^{2N}$. Conditioned on $\gcd(z,N)=1$, the element $z$ is uniform on $\ZNsqx$, so the output is uniform on $\GN$ by \cref{lem:dj-structure:b} of \cref{lem:dj-structure}. Moreover, since $p,q>2^{\ell(\secparam)-1}$ and $\ell(\secparam)\ge\secparam+2$,
	\[
		\Prob\left[\gcd(z,N)\ne1\mbar z\sample\ZNsq\right]=\frac{1}{p}+\frac{1}{q}-\frac{1}{N}<2^{2-\ell(\secparam)}\le2^{-\secparam},
	\]
	so $\SampG$ runs in strict polynomial time and satisfies the sampling condition with $\nuG(\secparam)\cqq2^{-\secparam}$.
\end{proofE}

\ifConferenceVersion
\PlaceRelocated{fallback-modulus}
\fi

To decode, we must recover the easy part of a product $c\Expo(x)$ for an unknown $c\in\GN$. \cite{C:BoyGilIsh16} computes distributed discrete logarithms in prime-order groups with inverse-polynomial error. \cite{EC:OrlSchYak21} and~\cite{C:RoySin21} independently compute them without error, over Paillier and \DamgardJurik groups, respectively. Following~\cite{C:RoySin21}, we recover the easy part of a unit by dividing it by its reduction modulo $N$. The resulting map equals the distance function of~\cite{C:RoySin21} at $s=1$ and the output of the algorithm of~\cite{EC:OrlSchYak21}, and we define it at every modulus $N\ge2$.

\begin{definition}[Distributed Discrete Logarithm]\label{def:ddlog}
	Let $N\ge2$ be an integer. For $c\in\ZNsqx$, let $w\in[0..N-1]$ be the least nonnegative representative of $c$ modulo $N$, viewed as an element of $\ZNsq$, and define $\DDLog_N(c)\cqq\Logg\left(c\cdot w^{-1}\right)\in\ZN$.
\end{definition}

Since $\gcd(w,N)=\gcd(c,N)=1$, we have $w\in[1..N-1]$, and $w$ is a unit modulo $N^2$ because $N$ and $N^2$ have the same prime divisors. The element $cw^{-1}$ is congruent to $1$ modulo $N$, so it lies in $\easypart$, the domain of $\Logg$ in \cref{lem:easy-part-arbitrary-modulus}, and $\DDLog_N$ is well defined and computable in time $\poly(\log N)$ at every modulus $N\ge2$.

\ifConferenceVersion
\PlaceRelocated{ddlog-shift}

\PlaceRelocated{dj-accounting}

Since every predicate value and every encoder share of our constructions is one residue in $\ZN$, we use only the modulus power $k=1$ of the general \DamgardJurik scheme~\cite{PKC:DamJur01}.

\begin{remark}[On the Choice of Modulus Power]\label{rem:modulus-power}
	\cite{PKC:DamJur01} defines the scheme modulo $N^{k+1}$ for every $k\ge1$, with message space $\ZZ/N^k\ZZ$, and~\cite{C:RoySin21} uses $k>1$ because its homomorphic secret-sharing construction encrypts its own secret key. In the \ELTDH{}, we write the input residues in binary, and its predicate value and each encoder share remain one residue, so the message space $\ZN$ of $k=1$ suffices for the trapdoor \ELF{} and the branch constructions.

	For $k>1$, the residue input space and every group-valued output become larger, the bounded-subset-sum argument for lossiness is unchanged, and we would need arithmetic modulo $N^{k+1}$ and the $k$-term series for $\Expo$ and $\Logg$, with $\ZN$ replaced by $\ZZ/N^k\ZZ$ and $\ZNsq$ by $\ZZ/N^{k+1}\ZZ$. \cite{C:RoySin21} amortizes one digest over many encodings in rate-one string oblivious transfer, where the ratio $(k+1)/k$ tends to $1$ as $k$ grows.
\end{remark}
\fi

Since $\ord(\GN)$ is not public, we sample exponents from an interval whose length is a power of two and exceeds the modulus by a factor of at least $2^\secparam$, and we bound once the statistical distance of the reduced exponent from uniform. \cite{C:BraGol10} bounds this distance by $M'/2^{\expl}$, and we prove the sharper bound $M'/(4\cdot2^{\expl})$, which is at most $2^{-\secparam-2}$ for $M'\le N$.

\begin{lemmaE}[Exponent Smoothing]\label{lem:exponent-smoothing}
	Let $N\ge2$ and $\secparam\in\NN$, let $\expl(\secparam,N)\cqq\clog{N}+\secparam$, and let $\rho$ be uniform on $[0..2^{\expl(\secparam,N)}-1]$. For every integer $M'$ with $1\le M'\le N$, the statistical distance between $\rho\bmod M'$ and the uniform distribution on $\ZZ/M'\ZZ$ is at most $2^{-\secparam-2}$.
\end{lemmaE}

\begin{proofE}
	We count the preimages of each residue modulo $M'$. We write $\expl\cqq\expl(\secparam,N)$ and $2^{\expl}=a M'+b$ with $a\in\NN$ and $b\in[0..M'-1]$. Exactly $b$ residues have $a+1$ preimages and the remaining $M'-b$ residues have $a$ preimages, so
	\begin{align*}
		(a+1)2^{-\expl}-\frac{1}{M'}
		  & =\frac{M'(a+1)-2^{\expl}}{2^{\expl}M'}
		=\frac{M'-b}{2^{\expl}M'},                \\
		\frac{1}{M'}-a2^{-\expl}
		  & =\frac{2^{\expl}-aM'}{2^{\expl}M'}
		=\frac{b}{2^{\expl}M'}.
	\end{align*}
	Both differences are nonnegative, and the $b$ residues of the first kind and the $M'-b$ residues of the second kind contribute equally to the statistical distance. Therefore, the distance is
	\[
		\frac{b(M'-b)}{2^{\expl}M'}
		\le\frac{M'}{4\cdot2^{\expl}}
		\le\frac{N}{4\cdot2^{\clog{N}+\secparam}}
		\le2^{-\secparam-2},
	\]
	where we use $b(M'-b)\le(M'/2)^2$ in the first inequality, $M'\le N$ in the second, and $N\le2^{\clog{N}}$ in the last.
\end{proofE}

In our proof of function privacy, we use the decisional composite residuosity assumption of~\cite{EC:Paillier99}, the indistinguishability of a uniform unit modulo $N^2$ from a uniform $N$-th power, which we state for the moduli of $\DJGen$.

\begin{definition}[Decisional Composite Residuosity]\label{def:dcr}
	\ifConferenceVersion
	For $\beta\in\bits$, we define the ensemble $\DDCR_\beta=(\DDCR_{\beta,\secparam})_{\secparam\in\NN}$ by
	\begin{align*}
		\DDCR_{0,\secparam} & \cqq
		\left\{(N,z)\ \middle|\
		\begin{array}{l}
			(N,g)\sample\DJGen(1^\secparam), \\[-1pt]
			z\sample\ZNsqx
		\end{array}\right\},      \distpairsep
		\DDCR_{1,\secparam} & \cqq
		\left\{(N,z^N)\ \middle|\
		\begin{array}{l}
			(N,g)\sample\DJGen(1^\secparam), \\[-1pt]
			z\sample\ZNsqx
		\end{array}\right\}.
	\end{align*}
	The decisional composite residuosity ($\DCR$) assumption holds for $\DJGen(\cdot)$ if $\DDCR_0\cindist\DDCR_1$.
	\else
	The decisional composite residuosity ($\DCR$) assumption holds for $\DJGen(\cdot)$ if
	\[
		\left\{(N,z)\ \middle|\
		\begin{array}{l}
			(N,g)\sample\DJGen(1^\secparam), \\[-1pt]
			z\sample\ZNsqx
		\end{array}\right\}_{\secparam\in\NN}
		\cindist
		\left\{(N,z^N)\ \middle|\
		\begin{array}{l}
			(N,g)\sample\DJGen(1^\secparam), \\[-1pt]
			z\sample\ZNsqx
		\end{array}\right\}_{\secparam\in\NN}.
	\]
	\fi
\end{definition}

\ifConferenceVersion
\PlaceRelocated{vector-game}

Under $\DCR$, membership in $\GN$ is hard to decide from $(N,g)$, because a uniform $N$-th power lies in $\GN$ with probability $1/4$, and a uniform unit with probability $1/(4N)$. For this reason, we omit a membership test from \cref{def:group-sampler}.

\begin{remark}[Hardness of Deciding Membership in $\GN$]\label{rem:no-membership-test}
	Suppose that $\DCR$ holds for $\DJGen(\cdot)$, and let $\varepsilon\in[0,1/5)$ be a constant. For every \PPT{} algorithm $\advM$, there are infinitely many $\secparam$ for which some $(N,g)\in\Supp(\DJGen(1^\secparam))$ and some $z\in\ZNsqx$ satisfy
	\[
		\Prob\left[\advM((N,g),z)\ne\ind{z\in\GN}\right]>\varepsilon,
	\]
	where the probability is over the coins of $\advM$.

	We suppose, toward a contradiction, that for all sufficiently large $\secparam$, every $(N,g)\in\Supp(\DJGen(1^\secparam))$, and every $z\in\ZNsqx$, $\advM$ errs with probability at most $\varepsilon$. Given a challenge $(N,z)$ from $\DDCR_\beta$, the reduction samples $w\sample\ZNsqx$, sets $g\gets w^{2N}$, and outputs $\advM((N,g),z)$, and on the challenge $\bot$, it outputs $0$. It runs $\advM$ once and otherwise performs one sampling and one exponentiation modulo $N^2$. By \cref{lem:dj-structure:b} of \cref{lem:dj-structure}, $g$ is uniform on $\GN$ and independent of $z$, and $g$ fails to generate $\GN$ with probability at most $1/p'+1/q'\le2^{3-\ell(\secparam)}$. If $g$ generates $\GN$, then $(N,g)\in\Supp(\DJGen(1^\secparam))$, because the generator step of \cref{def:dj-group} outputs every generator with positive probability.

	It remains to compare the two acceptance probabilities. We repeat the image computation of \cref{lem:dj-structure} with the exponent $N$ in place of $2N$ and obtain $\gcd(N,2pp')=p$ in the first cyclic factor and $\gcd(N,2qq')=q$ in the second, so the $N$-th powers form a subgroup of order $4p'q'=4M$. This subgroup contains $\GN$, and the map $z\mapsto z^N$ sends the uniform distribution on $\ZNsqx$ to the uniform distribution on it, so a challenge from $\DDCR_1$ lies in $\GN$ with probability $1/4$. By \cref{lem:dj-structure:a,lem:dj-structure:b} of \cref{lem:dj-structure}, a challenge from $\DDCR_0$ lies in $\GN$ with probability $M/(N\varphi)=1/(4N)$. Let $\nu(\secparam)$ be the probability that $\DJGen(1^\secparam)$ outputs $\bot$. For all sufficiently large $\secparam$, the reduction outputs $1$ on $\DDCR_1$ with probability at least $(1-\nu)(1-2^{3-\ell(\secparam)})(1-\varepsilon)/4$, and on $\DDCR_0$ with probability at most $(1-\nu)\left(2^{3-\ell(\secparam)}+1/(4N)+\varepsilon\right)$. Since $\nu$, $2^{3-\ell(\secparam)}$, and $1/N$ are negligible, the advantage of the reduction is at least $(1-5\varepsilon)/4$ minus a negligible function. This advantage is not negligible for $\varepsilon<1/5$, so $\DCR$ is false, a contradiction.
\end{remark}
\fi

\end{relocated}%

\section{Image-Size Bounds and Small-Exponent Assumptions}\label{sec:limits-small-exponent}
Every publicly evaluable \ELF{} has a lossy image of size $\Omega(\delta k^2)$ against observers that make $k$ evaluations and have advantage at most $1/\delta$, and in the generic group model, an algorithm that makes $T$ group operations distinguishes $d$ powers with exponents from an interval of length $\gamma$ from uniform group elements with advantage $\bigO((T+d)^2/\gamma)$ when the smallest prime factor of the group order is large. Since we can search every polynomially bounded interval in polynomial time, the interval threshold in the bounded-exponent assumption that we use depends on the size bound and the target advantage of the adversary.

\paperpart{A Lower Bound on the Image Size}{elf-birthday-bound}
Since evaluation of an \ELF{} is public, an observer can evaluate the function on independent uniform inputs and test whether two distinct inputs collide, so the image parameter cannot be arbitrarily small. We make quantitative the observation of~\cite{C:Zhandry16} that $\bigO(\sqrt r)$ evaluations detect an image of size $r$. \cite{Fel1968a} treats uniform outputs, \cite{Mun1977,Ber1980} prove that a nonuniform distribution on $r$ points has collision probability at least that of the uniform distribution on $r$ points, and \cite{EC:BelKoh04} gives the corresponding analysis for hash functions. With $k$ evaluations, the collision probability in an image of size $r$ is of order $k^2/r$, whereas distinct uniform inputs from $[m]$ collide with probability $\bigO(k^2/m)$, and from these two bounds we determine the smallest image size that an \ELF{} can hide, independently of the algebraic assumption used to construct it.

\ifConferenceVersion
When we invert the birthday bound for a fixed pair of functions, we obtain that every function on $[m]$ with an image of size $r$ that $k$ evaluations distinguish from an injective function with advantage at most $1/\delta$ satisfies $r\ge\delta k^2/12$, provided that $\delta\ge3$ and $m\ge\delta k(k-1)$. In an \ELF{}~\cite{C:Zhandry16}, we compare two generation algorithms, and since the collision distinguisher depends only on $m$ and $k$, we can run it on a sampled public description, average over the generation coins, and absorb the negligible failures of the two modes. We then solve the birthday inequality and obtain that the threshold $Q$ of every \ELF{} satisfies $Q(n)=\Omega(\delta(n)k(n)^2)$, where $k(n)$ is the number of evaluations that a circuit of the permitted size can perform.
\else
\phantomsection\label{app:collision-bounds}
\fi
\begin{relocated}{collision-bounds}
Our starting point is the birthday bound of~\cite{Fel1968a} for uniform outputs and its extension to nonuniform outputs~\cite{Mun1977,Ber1980}. We adapt it to a comparison with an injective oracle on a domain of size $m$, which adds the term $k(k-1)/(2m)$ for repeated inputs, and we use a distinguisher that depends only on $m$ and $k$, so that the same algorithm applies to every oracle, including a sampled function.

\begin{theoremE}[Birthday Collision Bound]\label{thm:elf-birthday-lower-bound}
	Let $m,k\in\NN$ with $m\ge1$ and $k\ge2$, let $\calY$ be a set, and let $\advA_{m,k}$ be the oracle algorithm that samples $(X_1,\dots,X_k)\sample[m]^k$, queries its oracle at $X_1,\dots,X_k$, receives the answers $Y_1,\dots,Y_k$, and returns $1$ if and only if $X_1,\dots,X_k$ are pairwise distinct and $Y_1,\dots,Y_k$ are not.
	\begin{parts}
		\item\label{thm:elf-birthday-lower-bound:a} Every injective function $\finj\colon[m]\to\calY$ satisfies
		\[
			\Prob\left[\advA_{m,k}^{\finj}=1\right]=0.
		\]
		\item\label{thm:elf-birthday-lower-bound:b} Every $r\in\NN_{>0}$ and every function $f\colon[m]\to\calY$ with $\abs{\Image(f)}\le r$ satisfy
		\[
			\Prob\left[\advA_{m,k}^f=1\right]
			\ge 1-\exp\left(-\frac{k(k-1)}{2r}\right)-\frac{k(k-1)}{2m}.
		\]
	\end{parts}
\end{theoremE}

\begin{proofE}
	We bound the probability that the outputs collide from below by a comparison with the uniform distribution on the image, and we bound the probability of repeated inputs by a union bound.

	We write $\Ein\cqq\{\exists i<j:X_i=X_j\}$ and $\Eout\cqq\{\exists i<j:Y_i=Y_j\}$, so that $\advA_{m,k}$ returns $1$ exactly on $\Ecoll\cqq\Eout\setminus\Ein$. If the oracle is injective, then $\Eout\subseteq\Ein$, so $\Ecoll=\emptyset$, and \cref{thm:elf-birthday-lower-bound:a} of \cref{thm:elf-birthday-lower-bound} holds.

	For \cref{thm:elf-birthday-lower-bound:b} of \cref{thm:elf-birthday-lower-bound}, we write $S\cqq\Image(f)$, $s\cqq\abs S\le r$, $\Edistinct\cqq\neg\Eout$, and
	\[
		p_y\cqq\Prob\left[f(X_1)=y\right]
	\]
	for every $y\in S$. If $s<k$, then two of the $k$ outputs coincide, so $\Edistinct=\emptyset$. We suppose that $s\ge k$ and write $e_k$ for the $k$th elementary symmetric polynomial. Since $X_1,\dots,X_k$ are independent and the $p_y$ sum to one, we obtain from Maclaurin's inequality~\cite{Mac1729,Ber1980}
	\[
		\Prob[\Edistinct]
		=k!\,e_k\bigl((p_y)_{y\in S}\bigr)
		\le k!\binom{s}{k}\left(\frac{1}{s}\sum_{y\in S}p_y\right)^k
		=k!\binom{s}{k}s^{-k}.
	\]
	Since $s\ge k$ and $r\ge s$, every factor $1-i/s$ and $1-i/r$ with $0\le i<k$ is nonnegative. We use $s\le r$ and $1-z\le e^{-z}$ for every real $z$, and obtain
	\begin{align*}
		k!\binom{s}{k}s^{-k}
		 & =\prod_{i=0}^{k-1}\left(1-\frac{i}{s}\right)
		\le\prod_{i=0}^{k-1}\left(1-\frac{i}{r}\right) \\
		 & \le\exp\left(-\frac{1}{r}\sum_{i=0}^{k-1}i\right)
		=\exp\left(-\frac{k(k-1)}{2r}\right).
	\end{align*}
	The resulting bound also holds when $s<k$, since then $\Edistinct=\emptyset$. By a union bound over the $\binom{k}{2}$ unordered pairs of inputs,
	\[
		\Prob[\Ein]
		\le\sum_{1\le i<j\le k}\Prob\left[X_i=X_j\right]
		=\binom{k}{2}\frac1m
		=\frac{k(k-1)}{2m}.
	\]
	Finally, $\Eout=\Ecoll\sqcup(\Eout\cap\Ein)$ and $\Eout\cap\Ein\subseteq\Ein$, so
	\begin{align*}
		\Prob\left[\advA_{m,k}^f=1\right]
		 & =\Prob[\Ecoll] \\*
		 & =\Prob[\Eout]-\Prob[\Eout\cap\Ein] \\*
		 & \ge 1-\Prob[\Edistinct]-\Prob[\Ein]                \\*
		 & \ge 1-\exp\left(-\frac{k(k-1)}{2r}\right)-\frac{k(k-1)}{2m}.
	\end{align*}
	Hence, \cref{thm:elf-birthday-lower-bound:b} of \cref{thm:elf-birthday-lower-bound} holds.
\end{proofE}

We invert \cref{thm:elf-birthday-lower-bound} and obtain a lower bound on the image size of every function that $k$ evaluations do not distinguish from an injective function.

\begin{corollaryE}[Quadratic Lower Bound on the Image Size]\label{cor:elf-quadratic-image-size}
	Let $m,r,k\in\NN$ and $\delta\in\RR$ with $m,r\ge1$, $k\ge2$, $\delta\ge3$, and $m\ge\delta k(k-1)$, let $f\colon[m]\to\calY$ satisfy $\abs{\Image(f)}\le r$, and let $\finj\colon [m] \to \calY$ be injective. Suppose that every oracle algorithm $\advA$ that makes at most $k$ oracle queries satisfies
	\[
		\abs{\Prob\left[\advA^f=1\right]-\Prob\left[\advA^{\finj}=1\right]}\le\frac1\delta.
	\]
	Then
	\[
		r \ge \frac{k(k-1)}{2\ln\left(1/\left(1-3/(2\delta)\right)\right)}
		\ge \frac{(2\delta-3)k(k-1)}{6}
		\ge \frac{\delta k^2}{12}.
	\]
\end{corollaryE}
\begin{proofE}
	We apply the hypothesis to $\advA_{m,k}$ of \cref{thm:elf-birthday-lower-bound}, solve the resulting inequality for $r$, and bound the logarithm.

	The algorithm $\advA_{m,k}$ makes $k$ oracle queries, so by \cref{thm:elf-birthday-lower-bound:a,thm:elf-birthday-lower-bound:b} of \cref{thm:elf-birthday-lower-bound} and the hypothesis, $1-\exp(-k(k-1)/(2r))-k(k-1)/(2m)\le1/\delta$. Since $m\ge\delta k(k-1)$, we have $k(k-1)/(2m)\le1/(2\delta)$, so
	\[
		1-\exp\left(-\frac{k(k-1)}{2r}\right)
		\le\frac{1}{\delta}+\frac{k(k-1)}{2m}
		\le\frac{3}{2\delta}.
	\]
	We write $z\cqq3/(2\delta)$. Since $\delta\ge3$, we have $0<z\le1/2$, and we take logarithms in $\exp(-k(k-1)/(2r))\ge1-z$ and obtain
	\[
		r\ge\frac{k(k-1)}{2\ln(1/(1-z))}.
	\]
	Since $\ln(1+x)\le x$ for every $x>-1$,
	\[
		\ln\left(\frac{1}{1-z}\right)
		=\ln\left(1+\frac{z}{1-z}\right)
		\le\frac{z}{1-z}
		=\frac{3}{2\delta-3},
	\]
	and we obtain $r\ge(2\delta-3)k(k-1)/6$. Since $2\delta-3\ge\delta$ and $k(k-1)\ge k^2/2$, we also obtain $r\ge\delta k^2/12$.
\end{proofE}

\end{relocated}%

\begin{relocated}{elf-threshold-corollary}
\Cref{cor:elf-quadratic-image-size} concerns one pair of functions, whereas the mode indistinguishability of an \ELF{} compares two generation algorithms. Since the algorithm $\advA_{m,k}$ of \cref{thm:elf-birthday-lower-bound} depends only on $m$ and $k$, we can run it on a sampled public description and transfer the bound to the threshold $Q$ of \cref{def:extremely-lossy-function}.

\begin{propositionE}[Lower Bound on the \ELF{} Threshold]\label{cor:elf-threshold-lower-bound}
	Let $\ELF=(\ELFGen,\ELFEval)$ satisfy \cref{def:extremely-lossy-function}, and for every $m\ge2$, write $n_m\cqq\clog{m}$. Let $t$ be a polynomial and $\ceval$ a positive polynomial such that, for every $m\ge2$ and every $k\in\NN$ with $2\le k\le t(n_m)$, the algorithm $\advA_{m,k}$ of \cref{thm:elf-birthday-lower-bound}, with its oracle queries replaced by evaluations of a public description of $\ELF$, is a probabilistic circuit of size at most $k\,\ceval(n_m)$. Let $\delta$ be a positive polynomial with $\delta(n)\ge3$ for every $n$, let $Q$ be a polynomial witnessing mode indistinguishability for $t$ and $\delta$ with $Q(n)\in\NN_{>0}$ for every sufficiently large $n$, and define
	\[
		k(n)\cqq\min\left\{\floor{t(n)},\left\lfloor\frac{t(n)}{\ceval(n)}\right\rfloor\right\}.
	\]
	Then $Q(n_m)\ge\delta(n_m)k(n_m)^2/24$ for every sufficiently large $m$ with $k(n_m)\ge2$.
\end{propositionE}
\begin{proofE}
	Since $\advA_{m,k}$ depends only on $m$ and $k$, we can run it on a sampled public description and average its advantage over the generation coins. Since the two failure events of the modes are negligible, we absorb them into the target advantage, and we then repeat the algebra of the proof of \cref{cor:elf-quadratic-image-size} for $r=Q(n_m)$.

	For every $m\ge2$ with $Q(n_m)\in[m]$, we write $\epsinj(m)$ for the probability that $f\sample\ELFGen(m,m)$ is not injective and $\epslos(m)$ for the probability that $f\sample\ELFGen(m,Q(n_m))$ satisfies $\abs{\Image(f)}>Q(n_m)$. We choose $m_0$ such that every $m\ge m_0$ satisfies the following conditions.
	\begin{enumerate}[label=(\roman*)]
		\item\label{item:elf-threshold-indist} The mode-indistinguishability clause of \cref{def:extremely-lossy-function} holds at $m$, which is the case from some $m$ on by the asymptotic quantifier of that clause.
		\item\label{item:elf-threshold-positive} $Q(n_m)\in\NN_{>0}$, which is the case from some $m$ on because $n_m$ grows without bound.
		\item\label{item:elf-threshold-small} $Q(n_m)\le m$, which is the case from some $m$ on because $m>2^{n_m-1}$ and $Q(n_m)$ is polynomial in $n_m$.
		\item\label{item:elf-threshold-failure} $\epslos(m)+\epsinj(m)\le1/(2\delta(n_m))$, which is the case from some $m$ on because $\epslos(m)$ and $\epsinj(m)$ are negligible in $n_m$ by the injective-mode and lossy-mode clauses of \cref{def:extremely-lossy-function}, while $\delta$ is polynomial.
		\item\label{item:elf-threshold-domain} $k(n_m)(k(n_m)-1)/(2m)\le1/(2\delta(n_m))$, which is the case from some $m$ on because $k$ and $\delta$ are polynomially bounded and $m>2^{n_m-1}$.
	\end{enumerate}
	We fix $m\ge m_0$ with $k(n_m)\ge2$, abbreviate $n\cqq n_m$ and $k\cqq k(n)$, and set $r\cqq Q(n)$. By \cref{item:elf-threshold-positive,item:elf-threshold-small}, we have $r\in[Q(n)..m]\subseteq[m]$, so $\ELFGen(m,r)$ is defined, and by \cref{item:elf-threshold-indist}, mode indistinguishability holds at the image parameter $r$. We take $\advA$ to be $\advA_{m,k}$ with its oracle queries replaced by evaluations of the public description. Since $k$ is the smaller of two floors, $k$ is an integer with $k\le t(n)$ and $k\,\ceval(n)\le t(n)$. By the first inequality, $k$ lies in the range where the cost hypothesis applies, and by the second, $\advA$ is a probabilistic circuit of size at most $t(n)$, so mode indistinguishability holds for $\advA$.

	We write $B\cqq1-\exp(-k(k-1)/(2r))$. By \cref{thm:elf-birthday-lower-bound:a,thm:elf-birthday-lower-bound:b} of \cref{thm:elf-birthday-lower-bound}, every $f$ with $\abs{\Image(f)}\le r$ and every injective $f'$ satisfy
	\[
		\Prob\left[\advA(f)=1\right]\ge B-\frac{k(k-1)}{2m},
		\qquad
		\Prob\left[\advA(f')=1\right]=0.
	\]
	After we average over the generation coins, we obtain
	\begin{align*}
		\Prob\left[\advA(f)=1\mbar f\sample\ELFGen(m,r)\right] & \ge B-\frac{k(k-1)}{2m}-\epslos(m), \\
		\Prob\left[\advA(f)=1\mbar f\sample\ELFGen(m,m)\right] & \le\epsinj(m).
	\end{align*}
	By mode indistinguishability, the difference of these two probabilities is at most $1/\delta(n)$. Therefore, by \cref{item:elf-threshold-failure,item:elf-threshold-domain},
	\[
		B\le\frac1{\delta(n)}+\frac{k(k-1)}{2m}+\epslos(m)+\epsinj(m)\le\frac2{\delta(n)}.
	\]
	We write $z\cqq2/\delta(n)$, so that $0<z\le2/3$ because $\delta(n)\ge3$. As in the proof of \cref{cor:elf-quadratic-image-size}, the inequality $\exp(-k(k-1)/(2r))\ge1-z$ implies $r\ge k(k-1)/(2\ln(1/(1-z)))$, and $\ln(1/(1-z))\le z/(1-z)=2/(\delta(n)-2)$, so $r\ge(\delta(n)-2)k(k-1)/4$. Finally, $\delta(n)-2\ge\delta(n)/3$, and since $k\ge2$, we have $k(k-1)\ge k^2/2$, so $r\ge\delta(n)k^2/24$. Since $r=Q(n)$, the claim holds.
\end{proofE}

\end{relocated}%

\ifConferenceVersion\else
\subsection{Small-Exponent Assumptions with a Fixed Bound}\label{app:fixed-bound-dsedl}

By \cref{cor:elf-quadratic-image-size}, an image that $k$ evaluations distinguish from an injective one with advantage at most $1/\delta$ has size at least $\delta k^2/12$, so the lossy mode needs bounded exponents from an interval whose length is polynomial and grows with the resources of the observer. We first state the bounded-exponent assumptions for a fixed interval and prove that every polynomially bounded fixed interval violates them.

\cite{EC:vOoWie96} studies short exponents in Diffie--Hellman key agreement for faster exponentiation, and \cite{C:PatSun98} uses them for pseudorandom generation. With \BSGS{}~\cite{Sha1971} or the kangaroo method~\cite{Pol1978,JC:Pollard00}, we recover an exponent from an interval of length $\gamma$ with $\bigO(\sqrt\gamma)$ group operations, and \cite{Pol1978} gives the rho method, a generic collision algorithm over the whole exponent space. \cite{EC:vOoWie96,C:PatSun98} take the interval large enough for generic security and also study bounded exponents modulo a composite, where \BSGS{}~\cite{Sha1971} and the kangaroo method~\cite{Pol1978,JC:Pollard00} recover bounded exponents in the same number of group operations. We are not aware of an attack on the subgroup of \cref{def:dj-group} that is faster than these square-root algorithms.

We sample bounded exponents from $[\gamma]=\{1,\dots,\gamma\}$. The interval $[0..\gamma-1]$ yields equivalent search and decision problems, since multiplication by the generator $g$ maps $g^u$ to $g^{u+1}$ and multiplication by $g^{-1}$ inverts this map. The order of the group of \cref{def:dj-group} is hidden, so we state every assumption of this subsection for an arbitrary cyclic group sampler $\CGen(\cdot)$ in the sense of \cref{def:group-sampler}, with a uniform group element as the alternative to a bounded power. Since $g$ generates $\GG(\pp)$, a uniform group element is distributed as $g^u$ for $u\sample\ZZ_{\ord(\pp)}$, so this alternative agrees with a uniform exponent whenever $\ord(\pp)$ is public.

The following search assumption adapts the short-exponent discrete-logarithm problem of~\cite{EC:vOoWie96,C:PatSun98} to a cyclic group sampler with a possibly hidden order and to exponents from $[\gamma(\secparam)]$.

\begin{definition}[Bounded \SEDL{}]\label{def:small-exponent-dlog}
	Let $\gamma\colon\NN\to\NN_{>0}$ satisfy $\gamma(\secparam)<\ord(\pp)$ for every sufficiently large $\secparam$ and every $\pp\in\Supp(\CGen(1^\secparam))$. For every adversary $\advA$, define
	\[
		\advantage{\SEDLgamma}{\CGen(\cdot),\advA}(\secparam)\cqq
		\Prob\left[
			\shat=s
			\mbar
			\begin{array}{l}
				\pp\sample\CGen(1^\secparam), \\
				s\sample[\gamma(\secparam)],        \\
				\shat\sample\advA(1^\secparam,\pp,\gamma(\secparam),g^s)
			\end{array}
			\right].
	\]
	The $\gamma$-bounded small-exponent discrete-logarithm assumption $\SEDLgamma$ holds for $\CGen(\cdot)$ if, for every \PPT{} adversary $\advA$, the function $\advantage{\SEDLgamma}{\CGen(\cdot),\advA}$ is negligible.
\end{definition}

Since $\gamma(\secparam)<\ord(\pp)$, the map $s\mapsto g^s$ is injective on $[\gamma(\secparam)]$, so $g^s$ determines $s$. In the decisional version, the adversary must distinguish a bounded power from a uniform group element, and a distinguisher need not recover the exponent. Decisional short-exponent problems appear as the $c$-decisional short discrete logarithm problem of~\cite{EC:AFLT12}, in which the exponent has $c$ bits, and we state ours for the interval $[\gamma]$ and a uniform group element as the alternative.

\begin{definition}[Bounded \DSEDL{}]\label{def:decisional-small-exponent-dlog}
	For every $\secparam\in\NN$ and every positive integer $\gamma$ smaller than the order of every group in $\Supp(\CGen(1^\secparam))$, define
	\begin{align*}
		\DDSEDL_{0,\secparam,\gamma} & \cqq
		\left\{(1^\secparam,\pp,\gamma,g^s)\ \middle|\
		\begin{array}{l}
			\pp\sample\CGen(1^\secparam), \\[-1pt]
			s\sample[\gamma]
		\end{array}\right\},               \\
		\DDSEDL_{1,\secparam,\gamma} & \cqq
		\left\{(1^\secparam,\pp,\gamma,G)\ \middle|\
		\begin{array}{l}
			\pp\sample\CGen(1^\secparam), \\[-1pt]
			G\sample\GG(\pp)
		\end{array}\right\}.
	\end{align*}
	Let $\gamma\colon\NN\to\NN_{>0}$ satisfy $\gamma(\secparam)<\ord(\pp)$ for every sufficiently large $\secparam$ and every $\pp\in\Supp(\CGen(1^\secparam))$. The $\gamma$-bounded decisional small-exponent discrete-logarithm assumption $\DSEDLgamma$ holds for $\CGen(\cdot)$ if
	\[
		\left(\DDSEDL_{0,\secparam,\gamma(\secparam)}\right)_{\secparam\in\NN}\cindist\left(\DDSEDL_{1,\secparam,\gamma(\secparam)}\right)_{\secparam\in\NN}.
	\]
\end{definition}

A recovered exponent is a witness that we can verify by exponentiation, so $\DSEDLgamma$ implies $\SEDLgamma$ whenever $[\gamma(\secparam)]$ has negligible density in the group. The implication from $\SEDLgamma$ to $\DSEDLgamma$ would require a search-to-decision reduction.

\begin{theoremE}[\DSEDL{} Implies \SEDL{}]\label{thm:dsedl-implies-sedl}
	Let $\gamma\colon\NN\to\NN_{>0}$, and suppose that there is a negligible function $\varepsilon_0$ such that $\gamma(\secparam)/\ord(\pp)\le\varepsilon_0(\secparam)$ for every sufficiently large $\secparam$ and every $\pp\in\Supp(\CGen(1^\secparam))$. If $\DSEDLgamma$ holds for $\CGen(\cdot)$, then $\SEDLgamma$ holds for $\CGen(\cdot)$.
\end{theoremE}

\begin{proofE}
	We verify a recovered exponent by exponentiation, and on a uniform challenge, the check passes only if the challenge is a power with exponent in $[\gamma(\secparam)]$, which has probability at most $\varepsilon_0(\secparam)$.

	We prove the contrapositive and suppose that $\SEDLgamma$ does not hold, so that there is a \PPT{} adversary $\advA$ whose success function $\alpha(\secparam)\cqq\advantage{\SEDLgamma}{\CGen(\cdot),\advA}(\secparam)$ is not negligible. We construct a \PPT{} distinguisher $\advB$ that returns $0$ on the failure symbol. On input $(1^\secparam,\pp,\gamma(\secparam),h)\ne\bot$, it samples $\shat\sample\advA(1^\secparam,\pp,\gamma(\secparam),h)$ and returns $\ind{\shat\in[\gamma(\secparam)]\land g^{\shat}=h}$.

	If $h=g^s$ for $s\sample[\gamma(\secparam)]$, then $\advB$ accepts whenever $\advA$ returns $s$, hence
	\[
		\Prob\left[
			\advB(1^\secparam,\pp,\gamma(\secparam),g^s)=1
			\mbar
			\begin{array}{l}
				\pp\sample\CGen(1^\secparam), \\[-1pt]
				s\sample[\gamma(\secparam)]
			\end{array}
			\right]\ge\alpha(\secparam).
	\]
	For a uniform group element $G\sample\GG(\pp)$, acceptance implies that $G=g^s$ for one of the $\gamma(\secparam)$ integers $s\in[\gamma(\secparam)]$. After we condition on parameters $\pp\in\Supp(\CGen(1^\secparam))$, the acceptance probability is at most $\gamma(\secparam)/\ord(\pp)$, and we average over the sampler and obtain, for every sufficiently large $\secparam$,
	\[
		\Prob\left[
			\advB(1^\secparam,\pp,\gamma(\secparam),G)=1
			\mbar
			\begin{array}{l}
				\pp\sample\CGen(1^\secparam), \\[-1pt]
				G\sample\GG(\pp)
			\end{array}
			\right]\le\varepsilon_0(\secparam).
	\]
	Hence, the distinguishing advantage of $\advB$ is at least $\alpha(\secparam)-\varepsilon_0(\secparam)$. Since $\alpha$ is not negligible, there are a positive polynomial $p$ and an infinite set $I\subseteq\NN$ with $\alpha(\secparam)\ge1/p(\secparam)$ for every $\secparam\in I$, and since $\varepsilon_0$ is negligible, $\varepsilon_0(\secparam)\le1/(2p(\secparam))$ for every sufficiently large $\secparam$. Hence, $\advB$ has advantage at least $1/(2p(\secparam))$ for infinitely many $\secparam$, which is impossible under $\DSEDLgamma$, so $\SEDLgamma$ holds.
\end{proofE}

For a polynomially bounded $\gamma$, the decisional assumption is false. We write $v\cqq\ceil{\sqrt{\gamma(\secparam)}}$, and with \BSGS{}~\cite{Sha1971} restricted to $[\gamma(\secparam)]$, we compute $g^j$ for $j\in[0..v]$ and $hg^{-iv}$ for $i\in[v]$ from a challenge $h\ne\bot$, and we accept if and only if $hg^{-iv}=g^j$ for some $i\in[0..v]$ and $j\in[0..v-1]$ with $iv+j\in[\gamma(\secparam)]$. Every $s\in[\gamma(\secparam)]$ satisfies $s\le v^2$, so by Euclidean division, $s=iv+j$ with $i\in[0..v]$ and $j\in[0..v-1]$, and we accept every bounded power. Conversely, after we condition on $\pp\in\Supp(\CGen(1^\secparam))$, we accept a uniform $G\sample\GG(\pp)$ only if $G\in\{g^s\mid s\in[\gamma(\secparam)]\}$, which has probability $\gamma(\secparam)/\ord(\pp)\le\gamma(\secparam)2^{-\secparam}$ by \cref{def:group-sampler}. Hence, we distinguish the two ensembles of \cref{def:decisional-small-exponent-dlog} with $\bigO(\sqrt{\gamma(\secparam)})$ group operations and advantage at least $1-\gamma(\secparam)2^{-\secparam}$ minus the negligible failure probability of $\CGen(1^\secparam)$, so $\DSEDLgamma$ is false for every polynomially bounded $\gamma$ against \PPT{} adversaries.
\fi

\ifConferenceVersion
\paperpart{Small-Exponent Assumptions}{subsec:small-exponent-assumptions}
Bounded \DSEDL{} for a fixed interval family $\gamma(\secparam)$ is the assumption that $g^s$ for $s\sample[\gamma(\secparam)]$ is indistinguishable from a uniform group element. With \BSGS{}~\cite{Sha1971}, we can search every polynomially bounded interval with $\bigO(\sqrt{\gamma(\secparam)})$ group operations, so no fixed polynomial interval satisfies this assumption against all polynomial-size adversaries. Therefore, for the extremely lossy construction, we use the quantifier order of~\cite{C:Zhandry16}: we first fix the circuit-size bound and the target advantage of the adversary and then choose a sufficiently large polynomial interval.
\else
\paperpart{Tunable Small-Exponent Assumptions}{subsec:small-exponent-assumptions}
Since $\DSEDLgamma$ is false for every polynomially bounded $\gamma$, we use the quantifier order of~\cite{C:Zhandry16} for the extremely lossy construction: we first fix the circuit-size bound and the target advantage of the adversary and then choose a sufficiently large polynomial interval.
\fi
Hence, in tunable \DSEDL{}, the interval threshold $Q$ depends on the size bound $p$ and the target advantage $1/\delta$. The definition follows the short-exponent assumptions of~\cite{EC:vOoWie96,C:PatSun98}, except that the alternative is a uniform group element and the interval threshold is chosen after the size bound and the target advantage, as in~\cite{C:Zhandry16}. The exponent bound $\gamma$ is part of both distributions, so the adversary knows it.

\begin{definition}[Tunable \DSEDL{}]\label{def:tunable-dsedl}
	\ifConferenceVersion
	For every $\secparam\in\NN$ and every positive integer $\gamma$ smaller than the order of every group in $\Supp(\CGen(1^\secparam))$, define
	\begin{align*}
		\DDSEDL_{0,\secparam,\gamma} & \cqq
		\left\{(1^\secparam,\pp,\gamma,g^s)\middle|
		\begin{array}{l}
			\pp\sample\CGen(1^\secparam), \\[-1pt]
			s\sample[\gamma]
		\end{array}\right\},               \distpairsep
		\DDSEDL_{1,\secparam,\gamma} & \cqq
		\left\{(1^\secparam,\pp,\gamma,G)\middle|
		\begin{array}{l}
			\pp\sample\CGen(1^\secparam), \\[-1pt]
			G\sample\GG(\pp)
		\end{array}\right\}.
	\end{align*}
	We say
	\else
	Let $\DDSEDL_{0,\secparam,\gamma}$ and $\DDSEDL_{1,\secparam,\gamma}$ be the distributions of \cref{def:decisional-small-exponent-dlog}. We say
	\fi
	that $\CGen(\cdot)$ satisfies tunable decisional small-exponent discrete logarithm if, for every polynomial $p$ and every positive polynomial $\delta$, there exist a polynomial $Q$ and $\secparam_0\in\NN$ such that $\DDSEDL_{0,\secparam,\gamma}\indist_{p,1/\delta}\DDSEDL_{1,\secparam,\gamma}$ for every $\secparam\ge\secparam_0$ and every $\gamma\in\NN$ satisfying $Q(\secparam)\le\gamma<\ord(\pp)$ for all $\pp\in\Supp(\CGen(1^\secparam))$.
\end{definition}

\ifConferenceVersion\else
We can sample both distributions of \cref{def:tunable-dsedl} from public parameters, so we can check a claimed violation of tunable \DSEDL{} at given $\secparam$ and $\gamma$ by sampling. In the sense of~\cite{C:Naor03,STOC:GenWic11}, an assumption is falsifiable when an efficient algorithm can check a claimed violation of it, and in the following lemma, we give such an algorithm for tunable \DSEDL{}.

\begin{lemmaE}[Falsifiability of Tunable \DSEDL{}]\label{lem:dsedl-falsifiable}
	Let $\CGen(\cdot)$ be a cyclic group sampler, let $\secparam\in\NN$, and let $\gamma$ be a positive integer smaller than the order of every group in $\Supp(\CGen(1^\secparam))$. For a probabilistic circuit $C$ and $b\in\bits$, write
	\[
		q_b(C)\cqq\Prob\left[C(X)=1\mbar X\sample\DDSEDL_{b,\secparam,\gamma}\right].
	\]
	There is a probabilistic algorithm $\advV$ that, on input $(1^\secparam,\gamma,1^d,C)$ with $d\in\NN_{>0}$, runs in time polynomial in $\secparam$, $d$, and the size of $C$, returns a bit, and satisfies the following.
	\begin{parts}
		\item\label{lem:dsedl-falsifiable:a} If $\abs{q_0(C)-q_1(C)}\ge2/d$, then $\advV$ returns $1$ with probability at least $1-2^{-\secparam}$.
		\item\label{lem:dsedl-falsifiable:b} If $\abs{q_0(C)-q_1(C)}\le1/d$, then $\advV$ returns $1$ with probability at most $2^{-\secparam}$.
	\end{parts}
\end{lemmaE}

\begin{proofE}
	\label{proof:dsedl-falsifiable}
	We estimate both acceptance probabilities empirically and compare the difference of the estimates with $3/(2d)$, the midpoint between the two cases.

	We write $k\cqq8d^2(\secparam+2)$. For each $b\in\bits$, the algorithm $\advV$ draws $k$ independent samples $X_{b,1},\dots,X_{b,k}$ of $\DDSEDL_{b,\secparam,\gamma}$: it runs $\CGen(1^\secparam)$ and then either samples $s\sample[\gamma]$ and computes $g^s$ when $b=0$, or calls $\SampG(\pp)$ when $b=1$, with the propagation convention of \cref{app:cyclic-group-samplers} on failures. It runs $C$ on every sample with fresh coins, computes the empirical frequencies
	\[
		\qhat_b\cqq\frac1k\sum_{i=1}^kC(X_{b,i}),
	\]
	and returns $1$ if and only if $\abs{\qhat_0-\qhat_1}>3/(2d)$. By \cref{def:group-sampler}, every group in $\Supp(\CGen(1^\secparam))$ has at most $2^{\grplen(\secparam)+1}$ elements, so $\gamma$ has at most $\grplen(\secparam)+1$ bits, and each of the $2k=16d^2(\secparam+2)$ samples costs time $\poly(\secparam)$. With the $2k$ evaluations of $C$, the running time is polynomial in $\secparam$, $d$, and the size of $C$.

	For each $b$, the samples and the coins of $C$ are independent, so $C(X_{b,1}),\dots,C(X_{b,k})$ are independent bits with mean $q_b(C)$. By Hoeffding's inequality~\cite{Hoe1963},
	\[
		\Prob\left[\abs{\qhat_b-q_b(C)}\ge\frac1{4d}\right]
		\le2\exp\left(-\frac{2k}{16d^2}\right)
		=2e^{-(\secparam+2)}
		\le2^{-\secparam-1}.
	\]
	By a union bound over $b\in\bits$, the inequality $\abs{\qhat_b-q_b(C)}<1/(4d)$ holds for both $b$ with probability at least $1-2^{-\secparam}$. On this event, we obtain from the triangle inequality
	\[
		\abs{\abs{\qhat_0-\qhat_1}-\abs{q_0(C)-q_1(C)}}<\frac1{2d}.
	\]
	If $\abs{q_0(C)-q_1(C)}\ge2/d$, then $\abs{\qhat_0-\qhat_1}>2/d-1/(2d)=3/(2d)$ on this event, and $\advV$ returns $1$. If $\abs{q_0(C)-q_1(C)}\le1/d$, then $\abs{\qhat_0-\qhat_1}<1/d+1/(2d)=3/(2d)$ on this event, and $\advV$ returns $0$.
\end{proofE}

A violation of \cref{def:tunable-dsedl} consists of polynomials $p$ and $\delta$ such that, for every polynomial $Q$, infinitely many $\secparam$ admit an interval bound $\gamma\ge Q(\secparam)$ below the group order and a circuit of size at most $p(\secparam)$ whose advantage exceeds $1/\delta(\secparam)$. With $d\cqq\ceil{2\delta(\secparam)}$, this advantage is at least $2/d$, so by \cref{lem:dsedl-falsifiable:a} of \cref{lem:dsedl-falsifiable}, we accept each such instance in time polynomial in $\secparam$ with probability at least $1-2^{-\secparam}$, and by \cref{lem:dsedl-falsifiable:b} of \cref{lem:dsedl-falsifiable}, we accept a circuit of advantage at most $1/d$ with probability at most $2^{-\secparam}$.
\fi

Since our constructions publish $\dlam$ bounded powers at once, we need tunable \DSEDL{} for vectors, and the coordinate hybrid costs a factor $\dlam$ in the target advantage.

\begin{lemmaE}[Tunable Vector \DSEDL{}]\label{lem:tunable-scalar-to-vec-dsedl}
	Suppose that $\CGen(\cdot)$ satisfies tunable \DSEDL{}. Let $d\colon\NN\to\NN_{>0}$ be polynomially bounded, and write $\dlam\cqq d(\secparam)$. For every $\secparam\in\NN$ and every positive integer $\gamma$ smaller than the order of every group in $\Supp(\CGen(1^\secparam))$, define
	\begin{align*}
		\DvDSEDL_{0,\secparam,\gamma} & \cqq
		\left\{(1^\secparam,\pp,\gamma,(g^{s_j})_{j\in[\dlam]})\ \middle|\
		\begin{array}{l}
			\pp\sample\CGen(1^\secparam), \\[-1pt]
			(s_j)_{j\in[\dlam]}\sample[\gamma]^{\dlam}
		\end{array}\right\},      \\
		\DvDSEDL_{1,\secparam,\gamma} & \cqq
		\left\{(1^\secparam,\pp,\gamma,(G_j)_{j\in[\dlam]})\ \middle|\
		\begin{array}{l}
			\pp\sample\CGen(1^\secparam), \\[-1pt]
			(G_j)_{j\in[\dlam]}\sample\GG(\pp)^{\dlam}
		\end{array}\right\}.
	\end{align*}
	Then, for every polynomial $p$ and every positive polynomial $\delta$, there exist a polynomial $Q$ and $\secparam_0\in\NN$ such that $\DvDSEDL_{0,\secparam,\gamma}\indist_{p,1/\delta} \DvDSEDL_{1,\secparam,\gamma}$ for every $\secparam\ge\secparam_0$ and every $\gamma\in\NN$ satisfying $Q(\secparam)\le\gamma<\ord(\pp)$ for all $\pp\in\Supp(\CGen(1^\secparam))$.
\end{lemmaE}
\ConferenceProofSketch{proof:tunable-vector-dsedl}{
	We replace the coordinates one at a time by uniform group elements. For a uniformly chosen transition, we embed the scalar challenge, sample the other coordinates ourselves, and apply tunable \DSEDL{} with target advantage $1/(\dstar\delta)$ for a polynomial $\dstar\ge d$. The $\dlam$ differences telescope, so the vector advantage is at most $\dlam/(\dstar(\secparam)\delta(\secparam))\le1/\delta(\secparam)$.
}{}{}
\begin{proofE}
	\label{proof:tunable-vector-dsedl}
	We apply the scalar assumption with the target advantage multiplied by a polynomial upper bound on the vector dimension, and we obtain the vector bound from the telescoping identity over the $\dlam$ transitions of a coordinate hybrid.

	We fix $p$ and $\delta$, and we let $\pstar$ be a polynomial that bounds the size of a circuit that samples $\dlam-1$ exponents or uniform group elements, performs the corresponding exponentiations, runs a circuit of size at most $p(\secparam)$, and performs one equality test against the failure symbol and the corresponding branch. Since $d$ is polynomially bounded, there is a positive polynomial $\dstar$ with $\dstar(\secparam)\ge d(\secparam)$ for every $\secparam$, and we define the positive polynomial $\deltastar(\secparam)\cqq\dstar(\secparam)\delta(\secparam)$. We apply \cref{def:tunable-dsedl} with $\pstar$ and $\deltastar$ and obtain $Q$ and $\secparam_0$.

\begin{proofbox}{One coordinate transition in the tunable vector-\DSEDL{} hybrid.}{fig:tunable-vector-dsedl-coordinate-hybrid}
		\begin{minipage}[t]{0.48\linewidth}
			\vspace{0pt}
			{\footnotesize\underline{$\Hybrid_{k-1}$}\codelinesrightalign{coordinate $k$ bounded}\par}
			\smallskip
			\begin{codelines}
				\item $\pp\sample\CGen(1^\secparam)$
				\item \pcif $\pp=\bot$ \pcreturn $\bot$
				\item $g\gets g(\pp)$
				\item $(G_j)_{j<k}\sample\GG(\pp)^{k-1}$
				\item $s_k\sample[\gamma]$; $G_k\gets g^{s_k}$
				\item $(s_j)_{j>k}\sample[\gamma]^{\dlam-k}$
				\item $(G_j)_{j>k}\gets(g^{s_j})_{j>k}$; $\bvec G\gets(G_j)_{j\in[\dlam]}$
				\item \pcreturn $(1^\secparam,\pp,\gamma,\bvec G)$
			\end{codelines}
		\end{minipage}\hfill%
		\begin{minipage}[t]{0.48\linewidth}
			\vspace{0pt}
			{\footnotesize\underline{$\Hybrid_k$}\codelinesrightalign{coordinate $k$ uniform}\par}
			\smallskip
			\begin{codelines}
				\item $\pp\sample\CGen(1^\secparam)$
				\item \pcif $\pp=\bot$ \pcreturn $\bot$
				\item $g\gets g(\pp)$
				\item $(G_j)_{j<k}\sample\GG(\pp)^{k-1}$
				\item \hybchange{$G_k\sample\GG(\pp)$}
				\item $(s_j)_{j>k}\sample[\gamma]^{\dlam-k}$
				\item $(G_j)_{j>k}\gets(g^{s_j})_{j>k}$; $\bvec G\gets(G_j)_{j\in[\dlam]}$
				\item \pcreturn $(1^\secparam,\pp,\gamma,\bvec G)$
			\end{codelines}
		\end{minipage}
	\end{proofbox}

	\subparhead{Hybrids} We fix $\secparam\ge\secparam_0$ and $\gamma\in\NN$ with $Q(\secparam)\le\gamma<\ord(\pp)$ for every $\pp\in\Supp(\CGen(1^\secparam))$, and we let $\advA$ be a circuit of size at most $p(\secparam)$. For every $k\in[0..\dlam]$, we let $\Hybrid_k$ be the distribution with uniform group elements in its first $k$ coordinates and bounded powers in its remaining coordinates, as in \cref{fig:tunable-vector-dsedl-coordinate-hybrid}, and we define
	\[
		p_k\cqq\Prob\left[\advA(X)=1\mbar X\sample\Hybrid_k\right].
	\]
	Under the propagation convention of \cref{app:cyclic-group-samplers}, if a parameter or group-element sample fails, then the hybrid outputs $\bot$, and this outcome contributes zero to $p_k$. In particular, $p_0$ and $p_{\dlam}$ are the acceptance probabilities of $\advA$ on $\DvDSEDL_{0,\secparam,\gamma}$ and $\DvDSEDL_{1,\secparam,\gamma}$, respectively.

	\subparhead{The Reduction} We construct a scalar circuit $\advB$ of size at most $\pstar(\secparam)$. On input $\bot$, the circuit $\advB$ returns $0$. On a challenge other than $\bot$, it samples $L\sample[\dlam]$, places its scalar challenge in coordinate $L$, samples the coordinates $j<L$ as uniform group elements of $\GG(\pp)$ and the coordinates $j>L$ as bounded powers, runs $\advA$ on the resulting vector, and returns the output of $\advA$, and it returns $0$ if a further group-element sample fails. We test for $\bot$ and branch on failures because the circuits in \cref{def:tunable-dsedl} output a bit and because $\advB$ cannot sample the remaining coordinates without $\pp$. For $b\in\bits$, we write
	\[
		q_b\cqq\Prob\left[\advB(X)=1\mbar X\sample\DDSEDL_{b,\secparam,\gamma}\right].
	\]
	For each value $L=k$, the circuit $\advB$ and the hybrids $\Hybrid_{k-1}$ and $\Hybrid_k$ make the same parameter and group-element sampling calls and assign acceptance probability zero if a call fails, so on a bounded scalar challenge the input of $\advA$ is distributed as $\Hybrid_{k-1}$, and on a uniform one as $\Hybrid_k$. Hence,
	\[
		q_0-q_1
		=\frac1{\dlam}\sum_{k=1}^{\dlam}\left(p_{k-1}-p_k\right)
		=\frac{p_0-p_{\dlam}}{\dlam}.
	\]
	By the choice of $Q$ and $\secparam_0$ in \cref{def:tunable-dsedl}, we have $\abs{q_0-q_1}\le1/\deltastar(\secparam)$. We multiply the identity for $q_0-q_1$ by $\dlam$ and obtain
	\[
		\abs{p_0-p_{\dlam}}
		=\dlam\abs{q_0-q_1}
		\le\frac{d(\secparam)}{\deltastar(\secparam)}
		=\frac{d(\secparam)}{\dstar(\secparam)\delta(\secparam)}
		\le\frac{1}{\delta(\secparam)},
	\]
	as required.
\end{proofE}

\paperpart{Generic-Group Bounds for Bounded Powers}{subsec:generic-group-dsedl}
\ifConferenceVersion
In the generic group model~\cite{Nechaev94,EC:Shoup97}, we prove that every generic algorithm $\advA$ that makes at most $T$ group-oracle queries satisfies
\[
	\AdvGG_{q,\gamma,d}(\advA)
	\le\min\left\{1,\binom{T+d+1}{2}\left(\frac1\gamma+\frac2{\piq}\right)\right\},
\]
where $\AdvGG_{q,\gamma,d}(\advA)$ is the difference between the acceptance probabilities of $\advA$ on $d$ powers of a generator of a group of order $q$ with exponents uniform in $[\gamma]^d$ and with exponents uniform in $\ZZ_q^d$, and $\piq$ is the smallest prime factor of $q$. For $d=1$, with \BSGS{}~\cite{Sha1971} restricted to $[\gamma]$, we obtain a distinguisher that makes at most $2\ceil{\sqrt\gamma}+1$ queries and has advantage $1-\gamma/q$, and with a truncated version of this distinguisher, we obtain the same quadratic dependence on $T$ at every inverse-polynomial advantage. These bounds hold in the generic group model, also with encoding-dependent preprocessing~\cite{C:CorDodGuo18}, so concrete parameters for tunable \DSEDL{} remain open, and attacks that use the \DamgardJurik representation are outside their scope.

The factor $\dlam$ in the target advantage of \cref{lem:tunable-scalar-to-vec-dsedl} comes from the $\dlam$ transitions of the coordinate hybrid. For $T\ge d$, we have
\[
	\binom{T+d+1}{2}\le\binom{2T+1}{2}=T(2T+1)\le6\binom{T+2}{2},
\]
so the generic bound for $d$ challenge powers is at most six times the bound for one challenge power with the same $T$. Hence, a vector version of tunable \DSEDL{}, assumed directly, would remove the factor $\dlam$ from every interval threshold that we derive from \cref{lem:tunable-scalar-to-vec-dsedl}, and we take the scalar assumption as our starting point and keep this factor.
\else
\label{app:generic-group-evidence}
\fi
\begin{relocated}{generic-evidence}
Tunable \DSEDL{} is the only nonstandard decisional assumption in our constructions, and we bound its dependence on the resources of the adversary in the generic group model~\cite{Nechaev94,EC:Shoup97}. We must allow a composite order, since the group of \cref{def:dj-group} has composite order and its exponent ring is not a field. Therefore, we carry out the analysis for an arbitrary order $q$ and write $\piq$ for its smallest prime factor, so that $\piq=q$ for prime $q$. A nonzero coefficient in $\ZZ_q$ need not be invertible, so a nonzero affine form can vanish at many points, but after we fix all variables except one with a nonzero coefficient, the solutions of the resulting univariate congruence form an arithmetic progression with common difference at least $\piq$.

\ifConferenceVersion
For finite sets $X$ and $Y$ of equal cardinality, we write $\psi(X,Y)$ for the set of bijections from $X$ to $Y$. We let $\calL$ be a set of cardinality $q$, sample $\xi\sample\psi(\ZZ_q,\calL)$, and define the generic group oracle by $\Oxi(\xi(a),\xi(b),\varepsilon)\cqq\xi\left(a+(-1)^\varepsilon b\right)$ for $a,b\in\ZZ_q$ and $\varepsilon\in\bits$. A generic algorithm queries $\Oxi$ only on labels that it received as input or as answers of $\Oxi$, and its nonuniform advice is a function of $q$, $\gamma$, and $d$ and is independent of $\xi$. The bounded-power and uniform-power experiments are those of \cref{fig:generic-dsedl-games}.

\begin{proofbox}{Generic bounded-power and uniform-power experiments.}{fig:generic-dsedl-games}
		\begin{minipage}[t]{0.48\linewidth}
			\vspace{0pt}
			{\footnotesize\underline{$\bndexp^{\advA}_{q,\gamma,d}$}\par}
			\smallskip
			\begin{codelines}
				\item \label{line:gg-bounded-exponents}$(x_j)_{j\in[d]}\sample[\gamma]^d$
				\item \label{line:gg-bounded-encoding}$\xi\sample\psi(\ZZ_q,\calL)$
				\item \label{line:gg-bounded-adversary}$b\sample\advA^{\Oxi}(q,\gamma,\xi(1),(\xi(x_j))_{j\in[d]})$
				\item \pcreturn $b$
			\end{codelines}
		\end{minipage}\hfill%
		\begin{minipage}[t]{0.48\linewidth}
			\vspace{0pt}
			{\footnotesize\underline{$\uniformexp^{\advA}_{q,\gamma,d}$}\par}
			\smallskip
			\begin{codelines}
				\item \label{line:gg-uniform-exponents}$(x_j)_{j\in[d]}\sample\ZZ_q^d$
				\item \label{line:gg-uniform-encoding}$\xi\sample\psi(\ZZ_q,\calL)$
				\item \label{line:gg-uniform-adversary}$b\sample\advA^{\Oxi}(q,\gamma,\xi(1),(\xi(x_j))_{j\in[d]})$
				\item \pcreturn $b$
			\end{codelines}
		\end{minipage}
	\end{proofbox}

For a generic algorithm $\advA$, we define
\[
	\AdvGG_{q,\gamma,d}(\advA)
	\cqq\abs{\Prob\left[\bndexp^{\advA}_{q,\gamma,d}=1\right]-\Prob\left[\uniformexp^{\advA}_{q,\gamma,d}=1\right]}.
\]
\else
In the generic group model, the group elements are random labels, and the adversary computes with them only through an oracle for the group operation, so its view is determined by which of its affine combinations of the challenge exponents coincide.

\begin{definition}[Generic Bounded-Power Problem]\label{def:generic-dsedl}
	For finite sets $X$ and $Y$ with $\abs X=\abs Y$, we write $\psi(X,Y)$ for the set of bijections from $X$ to $Y$. Let $q,\gamma,d\in\NN$ with $q\ge2$ and $1\le\gamma<q$, and let $\calL$ be a set with $\abs\calL=q$. For $\xi\in\psi(\ZZ_q,\calL)$, the generic group oracle $\Oxi$ is defined by $\Oxi(\xi(a),\xi(b),\varepsilon)\cqq\xi\left(a+(-1)^\varepsilon b\right)$ for $a,b\in\ZZ_q$ and $\varepsilon\in\bits$. A generic algorithm is an oracle algorithm $\advA$ that queries $\Oxi$ only on labels that it received as input or as answers of $\Oxi$, and whose nonuniform advice depends only on $q$, $\gamma$, and $d$. For $b\in\bits$, let $\calX_0\cqq[\gamma]^d$ and $\calX_1\cqq\ZZ_q^d$, and define
	\[
		p_b(\advA)\cqq
		\Prob\left[
			\advA^{\Oxi}\left(q,\gamma,\xi(1),(\xi(x_j))_{j\in[d]}\right)=1
			\mbar
			\begin{array}{l}
				(x_j)_{j\in[d]}\sample\calX_b, \\
				\xi\sample\psi(\ZZ_q,\calL)
			\end{array}
		\right].
	\]
	The generic advantage of $\advA$ is $\AdvGG_{q,\gamma,d}(\advA)\cqq\abs{p_0(\advA)-p_1(\advA)}$.
\end{definition}
\fi

For bounded exponents and composite order, we adapt the symbolic-transcript argument of~\cite{Nechaev94,EC:Shoup97}. We have not found this decisional bounded-interval statement in the literature.

\begin{theoremE}[Generic-Group Bound]\label{thm:generic-vector-dsedl}
	Let $q,\gamma,d,T\in\NN$ with $q\ge2$, $1\le\gamma<q$, and $d\ge1$, let $\piq\cqq\min\{p\in\Primes\mid q\in p\ZZ\}$, and let $\advA$ be a generic algorithm that makes at most $T$ oracle queries. Then
	\[
		\AdvGG_{q,\gamma,d}(\advA)
		\le\min\left\{1,\binom{T+d+1}{2}\left(\frac{1}{\gamma}+\frac{2}{\piq}\right)\right\}.
	\]
\end{theoremE}

\begin{proofE}
	We compare both experiments with one symbolic execution in which each label stands for an affine form in the challenge exponents. The transcript of each experiment has the same distribution as the symbolic transcript unless two distinct forms of the execution take the same value, and we bound the probability of this collision event separately for bounded and for uniform exponents.

	We write $\numforms\cqq T+d+1$ within this proof. If
	\[
		\binom{\numforms}{2}\left(\frac1\gamma+\frac2\piq\right)\ge1,
	\]
	then the bound holds because $\AdvGG_{q,\gamma,d}(\advA)\le1$. We suppose henceforth that the left-hand side is smaller than one, so that $\binom{\numforms}{2}<\piq/2\le q/2$. Then $\numforms\le q$, since $\binom{v}{2}\ge v$ for every integer $v\ge3$, and $\numforms\le2\le q$ otherwise. Hence, the at most $\numforms$ forms of the symbolic execution can receive distinct labels from $\calL$, and the symbolic oracle always has an unused label.

	We fix the coins and the nonuniform advice of $\advA$. Each label of the symbolic execution stands for an affine form
	\[
		F(\bvec{X})=\alpha_0+\sum_{j=1}^d\alpha_jX_j\in\ZZ_q[X_1,\dots,X_d],
	\]
	and the initial labels stand for $1,X_1,\dots,X_d$. The symbolic oracle keeps an injective table from coefficient vectors in $\ZZ_q^{d+1}$ to labels and assigns uniformly random distinct labels to these $d+1$ forms. On a query with the labels of forms $F_1$ and $F_2$ and with $\varepsilon\in\bits$, the oracle computes $F_1+(-1)^\varepsilon F_2$, returns the existing label if its coefficient vector is in the table, and otherwise assigns a uniformly random unused label. The execution begins with $d+1$ forms, and each query adds at most one form, so at most $\numforms$ forms occur. The symbolic transcript depends on the coins of $\advA$ and on the assigned labels and is independent of the values of $X_1,\dots,X_d$.

	We fix a complete symbolic transcript and, for a distribution $\Distribution$ on $\ZZ_q^d$, write $\Event_{\Distribution}$ for the event that two distinct forms of the transcript take the same value at $\bvec{X}\sample\Distribution$. We fix distinct forms $F_1$ and $F_2$ of the transcript and write
	\[
		R(\bvec{X})\cqq F_1(\bvec{X})-F_2(\bvec{X})=\beta_0+\sum_{j=1}^d\beta_jX_j,
	\]
	which is nonzero because $F_1\ne F_2$. If $\beta_1=\dots=\beta_d=0$, then $\beta_0\neq0$, and $R$ does not vanish. Otherwise, we fix $j$ with $\beta_j\ne0$. After we fix every variable other than $X_j$, the condition $R(\bvec{X})=0$ becomes the congruence $\beta_jX_j\equiv c\pmod q$ for some $c\in\ZZ_q$, whose solution set is either empty or an arithmetic progression of $\gcd(\beta_j,q)$ residues with common difference $q/\gcd(\beta_j,q)$. Since $\beta_j\neq0$, the integer $q/\gcd(\beta_j,q)$ is a divisor of $q$ larger than one, hence at least $\piq$. Since $\gamma<q$, the elements of $[\gamma]$ are distinct modulo $q$, and an interval of $\gamma$ consecutive integers meets an arithmetic progression with common difference at least $\piq$ in at most $\ceil{\gamma/\piq}\le\gamma/\piq+1$ points. Since $X_j$ is uniform on $[\gamma]$ and independent of the other coordinates, we obtain
	\[
		\Prob\left[R(\bvec{X})=0\mbar\bvec{X}\sample[\gamma]^d\right]\le\frac{1}{\gamma}\left(\frac{\gamma}{\piq}+1\right)=\frac{1}{\piq}+\frac{1}{\gamma}.
	\]
	For uniform exponents, after the same conditioning, $\gcd(\beta_j,q)$ of the $q$ residues of $X_j$ satisfy the congruence, so
	\[
		\Prob\left[R(\bvec{X})=0\mbar\bvec{X}\sample\ZZ_q^d\right]\le\frac{\gcd(\beta_j,q)}{q}=\frac{1}{q/\gcd(\beta_j,q)}\le\frac{1}{\piq}.
	\]
	By a union bound over the at most $\binom{\numforms}{2}$ unordered pairs of forms,
	\[
		\Prob\left[\Event_{[\gamma]^d}\right]\le\binom{\numforms}{2}\left(\frac{1}{\gamma}+\frac{1}{\piq}\right),
		\qquad
		\Prob\left[\Event_{\ZZ_q^d}\right]\le\binom{\numforms}{2}\frac{1}{\piq}.
	\]

	Outside the corresponding collision event, two evaluated forms take the same value exactly when their coefficient vectors agree, so the equality pattern among the group elements of the experiment agrees with the symbolic table. The symbolic oracle assigns uniformly random distinct labels to the forms that occur, which is the distribution of the labels of their values under a uniformly random bijection $\xi\colon\ZZ_q\to\calL$. Hence, outside the corresponding collision event, the transcript of each of the two experiments has the same distribution as the symbolic transcript, and we obtain
	\[
		\AdvGG_{q,\gamma,d}(\advA)
		\le\Prob\left[\Event_{[\gamma]^d}\right]+\Prob\left[\Event_{\ZZ_q^d}\right]
		\le\binom{\numforms}{2}\left(\frac{1}{\gamma}+\frac{2}{\piq}\right).
	\]
	This bound holds for every fixed choice of coins, hence also after we average over the coins of a randomized $\advA$, and with the case treated at the start of the proof, the bound in the statement holds.
\end{proofE}

\ifConferenceVersion\else
\Cref{thm:generic-vector-dsedl} holds for every $\gamma\in[q-1]$, including $\gamma\ge\piq$, for which a nonzero affine form can vanish at several points of an interval of length $\gamma$. For the sampler of \cref{def:dj-group}, the group order is $q=p'q'$ with $p',q'\ge2^{\ell(\secparam)-2}$, and since $\ell(\secparam)\ge\secparam+2$, we have $\piq\ge2^{\ell(\secparam)-2}\ge2^\secparam$. Hence, every interval with $\gamma<2^\secparam$ satisfies $\gamma<\piq$, and the term $2/\piq$ is at most $2^{1-\secparam}$.

The factor $\dlam$ in the target advantage of \cref{lem:tunable-scalar-to-vec-dsedl} comes from the $\dlam$ transitions of the coordinate hybrid. For $T\ge d$, we have
\[
	\binom{T+d+1}{2}\le\binom{2T+1}{2}=T(2T+1)\le6\binom{T+2}{2},
\]
so the bound of \cref{thm:generic-vector-dsedl} for $d$ challenge powers is at most six times the bound for one challenge power with the same $T$. Hence, a vector version of tunable \DSEDL{}, assumed directly, would remove the factor $\dlam$ from every interval threshold that we derive from \cref{lem:tunable-scalar-to-vec-dsedl}, and we take the scalar assumption as our starting point and keep this factor.
\fi

With square-and-multiply, we can compute $\xi(\sigma x)$ from $\xi(x)$ for $\sigma\in[0..q-1]$ with at most $2\clog{q}$ oracle queries. Hence, a generic algorithm that makes $T$ queries and, in addition, computes the $\sigma$th powers of $D$ challenge elements makes at most $T+2D\clog{q}$ queries, and \cref{thm:generic-vector-dsedl} applies with $T$ replaced by $T+2D\clog{q}$.

To show that the dependence on $\gamma$ in \cref{thm:generic-vector-dsedl} is tight at constant advantage, we restrict \BSGS{}~\cite{Sha1971} to the interval $[\gamma]$, as in~\cite{EC:vOoWie96}.

\begin{theoremE}[Matching Generic Attack]\label{thm:generic-group-dsedl-attack}
	For all $q,\gamma\in\NN$ with $q\ge2$ and $1\le\gamma<q$, there is a deterministic generic algorithm $\advA$ that makes at most $2\ceil{\sqrt\gamma}+1$ oracle queries and satisfies $\AdvGG_{q,\gamma,1}(\advA)=1-\gamma/q$.
\end{theoremE}

\begin{proofE}
	We run \BSGS{} on the interval $[\gamma]$ and add a range check. Because of the range check, acceptance is equivalent to membership of the challenge exponent in $[\gamma]$.

	We set $v\cqq\ceil{\sqrt\gamma}$. On input $\xi(1)$ and the challenge label $\xi(x)$, the algorithm $\advA$ computes and stores the labels $\xi(j)$ for $j\in[0..v]$ with at most $v$ oracle queries, then computes $\xi(x-iv)$ for every $i\in[v]$ with $v$ further queries, and accepts if and only if there are $i\in[0..v]$ and $j\in[0..v-1]$ with $\xi(x-iv)=\xi(j)$ and $iv+j\in[\gamma]$.

	If $x\in[\gamma]$, then $x\le\gamma\le v^2$, and by Euclidean division, $x=iv+j$ with $i\cqq\floor{x/v}\in[0..v]$ and $j\in[0..v-1]$, so $\advA$ accepts every bounded challenge. Conversely, if $\advA$ accepts, then $x\equiv iv+j\pmod q$ for an integer $iv+j\in[\gamma]$. Since $\gamma<q$, both $x$ and $iv+j$ are their own representatives in $[0..q-1]$, so $x=iv+j\in[\gamma]$. Hence, $\advA$ accepts a uniform $x\sample\ZZ_q$ with probability $\gamma/q$, so $\AdvGG_{q,\gamma,1}(\advA)=1-\gamma/q$, and $\advA$ makes at most $2v+1$ queries.
\end{proofE}

\end{relocated}%

\begin{relocated}{aux-generic-evidence}
\subsection{Tightness and Preprocessing}\label{app:auxiliary-generic-evidence}

By \cref{thm:generic-group-dsedl-attack}, the quadratic dependence on the number of queries in \cref{thm:generic-vector-dsedl} is tight at constant advantage, and we show that it is tight at every advantage and that the upper bound extends to encoding-dependent preprocessing. When we truncate the grid of \cref{thm:generic-group-dsedl-attack}, we accept exactly the challenge exponents in an initial segment of $[\gamma]$ whose size is quadratic in the grid parameter.

\begin{propositionE}[Truncated Generic Attack]\label{cor:generic-dsedl-attack-truncated}
	Let $q,\gamma,v\in\NN$ with $q\ge2$, $1\le\gamma<q$, and $1\le v\le\floor{\sqrt\gamma}$. There is a deterministic generic algorithm $\advA$ that makes at most $2v-1$ oracle queries and satisfies
	\[
		\AdvGG_{q,\gamma,1}(\advA)=\left(v^2-1\right)\left(\frac1\gamma-\frac1q\right).
	\]
	In particular, if $\gamma\le q/2$, then for every odd $T\in\NN$ with $3\le T\le2\floor{\sqrt\gamma}-1$, the algorithm for $v=(T+1)/2$ makes at most $T$ queries and satisfies $\AdvGG_{q,\gamma,1}(\advA)\ge T^2/(8\gamma)$.
\end{propositionE}
\begin{proofE}
	The two acceptance probabilities are the densities of the segment $[v^2-1]$ in $[\gamma]$ and in $\ZZ_q$.

	On input $\xi(1)$ and the challenge label $\xi(x)$, the algorithm $\advA$ computes and stores the labels $\xi(j)$ for $j\in[0..v]$ with at most $v$ oracle queries, then computes $\xi(x-iv)$ for every $i\in[v-1]$ with at most $v-1$ further queries, and accepts if and only if there are $i,j\in[0..v-1]$ with $\xi(x-iv)=\xi(j)$ and $iv+j\in[\gamma]$.

	Since $v\le\floor{\sqrt\gamma}$, we have $v^2-1<\gamma<q$, so every element of $[v^2-1]$ is its own representative modulo $q$ and lies in $[\gamma]$. Since $\{iv+j\mid i,j\in[0..v-1]\}=[0..v^2-1]$, the algorithm $\advA$ accepts if and only if $x\bmod q\in[v^2-1]$. Hence, $\advA$ accepts a bounded challenge $x\sample[\gamma]$ with probability $(v^2-1)/\gamma$ and a uniform challenge $x\sample\ZZ_q$ with probability $(v^2-1)/q$, and $\AdvGG_{q,\gamma,1}(\advA)$ is their difference.

	For the last claim, $v=(T+1)/2$ satisfies $1\le v\le\floor{\sqrt\gamma}$, and since $T\ge3$, we have $v^2-1=T^2/4+(2T-3)/4\ge T^2/4$. Since $\gamma\le q/2$, we have $1/\gamma-1/q\ge1/(2\gamma)$, so $\AdvGG_{q,\gamma,1}(\advA)\ge T^2/(8\gamma)$.
\end{proofE}

\ifConferenceVersion
\parhead{Preprocessing} A concrete nonuniform adversary can preprocess one fixed group representation. In the auxiliary-input generic-group model of~\cite{C:CorDodGuo18}, a two-stage algorithm $\advA$ consists of a computationally unbounded preprocessing algorithm, which receives the bijection $\xi$ and outputs a string $z\in\bits^{\le S}$, and an online algorithm, which receives $z$ and the challenge of \cref{fig:generic-dsedl-games} and makes at most $T$ queries to $\Oxi$. We write $\Adv^{\aiggdsedl}_{q,\gamma,d,S,T}(\advA)$ for the difference between the acceptance probabilities of $\advA$ with exponents $(x_j)_{j\in[d]}\sample[\gamma]^d$ and with exponents $(x_j)_{j\in[d]}\sample\ZZ_q^d$. We combine the presampling theorem of~\cite{C:CorDodGuo18} with the symbolic argument of \cref{thm:generic-vector-dsedl} and obtain a polynomial interval threshold against such algorithms.
\else
\parhead{Preprocessing} In \cref{thm:generic-vector-dsedl}, the nonuniform advice is independent of the encoding $\xi$, whereas a concrete nonuniform adversary can preprocess one fixed group representation. We extend \cref{thm:generic-vector-dsedl} to the auxiliary-input generic-group model of~\cite{C:CorDodGuo18}, in which an unbounded preprocessing algorithm receives the complete encoding $\xi$ and passes an $S$-bit advice string to an online algorithm that makes at most $T$ group-oracle queries.

We first bound the advantage in the bit-fixing model of~\cite{C:CorDodGuo18}, in which at most $P$ points of the encoding are fixed in advance. For a list $\calF = \bigl((r_i,\ell_i)\bigr)_{i\in[P']}$ with $P'\le P$, pairwise distinct $r_i\in\ZZ_q$, and pairwise distinct labels $\ell_i\in\calL$, we write $\psi_\calF$ for the set of bijections $\xi\in\psi(\ZZ_q,\calL)$ with $\xi(r_i)=\ell_i$ for every $i\in[P']$. The online algorithm receives the complete list $\calF$ and arbitrary advice derived from it, and since it can ignore $\calF$, this model contains the bit-fixing model of~\cite{C:CorDodGuo18}. As in the bit-fixing bounds of~\cite{C:CorDodGuo18}, we adapt the symbolic argument of \cref{thm:generic-vector-dsedl} to this model.

\begin{theoremE}[Bit-Fixing Generic Bound]\label{thm:bit-fixing-generic-dsedl}
	Let $q,\gamma,d,P,T\in\NN$ with $q\ge2$ and $1\le\gamma<q$, let $\piq$ be the smallest prime factor of $q$, let $\calF=((r_i,\ell_i))_{i\in[P']}$ with $P'\le P$, pairwise distinct $r_i\in\ZZ_q$, and pairwise distinct $\ell_i\in\calL$, and let $z$ be a string. Let $\advA$ be a generic algorithm that makes at most $T$ queries to $\Oxi$, and for $b\in\bits$, let $\calX_0\cqq[\gamma]^d$, $\calX_1\cqq\ZZ_q^d$, and
	\[
		p_b\cqq
		\Prob\left[
			\advA^{\Oxi}\left(z,\calF,q,\gamma,\xi(1),(\xi(x_j))_{j\in[d]}\right)=1
			\mbar
			\begin{array}{l}
				(x_j)_{j\in[d]}\sample\calX_b, \\
				\xi\sample\psi_\calF
			\end{array}
		\right].
	\]
	Then
	\[
		\abs{p_0-p_1}\le
		\min\left\{
			1,
			\binom{T+d+P+1}{2}
			\left(
				\frac1\gamma+
				\frac2{\piq}
			\right)
		\right\}.
	\]
\end{theoremE}

\begin{proofE}
	We repeat the symbolic execution of the proof of \cref{thm:generic-vector-dsedl}, which now begins with the constant forms of the exponents $r_i$ of $\calF$ in addition to $1,X_1,\ldots,X_d$. The labels of the constant forms are prescribed by $\calF$, every new coefficient vector receives a uniformly random unused label, and the collision calculation does not depend on $z$ and $\calF$, which are fixed.

	We fix the coins of $\advA$ and write $\numforms\cqq T+d+P+1$ within this proof. At most $P+d+1$ forms occur initially, and every query adds at most one form, so the transcript contains at most $\numforms$ forms. For two distinct forms $F_1$ and $F_2$, the difference $F_1-F_2$ is a nonzero affine polynomial in $X_1,\ldots,X_d$. If one of its variable coefficients is nonzero, we repeat the progression argument of the proof of \cref{thm:generic-vector-dsedl}, conditioned on the other variables, and obtain that $F_1-F_2$ vanishes with probability at most $1/\gamma+1/\piq$ for bounded exponents and at most $1/\piq$ for uniform exponents. If all its variable coefficients are zero, then $F_1-F_2$ is a nonzero constant and does not vanish.

	As in the proof of \cref{thm:generic-vector-dsedl}, if
	\[
		\binom{\numforms}{2}\left(\frac1\gamma+\frac2\piq\right)\ge1,
	\]
	then the bound holds because the advantage is at most one. Otherwise, $\binom{\numforms}{2}<\piq/2\le q/2$, so $\numforms\le q$ by the argument of that proof, and the symbolic execution always has an unused label. By a union bound over the unordered pairs of forms, the collision event has probability at most
	\[
		\binom{\numforms}{2}\left(\frac1\gamma+\frac1\piq\right)
		,\qquad
		\binom{\numforms}{2}\frac1\piq
	\]
	for bounded and for uniform exponents, respectively. Outside the collision event, the equality pattern among the evaluated exponents agrees with the symbolic table, and the labels of the exponents outside $\calF$ are uniformly distributed without replacement among the unused labels, which is their distribution under a uniformly random bijection $\xi\sample\psi_\calF$. We compare each of the two experiments with the symbolic execution, average over the coins of $\advA$, and obtain the bound.
\end{proofE}

We state the auxiliary-input version of the problem of \cref{def:generic-dsedl}, adapted from the auxiliary-input generic group of~\cite{C:CorDodGuo18}.

\begin{definition}[Auxiliary-Input Generic Bounded-Power Problem]\label{def:auxiliary-input-generic-dsedl}
	Let $q,\gamma,d,S,T\in\NN$ with $q\ge2$ and $1\le\gamma<q$, and let $\calX_0\cqq[\gamma]^d$ and $\calX_1\cqq\ZZ_q^d$. A two-stage auxiliary-input generic algorithm is a pair $\advA=(\advA_{\var{pre}},\advA_{\var{on}})$ of a computationally unbounded algorithm $\advA_{\var{pre}}$ that, on input $\xi\in\psi(\ZZ_q,\calL)$, outputs a string $z\in\bits^{\le S}$, and a generic algorithm $\advA_{\var{on}}$ that makes at most $T$ queries to $\Oxi$. For $b\in\bits$, define
	\[
		p_b^{\aiggdsedl}(\advA)\cqq
		\Prob\left[
			\advA_{\var{on}}^{\Oxi}\left(z,q,\gamma,\xi(1),(\xi(x_j))_{j\in[d]}\right)=1
			\mbar
			\begin{array}{l}
				\xi\sample\psi(\ZZ_q,\calL), \\
				z\sample\advA_{\var{pre}}(\xi), \\
				(x_j)_{j\in[d]}\sample\calX_b
			\end{array}
		\right],
	\]
	and $\Adv^{\aiggdsedl}_{q,\gamma,d,S,T}(\advA)\cqq\abs{p_0^{\aiggdsedl}(\advA)-p_1^{\aiggdsedl}(\advA)}$.
\end{definition}

We combine \cref{thm:bit-fixing-generic-dsedl} with the additive part of the auxiliary-input-to-bit-fixing theorem of~\cite{C:CorDodGuo18} and obtain the following bound.

\begin{theoremE}[Auxiliary-Input Generic Bound]\label{thm:auxiliary-input-generic-dsedl}
	Let $q,\gamma,d,S,T\in\NN$ with $q\ge16$ and $1\le\gamma<q$, let $\piq$ be the smallest prime factor of $q$, let $\vartheta\in(0,1)$, and let $P\in[q]$. Every two-stage auxiliary-input generic algorithm $\advA$ in the sense of \cref{def:auxiliary-input-generic-dsedl} satisfies
	\begin{align*}
		\Adv^{\aiggdsedl}_{q,\gamma,d,S,T}(\advA)
		\le{}
		 & \min\left\{
			1,
			\binom{T+d+P+1}{2}
			\left(
				\frac1\gamma+
				\frac2{\piq}
			\right)
		\right\} \\
		 & {}+
		\frac{12\bigl(S+\log_2(1/\vartheta)\bigr)(T+d+1)}{P}
		+2\vartheta.
	\end{align*}
\end{theoremE}

\begin{proofE}
	We apply the additive part of the auxiliary-input-to-bit-fixing theorem of~\cite{C:CorDodGuo18} to the problem of \cref{def:auxiliary-input-generic-dsedl}. We take the auxiliary-input and bit-fixing generic groups of~\cite{C:CorDodGuo18} to be $\AIGG(q,q)$ and $\BFGG(P,q,q)$, so the random injection of~\cite{C:CorDodGuo18} is the bijection $\xi$, and the range set published by the bit-fixing oracle of~\cite{C:CorDodGuo18} is $\calL$, which carries no information about $\xi$. The theorem of~\cite{C:CorDodGuo18} concerns random injections, so it holds for every $q$, including a composite order. The slack parameter of~\cite{C:CorDodGuo18}, which~\cite{C:CorDodGuo18} calls $\gamma$, is our $\vartheta$, and the advantage of~\cite{C:CorDodGuo18} for an indistinguishability application is twice the deviation of the success probability from one half, which equals $\Adv^{\aiggdsedl}_{q,\gamma,d,S,T}(\advA)$. The combined query complexity of the online algorithm and the challenger is at most $T+d+1$, since the challenger reveals $\xi(1)$ and the $d$ challenge labels in addition to the $T$ queries of $\advA_{\var{on}}$. By the theorem of~\cite{C:CorDodGuo18}, we can replace the auxiliary-input experiment by a convex combination of bit-fixing experiments with at most $P$ fixed points, at additive loss
	\[
		\frac{12\bigl(S+\log_2(1/\vartheta)\bigr)(T+d+1)}{P}
		+2\vartheta.
	\]
	Since the bound of \cref{thm:bit-fixing-generic-dsedl} holds for every list of at most $P$ pairs and every auxiliary string, we apply it to every member of the convex combination and substitute the result into the presampling bound.
\end{proofE}

We choose the free parameters of \cref{thm:auxiliary-input-generic-dsedl} and obtain a polynomial interval threshold against encoding-dependent preprocessing.
\fi

\begin{corollaryE}[Generic Bound with Preprocessing]\ifConferenceVersion\label{thm:auxiliary-input-generic-dsedl}\else\label{cor:auxiliary-input-threshold}\fi
	Let $\delta,S,T,d\colon\NN\to\NN$ be polynomially bounded with $\delta(\secparam)\ge1$ for every $\secparam$, and for every $\secparam\in\NN$, define
	\begin{align*}
		\vartheta(\secparam) & \cqq\frac{1}{8\delta(\secparam)},                                          \\
		P(\secparam)            & \cqq\ceil{48\delta(\secparam)\bigl(S(\secparam)+\ceil{\log_2(8\delta(\secparam))}\bigr)(T(\secparam)+d(\secparam)+1)}, \\
		\Tai(\secparam) & \cqq T(\secparam)+d(\secparam)+P(\secparam)+1.
	\end{align*}
	Let $\secparam,q,\gamma\in\NN$ satisfy
	\[
		\piq\ge8\delta(\secparam)\binom{\Tai(\secparam)}{2}
		,\qquad
		4\delta(\secparam)\binom{\Tai(\secparam)}{2}\le\gamma<q,
	\]
	where $\piq$ is the smallest prime factor of $q$. Then every two-stage auxiliary-input generic algorithm $\advA$ whose preprocessing stage outputs at most $S(\secparam)$ bits and whose online stage makes at most $T(\secparam)$ queries to $\Oxi$ satisfies $\Adv^{\aiggdsedl}_{q,\gamma,d(\secparam),S(\secparam),T(\secparam)}(\advA)\le1/\delta(\secparam)$. Moreover, the function $\secparam\mapsto4\delta(\secparam)\binom{\Tai(\secparam)}{2}$ is polynomially bounded.
\end{corollaryE}
\begin{proofE}
	We bound the three terms of the presampling bound separately by $1/(2\delta)$, $1/(4\delta)$, and $1/(4\delta)$, where we write $\delta$, $S$, $T$, $d$, $\vartheta$, $P$, and $\Tai$ for their values at $\secparam$.

	Since $\Tai\ge P+1\ge49$, we have $\binom{\Tai}{2}\ge\Tai>P$, so $q\ge\piq\ge8\delta\binom{\Tai}{2}\ge P$ and $q\ge16$.
	\ifConferenceVersion
	By the additive presampling theorem of~\cite{C:CorDodGuo18}, we can replace an execution with $S$-bit auxiliary input by a convex combination of executions in which at most $P$ points of the encoding are fixed. After we condition on those points, the symbolic execution in the proof of \cref{thm:generic-vector-dsedl} has at most $P$ further constant forms, and we repeat the collision calculation of that proof and obtain that the advantage of $\advA$ is at most
	\[
		\binom{\Tai}{2}\left(\frac1\gamma+\frac2{\piq}\right)
		+\frac{12(S+\log_2(1/\vartheta))(T+d+1)}P+2\vartheta.
	\]
	\else
	Hence, \cref{thm:auxiliary-input-generic-dsedl} applies with these $\vartheta$ and $P$, and its bound is at most
	\[
		\binom{\Tai}{2}\left(\frac1\gamma+\frac2{\piq}\right)
		+\frac{12(S+\log_2(1/\vartheta))(T+d+1)}P+2\vartheta.
	\]
	\fi
	By the conditions on $\gamma$ and $\piq$, the first term is at most $1/(4\delta)+1/(4\delta)=1/(2\delta)$. Since $\log_2(1/\vartheta)=\log_2(8\delta)\le\ceil{\log_2(8\delta)}$, the choice of $P$ bounds the second term by $12/(48\delta)=1/(4\delta)$, and $2\vartheta=1/(4\delta)$. Hence, the advantage is at most $1/\delta$. Since $S$, $T$, $d$, and $\delta$ are polynomially bounded, so are $P$, $\Tai$, and $4\delta\binom{\Tai}{2}$.
\end{proofE}

\ifConferenceVersion\else
\parhead{Parameter Scale} By \cref{thm:generic-vector-dsedl}, every generic algorithm that makes at most $T$ queries satisfies $\AdvGG_{q,\gamma,d}(\advA)\le1/\delta$ whenever
\[
	\gamma\ge2\delta\binom{T+d+1}{2}
	,\qquad
	\piq\ge4\delta\binom{T+d+1}{2},
\]
since each of the two terms of the bound is then at most $1/(2\delta)$. By \cref{cor:auxiliary-input-threshold}, a polynomial interval threshold also suffices against polynomial encoding-dependent advice. By \cref{thm:generic-group-dsedl-attack,cor:generic-dsedl-attack-truncated}, the quadratic dependence on $T$ is tight for $d=1$, at advantage $1-\gamma/q$ and at every inverse-polynomial advantage, respectively.

\parhead{Scope of the Generic Bounds}\phantomsection\label{rem:ggm-scope}\label{rem:ggm-scope-full} \Cref{thm:generic-vector-dsedl,thm:auxiliary-input-generic-dsedl} hold in the generic group model, so concrete parameters for tunable \DSEDL{} remain open, and \DCR{}, which \cref{def:dcr} states, lies outside this analysis. In \cref{thm:auxiliary-input-generic-dsedl}, the advice depends on the complete random encoding, and the interval threshold grows with the advice length $S$, which is necessary, since advice that contains the labels of a set $\calI\subseteq[\gamma]$ yields a distinguisher with zero queries and advantage $\abs{\calI}/\gamma-\abs{\calI}/q$. In both theorems, the algorithm receives the group order $q$, so the bounds hold without the secrecy of the order in \cref{def:dj-group}. In the structured generic group model of~\cite{EC:CorHenWu26}, labels have a controlled amount of structure, but the bounds of~\cite{EC:CorHenWu26} concern search discrete logarithms and do not imply the bounded-versus-uniform statement of \cref{thm:generic-vector-dsedl}, so generic bounds in a model that exposes the group representation remain open.

The generic model does not cover the following attacks that use the representation of the \DamgardJurik group: arithmetic modulo $N$ and modulo $N^2$, both of which apply to a \DamgardJurik element; factoring $N$, which splits the group by the Chinese remainder theorem; index-calculus and number-field-sieve methods, which use the arithmetic of the ambient ring~\cite{EC:Weber95,PKC:ComSem06}; preprocessing against a fixed family of moduli, which is a number-theoretic object shared by every instance, whereas the advice in \cref{thm:auxiliary-input-generic-dsedl} depends on a random encoding; and statistics of the representation that are independent of the exponent. Whether one of these attacks distinguishes the two distributions of \cref{def:tunable-dsedl} faster than \BSGS{}~\cite{Sha1971} is open.
\fi

\end{relocated}%

\section{Extremely Lossy Trapdoor Hashing}\label{sec:hashing}
Trapdoor hashing~\cite{C:DGIMMO19} and extremely lossy functions~\cite{C:Zhandry16} have been developed independently, and we construct one algebraic object with the properties of both. From its short digest, we can recover a predicate value with a trapdoor, as from a trapdoor hash, and under an indistinguishable choice of public parameters, its hash maps an exponential domain onto a polynomial image, as an extremely lossy function does. We define this primitive and build it over a \DamgardJurik group by replacing the uniform hash bases with bounded powers of a generator. From~\cite{C:RoySin21}, we adapt the distributed discrete logarithm $\DDLog_N$, which maps the digest and the encoding to exact additive shares of a linear predicate.

In our formulation, the inputs are residue vectors modulo the modulus $N$ of the hash key. Inside the construction, we decompose each residue into bits and multiply the offset of bit $k$ by $2^k$, so that the predicate value is a linear function of the residues and the binary expansion appears only in hashing and encoding.

We extend trapdoor hashing~\cite{C:DGIMMO19}\ifConferenceVersion\else, recalled in \cref{def:trapdoor-hashing},\fi{} by an image parameter $r$ of the setup algorithm and change its clauses as follows. Since in the \DamgardJurik construction the inputs are residue vectors modulo the modulus of the hash key, the input set $\calX$ and the predicate range $\KK$ depend on the hash key. The encoding takes values in an abelian group $\KK$, decoding returns one element of $\KK$, and for correctness, we require the encoding and the decoded value to differ by the predicate value. At the full parameter $\Rfull(\secparam,d)$, the hash key has the full-range distribution, and at every smaller parameter $r$, the hash with the all-zero tape must take at most $r$ values, so that the image bound is independent of the encoding key. For mode indistinguishability, we require the image parameters it covers to form a nonempty interval for all sufficiently large $\secparam$, since otherwise the clause would hold vacuously. For input privacy and function privacy, we require only the full parameter, since at nonfull parameters with polynomially bounded exponents, we can recover the predicate offsets from the encoding key.

\begin{definition}[Extremely Lossy Trapdoor Hashing]\label{def:eltdh}
	Let $\calX$ and $\KK$ be families of finite sets and finite abelian groups indexed by $\secparam\in\NN$, $d\in\NN_{>0}$, and a hash key. Let $\FuncClass$ be a function class indexed in the same way, with every $f\in\FuncClass_{\secparam,d,\hk}$ mapping $\calX_{\secparam,d,\hk}$ to $\KK_{\secparam,d,\hk}$. When no ambiguity arises, we write $\calX = \calX_{\secparam,d,\hk}$, $\KK =\KK_{\secparam,d,\hk}$, and $\FuncClass = \FuncClass_{\secparam,d,\hk}$. We require efficient membership testing and canonical representations for $\calX$ and $\KK$, efficient equality and group operations in $\KK$, and polynomial-size descriptions and polynomial-time evaluation for every predicate in $\FuncClass$. Let $\Rfull\colon\NN\times\NN_{>0}\to\NN$ be efficiently computable with $\log \Rfull(\secparam,d)=\poly(\secparam+d)$.

	An extremely lossy trapdoor-hashing scheme for $\FuncClass$ is a tuple of probabilistic polynomial-time algorithms $\ELTDH\cqq(\ELTDHSetupLos,\allowbreak\ELTDHHash,\allowbreak\ELTDHGen,\allowbreak\ELTDHEnc,\allowbreak\ELTDHDec)$ with the following syntax.
	\begin{itemize}
		\item $\hk\sample\ELTDHSetupLos(1^\secparam,1^d,r)$. On input the security parameter, the arity $d$, and an image parameter $r\in[\Rfull(\secparam,d)]$, the setup algorithm samples a hash key $\hk$, which fixes $\calX$ and $\KK$.
		\item $h\gets\ELTDHHash(\hk,x;\rho)$. The hashing algorithm maps the hash key, an input $x\in\calX$, and a random tape $\rho$ to a digest $h$.
		\item $(\ek,\td)\sample\ELTDHGen(\hk,f)$. The generation algorithm samples a public encoding key $\ek$ and a trapdoor $\td$ for a predicate $f\in\FuncClass$.
		\item $e\gets\ELTDHEnc(\ek,x;\rho)$. The encoding algorithm maps the encoding key, an input $x\in\calX$, and the random tape $\rho$ of hashing to an encoding $e\in\KK$.
		\item $e'\gets\ELTDHDec(\td,h)$. The deterministic decoding algorithm maps the trapdoor $\td$ and a digest $h$ to a decoded value $e'\in\KK$.
	\end{itemize}
	For fixed keys, we write
	\[
		\begin{aligned}
			\Phihk(x)   & \cqq\ELTDHHash(\hk,x;0),                               \\
			\Deltahk(x) & \cqq\ELTDHEnc(\ek,x;0)-\ELTDHDec(\td,\Phihk(x))
		\end{aligned}
	\]
	for the hash and for the difference between the encoding and the decoded value at the all-zero tape $0$. The scheme satisfies the following properties.
	\begin{itemize}
		\item \clauselabel{Compact Descriptions.} Hash keys, encoding keys, and trapdoors have length $\poly(\secparam+d)$, digests and encodings have length $\poly(\secparam)$ independent of $d$, and every algorithm runs in time $\poly(\secparam+d)$.
		\item \clauselabel{Correctness.} For a function $\delta\colon\NN\to[0,1]$, the scheme is $\delta$-correct if, for every $\secparam\in\NN$, every $d\in\NN_{>0}$, every $r\in[\Rfull(\secparam,d)]$, every $\hk\in\Supp(\ELTDHSetupLos(1^\secparam,1^d,r))$, every $f\in\FuncClass$, and every $x\in\calX$,
		      \[
			      \Prob\left[
				      e-e'=f(x)
				      \mbar
				      \begin{array}{l}
					      \rho\sample\bits^*,                     \\
					      h\gets\ELTDHHash(\hk,x;\rho),           \\
					      (\ek,\td)\sample\ELTDHGen(\hk,f), \\
					      e\gets\ELTDHEnc(\ek,x;\rho),           \\
					      e'\gets\ELTDHDec(\td,h)
				      \end{array}
				      \right]\ge 1-\delta(\secparam).
		      \]
		      The scheme is fully correct if it is $\varepsilon$-correct for some negligible function $\varepsilon$, and perfectly correct if it is $0$-correct.
		\item \clauselabel{Lossiness.} There exists a negligible function $\nulos$ such that, for every $\secparam\in\NN$, every $d\in\NN_{>0}$, and every $r\in[\Rfull(\secparam,d)]$,
		      \[
			      \Prob\left[
				      \abs{\Image(\Phihk)}\le r
				      \mbar
				      \hk\sample\ELTDHSetupLos(1^\secparam,1^d,r)
				      \right]\ge1-\nulos(\secparam).
		      \]
		\item \clauselabel{Mode Indistinguishability.} Let $d\colon\NN\to\NN_{>0}$ be polynomially bounded, write $\dlam\cqq d(\secparam)$, and define
		      \[
			      \Dmode_{\secparam,r}
			      \cqq\left\{\hk\ \middle|\ \hk\sample\ELTDHSetupLos(1^\secparam,1^{\dlam},r)\right\}.
		      \]
		      For every polynomial $t$ and every positive polynomial $\delta$, there exist a polynomial $Q$ and $\secparam_0\in\NN$ such that, for every $\secparam\ge\secparam_0$, the interval $[Q(\secparam)..\Rfull(\secparam,\dlam)]$ is nonempty and $\Dmode_{\secparam,r_0} \indist_{t,1/\delta} \Dmode_{\secparam,r_1}$ for all $r_0,r_1\in[Q(\secparam)..\Rfull(\secparam,\dlam)]$.
		\item \clauselabel{Input Privacy.} Let $d\colon\NN\to\NN_{>0}$ be polynomially bounded, and let $\advS$ be a nonuniform $\PPT$ selector that on input $(1^\secparam,\hk)$ returns $(x^0,x^1,\var{st})$ with $x^0,x^1\in\calX$. For $b\in\bits$, define the ensemble $\Dinpriv_b=(\Dinpriv_{b,\secparam})_{\secparam\in\NN}$ by
		      \[
			      \Dinpriv_{b,\secparam}\cqq
			      \left\{(\hk,\var{st},h^b)\ \middle|\
			      \begin{array}{l}
				      \hk\sample\ELTDHSetupLos(1^\secparam,1^{d(\secparam)},\Rfull(\secparam,d(\secparam))), \\
				      (x^0,x^1,\var{st})\sample\advS(1^\secparam,\hk),                                     \\
				      \rho\sample\bits^*,\quad h^b\gets\ELTDHHash(\hk,x^b;\rho)
			      \end{array}\right\}.
		      \]
		      The scheme has input privacy if $\Dinpriv_0\cindist\Dinpriv_1$ for every such $\advS$, statistical input privacy if $\Dinpriv_0\sindist\Dinpriv_1$ for every such $\advS$, and perfect input privacy if $\Dinpriv_0\equivdist\Dinpriv_1$ for every such $\advS$.
		\item \clauselabel{Function Privacy.} Let $d\colon\NN\to\NN_{>0}$ be polynomially bounded, and let $\advS$ be a nonuniform $\PPT$ selector that on input $(1^\secparam,\hk)$ returns $(f^0,f^1,\var{st})$ with $f^0,f^1\in\FuncClass$. For $b\in\bits$, define the ensemble $\Dfnpriv_b=(\Dfnpriv_{b,\secparam})_{\secparam\in\NN}$ by
		      \[
			      \Dfnpriv_{b,\secparam}\cqq
			      \left\{(\hk,\var{st},\ek^b)\ \middle|\
			      \begin{array}{l}
				      \hk\sample\ELTDHSetupLos(1^\secparam,1^{d(\secparam)},\Rfull(\secparam,d(\secparam))), \\
				      (f^0,f^1,\var{st})\sample\advS(1^\secparam,\hk),                                     \\
				      (\ek^b,\td^b)\sample\ELTDHGen(\hk,f^b)
			      \end{array}\right\}.
		      \]
		      The scheme is function private if $\Dfnpriv_0\cindist\Dfnpriv_1$ for every such $\advS$.
		\item \clauselabel{Rate.} For $\mu\colon\NN\to\RR_{\ge0}$, the scheme has rate at least $\mu$ if, for all $\secparam$, $d$, $r$, $\hk$, $f$, and $x$ quantified as in the correctness clause, every $(\ek,\td)\in\Supp(\ELTDHGen(\hk,f))$ and every random tape $\rho$ satisfy
		      \[
			      \frac{\replen{f(x)}}{\replen{\ELTDHEnc(\ek,x;\rho)}}\ge\mu(\secparam),
		      \]
		      where $\replen{u}$ denotes the length of the canonical representation of $u$. The scheme has rate one if the two lengths are equal for all these objects.
	\end{itemize}
\end{definition}

\ifConferenceVersion\else
For $\calX=\bits^d$, $\KK=\ZZ/2\ZZ$, and the full parameter, a scheme as in \cref{def:eltdh} is a trapdoor-hashing scheme in the sense of \cref{def:trapdoor-hashing}, with setup $\ELTDHSetupLos(1^\secparam,1^d,\Rfull(\secparam,d))$, encoder $\ELTDHEnc$, and decoder $\TDHDec(\td,h)\cqq(e',e'+1)$ for $e'\gets\ELTDHDec(\td,h)$. Since $e-e'=f(x)$ holds if and only if $e=e'+f(x)$, the candidate indexed by $f(x)$ equals the encoding, and $e'\ne e'+1$, so we obtain the correctness condition of \cref{def:tdh-correctness}. The input-privacy and function-privacy clauses of \cref{def:eltdh} at the full parameter imply those of \cref{def:tdh-input-privacy,def:tdh-function-privacy}, because the selector of \cref{def:eltdh} may choose the inputs after it sees the hash key.

\fi
The lossiness clause of \cref{def:eltdh} concerns only the hash, and its correctness clause holds for each input separately. To obtain a trapdoor extremely lossy function from the scheme, we also need correctness for every input under one key, for an injective predicate and for a constant predicate.

\begin{remark}[From Lossy Hashing to Trapdoor Extremely Lossy Functions]\label{rem:eltdh-gives-telf}
	Suppose that a predicate $f$ of a scheme is injective with an efficient inverse, that a predicate $\fstar$ of the scheme is constant with value $c$, and that there is a negligible function $\nu_{\var{sim}}$ such that both of the following conditions hold.
	\begin{parts}
		\item For every supported hash key $\hk$ at the full parameter, with probability at least $1-\nu_{\var{sim}}(\secparam)$ over $(\ek,\td)\sample\ELTDHGen(\hk,f)$, we have $\Deltahk(x)=f(x)$ for every input $x$.
		\item For every image parameter $r$ of the lossy mode and every supported hash key $\hk$ at $r$, with probability at least $1-\nu_{\var{sim}}(\secparam)$ over $(\ek,\td)\sample\ELTDHGen(\hk,\fstar)$, we have $\Deltahk(x)=\fstar(x)$ for every input $x$.
	\end{parts}
	We define the deterministic public map $\Fhk(x)\cqq(\Phihk(x),\ELTDHEnc(\ek,x;0))$. At the full parameter and under the first condition, we recover $f(x)=e-\ELTDHDec(\td,h)$ from $(h,e)=\Fhk(x)$ and $\td$, and then $x$ with the inverse of $f$. Under the second condition, $e=c+\ELTDHDec(\td,h)$, so
	\[
		\Fhk(x)=\Jtdc(\Phihk(x)),
		\qquad
		\Jtdc(h)\cqq(h,c+\ELTDHDec(\td,h)).
	\]
	Since the first coordinate of $\Jtdc(h)$ is $h$, the map $\Jtdc$ is injective and $\abs{\Image(\Fhk)}=\abs{\Image(\Phihk)}$. We first change $f$ to $\fstar$ at the full parameter under function privacy and then change the image parameter of the hash key to $r$ under mode indistinguishability. By hash lossiness and the two conditions, $\Fhk$ is a trapdoor extremely lossy function.
\end{remark}

\paperpart{The \DamgardJurik Construction}{subsec:eltdh-construction}
We realize \cref{def:eltdh} over a \DamgardJurik group, with uniform hash bases at the full parameter and bounded powers of a generator at every other parameter.
\begin{relocated}{fallback-modulus}
We use the sampler $\DJGen$ of \cref{def:dj-group} in every setup, and support quantifiers exclude $\bot$ as in \cref{app:cyclic-group-samplers}. We make every construction total by a fallback output on the event that $\DJGen$ outputs $\bot$, and on this event, we use one fixed modulus that satisfies the same two-sided size bound as the honest moduli. Since we use the same fallback hash key at every image parameter, the fallback branch adds no term to the bound for mode indistinguishability, and its probability appears only in the bounds that we derive from function privacy.

\begin{definition}[Fallback Modulus]\label{def:fallback-modulus}
	Let $\ell$ be the integer-valued polynomial of \cref{def:dj-group}, and define $\Nbot(\secparam)\cqq2^{2\ell(\secparam)-1}+1$.
\end{definition}

The fallback modulus $\Nbot$ is odd and satisfies
\[
	2^{2\ell(\secparam)-2}<\Nbot(\secparam)<2^{2\ell(\secparam)},
	\qquad
	\ceil{\log_2\Nbot(\secparam)}=2\ell(\secparam),
\]
and since $2^{\ell(\secparam)-1}<p,q<2^{\ell(\secparam)}$, every honest modulus $N=pq$ of \cref{def:dj-group} satisfies $2^{2\ell(\secparam)-2}<N<2^{2\ell(\secparam)}$ as well.

\end{relocated}%
We fix $\secparam\in\NN$ and $d\in\NN_{>0}$, and we set $w\cqq 2\ell(\secparam)$ and $\Rfull(\secparam,d) \cqq 2^{dw}$, where $\ell(\secparam)$ is the bit length of each prime factor of a modulus sampled by $\DJGen$\ifConferenceVersion~\cite{PKC:DamJur01,C:RoySin21}\else{} in \cref{def:dj-group}\fi, so that $w$ is the width of one residue. Every modulus of a hash key is smaller than $2^w$, including the fallback modulus $\Nbot(\secparam)=2^{2\ell(\secparam)-1}+1$\ifConferenceVersion\else{} of \cref{def:fallback-modulus}\fi, which is the modulus of the hash key when parameter generation fails. For a hash key with modulus $N$, we take $\calX\cqq (\ZN)^d$ and $\KK\cqq \ZN$, and the predicate class $\Flin = \Flin_{\secparam,d,\hk}$ consists of the residue-linear predicates
\[
	\Lina(\bvec x)\cqq \sum_{j=1}^d a_jx_j\in\ZN,
	\qquad
	\bvec a=(a_1,\dots,a_d)\in(\ZN)^d.
\]
We give the algorithms in \cref{fig:eltdh-dcr}.

For a nonfull image parameter $r$, we write $\rcap\cqq\min\{r,2^\secparam\}$, and for $r>2^\secparam$, we sample the hash key from the distribution for $2^\secparam$, whose image bound $2^\secparam$ is smaller than $r$. With this cap, every exponent bound $\gamma$ in \cref{fig:eltdh-dcr} satisfies $\gamma<2^\secparam<\abs{\GN}$. The hash key contains one base for every bit of the padded binary expansions of the $d$ residues, and for $\binw(x_j)=(x_{j,0},\dots,x_{j,w-1})$, the encoding-key entry indexed by $(j,k)$ is $g_{j,k}^\sigma\Expo(2^ka_j)$. Since $\binw(x_j)$ is the binary expansion of $\rep(x_j)$, the product of the offsets selected by an input $\bvec x$ is
\[
	\Expo\left(\sum_{j=1}^d\sum_{k=0}^{w-1}2^ka_jx_{j,k}\right)
	=\Expo\left(\sum_{j=1}^da_jx_j\right).
\]

The same exponent $\rho'\cqq\rho\vert_{\expl}$, with $\expl\cqq\clog{N}+\secparam$, appears in hashing and in encoding. We totalize the setup with one fixed fallback key, independent of the image parameter,
\[
	\hkbot\cqq\left((\Nbot(\secparam),1),d,w,\explbot,1,(1)_{j,k}\right),
	\qquad
	\explbot\cqq\clog{\Nbot(\secparam)}+\secparam,
\]
which the setup returns when $\DJGen$ fails and, at the full parameter, also when one of the $dw+1$ samples from $\GN$ fails. Every other hash key differs from $\hkbot$, because its generator $g$ generates $\GN\ne\{1\}$. We sample group elements only at the full parameter, so at every nonfull parameter, the setup returns $\hkbot$ exactly when $\DJGen$ fails. We write $\fDJ(\secparam)$ for the failure probability of $\DJGen(1^\secparam)$, $\Efull$ for the event that the setup at the full parameter returns $\hkbot$, and
\[
	\ffull(\secparam,d)\cqq\Prob\left[\hk=\hkbot\mbar\hk\sample\ELTDHSetupLos(1^\secparam,1^d,\Rfull(\secparam,d))\right]
\]
for its probability. By the union bound, $\ffull(\secparam,d)\le \fDJ(\secparam)+(dw+1)\nuG(\secparam)$, where $\nuG(\secparam)$ bounds the failure probability of one sample from $\GN$\ifConferenceVersion\else{} as in \cref{def:group-sampler}\fi, so $\ffull$ is negligible whenever $d$ is polynomially bounded.

\begin{proofbox}{Extremely lossy trapdoor hash for residue-linear predicates. Here $\binw$ is the padded $w$-bit binary expansion, $\rho\vert_{\expl}$ is the integer formed by the first $\expl$ bits of $\rho$, $[\gamma]=\{1,\dots,\gamma\}$, $\hkbot$ is the fallback key, and an algorithm whose \pcassert{}condition fails returns $\bot$.}{fig:eltdh-dcr}
		\begin{minipage}[t]{0.49\linewidth}
			\vspace{0pt}
			{\footnotesize\underline{$\ELTDHSetupLos(1^\secparam,1^d,r)$}\par}
			\smallskip
			\begin{codelines}
				\item $w\gets2\ell(\secparam)$; $R\gets2^{dw}$
				\item \pcassert $d\ge1\wedge r\in[R]$
				\item $(N,g)\sample\DJGen(1^\secparam)$
				\item \pcif $(N,g)=\bot$ \pcreturn $\hk\gets\hkbot$
				\item $\expl\gets\clog{N}+\secparam$
				\item \pcif $r=R$
				\codelinesplus
					\item $g_0\sample\GN$; $(g_{j,k})_{j,k}\sample\GN^{dw}$
					\item \pcif $g_0=\bot\vee(\exists j,k\colon g_{j,k}=\bot)$ \pcreturn $\hk\gets\hkbot$
				\codelinesminus
				\item \pcelse
				\codelinesplus
				\item $\rcap\gets\min\{r,2^\secparam\}$; $\gamma\gets\floor{(\rcap-1)/(dw)}$
				\item \pcif $\gamma=0$
				\codelinesplus
				\item $g_0\gets1$; $(g_{j,k})_{j,k}\gets(1)_{j,k}$
				\codelinesminus
				\item \pcelse
				\codelinesplus
				\item $s_0\sample[\gamma]$; $g_0\gets g^{s_0}$
				\item \pcfor $(j,k)\in[d]\times[0..w-1]$
				\codelinesplus
				\item $s_{j,k}\sample[\gamma]$; $g_{j,k}\gets g^{s_{j,k}}$
				\codelinesminus
				\codelinesminus
				\codelinesminus
				\item \pcreturn $\hk\gets((N,g),d,w,\expl,g_0,(g_{j,k})_{j,k})$
			\end{codelines}

			\medskip
			{\footnotesize\underline{$\ELTDHHash(\hk,\bvec x;\rho)$}\par}
			\smallskip
			\begin{codelines}
				\item \pcparse $\hk=((N,g),d,w,\expl,g_0,(g_{j,k})_{j,k})$
				\item \pcassert $\bvec x\in(\ZN)^d$
				\item \pcfor $j\in[d]$
				\codelinesplus
				\item $(x_{j,0},\dots,x_{j,w-1})\gets\binw(x_j)$
				\codelinesminus
				\item $\rho'\gets\rho\vert_{\expl}$
				\item \pcreturn $h\gets g_0^{\rho'}\prod_{j=1}^d\prod_{k=0}^{w-1}g_{j,k}^{x_{j,k}}$
			\end{codelines}
		\end{minipage}\hfill%
		\begin{minipage}[t]{0.49\linewidth}
			\vspace{0pt}
			{\footnotesize\underline{$\ELTDHGen(\hk,\bvec a)$}\par}
			\smallskip
			\begin{codelines}
				\item \pcparse $\hk=((N,g),d,w,\expl,g_0,(g_{j,k})_{j,k})$
				\item \pcassert $\bvec a=(a_1,\dots,a_d)\in(\ZN)^d$
				\item $\sigma\sample[0..2^{\expl}-1]$
				\item $K_0\gets g_0^\sigma$
				\item \pcfor $(j,k)\in[d]\times[0..w-1]$
				\codelinesplus
				\item $K_{j,k}\gets g_{j,k}^\sigma\Expo(2^ka_j)$
				\codelinesminus
				\item $\ek\gets(N,d,w,\expl,K_0,(K_{j,k})_{j,k})$; $\td\gets(N,\sigma)$
				\item \pcreturn $(\ek,\td)$
			\end{codelines}

			\medskip
			{\footnotesize\underline{$\ELTDHEnc(\ek,\bvec x;\rho)$}\par}
			\smallskip
			\begin{codelines}
				\item \pcparse $\ek=(N,d,w,\expl,K_0,(K_{j,k})_{j,k})$
				\item \pcassert $\bvec x\in(\ZN)^d$
				\item \pcfor $j\in[d]$
				\codelinesplus
				\item $(x_{j,0},\dots,x_{j,w-1})\gets\binw(x_j)$
				\codelinesminus
				\item $\rho'\gets\rho\vert_{\expl}$
				\item $c\gets K_0^{\rho'}\prod_{j=1}^d\prod_{k=0}^{w-1}K_{j,k}^{x_{j,k}}$
				\item \pcreturn $e\gets\DDLog_N(c)$
			\end{codelines}

			\medskip
			{\footnotesize\underline{$\ELTDHDec(\td,h)$}\par}
			\smallskip
			\begin{codelines}
				\item \pcparse $\td=(N,\sigma)$
				\item \pcassert $h\in\ZNsqx$
				\item \pcreturn $e'\gets\DDLog_N(h^\sigma)$
			\end{codelines}
		\end{minipage}
	\end{proofbox}

\paperpart{Correctness, Input Privacy, and Lossiness}{subsec:eltdh-unconditional} Correctness, input privacy, and lossiness of the construction hold unconditionally, and we need computational assumptions only for function privacy and mode indistinguishability. Correctness follows from the common exponent $\sigma$ of the encoding key and the trapdoor, input privacy from the uniform randomizer base $g_0$ at the full parameter, and lossiness from the bounded exponents of the hash bases at nonfull parameters.

\begin{relocated}{dj-accounting}
Every element that we multiply in the construction is a unit modulo $N^2$. At a modulus that $\DJGen$ samples successfully, the published bases $g_0$ and $g_{j,k}$ lie in the subgroup $\GN$ of order $p'q'$, and so do every digest $h$ and the randomizer entry $K_0=g_0^\sigma$, since each of them is a product of powers of those bases. The other encoding-key entries lie in cosets of $\GN$. The entry $K_{j,k}=g_{j,k}^\sigma\Expo(2^ka_j)$ lies in $\GN\Expo(2^ka_j)$, and the encoder product
\[
	K_0^{\rho'}\prod_{j,k}K_{j,k}^{x_{j,k}}=h^\sigma\Expo(\Lina(\bvec x))
\]
lies in $\GN\Expo(\Lina(\bvec x))$. By \cref{lem:dj-structure:c} of \cref{lem:dj-structure}, each of these products has a unique factorization $c\Expo(y)$ with $c\in\GN$ and $y\in\ZN$, so we can recover the predicate value $\Lina(\bvec x)$ with $\DDLog_N$ without knowing $\sigma$. The message space $\ZN$, which $\Expo$ maps isomorphically onto the easy subgroup, contains the \ELTDH{} input residues, the predicate coefficients, the predicate value, and the two shares of the encoder and the decoder.

\end{relocated}%

\begin{relocated}{ddlog-shift}
Two parties who hold $c\Expo(x)$ and $c$ for an unknown $c\in\ZNsqx$ obtain additive shares of $x$ without communication, since $\DDLog_N$ maps $c\Expo(x)$ and $c$ to two residues whose difference is $x$. \cite{C:RoySin21} and~\cite{EC:OrlSchYak21} prove this identity, and we state it at the modulus $N^2$ for every $N\ge2$, which includes the fallback modulus.

\begin{lemmaE}[Shift Invariance]\label{thm:ddlog-shift}
	Let $N\ge2$ be an integer. For every $c\in\ZNsqx$ and every $x\in\ZN$, $\DDLog_N\left(c\Expo(x)\right)-\DDLog_N(c)=x$.
\end{lemmaE}

\begin{proofE}
	Since $\Expo(x)=1+Nx\equiv1\pmod N$, we have $c\Expo(x)\equiv c\pmod N$, so the two arguments have the same representative $w\in[1..N-1]$ in \cref{def:ddlog}. Since $\Logg$ is an isomorphism from $\easypart$ onto $(\ZN)^+$ by \cref{lem:easy-part-arbitrary-modulus}, we obtain
	\begin{align*}
		\DDLog_N\left(c\Expo(x)\right)-\DDLog_N(c)
		 & =\Logg\left(c\Expo(x)w^{-1}\right)
		-\Logg\left(cw^{-1}\right)                      \\
		 & =\Logg\left(
		\left(c\Expo(x)w^{-1}\right)
		\left(cw^{-1}\right)^{-1}
		\right)                                         \\
		 & =\Logg(\Expo(x))                             \\
		 & =x,
	\end{align*}
	where we use $\Logg\circ\Expo=\IdMap$ in the last equality.
\end{proofE}

\end{relocated}%

Since the encoding key and the trapdoor share the exponent $\sigma$, we obtain an exact identity between the two shares from shift invariance of $\DDLog_N$~\cite{EC:OrlSchYak21,C:RoySin21}. We adapt this identity of~\cite{C:RoySin21} to hash bases with bounded exponents and to the fallback key $\hkbot$.

\begin{theoremE}[Exact Share Identity]
	\label{thm:eltdh-share-correctness}
	\ifConferenceVersion\label{cor:eltdh-correctness}\fi
	Let $\secparam\in\NN$, $d\in\NN_{>0}$, $r\in[\Rfull(\secparam,d)]$, and $\hk\in\Supp(\ELTDHSetupLos(1^\secparam,1^d,r))$ with modulus $N$. For every $\bvec a\in(\ZN)^d$, every $(\ek,\td)\in\Supp(\ELTDHGen(\hk,\bvec a))$, every $\bvec x\in(\ZN)^d$, and every tape $\rho$, the digest $h\cqq\ELTDHHash(\hk,\bvec x;\rho)$, the encoding $e\cqq\ELTDHEnc(\ek,\bvec x;\rho)$, and the decoded value $e'\cqq\ELTDHDec(\td,h)$ satisfy $e-e'=\Lina(\bvec x)$.\ifConferenceVersion\ Hence, the scheme is perfectly correct, and $\Deltahk(\bvec x)=\Lina(\bvec x)$ holds simultaneously for every input under one fixed key.\fi
\end{theoremE}

\ConferenceProofSketch{proof:eltdh-share-identity}{
	The encoding key multiplies every hash base by the common exponent $\sigma$, and the offset selected by bit $k$ of $x_j$ is $\Expo(2^ka_j)$, so the encoder product equals $h^\sigma\Expo(\Lina(\bvec x))$. By shift invariance of $\DDLog_N$~\cite{EC:OrlSchYak21,C:RoySin21}, $\DDLog_N(h^\sigma\Expo(\Lina(\bvec x)))-\DDLog_N(h^\sigma)=\Lina(\bvec x)$ at every modulus $N\ge2$, so the identity holds for every supported key, including $\hkbot$, and every tape. }{}{}

\begin{proofE}
	\label{proof:eltdh-share-identity}
	We factor the encoder product as $h^\sigma\Expo(\Lina(\bvec x))$ and apply shift invariance of $\DDLog_N$. We write $\binw(x_j)=(x_{j,0},\dots,x_{j,w-1})$. By construction, $K_0=g_0^\sigma$ and $K_{j,k}=g_{j,k}^\sigma\Expo(2^ka_j)$, and the same $\rho'$ appears in hashing and in encoding. Since $\Expo$ is a homomorphism, the product $c$ computed by $\ELTDHEnc$ satisfies
	\begin{align*}
		c
		&=K_0^{\rho'}\prod_{j=1}^d\prod_{k=0}^{w-1}K_{j,k}^{x_{j,k}}\\
		&=h^\sigma\Expo\left(\sum_{j=1}^d\sum_{k=0}^{w-1}2^ka_jx_{j,k}\right)\\
		&=h^\sigma\Expo\left(\sum_{j=1}^da_jx_j\right),
	\end{align*}
	where the last equality holds because $\binw(x_j)$ is the binary expansion of $\rep(x_j)$. Hence, by \cref{thm:ddlog-shift} at $h^\sigma\in\ZNsqx$ and $\Lina(\bvec x)\in\ZN$,
	\[
		e-e'=\DDLog_N\left(h^\sigma\Expo(\Lina(\bvec x))\right)-\DDLog_N\left(h^\sigma\right)=\Lina(\bvec x).
	\]
	The calculation holds for every choice of bases, and \cref{lem:easy-part-arbitrary-modulus,thm:ddlog-shift} hold at every modulus $N\ge2$, so the identity also holds at $\hkbot$.
\end{proofE}

\ifConferenceVersion\else
Since the identity holds for every tape, it implies perfect correctness and, at the all-zero tape, one identity for all inputs under one key.

\begin{corollaryE}[Perfect Correctness]
	\label{cor:eltdh-correctness}
	The scheme $\ELTDH$ is $0$-correct. Moreover, for every $\secparam\in\NN$, $d\in\NN_{>0}$, $r\in[\Rfull(\secparam,d)]$, $\hk\in\Supp(\ELTDHSetupLos(1^\secparam,1^d,r))$ with modulus $N$, $\bvec a\in(\ZN)^d$, and $(\ek,\td)\in\Supp(\ELTDHGen(\hk,\bvec a))$, we have $\Deltahk(\bvec x)=\Lina(\bvec x)$ for every $\bvec x\in(\ZN)^d$.
\end{corollaryE}

\begin{proofE}
	The encoder returns $e$ and the decoder returns $e'$, so $0$-correctness follows from \cref{thm:eltdh-share-correctness}. At the all-zero tape, we obtain $\Deltahk(\bvec x)=\Lina(\bvec x)$ for every $\bvec x\in(\ZN)^d$ under the same keys from \cref{thm:eltdh-share-correctness}.
\end{proofE}
\fi

The digest must hide the input at the full parameter. There, the randomizer base $g_0$ is uniform in $\GN$, so it generates $\GN$ except with probability at most $2^{-\ell(\secparam)+3}$, and for a generator $g_0$, the digest is close to uniform on $\GN$ for every input. We adapt the statistical input-privacy argument of~\cite{C:RoySin21} to the exponent range $[0..2^{\expl}-1]$ and to the fallback key.

\begin{theoremE}[Statistical Input Privacy]
	\label{thm:eltdh-input-privacy}
	Let $d\colon\NN\to\NN_{>0}$ be polynomially bounded, and let $\advS$ be a selector as in the input-privacy clause of \cref{def:eltdh}. For every $\secparam\in\NN$, the statistical distance between $\Dinpriv_{0,\secparam}$ and $\Dinpriv_{1,\secparam}$ is at most $2^{-\secparam-1}+2^{-\ell(\secparam)+4}$. In particular, $\ELTDH$ has statistical input privacy.
\end{theoremE}

\ConferenceProofSketch{proof:eltdh-input-privacy}{
	For a fixed hash key and a fixed input, the digest is $g_0^{\rho'}$ times a fixed element of $\GN$. If $g_0$ generates $\GN$, then by exponent smoothing~\cite{C:BraGol10}, $g_0^{\rho'}$ is within $2^{-\secparam-2}$ of uniform on $\GN$, and multiplication by the fixed element preserves the uniform distribution. The base $g_0$ fails to generate $\GN$ with probability at most $2^{-\ell(\secparam)+3}$, the digest at $\hkbot$ is $1$ for both inputs, and we compare both inputs with the same reference distribution. }{}{}

\begin{proofE}
	\label{proof:eltdh-input-privacy}
	We compare both ensembles with one reference distribution and apply the triangle inequality. We write $\Dref_\secparam$ for the distribution of $(\hk,\var{st},u)$, where $\hk$ and $\var{st}$ are sampled as in $\Dinpriv_{b,\secparam}$, $u$ is uniform in $\GN$ if $\hk\ne\hkbot$, and $u=1$ if $\hk=\hkbot$.

	We fix $\hk\ne\hkbot$ in the support of the setup at the full parameter and an output $(x^0,x^1,\var{st})$ of $\advS(1^\secparam,\hk)$, and we write $c_b\in\GN$ for the product of the bases selected by the bits of $x^b$, so that $h^b=g_0^{\rho'}c_b$. Conditioned on these values, $\rho'$ is uniform on $[0..2^{\expl}-1]$. If $g_0$ generates $\GN$, then $s\mapsto g_0^s$ is a bijection from $\ZZ/\abs{\GN}\ZZ$ onto $\GN$, so by \cref{lem:exponent-smoothing} with $M'=\abs{\GN}\le N$, the element $g_0^{\rho'}$ is within statistical distance $2^{-\secparam-2}$ of uniform on $\GN$. Multiplication by $c_b$ is a bijection of $\GN$, so $h^b$ is within the same distance of uniform on $\GN$. If $g_0$ does not generate $\GN$, we bound the distance by $1$.

	Conditioned on $\hk\ne\hkbot$ and on the modulus $N=pq$, the base $g_0$ is uniform in the cyclic group $\GN$ of order $p'q'$ by \cref{lem:dj-structure}. Since $p'\ge2^{\ell(\secparam)-2}$ and $q'\ge2^{\ell(\secparam)-2}$, we obtain
	\[
		\Prob\left[\langle g_0\rangle\ne\GN\mbar g_0\sample\GN\right]\le\frac1{p'}+\frac1{q'}\le2^{-\ell(\secparam)+3}.
	\]
	At $\hkbot$, the digest is $1$ for both inputs, as in $\Dref_\secparam$. Hence, for each $b\in\bits$, the statistical distance between $\Dinpriv_{b,\secparam}$ and $\Dref_\secparam$ is at most $2^{-\secparam-2}+2^{-\ell(\secparam)+3}$, and by the triangle inequality, we obtain the bound of the theorem.
\end{proofE}

At a nonfull parameter, the hash with the all-zero tape depends on the input only through one subset sum of the base exponents, which takes at most $r$ values.

\begin{theoremE}[Image Bound]
	\label{thm:eltdh-lossy-image}
	Let $\secparam\in\NN$, $d\in\NN_{>0}$, and $r\in[\Rfull(\secparam,d)]$. For every $\hk\in\Supp(\ELTDHSetupLos(1^\secparam,1^d,r))$, we have $\abs{\Image(\Phihk)}\le r$. In particular, $\ELTDH$ satisfies the lossiness clause of \cref{def:eltdh} with $\nulos=0$.
\end{theoremE}

\ConferenceProofSketch{proof:eltdh-lossy-image}{
	At a hash key other than $\hkbot$ with $\gamma\ge1$, every digest at the all-zero tape equals $g^{A_{\bvec x}}$ for an integer subset sum $A_{\bvec x}\in[0..dw\gamma]$ of the base exponents, and since $\gamma=\floor{(\rcap-1)/(dw)}$, this interval contains at most $\rcap\le r$ integers. At the full parameter, the bound follows from $N^d<2^{dw}$, and at $\hkbot$ and for $\gamma=0$, the hash is constant. }{}{}

\begin{proofE}
	\label{proof:eltdh-lossy-image}
	We count the values of one bounded subset sum at nonfull parameters and treat the full parameter, the fallback key, and the branch $\gamma=0$ separately. We write $w=2\ell(\secparam)$ and $R=2^{dw}$. If $r=R$, then $\abs{\Image(\Phihk)}\le\abs{\calX}=N^d<2^{dw}=r$, since $N<2^w$. If $r<R$, we write $\rcap=\min\{r,2^\secparam\}$ and $\gamma\cqq\floor{(\rcap-1)/(dw)}$.

	At $\hkbot$ and in the branch $\gamma=0$, all hash bases are $1$, so $\Phihk$ is constant and $\abs{\Image(\Phihk)}=1\le r$. We suppose that $\gamma\ge1$, and we write $s_{j,k}\in[\gamma]$ for the exponents of the bases $g_{j,k}=g^{s_{j,k}}$. For $\binw(x_j)=(x_{j,0},\dots,x_{j,w-1})$, we define
	\[
		A_{\bvec x}\cqq\sum_{j=1}^d\sum_{k=0}^{w-1}s_{j,k}x_{j,k}.
	\]
	Then $0\le A_{\bvec x}\le dw\gamma$, and at the all-zero tape, $\Phihk(\bvec x)=g^{A_{\bvec x}}$. Hence, $\Phihk$ takes at most $dw\gamma+1$ values, and $\abs{\Image(\Phihk)}\le dw\gamma+1 \le \rcap \le r$.
\end{proofE}

\ifConferenceVersion\else
We count group elements and bits to verify the compactness and rate clauses of \cref{def:eltdh}.

\begin{lemmaE}[Efficiency, Rate, and Compactness]
	\label{lem:eltdh-efficiency}
	Let $\secparam\in\NN$, $d\in\NN_{>0}$, and $r\in[\Rfull(\secparam,d)]$, let $\hk\in\Supp(\ELTDHSetupLos(1^\secparam,1^d,r))$ have modulus $N$ and exponent length $\expl$, and let $w=2\ell(\secparam)$. Excluding the modulus, the generator, and integer metadata, the hash key contains $dw+1$ group elements, every encoding key generated from it contains $dw+1$ group elements, and every corresponding trapdoor contains one exponent of $\expl$ bits. A digest occupies at most $2\clog{N}$ bits and an encoding exactly $\clog{N}$ bits. Every algorithm runs in time $\poly(\secparam+d)$, and the scheme has rate one.
\end{lemmaE}

\begin{proofE}
	We count the entries of the two public key vectors and their representation lengths. Each vector contains one randomizer entry and one entry for each pair $(j,k)\in[d]\times[0..w-1]$. Since $w=2\ell(\secparam)=\poly(\secparam)$ and every exponent has at most $\expl=\poly(\secparam)$ bits, the lengths of both vectors and the work needed to form or evaluate them are polynomial in $\secparam+d$. The digest is one unit modulo $N^2$, of at most $2\clog{N}$ bits, and the encoding is one residue modulo $N$, of $\clog{N}$ bits. Apart from $N$, the trapdoor consists of $\sigma\in[0..2^{\expl}-1]$. The predicate value $\Lina(\bvec x)$ and the encoder output both lie in $\ZN$, so they have the same representation length and the rate is one. By \cref{def:fallback-modulus}, the same counts hold at $\hkbot$.
\end{proofE}
\fi

\paperpart{Function Privacy and Mode Indistinguishability}{subsec:eltdh-security} For the two remaining properties, we change the predicate offsets of the encoding key and the distribution of the hash bases, one at a time. With the interactive vector game of~\cite{C:BraGol10} under decisional composite residuosity~\cite{EC:Paillier99}, as in~\cite{C:DGIMMO19}, we remove polynomially many offsets in $\easypart$ from elements of $\GN$ with one common exponent, and under tunable \DSEDL{} as in \cref{def:tunable-dsedl}, we replace bounded bases by uniform ones.

\begin{relocated}{vector-game}
We remove the predicate-dependent offsets from the encoding key with the interactive vector game of~\cite{C:BraGol10}, in which the adversary receives a vector of subgroup elements and then chooses offsets adaptively. \cite{C:BraGol10} bounds the advantage in this game under quadratic residuosity, and~\cite{C:BraGol10} proves the same bound for every subgroup-indistinguishability problem, which we recall in our notation.

\begin{lemma}[Interactive Vector Game]\label{lem:interactive-vector-game}
	Let $\GG_U=\GG_M\times\GG_L$ be a finite abelian group with $\gcd(\abs{\GG_M},\abs{\GG_L})=1$, where $\GG_M$ is cyclic and has a generator that is efficiently computable from the description of $\GG_U$. Suppose that the group operations in $\GG_U$ and uniform sampling from $\GG_M$ and from $\GG_L$ are efficient, and let $T\ge\abs{\GG_U}$ be a known integer. For an algorithm $\advB$, let
	\[
		\varepsilon_{\var{SG}}(\advB)\cqq\abs{\Prob\left[\advB(u)=1\mbar u\sample\GG_U\right]-\Prob\left[\advB(u)=1\mbar u\sample\GG_L\right]}.
	\]
	Let $D\in\NN_{>0}$. In the interactive $D$-vector game, the challenger samples $b\sample\bits$ and $\bvec g\sample\GG_L^D$ and gives $\bvec g$ to the adversary. On each query $\bvec a\in\GG_M^D$, which the adversary chooses adaptively after it receives $\bvec g$, the challenger samples $r\sample[T^2]$ and returns the coordinatewise product $\bvec a^{1-b}\cdot\bvec g^{\,r}$. Finally, the adversary outputs a bit $b'$. For an adversary $\advA$, let
	\[
		\varepsilon_{\var{IV}}(\advA)\cqq\abs{\Prob[b'=1\mid b=0]-\Prob[b'=1\mid b=1]}.
	\]
	Then for every adversary $\advA$ that makes at most $t$ queries, there exists an adversary $\advB$ such that
	\[
		\varepsilon_{\var{IV}}(\advA)\le4tD\cdot\varepsilon_{\var{SG}}(\advB)+\frac{2tD}{T}.
	\]
\end{lemma}

\cite{C:BraGol10} gives no bound on the running time of $\advB$ in this lemma. In the proof of~\cite{C:BraGol10}, $\advB$ runs $\advA$ once and otherwise performs $\poly(t,D,\log T)$ group operations in $\GG_U$ and samplings, so $\advB$ is \PPT{} whenever $\advA$ is \PPT{} and $t$, $D$, and $\log T$ are polynomially bounded.

Since the adversary chooses the offsets after it receives the bases, the predicates can depend on the modulus, and every coordinate of an answer contains a power of the corresponding base with the same exponent $r$. By~\cite{C:BraGol10}, the \DamgardJurik group is an instance, with $\GG_U=\ZNsqx$, with $\GG_M=\easypart$, which is cyclic of order $N$ and generated by $1+N$, and with $\GG_L$ the subgroup of $N$-th residues, of order $\varphi$. These two orders are coprime, because $N$ is odd and $p',q'\notin\{p,q\}$, and $T=N^2$ is a valid bound, since $\abs{\GG_U}=N\varphi<N^2$ by \cref{lem:dj-structure:a} of \cref{lem:dj-structure}.

\end{relocated}%
\ifConferenceVersion\else

In our encoding keys, the bases are uniform in $\GN$, the exponent is uniform in $[0..2^{\expl}-1]$, and the offsets are chosen after the modulus and the bases. The function-privacy arguments of~\cite{C:DGIMMO19} and~\cite{C:RoySin21} remove such offsets for bases $r^N$ with $r$ uniform in $\ZNsqx$ and for exponents uniform in $[0..N-1]$. We adapt the offset removal of~\cite{C:BraGol10,C:DGIMMO19} to our distributions and reduce it to \cref{lem:interactive-vector-game} at one query.

\fi
\begin{relocated}{offset-removal}
\begin{lemmaE}[Offset Removal]
	\label{lem:dcr-offset-removal}
	Suppose that $\DJGen(\cdot)$ satisfies $\DCR$ against nonuniform $\PPT$ adversaries, and let $d\colon\NN\to\NN_{>0}$ be polynomially bounded. For a nonuniform $\PPT$ algorithm $\advA$ that on input $(N,\bvec g)$ returns a state $\var{st}$ and offsets $\beta_1,\dots,\beta_{d(\secparam)}\in\ZN$, and for $b\in\bits$, define the ensemble $\Doff_b=(\Doff_{b,\secparam})_{\secparam\in\NN}$ by
	\[
		\Doff_{b,\secparam} \cqq
		\left\{\left(N,\bvec g,\var{st},\left(g_j^\sigma\Expo(b\beta_j)\right)_{j\in[d(\secparam)]}\right)\ \middle|\
		\begin{array}{l}
			(N,g)\sample\DJGen(1^\secparam),                                  \\[-1pt]
			\bvec g\cqq(g_j)_{j\in[d(\secparam)]}\sample\GN^{d(\secparam)}, \\[-1pt]
			(\var{st},(\beta_j)_j)\sample\advA(N,\bvec g),                \\[-1pt]
			\sigma\sample[0..2^{\expl}-1]
		\end{array}\right\},
	\]
	where $\expl=\clog{N}+\secparam$, and where the sample is $\bot$ if $\DJGen(1^\secparam)$ or one of the $d(\secparam)$ samples from $\GN$ returns $\bot$. Then $\Doff_1\cindist\Doff_0$ for every such $\advA$.
\end{lemmaE}

\begin{proofE}
	\label{proof:dcr-offset-removal}
	We reduce to \cref{lem:interactive-vector-game} at one query. Since squaring maps the $N$-th powers of the game onto $\GN$ and doubles the offsets, we halve the offsets first, and we then account for the exponent range of the game and for the failures of the sampler of $\GN$.

	\subparhead{The Instance} For a successful output $(N,g)$ of $\DJGen(1^\secparam)$, we write $N=pq$, $\varphi=(p-1)(q-1)$, and $\GG_L\cqq\{z^N\mid z\in\ZNsqx\}$. We repeat the computation of \cref{lem:dj-structure} with exponent $N$ in place of $2N$ and find $\abs{\GG_L}=\varphi$. As we noted after \cref{lem:interactive-vector-game}, the hypotheses of that lemma hold for $D\cqq d(\secparam)$, $\GG_U\cqq\ZNsqx$, $\GG_M\cqq\easypart$, this $\GG_L$, and $T\cqq N^2$. Since $z\mapsto z^N$ maps the uniform distribution on $\ZNsqx$ to the uniform distribution on $\GG_L$, distinguishing a uniform element of $\GG_U$ from a uniform element of $\GG_L$ is distinguishing the two ensembles of \cref{def:dcr} at the modulus $N$. The group in \cref{lem:interactive-vector-game} is fixed, and~\cite{C:BraGol10} proves the lemma for a group output by a generator, with both advantages taken over the generator. We apply it in that form, with the group determined by the modulus of $\DJGen(1^\secparam)$ conditioned on success.

	\subparhead{The Reduction} We fix $\advA$ and a nonuniform $\PPT$ distinguisher $\advD$, which outputs a bit on every input, including $\bot$. We define an adversary $\advC$ for the interactive $D$-vector game as follows. On input $N$ and the bases $\bvec u=(u_j)_{j\in[D]}\in\GG_L^D$, the adversary $\advC$ sets $g_j\gets u_j^2$ for $j\in[D]$ and $\bvec g\cqq(g_j)_{j\in[D]}$, runs $(\var{st},(\beta_j)_j)\sample\advA(N,\bvec g)$, and makes the single query $\bvec a\cqq(\Expo(2^{-1}\beta_j))_{j\in[D]}\in\GG_M^D$, where $2^{-1}$ is the inverse of $2$ in $\ZN$, which exists because $N$ is odd. On the answer $\bvec y=(y_j)_{j\in[D]}$, it outputs $\advD(N,\bvec g,\var{st},(y_j^2)_{j\in[D]})$. The adversary $\advC$ is nonuniform $\PPT$ because $\advA$ and $\advD$ are.

	\subparhead{The Simulated Distribution} Since $\GN=\{(z^N)^2\mid z\in\ZNsqx\}$, squaring is a surjective homomorphism from $\GG_L$ onto $\GN$, so $\bvec g$ is uniform on $\GN^D$. For the challenge bit $b'$ of the game, the answer has the coordinates $y_j=\Expo(2^{-1}\beta_j)^{1-b'}u_j^{r}$ with $r\sample[T^2]$. Since $\Expo$ is a homomorphism and $2\cdot2^{-1}\beta_j=\beta_j$ in $\ZN$,
	\[
		y_j^2=\Expo(\beta_j)^{1-b'}g_j^{r}
		\qquad(j\in[D]).
	\]
	Hence, for $b=1-b'$ and conditioned on the modulus $N$, the input of $\advD$ is distributed as a sample of $\Doff_{b,\secparam}$ conditioned on $N$ and on the success of the $D$ samples from $\GN$, except that the exponent is $r\sample[T^2]$.

	\subparhead{The Exponent Range} Every $g_j$ lies in $\GN$, so the input of $\advD$ depends on the exponent only through its residue modulo $M\cqq\abs{\GN}\le N$. For $L$ consecutive integers, the computation in the proof of \cref{lem:exponent-smoothing} bounds the statistical distance of their residue modulo $M$ from uniform on $\ZZ/M\ZZ$ by $M/(4L)$. For $r\sample[T^2]$, where $L=N^4$, this distance is at most $1/(4N^3)\le2^{-\secparam-2}$, and for $\sigma\sample[0..2^{\expl}-1]$, it is at most $2^{-\secparam-2}$ by \cref{lem:exponent-smoothing}. Hence, replacing $r$ by $\sigma$ changes each acceptance probability of $\advD$ by at most $2^{-\secparam-1}$.

	\subparhead{The Bound} For a successful output of $\DJGen(1^\secparam)$ with modulus $N$, we write $\pi_N$ for the probability that $\DJGen(1^\secparam)$ returns a pair with modulus $N$ and $f(N)$ for the probability that one of the $D$ samples from $\GN$ fails. We write $\Delta(N)$ for the probability that $\advD$ outputs $1$ on $\Doff_{1,\secparam}$ minus the probability that it outputs $1$ on $\Doff_{0,\secparam}$, both conditioned on $N$ and on the success of these samples, and $\Delta_{\var{IV}}(N)$ for the probability that $\advC$ outputs $1$ at $b'=0$ minus the probability that it outputs $1$ at $b'=1$ in the game at the modulus $N$. After a failed sample, the output is $\bot$ for both values of $b$, $f(N)\le D\nuG(\secparam)$ by \cref{def:group-sampler}, and $\abs{\Delta(N)-\Delta_{\var{IV}}(N)}\le2^{-\secparam}$ by the exponent range. For the adversary $\advB$ that \cref{lem:interactive-vector-game} constructs from $\advC$, with output $0$ on input $\bot$, we write
	\[
		\varepsilon_{\DCR}(\advB)\cqq\abs{\Prob\left[\advB(X)=1\mbar X\sample\DDCR_{0,\secparam}\right]-\Prob\left[\advB(X)=1\mbar X\sample\DDCR_{1,\secparam}\right]}.
	\]
	We sum over the successful outputs of $\DJGen(1^\secparam)$ and apply \cref{lem:interactive-vector-game} with one query and $N>2^{2\ell(\secparam)-2}$, and we obtain
	\begin{align*}
		&\abs{\Prob\left[\advD(X)=1\mbar X\sample\Doff_{1,\secparam}\right]-\Prob\left[\advD(X)=1\mbar X\sample\Doff_{0,\secparam}\right]}\\
		&\qquad=\abs{\sum_N\pi_N\left(1-f(N)\right)\Delta(N)}\\
		&\qquad\le\abs{\sum_N\pi_N\Delta_{\var{IV}}(N)}+2^{-\secparam}+D\nuG(\secparam)\\
		&\qquad\le4D\cdot\varepsilon_{\DCR}(\advB)+2D\cdot2^{-4\ell(\secparam)+4}+2^{-\secparam}+D\nuG(\secparam).
	\end{align*}
	The adversary $\advB$ is nonuniform $\PPT$, so $\varepsilon_{\DCR}(\advB)$ is negligible under $\DCR$, and since $D=d(\secparam)$ is polynomially bounded and $\nuG$ is negligible, the right-hand side is negligible.
\end{proofE}

\end{relocated}%
\ifConferenceVersion We remove the offsets of each of the two predicates in turn\else We apply \cref{lem:dcr-offset-removal} once for each of the two predicates\fi{} and obtain function privacy at the full parameter, adapting the function-privacy argument of~\cite{C:RoySin21} to encoding keys with one entry for every bit of every residue and to the fallback key.

\begin{theoremE}[Function Privacy]
	\label{thm:eltdh-function-privacy}
	Let $d\colon\NN\to\NN_{>0}$ be polynomially bounded. If $\DJGen(\cdot)$ satisfies $\DCR$ against nonuniform $\PPT$ adversaries, then $\ELTDH$ is function private for $\Flin$ in the sense of \cref{def:eltdh}.
\end{theoremE}

\ConferenceProofSketch{proof:eltdh-function-privacy}{
	We compare both encoding keys with the encoding key in which every offset is removed. By offset removal under \DCR{}~\cite{C:BraGol10,C:DGIMMO19} at $dw+1$ coordinates, each of the two transitions changes the acceptance probability of a distinguisher by a negligible amount, up to a statistical loss of at most $2^{-\ell(\secparam)+4}$ at each end of a transition, which we pay because our reduction publishes $z^{2N}$ for a uniform $z$ in place of the generator of $\DJGen$. The fallback event $\Efull$ adds its negligible probability $\ffull$. }{}{}

\begin{proofE}
\label{proof:eltdh-function-privacy}
We compare both encoding keys with an encoding key without offsets and apply \cref{lem:dcr-offset-removal} at each of the two transitions, and we account separately for the fallback event $\Efull$ and for replacing the generator of $\DJGen$ with a uniform element of $\GN$. In the hybrids of \cref{fig:eltdh-function-privacy-hybrids}, the hash key at the full parameter is sampled in the same way, and we write $\bvec g=(g_{j,k})_{j,k}$ and $\bvec K=(K_{j,k})_{j,k}$.

\begin{proofbox}{Function-privacy hybrids for \ELTDH{}. Each hybrid outputs $(\hk,\var{st},\ek)$, and the highlighted line is the line changed from the previous hybrid.}{fig:eltdh-function-privacy-hybrids}
		\begin{minipage}[t]{0.32\linewidth}
			\vspace{0pt}
			{\footnotesize\underline{$\Hfp_0$}\codelinesrightalign{first predicate}\par}
			\smallskip
			\begin{codelines}
				\item $w\gets2\ell(\secparam)$; $(N,g)\sample\DJGen(1^\secparam)$
				\item $\expl\gets\clog{N}+\secparam$; $(g_0,\bvec g)\sample\GN^{dw+1}$
				\item $\hk\gets((N,g),d,w,\expl,g_0,\bvec g)$
				\item $(\bvec a^0,\bvec a^1,\var{st})\sample\advS(1^\secparam,\hk)$
				\item $\sigma\sample[0..2^{\expl}-1]$; $K_0\gets g_0^\sigma$
				\item \pcfor $(j,k)\in[d]\times[0..w-1]$
				\codelinesplus
				\item $K_{j,k}\gets g_{j,k}^\sigma\Expo(2^ka^0_j)$
				\codelinesminus
				\item $\ek\gets(N,d,w,\expl,K_0,\bvec K)$
				\item \pcreturn $(\hk,\var{st},\ek)$
			\end{codelines}
		\end{minipage}\hfill%
		\begin{minipage}[t]{0.32\linewidth}
			\vspace{0pt}
			{\footnotesize\underline{$\Hfp_1$}\codelinesrightalign{remove offsets}\par}
			\smallskip
			\begin{codelines}
				\item $w\gets2\ell(\secparam)$; $(N,g)\sample\DJGen(1^\secparam)$
				\item $\expl\gets\clog{N}+\secparam$; $(g_0,\bvec g)\sample\GN^{dw+1}$
				\item $\hk\gets((N,g),d,w,\expl,g_0,\bvec g)$
				\item $(\bvec a^0,\bvec a^1,\var{st})\sample\advS(1^\secparam,\hk)$
				\item $\sigma\sample[0..2^{\expl}-1]$; $K_0\gets g_0^\sigma$
				\item \pcfor $(j,k)\in[d]\times[0..w-1]$
				\codelinesplus
				\item \label{line:eltdh-fp-zero}\hybchange{$K_{j,k}\gets g_{j,k}^\sigma$}
				\codelinesminus
				\item $\ek\gets(N,d,w,\expl,K_0,\bvec K)$
				\item \pcreturn $(\hk,\var{st},\ek)$
			\end{codelines}
		\end{minipage}\hfill%
		\begin{minipage}[t]{0.32\linewidth}
			\vspace{0pt}
			{\footnotesize\underline{$\Hfp_2$}\codelinesrightalign{second predicate}\par}
			\smallskip
			\begin{codelines}
				\item $w\gets2\ell(\secparam)$; $(N,g)\sample\DJGen(1^\secparam)$
				\item $\expl\gets\clog{N}+\secparam$; $(g_0,\bvec g)\sample\GN^{dw+1}$
				\item $\hk\gets((N,g),d,w,\expl,g_0,\bvec g)$
				\item $(\bvec a^0,\bvec a^1,\var{st})\sample\advS(1^\secparam,\hk)$
				\item $\sigma\sample[0..2^{\expl}-1]$; $K_0\gets g_0^\sigma$
				\item \pcfor $(j,k)\in[d]\times[0..w-1]$
				\codelinesplus
				\item \label{line:eltdh-fp-second}\hybchange{$K_{j,k}\gets g_{j,k}^\sigma\Expo(2^ka^1_j)$}
				\codelinesminus
				\item $\ek\gets(N,d,w,\expl,K_0,\bvec K)$
				\item \pcreturn $(\hk,\var{st},\ek)$
			\end{codelines}
		\end{minipage}
	\end{proofbox}

	We fix a selector $\advS$ and a distinguisher $\advA$, and we write $d=d(\secparam)$, $w=2\ell(\secparam)$, and $\ffull(\secparam)\cqq \ffull(\secparam,d)$ for the probability of the fallback event $\Efull$ of \cref{subsec:eltdh-construction}. For $b\in\bits$, we write $p_b$ for the probability that $\advA$ outputs $1$ on a sample of $\Dfnpriv_{b,\secparam}$. The hybrids make the same calls to $\DJGen$ and to the sampler of $\GN$ as the setup at $\Rfull(\secparam,d)$, so all these calls succeed with probability $1-\ffull(\secparam)$. For $i\in\{0,1,2\}$, we write $\pcond_i$ for the probability that $\advA$ outputs $1$ on the output of $\Hfp_i$, conditioned on the success of all these calls, and $q_i\cqq(1-\ffull(\secparam))\pcond_i$.

	\subparhead{The Fallback Event} Conditioned on the complement of $\Efull$, the setup at $\Rfull(\secparam,d)$ samples $(N,g)$ and the $dw+1$ bases as the hybrids do, and $\ELTDHGen$ samples $\sigma$ from the same range $[0..2^{\expl}-1]$. Hence, conditioned on the complement of $\Efull$, a sample of $\Dfnpriv_{0,\secparam}$ is distributed as the output of $\Hfp_0$ and a sample of $\Dfnpriv_{1,\secparam}$ as the output of $\Hfp_2$, each conditioned on the success of all sampling calls, so
	\[
		p_b=(1-\ffull(\secparam))\pcond_{2b}+\ffull(\secparam)\pbot_b
		\qquad(b\in\bits),
	\]
	where $\pbot_b$ is the probability that $\advA$ outputs $1$ conditioned on $\Efull$, that is, at $\hk=\hkbot$. At $\hkbot$, the encoding key has the entries $K_{j,k}=\Expo(2^ka^b_j)$, so we use only the bound $\abs{\pbot_0-\pbot_1}\le1$. Hence,
	\[
		\abs{p_0-p_1}
		\le \ffull(\secparam)+(1-\ffull(\secparam))\abs{\pcond_0-\pcond_2}
		=\ffull(\secparam)+\abs{q_0-q_2}.
	\]

	\subparhead{The Two Transitions} We bound $\abs{q_0-q_1}$ and $\abs{q_1-q_2}$ with \cref{lem:dcr-offset-removal} at $dw+1$ coordinates, a number that is polynomially bounded because $d$ and $w$ are. For $i\in\bits$, we define an offset-choosing algorithm $\advB_i$ and a distinguisher $\advD_i$ for that lemma. On input $(N,(g_0,\bvec g))$, the algorithm $\advB_i$ samples $z\sample\ZZ/N^2\ZZ$, sets $g\gets z^{2N}$, $\expl\gets\clog{N}+\secparam$, and $\hk\gets((N,g),d,w,\expl,g_0,\bvec g)$, runs $(\bvec a^0,\bvec a^1,\var{st})\sample\advS(1^\secparam,\hk)$, and returns the state $(\hk,\var{st})$ and the offsets $(0,(2^ka^i_j)_{(j,k)\in[d]\times[0..w-1]})$. On input $\bot$, the distinguisher $\advD_i$ outputs $0$, and on input $(N,(g_0,\bvec g),(\hk,\var{st}),(K_0,\bvec K))$, it sets $\ek\gets(N,d,w,\expl,K_0,\bvec K)$ and outputs $\advA(\hk,\var{st},\ek)$. Both algorithms are nonuniform $\PPT$ because $\advS$ and $\advA$ are. We write $\Doff_{b,\secparam}$ for the ensemble of \cref{lem:dcr-offset-removal} at $dw+1$ coordinates with the algorithm $\advB_i$, and
	\[
		\varepsilon_i(\secparam)\cqq\abs{\Prob\left[\advD_i(X)=1\mbar X\sample\Doff_{1,\secparam}\right]-\Prob\left[\advD_i(X)=1\mbar X\sample\Doff_{0,\secparam}\right]},
	\]
	which is negligible by \cref{lem:dcr-offset-removal}.

	A sample of $\Doff_{b,\secparam}$ is $\bot$ exactly when $\DJGen$ or one of the $dw+1$ samples from $\GN$ fails, which has probability $\ffull(\secparam)$, and then $\advD_i$ outputs $0$. Conditioned on success, the input of $\advA$ is distributed as the output of $\Hfp_{2i}$ for $b=1$ and as the output of $\Hfp_1$ for $b=0$, each conditioned on the success of all sampling calls, except for the generator: it is $z^{2N}$ with $z$ uniform in $\ZZ/N^2\ZZ$, whereas the hybrids use the generator of $\DJGen$. Conditioned on $N$, both generators are independent of the bases, and the generator of $\DJGen$ is uniform among the generators of $\GN$. By \cref{lem:dj-structure}, the element $z^{2N}$ is within statistical distance $\Prob[\gcd(z,N)\ne1]<2^{-\ell(\secparam)+2}$ of uniform on $\GN$, and a uniform element of $\GN$ fails to generate $\GN$ with probability at most $1/p'+1/q'\le2^{-\ell(\secparam)+3}$. Hence, conditioned on $N$, the two generators are within statistical distance $\eta(\secparam)\cqq2^{-\ell(\secparam)+4}$. Hence, each of the two acceptance probabilities of $\advD_i$ differs from the corresponding $q$ by at most $(1-\ffull(\secparam))\eta(\secparam)$, and we obtain
	\[
		\abs{q_i-q_{i+1}}
		\le\varepsilon_i(\secparam)+2\left(1-\ffull(\secparam)\right)\eta(\secparam)
		\qquad(i\in\bits).
	\]

	\subparhead{The Bound} By the triangle inequality and $\ell(\secparam)\ge\secparam+2$,
	\[
		\begin{aligned}
			\abs{p_0-p_1}
			& \le \ffull(\secparam)+\varepsilon_0(\secparam)+\varepsilon_1(\secparam)+4\left(1-\ffull(\secparam)\right)\eta(\secparam) \\
			& \le \ffull(\secparam)+\varepsilon_0(\secparam)+\varepsilon_1(\secparam)+2^{-\secparam+4}.
		\end{aligned}
	\]
	In \cref{subsec:eltdh-construction}, we bounded $\ffull$ by a negligible function, so $\abs{p_0-p_1}$ is negligible.
\end{proofE}

For mode indistinguishability, we change only the distribution of the hash bases, from bounded powers to uniform elements, with one application of \cref{lem:tunable-scalar-to-vec-dsedl} to all $dw+1$ bases. Since $\rcap\le2^\secparam$ at every nonfull parameter, we have $\gamma<2^\secparam$, so $\gamma$ is smaller than the order of every \DamgardJurik subgroup $\GN$, as that lemma requires.\ifConferenceVersion\else{} The same inequality places these parameters within the scope of \cref{thm:generic-vector-dsedl}, since $\ell(\secparam)\ge\secparam+2$ and the smallest prime factor $\piq$ of $\abs{\GN}=p'q'$ is at least $2^{\ell(\secparam)-2}$, so $\gamma<2^\secparam\le\piq$.\fi

\begin{theoremE}[Mode Indistinguishability]
	\label{thm:eltdh-mode-indistinguishability}
	Let $d\colon\NN\to\NN_{>0}$ be polynomially bounded. If $\DJGen(\cdot)$ satisfies tunable $\DSEDL$, then $\ELTDH$ satisfies the mode-indistinguishability clause of \cref{def:eltdh} at the family $d$.
\end{theoremE}

\ConferenceProofSketch{proof:eltdh-mode-indistinguishability}{
	With one application of \cref{lem:tunable-scalar-to-vec-dsedl} to the $dw+1$ public bases, we change them from bounded powers to uniform elements, which maps the setup distribution at a nonfull parameter $r$ to the one at $\Rfull(\secparam,d)$. We choose the image threshold $Q$ so that every nonfull parameter in $[Q(\secparam)..\Rfull(\secparam,d)]$ has an exponent bound $\gamma$ above the threshold of that lemma, enlarge the security threshold until $Q(\secparam)\le2^\secparam$, and compare two arbitrary parameters through the full one. }{}{}

\begin{proofE}
	\label{proof:eltdh-mode-indistinguishability}
	We change all $dw+1$ hash bases at once with \cref{lem:tunable-scalar-to-vec-dsedl}, choose the image threshold so that every nonfull parameter above it has an exponent bound above the threshold of that lemma, and compare two arbitrary parameters through the full one.

	We fix a polynomial $t$ and a positive polynomial $\delta$, write $d=d(\secparam)$ and $w=2\ell(\secparam)$, and choose a polynomial $p_d$ with $p_d(\secparam)\ge d(\secparam)$ for every $\secparam$. We write $D(\secparam)\cqq d(\secparam)w(\secparam)+1$ and choose a polynomial $p_{\hk}$ that bounds the size of a circuit computing the map $\Psi$ that sends $\bot$ to $\hkbot$ and a challenge $(1^\secparam,(N,g),\gamma,(G_0,(G_{j,k})_{j,k}))$ to the hash key $((N,g),d,w,\clog{N}+\secparam,G_0,(G_{j,k})_{j,k})$. We apply \cref{lem:tunable-scalar-to-vec-dsedl} to the cyclic group sampler $\DJGen$ of \cref{lem:dj-structure} with dimension $D$, size bound $t+p_{\hk}$, and target polynomial $2\delta$, and obtain a polynomial $\Qvec$ and a threshold $\secparam_{\var{vec}}\in\NN$. We choose an integer-valued polynomial $\Qbar$ with $\Qbar(\secparam)\ge\max\{1,\ceil{\Qvec(\secparam)}\}$ for every $\secparam$ and define the polynomial
	\[
		Q(\secparam)\cqq 2p_d(\secparam)\ell(\secparam)\Qbar(\secparam)+1.
	\]
	Since $Q$ is a polynomial, there is $\secparam_1\in\NN$ with $Q(\secparam)\le2^\secparam$ for every $\secparam\ge\secparam_1$, and we take $\secparam_0\cqq\max\{\secparam_{\var{vec}},\secparam_1\}$.

	We fix $\secparam\ge\secparam_0$. Since $d\ge1$ and $2\ell(\secparam)\ge2\secparam+4$, we have $\Rfull(\secparam,d)=2^{dw}\ge2^{2\ell(\secparam)}>2^\secparam\ge Q(\secparam)$, so the interval $[Q(\secparam)..\Rfull(\secparam,d)]$ is nonempty. For $r\in[Q(\secparam)..\Rfull(\secparam,d)-1]$, we write $\rcap=\min\{r,2^\secparam\}$, so $\rcap\ge Q(\secparam)$, and the exponent bound of the setup at $r$ satisfies
	\[
		\gamma=\floor{\frac{\rcap-1}{dw}}\ge\floor{\frac{p_d(\secparam)w\Qbar(\secparam)}{dw}}\ge\Qbar(\secparam)\ge\max\{1,\Qvec(\secparam)\}.
	\]
	Moreover, $\gamma\le\rcap-1<2^\secparam$, and every modulus output by $\DJGen(1^\secparam)$ satisfies $\abs{\GN}=p'q'>2^{2\ell(\secparam)-4}\ge2^{2\secparam}$ by \cref{lem:dj-structure}, so $\gamma$ is smaller than the order of every group in the support of $\DJGen(1^\secparam)$. Since also $\secparam\ge\secparam_{\var{vec}}$, \cref{lem:tunable-scalar-to-vec-dsedl} implies
	\[
		\DvDSEDL_{0,\secparam,\gamma}\indist_{t+p_{\hk},1/(2\delta)}\DvDSEDL_{1,\secparam,\gamma}
	\]
	at dimension $D(\secparam)$.

	The map $\Psi$ sends $\DvDSEDL_{0,\secparam,\gamma}$ to $\Dmode_{\secparam,r}$, because the setup at $r$ returns $\hkbot$ when $\DJGen$ fails and otherwise publishes $g_0=g^{s_0}$ and $g_{j,k}=g^{s_{j,k}}$ for independent $s_0,s_{j,k}\sample[\gamma]$, with this $\gamma\ge1$. It sends $\DvDSEDL_{1,\secparam,\gamma}$ to $\Dmode_{\secparam,\Rfull(\secparam,d)}$, because the setup at $\Rfull(\secparam,d)$ returns $\hkbot$ when $\DJGen$ or one of the $dw+1$ samples from $\GN$ fails and otherwise publishes $dw+1$ independent uniform elements of $\GN$, and the challenge of $\DvDSEDL_{1,\secparam,\gamma}$ is $\bot$ in the same event and otherwise consists of the same elements. For every circuit $C$ of size at most $t(\secparam)$, the composition of $C$ with $\Psi$ has size at most $t(\secparam)+p_{\hk}(\secparam)$, so
	\[
		\Dmode_{\secparam,r}\indist_{t,1/(2\delta)}\Dmode_{\secparam,\Rfull(\secparam,d)}.
	\]
	For $r_0,r_1\in[Q(\secparam)..\Rfull(\secparam,d)]$, we compare each nonfull one with $\Rfull(\secparam,d)$, and by the triangle inequality, $\Dmode_{\secparam,r_0}\indist_{t,1/\delta}\Dmode_{\secparam,r_1}$. This holds for every $\secparam\ge\secparam_0$, which is the mode-indistinguishability clause of \cref{def:eltdh}.
\end{proofE}

We combine \cref{thm:eltdh-share-correctness,thm:eltdh-input-privacy,thm:eltdh-lossy-image,thm:eltdh-function-privacy,thm:eltdh-mode-indistinguishability}\ifConferenceVersion\else{} and \cref{lem:eltdh-efficiency}\fi{} to verify every clause of \cref{def:eltdh}.

\begin{theoremE}[Extremely Lossy Trapdoor Hash]
	\label{thm:eltdh}
	\ifConferenceVersion\label{lem:eltdh-efficiency}\fi
	Suppose that $\DJGen(\cdot)$ satisfies $\DCR$ against nonuniform $\PPT$ adversaries and tunable $\DSEDL$. Then, for every polynomially bounded $d\colon\NN\to\NN_{>0}$, the construction of \cref{fig:eltdh-dcr} is a perfectly correct, statistically input-private, function-private, rate-one extremely lossy trapdoor-hashing scheme for residue-linear predicates, with $\nulos=0$. It also satisfies the perfect simultaneous identity of \cref{cor:eltdh-correctness} at every image parameter.\ifConferenceVersion\ Excluding common parameters and integer metadata, each hash or encoding key contains $dw+1$ group elements for $w=2\ell(\secparam)$, a trapdoor contains one $\expl$-bit exponent, a digest contains at most $2\clog{N}$ bits, and an encoding contains $\clog{N}$ bits.\fi
\end{theoremE}

\ConferenceProofSketch{proof:eltdh}{
	We obtain perfect correctness from \cref{thm:eltdh-share-correctness}, statistical input privacy from \cref{thm:eltdh-input-privacy}, lossiness with $\nulos=0$ from \cref{thm:eltdh-lossy-image}, function privacy from \cref{thm:eltdh-function-privacy}, and mode indistinguishability from \cref{thm:eltdh-mode-indistinguishability}. We count the $dw+1$ public bases and their fixed representation lengths and obtain compactness and rate one, since the predicate value and the encoding both lie in $\ZN$. }{}{}

\begin{proofE}
	\label{proof:eltdh}
	\ifConferenceVersion The key counts follow from the $dw+1$ base positions in \cref{fig:eltdh-dcr}. All group elements and exponents have representation length $\poly(\secparam)$, so every algorithm runs in time $\poly(\secparam+d)$, a digest, which is one unit modulo $N^2$, has at most $2\clog{N}$ bits, and an encoding, which is one residue modulo $N$, has $\clog{N}$ bits. The predicate value and the encoding both lie in $\ZN$, so the rate is one.\else We obtain compact descriptions, the rate, and the running-time bounds from \cref{lem:eltdh-efficiency}.\fi\ Perfect correctness follows from \cref{cor:eltdh-correctness}, statistical input privacy from \cref{thm:eltdh-input-privacy}, and lossiness from \cref{thm:eltdh-lossy-image}. Function privacy follows from \cref{thm:eltdh-function-privacy} and mode indistinguishability from \cref{thm:eltdh-mode-indistinguishability}. Hence, $\ELTDH$ satisfies every clause of \cref{def:eltdh}.
\end{proofE}

\ifConferenceVersion
Restricting the inputs to bits and reducing both shares modulo $2$, we recover the rate-one trapdoor-hashing setting of~\cite{C:RoySin21} from \cref{thm:eltdh}. At nonfull parameters with polynomially bounded $\gamma$, in contrast, we can recover the offsets from the encoding key, so we require function privacy only at $\Rfull(\secparam,d)$.
\else
At nonfull parameters with polynomially bounded $\gamma$, we can recover the offsets from the encoding key, so we require function privacy only at $\Rfull(\secparam,d)$.
\fi

\begin{remark}[Offset Recovery for Polynomially Bounded $\gamma$]\label{rem:tdh-polynomial-gamma}
	At $\hkbot$ and in the branch $\gamma=0$, the encoding-key entries are $K_{j,k}=\Expo(2^ka_j)$, which determine $\bvec a$. We suppose that the setup returned a hash key other than $\hkbot$ at a nonfull parameter with polynomially bounded $\gamma\ge1$. With \BSGS~\cite{Sha1971}, we recover $s_0$ from $g_0=g^{s_0}$ and $s_{j,0}$ from $g_{j,0}=g^{s_{j,0}}$, for every $j\in[d]$, with $\bigO(\sqrt\gamma)$ group operations each. The corresponding encoding-key entries are $K_0=g^{s_0\sigma}$ and $K_{j,0}=g^{s_{j,0}\sigma}\Expo(a_j)$, so $K_{j,0}^{s_0}K_0^{-s_{j,0}}=\Expo(s_0a_j)$. Since $s_0\le\gamma<2^\secparam$ and the prime factors of $N$ exceed $2^{\ell(\secparam)-1}\ge2^{\secparam+1}$, the integer $s_0$ is invertible modulo $N$, and we recover $a_j=s_0^{-1}\Logg(\Expo(s_0a_j))$ exactly. Hence, at such a parameter, a selector that returns two distinct coefficient vectors $\bvec a^0\ne\bvec a^1$ and a distinguisher that recovers the coefficient vector from the encoding key have advantage one, so function privacy fails there.
\end{remark}

\ifConferenceVersion
\begin{remark}[Pointwise Versus Simultaneous Correctness]
	\label{rem:simultaneous-correctness}
	In \cref{def:eltdh}, we require pointwise correctness. To build a trapdoor function, we need one key that decodes every point of the domain, and since lossiness concerns only the hash, we also need the encoding for the constant predicate to be determined by the digest under that key. Both facts follow from the simultaneous identity. The \DamgardJurik construction has this stronger property because \cref{thm:eltdh-share-correctness} holds for every input and every tape under one fixed key.
\end{remark}
\fi

An \ELF{} that is injective in one mode has at least $2^n$ outputs on $\bits^n$, and since fewer than $2^n$ bit strings have length less than $n$, some output has length at least $n$. In contrast, the public map of \ELTDH{} has outputs of length $\poly(\secparam)$ for every arity $d$, and with the trapdoor, we recover the predicate value from each output at a lossy parameter.

\begin{corollaryE}[Image and Output Length of the Public Map]
	\label{cor:eltdh-compressing-lossy}
	Let $\secparam\in\NN$ and $d\in\NN_{>0}$, let $r\in[\Rfull(\secparam,d)]$ and $\hk\in\Supp(\ELTDHSetupLos(1^\secparam,1^d,r))$ with modulus $N$, let $\bvec a\in(\ZN)^d$ and $(\ek,\td)\in\Supp(\ELTDHGen(\hk,\bvec a))$, and let $\Fhk(\bvec x)\cqq(\Phihk(\bvec x),\ELTDHEnc(\ek,\bvec x;0))$ be the public map of \cref{rem:eltdh-gives-telf}. Then $\abs{\Fhk(\calD)}\le r\cdot\abs{\Lina(\calD)}$ for every $\calD\subseteq(\ZN)^d$, every $\bvec x\in(\ZN)^d$ satisfies $e-\ELTDHDec(\td,h)=\Lina(\bvec x)$ at $(h,e)=\Fhk(\bvec x)$, and every output of $\Fhk$ has at most $3\clog{N}$ bits.
\end{corollaryE}

\begin{proofE}
	Since the share identity determines the second coordinate from the digest and the predicate value, the image has at most the product of their two ranges. By \cref{cor:eltdh-correctness}, every $\bvec x$ satisfies $e-\ELTDHDec(\td,h)=\Lina(\bvec x)$ at the all-zero tape, so $e$ is determined by $h$ and $\Lina(\bvec x)$. Hence, $\Fhk$ factors through $\bvec x\mapsto(\Phihk(\bvec x),\Lina(\bvec x))$, and $\abs{\Fhk(\calD)}\le\abs{\Phihk(\calD)}\cdot\abs{\Lina(\calD)} \le r\cdot\abs{\Lina(\calD)}$ by \cref{thm:eltdh-lossy-image}. By \cref{lem:eltdh-efficiency}, a digest has at most $2\clog{N}$ bits and an encoding has $\clog{N}$ bits, so every output has at most $3\clog{N}$ bits.
\end{proofE}

\ifConferenceVersion
\begin{remark}[A Compressing Lossy Map]\label{rem:eltdh-compression}
\fi
We take $\calD=\bits^d$ and $\bvec a=(1,\dots,1)$, so that $\Lina(\bvec x)$ is the Hamming weight of $\bvec x$, and we take $6\ell(\secparam)<d<2^{2\ell(\secparam)-2}$. Every modulus $N$ of a hash key satisfies $2^{2\ell(\secparam)-2}<N<2^{2\ell(\secparam)}$, so $3\clog{N}\le6\ell(\secparam)<d$ and $d<N$. Since $d<N$, the $d+1$ possible weights are distinct in $\ZN$, so $\abs{\Lina(\calD)}=d+1$, and since $d>3\clog{N}$, every output of $\Fhk$ is shorter than the $d$-bit input. Hence, at a polynomial image parameter $r$, the map $\Fhk$ has at most $r(d+1)$ outputs on a domain of $2^d$ points, and with the trapdoor, we recover the Hamming weight from each output.
\ifConferenceVersion
\end{remark}
\fi

\ifConferenceVersion\else
\subsection{Rate-One Trapdoor Hashing}
\label{subsec:eltdh-instantiations}

The predicate value of \ELTDH{} and its encoding both lie in $\ZN$, so \ELTDH{} has rate one in the sense of the rate clause of \cref{def:eltdh} by \cref{lem:eltdh-efficiency}. Trapdoor hashing in the sense of \cref{def:trapdoor-hashing} has bit-valued predicates, so we restrict \ELTDH{} to bit inputs, $\FF_2$-linear predicates, and one-bit encodings, and we recover the rate-one trapdoor hash of~\cite{C:RoySin21}. For this restriction, we measure the rate with the definition of~\cite{C:DGIMMO19}, which we state for the syntax of \cref{def:trapdoor-hashing}.

\begin{definition}[Rate]\label{def:tdh-rate}
	Let $\mu\colon\NN\to\RR_{\ge 0}$. A trapdoor-hashing scheme $\THF$ for $\FuncClass$ has rate at least $\mu$ if $\replen{f(x)}/\replen{\TDHEnc(\ek,x;\rho)}\ge\mu(\secparam)$ for every $\secparam,n\in\NN$, every $x\in\bits^n$, every $\hk\in\Supp(\TDHSetup(1^\secparam,1^n))$, every $f\in\FuncClass_{\secparam,n,\hk}$, every $(\ek,\td)\in\Supp(\TDHGen(\hk,f))$, and every random tape $\rho$, where $\replen{u}$ denotes the length of the canonical representation of $u$. The scheme has rate one if $\replen{f(x)}=\replen{\TDHEnc(\ek,x;\rho)}$ for all these objects.
\end{definition}

To obtain one-bit encodings, we add a uniform mask $\alpha\in\ZN$ to both shares of \cref{thm:eltdh-share-correctness} and keep the parities of their representatives. We adapt the rate-one trapdoor hash of~\cite{C:RoySin21} in this way to the bounded-exponent hash keys of \ELTDH{}, and the parities differ by the predicate value except with the probability of a wraparound.

\begin{corollaryE}[Rate-One Trapdoor Hash for $\FF_2$-Linear Predicates]
	\label{cor:eltdh-parity}
	Let $n\colon\NN\to\NN$ be polynomially bounded with $n(\secparam)\ge1$ for every $\secparam$, and suppose that $\DJGen(\cdot)$ satisfies $\DCR$ against nonuniform $\PPT$ adversaries. Take $\ELTDH$ at arity $d=n(\secparam)$ and at the full parameter, restrict the input coordinates and the coefficient vector $\bvec a$ to $\bits$, and take $\rep(\Lina(\bvec x))\bmod2$ as the $\FF_2$-linear predicate. Generation additionally samples $\alpha\sample\ZN$ and places it in both $\ek$ and $\td$. The encoder returns $t\cqq\rep(e+\alpha)\bmod2$. Given $e'\gets\ELTDHDec(\td,h)$, the decoder sets $t_0\cqq\rep(e'+\alpha)\bmod2$ and returns the candidate pair $(t_0,t_0\oplus1)$. The resulting trapdoor-hashing scheme is $\deltapar$-correct, statistically input-private, function-private, and rate one, where $\deltapar(\secparam)\cqq\min\{1,n(\secparam)2^{2-2\ell(\secparam)}\}$. Its digest is one element of $\ZNsqx$.
\end{corollaryE}

\begin{proofE}
	The only correctness error comes from a wraparound when we add the integer predicate value to the masked decoded share. Since the mask is uniform, the masked share is uniform, so we first bound this event and then obtain privacy and rate from \cref{thm:eltdh-input-privacy,thm:eltdh-function-privacy}.

	For a bit input $\bvec x$ and a bit coefficient vector $\bvec a$, we write
	\[
		z\cqq\sum_{j=1}^{n(\secparam)}a_jx_j\in[0..n(\secparam)],
	\]
	so that $\rep(\Lina(\bvec x))=z$. By \cref{thm:eltdh-share-correctness}, $e=e'+z$ in $\ZN$. Since $\alpha$ is uniform in $\ZN$ and independent of $e'$, the element $e'+\alpha$ is uniform in $\ZN$. If $\rep(e'+\alpha)+z<N$, then $\rep(e+\alpha)=\rep(e'+\alpha)+z$, so $t=t_0\oplus(z\bmod2)$, which is the candidate indexed by $z\bmod2=\rep(\Lina(\bvec x))\bmod2$, and the two candidates differ. The complementary event $\rep(e'+\alpha)+z\ge N$ has probability $z/N\le n(\secparam)/N$. Every supported modulus exceeds $2^{2\ell(\secparam)-2}$, so the failure probability is at most $\deltapar(\secparam)$.

	The mask $\alpha$ is independent of the predicate and of the remaining key material, and the digest is the digest of \ELTDH{}. Therefore, statistical input privacy follows from \cref{thm:eltdh-input-privacy}, and function privacy follows from \cref{thm:eltdh-function-privacy}, since a reduction can append a fresh mask to the encoding key of \ELTDH{}. The predicate value and the encoding are both one bit, so the rate is one.
\end{proofE}

The error $\deltapar$ is negligible, whereas the composite-residuosity variant of~\cite{C:DGIMMO19} has an inverse-polynomial error. From quadratic residuosity, \cite{C:DGIMMO19} constructs a fully correct rate-one trapdoor hash for $\FF_2$-linear predicates with a message group of order two. The composite-residuosity variant of~\cite{C:DGIMMO19} has a message group of order $N$, and the compilation of its group-valued encoding to one bit in~\cite{C:DGIMMO19} has an inverse-polynomial error. The work~\cite{C:RoySin21} replaces the distance function of~\cite{C:DGIMMO19} by $\DDLog_N$, which reduces the error to the probability of a wraparound, and for the scheme of \cref{cor:eltdh-parity}, this probability is at most $n(\secparam)/N$.
\fi

\section{Extremely Lossy Functions}\label{sec:elf}
From the extremely lossy trapdoor hash of \cref{sec:hashing} at residue arity one, we obtain a trapdoor \ELF{} on the domain $[m]$ by hashing the single residue $x=u-1$ of an input $u\in[m]$, with the identity predicate in the injective mode and the zero predicate in the lossy mode.

\ifConferenceVersion\else
\phantomsection\label{subsec:telf-syntax}

\fi
\begin{relocated}{telf-definition}
The injective generator of a trapdoor extremely lossy function also returns an inversion trapdoor. We use the trapdoor extremely lossy functions of~\cite{EC:Zhandry19a} with four changes. We write $m$ and $\Nout(m)$ for the domain size and the bound on the range, which~\cite{EC:Zhandry19a} calls $M$ and $N$, because $N$ denotes a modulus in this paper. We replace the inverse function by an inversion trapdoor and a deterministic inversion algorithm, and we give evaluation as a separate deterministic algorithm. We bound the lengths of descriptions, trapdoors, and outputs by $\poly(\log m)$. We state mode indistinguishability for probabilistic circuits of size $t$, where~\cite{EC:Zhandry19a} quantifies over adversaries whose running time is at most a polynomial $p$.

\begin{definition}[Trapdoor \ELF{}]\label{def:trapdoor-extremely-lossy-function}
	For every $m\ge2$, write $n_m\cqq\clog{m}$. A trapdoor extremely lossy function $\TELF=(\TELFGenInj,\allowbreak\TELFGenLos,\allowbreak\TELFEval,\allowbreak\TELFInv)$ consists of a function $\Nout$ with $\log\Nout(m)=\poly(n_m)$ and the following four algorithms.
	\begin{itemize}
		\item $(f,\td)\sample\TELFGenInj(m)$. On input $m\ge2$, the probabilistic injective-generation algorithm samples a public description $f$ of a function $f\colon[m]\to[\Nout(m)]$ and an inversion trapdoor $\td$.
		\item $f\sample\TELFGenLos(m,r)$. On input $m\ge2$ and $r\in[m]$, the probabilistic lossy-generation algorithm samples a public description $f$ of a function $f\colon[m]\to[\Nout(m)]$.
		\item $y\gets\TELFEval(f,u)$. On input a public description $f$ and $u\in[m]$, the deterministic evaluation algorithm outputs $y=f(u)$.
		\item $u\gets\TELFInv(f,\td,y)$. On input a public description $f$, a trapdoor $\td$, and $y\in[\Nout(m)]$, the deterministic inversion algorithm outputs a candidate $u\in[m]$ or $\bot$.
	\end{itemize}
	We identify a public description $f$ with the function $u\mapsto\TELFEval(f,u)$. The tuple $\TELF$ satisfies the following four properties.
	\begin{itemize}
		\item \textbf{Compact Descriptions.} Public descriptions, trapdoors, and outputs have length $\poly(n_m)$, and every algorithm of $\TELF$ runs in time polynomial in the length of its input.
		\item \textbf{Injective Mode.} If $(f,\td)\sample\TELFGenInj(m)$, then except with probability negligible in $n_m$, every $u\in[m]$ satisfies $\TELFInv(f,\td,\TELFEval(f,u))=u$.
		\item \textbf{Lossy Mode.} For every $r\in[m]$, if $f\sample\TELFGenLos(m,r)$, then $\abs{\Image(f)}\le r$ except with probability negligible in $n_m$.
		\item \textbf{Mode Indistinguishability.} For every polynomial $t$ and every positive polynomial $\delta$, there exists a polynomial $Q$ such that for every sufficiently large $m$ and every $r\in[Q(n_m)..m]$,
		      \[
			      \left\{
			      \begin{array}{@{}c|l@{}}
				      f & (f,\td)\sample\TELFGenInj(m)
			      \end{array}
			      \right\}
			      \indist_{t,1/\delta}
			      \left\{
			      \begin{array}{@{}c|l@{}}
				      f & f\sample\TELFGenLos(m,r)
			      \end{array}
			      \right\},
		      \]
		      where $t$ and $\delta$ are evaluated at $n_m$.
	\end{itemize}
\end{definition}

\end{relocated}%
\begin{relocated}{enumerability}
For some applications of extremely lossy functions, \cite{C:Zhandry16} also requires an efficient algorithm that outputs the lossy image, and it distinguishes an algorithm that outputs a superset of the image from one that outputs the image exactly. We use these two notions of~\cite{C:Zhandry16}, with the overwhelming probability written as $1-\varepsilon(\clog{m})$ for a negligible function $\varepsilon$.

\begin{definition}[Enumerable Image]\label{def:efficiently-enum-img}
	Let $\ELF=(\ELFGen,\ELFEval)$ be an extremely lossy function. For a randomized algorithm $\Enumerate$, integers $m\ge2$, and $r\in[m]$, define
	\[
		p_{\subseteq}^{\Enumerate}(m,r)\cqq\Prob\left[\Image(f)\subseteq S\mbar
		\begin{array}{l}
			f\sample\ELFGen(m,r), \\[-1pt]
			S\sample\Enumerate(f,r)
		\end{array}\right],
	\]
	and define $p_{=}^{\Enumerate}(m,r)$ in the same way, with the event $\Image(f)=S$. The image of $\ELF$ is weakly efficiently enumerable if there exist a randomized algorithm $\Enumerate$ running in time $\poly(r,\log m)$ and a negligible function $\varepsilon$ such that $p_{\subseteq}^{\Enumerate}(m,r)\ge1-\varepsilon(\clog{m})$ for every $m\ge2$ and $r\in[m]$. It is strongly efficiently enumerable if the same holds with $p_{=}^{\Enumerate}$ in place of $p_{\subseteq}^{\Enumerate}$.
\end{definition}

\end{relocated}%
\begin{relocated}{enumerability-telf}
Following~\cite{C:Zhandry16}, we apply the same two notions to the lossy generator $\TELFGenLos$ of a trapdoor extremely lossy function.

\begin{definition}[Enumerable Trapdoor \ELF{}]\label{def:efficiently-enumerable-telf}
	A trapdoor extremely lossy function $\TELF$ has a weakly (respectively, strongly) efficiently enumerable image if \cref{def:efficiently-enum-img} holds for $\TELF$ with $\ELFGen(m,r)$ replaced by $\TELFGenLos(m,r)$.
\end{definition}

\end{relocated}%
\ifConferenceVersion\else
\subsection{The Trapdoor Construction}
\label{subsec:telf-from-eltdh}
\fi

For every $m\ge2$, we choose the security parameter of the hash as a function of $n\cqq\clog{m}$. We fix an efficiently computable, strictly increasing, polynomially bounded function $\groupsec\colon\NN\to\NN$ with $\groupsec(n)\ge n$ for every $n\in\NN$, and we set $\secparam\cqq\groupsec(n)$, $\wlam\cqq2\ell(\secparam)$, and $\Rlam\cqq2^{\wlam}$. Every modulus $N$ of a hash key in the support of $\ELTDHSetupLos(1^\secparam,1^1,\cdot)$, including the fallback modulus, satisfies $m\le2^n<N$, so for every $u\in[m]$, the least nonnegative representative of the residue $u-1\in\ZN$ is $u-1$.

For a public description $\fk$ of this construction, we write $\ffk\colon[m]\to\ZNsqx\times\ZN$ for the function described by $\fk$. We encode the first coordinate of an output in exactly $4\ell(\secparam)$ bits and the second in exactly $2\ell(\secparam)$ bits, with leading zeros, and we interpret the resulting $6\ell(\secparam)$-bit string as its nonnegative integer value plus one. Hence, under this fixed-width encoding, every output lies in $[\NTelf(m)]$ for $\NTelf(m)\cqq2^{6\ell(\groupsec(\clog{m}))}$, and $\log\NTelf(m)=\poly(\log m)$, as required in \cref{def:trapdoor-extremely-lossy-function}.

\ifConferenceVersion
\subsection{The Trapdoor Construction}
\label{subsec:telf-from-eltdh}
\fi

Throughout this subsection, $\ELTDH$ denotes the construction of \cref{thm:eltdh} at residue arity $d=1$, and we write $\IdMap(x)=x$ and $0(x)=0$ for the two linear predicates of the construction. In injective generation, we take a hash key at the full image parameter $\Rlam$ and an encoding key for $\IdMap$, so that the trapdoor recovers $x=u-1$ from one decoded share. In lossy generation with $r<m$, we take a hash key at the image parameter $r$ and an encoding key for $0$.

\begin{proofbox}{Trapdoor \ELF{} from \ELTDH{} at residue arity one, for $m\ge2$ and $r\in[m]$. Here $\ell(\secparam)$ is the bit length of the prime factors of $N$, and the hash key carries the exponent length $\expl$ and the bases $g_{1,k}$ for $k\in[0..\wlam-1]$ of \cref{fig:eltdh-dcr}.}{fig:small-exponent-telf}
		\begin{minipage}[t]{0.51\linewidth}
			\vspace{0pt}
			{\footnotesize\underline{$\TELFGenInj(m)$}\par}
			\smallskip
			\begin{codelines}
				\item \pcassert $m\ge2$
				\item $n\gets\clog{m}$; $\secparam\gets\groupsec(n)$
				\item $\Rlam\gets2^{2\ell(\secparam)}$
				\item $\hk\sample\ELTDHSetupLos(1^\secparam,1^1,\Rlam)$
				\item \label{line:telf-inj-offsets}$(\ek,\td)\sample\ELTDHGen(\hk,(1))$
				\item $\fk\gets(m,n,\secparam,\hk,\ek)$
				\item \pcreturn $(\fk,\td)$
			\end{codelines}

			\medskip
			{\footnotesize\underline{$\TELFGenLos(m,r)$}\par}
			\smallskip
			\begin{codelines}
				\item \pcassert $m\ge2\land r\in[m]$
				\item \pcif $r=m$
				\codelinesplus
				\item $(\fk,\td)\sample\TELFGenInj(m)$
				\item \pcreturn $\fk$
				\codelinesminus
				\item $n\gets\clog{m}$; $\secparam\gets\groupsec(n)$
				\item \label{line:telf-los-bases}$\hk\sample\ELTDHSetupLos(1^\secparam,1^1,r)$
				\item \label{line:telf-los-offsets}$(\ek,\td)\sample\ELTDHGen(\hk,(0))$
				\item \pcreturn $\fk\gets(m,n,\secparam,\hk,\ek)$
			\end{codelines}
		\end{minipage}\hfill%
		\begin{minipage}[t]{0.44\linewidth}
			\vspace{0pt}
			{\footnotesize\underline{$\TELFEval(\fk,u)$}\par}
			\smallskip
			\begin{codelines}
				\item \pcparse $\fk=(m,n,\secparam,\hk,\ek)$
				\item \pcparse $\hk=((N,g),1,\wlam,\expl,g_0,(g_{1,k})_k)$
				\item \pcassert $u\in[m]$
				\item $x\gets u-1\in\ZN$
				\item $h\gets\ELTDHHash(\hk,(x);0)$
				\item $e\gets\ELTDHEnc(\ek,(x);0)$
				\item \pcreturn $y\gets(h,e)$
			\end{codelines}

			\medskip
			{\footnotesize\underline{$\TELFInv(\fk,\td,y)$}\par}
			\smallskip
			\begin{codelines}
				\item \pcparse $\fk=(m,n,\secparam,\hk,\ek)$
				\item \pcparse $\hk=((N,g),1,\wlam,\expl,g_0,(g_{1,k})_k)$
				\item \pcif $y\notin\ZNsqx\times\ZN$ \pcreturn $\bot$
				\item \pcparse $y=(h,e)$
				\item $v\gets e-\ELTDHDec(\td,h)$; $x\gets\rep(v)$
				\item \pcif $x\ge m$ \pcreturn $\bot$
				\item \pcreturn $u\gets x+1$
			\end{codelines}
		\end{minipage}
	\end{proofbox}

For the compact-description condition of \cref{def:trapdoor-extremely-lossy-function} and for the costs compared in \cref{sec:comparison}, we count the group elements in public descriptions, trapdoors, and outputs.

\begin{lemmaE}[Efficiency]\label{lem:telf-efficiency}
	Let $m\ge2$ and $r\in[m]$. The algorithms in \cref{fig:small-exponent-telf} run in time $\poly(\log m)$. Excluding the modulus, the generator, and integer metadata, a public description contains $2\wlam+2$ elements of $\ZNsqx$, a trapdoor contains one exponent, and an output contains one element of $\ZNsqx$ and one element of $\ZN$.
\end{lemmaE}
\ConferenceProofSketch{proof:telf-efficiency}{
	At residue arity one, each of $\hk$ and $\ek$ contains $\wlam+1$ group elements by \cref{lem:eltdh-efficiency}, the trapdoor is the exponent sampled by $\ELTDHGen$, and an output is one digest and one residue. Since all representation lengths are fixed and $\secparam=\groupsec(\clog{m})=\poly(\log m)$, every algorithm runs in time $\poly(\log m)$. }{}{}
\begin{proofE}
	\label{proof:telf-efficiency}
	The underlying \ELTDH{} construction has residue arity one, so by \cref{lem:eltdh-efficiency} each of $\hk$ and $\ek$ contains $\wlam+1$ group elements, the trapdoor contains one exponent, and an output contains one digest and one residue. Since $\wlam=\poly(\secparam)$ and $\secparam=\poly(n)$, all descriptions and algorithms have size and running time $\poly(n)=\poly(\log m)$. For $u\in[m]$, the bits of $u-1$ at the positions $k\in[n..\wlam-1]$ of its padded binary expansion are zero, so we compute each of the products in $\ELTDHHash$ and $\ELTDHEnc$ with $\bigO(n)$ group multiplications, and $\ELTDHEnc$ adds one evaluation of $\DDLog_N$.
\end{proofE}

With the single trapdoor $\td$, we invert $\ffk$ on all of $[m]$ by the identity $e-\ELTDHDec(\td,h)=x$ of \cref{cor:eltdh-correctness} at coefficient vector $(1)$.

\begin{theoremE}[Perfect Inversion]\label{thm:telf-injective-inversion}
	Let $m\ge2$ and $u\in[m]$. For every pair $(\fk,\td)$ in the support of $\TELFGenInj(m)$ and for $y=\TELFEval(\fk,u)$, we have $\TELFInv(\fk,\td,y)=u$. In particular, $\ffk$ is injective.
\end{theoremE}

\ConferenceProofSketch{proof:telf-injective-inversion}{
	By \cref{cor:eltdh-correctness} at coefficient vector $(1)$, every output $(h,e)$ at $u\in[m]$ satisfies $e-\ELTDHDec(\td,h)=x$ for $x=u-1$. Since $0\le x<m<N$, the least nonnegative representative of this residue is $x$, and $\TELFInv$ returns $x+1=u$. }{}{}

\begin{proofE}
	\label{proof:telf-injective-inversion}
	We fix $u\in[m]$ and write $x=u-1$ and $(h,e)=\TELFEval(\fk,u)$, so that $h\in\ZNsqx$ and $e\in\ZN$. By \cref{cor:eltdh-correctness} at coefficient vector $(1)$, we have $e-\ELTDHDec(\td,h)=x$ in $\ZN$. Since $0\le x<m\le2^n<N$, the least nonnegative representative of this residue is $x<m$, so $\TELFInv$ returns $x+1=u$. Since $\TELFInv$ is deterministic, $\ffk(u)=\ffk(u')$ for $u,u'\in[m]$ implies
	\[
		u=\TELFInv(\fk,\td,\ffk(u))=\TELFInv(\fk,\td,\ffk(u'))=u'.
	\]
\end{proofE}

\ifConferenceVersion\else
In \cref{fig:small-exponent-telf}, we recover the complete $n$-bit input from one decoded residue in $\ZN$, which is possible because $N>2^n$. When every decoded value lies in a set of at most $B$ elements, we need at least $n/\log_2B$ of them, by a count of the possible tuples of decoded values.

\fi
\begin{relocated}{packed-row-remark}
\ifConferenceVersion
\begin{remark}[Lower Bound on the Number of Decoded Coordinates]\label{rem:packed-row-counting}
\else
\begin{lemma}[Lower Bound on the Number of Decoded Coordinates]\label{rem:packed-row-counting}
\fi
	Let $f\colon\bits^n\to\calY$ be injective, let $B\ge2$ and $k\ge1$ be integers, and let $d_i\colon\calY\to\calV_i$ be maps with $\abs{\calV_i}\le B$ for $i\in[k]$ such that $x\mapsto(d_i(f(x)))_{i\in[k]}$ is injective on $\bits^n$. Then $k\ge n/\log_2B$.\ifConferenceVersion\ Indeed, the map $x\mapsto(d_i(f(x)))_{i\in[k]}$ embeds $\bits^n$ into the product of the sets $\calV_i$, so
	\[
		2^n\le\prod_{i=1}^k\abs{\calV_i}\le B^k,
	\]
	and hence $k\ge n/\log_2B$. In particular, $k=\Omega(n/\log n)$ when $B=\poly(n)$. In \cref{fig:small-exponent-telf}, we decode with $k=1$ and $\calV_1=\ZN$, which suffices for the complete $n$-bit input because $N>2^n$.
\end{remark}
\else
\end{lemma}
\begin{proof}
	The map $x\mapsto(d_i(f(x)))_{i\in[k]}$ is an injection of $\bits^n$ into the product of the sets $\calV_i$, so
	\[
		2^n\le\prod_{i=1}^k\abs{\calV_i}\le B^k,
	\]
	and we take logarithms to obtain $k\ge n/\log_2B$.
\end{proof}

By \cref{rem:packed-row-counting}, $k=\Omega(n/\log n)$ when $B=\poly(n)$, whereas in \cref{fig:small-exponent-telf}, we decode with $k=1$ and $\calV_1=\ZN$.
\fi

\end{relocated}%

In lossy generation with $r<m$, we use the predicate $0$, so by \cref{cor:eltdh-correctness} the encoding coordinate of $\ffk(u)$ is a function of the digest $\Phihk(u-1)$, which is the first coordinate of $\ffk(u)$. Hence, $\abs{\Image(\ffk)}\le\abs{\Image(\Phihk)}$, and for $r=m$, the domain has $m=r$ points.

\begin{theoremE}[Lossy Image Bound]\label{thm:telf-lossy-image}
	Let $m\ge2$ and $r\in[m]$. Every $\fk\in\Supp(\TELFGenLos(m,r))$ satisfies $\abs{\Image(\ffk)}\le r$.
\end{theoremE}
\ConferenceProofSketch{proof:telf-lossy-image}{
	For $r=m$, we have $\abs{\Image(\ffk)}\le m=r$. For $r<m$ and a trapdoor $\td$ sampled with $\ek$, \cref{cor:eltdh-correctness} at coefficient vector $(0)$ implies $\ffk(u)=(h,\ELTDHDec(\td,h))$ for $h=\Phihk(u-1)$, so $\abs{\Image(\ffk)}\le\abs{\Image(\Phihk)}\le r$ by \cref{thm:eltdh-lossy-image}. }{}{}
\begin{proofE}
	\label{proof:telf-lossy-image}
	If $r=m$, then $\abs{\Image(\ffk)}\le\abs{[m]}=r$. We suppose that $r<m$ and fix $\fk=(m,n,\secparam,\hk,\ek)\in\Supp(\TELFGenLos(m,r))$. By the definition of $\TELFGenLos$, there is a trapdoor $\td$ with $(\ek,\td)\in\Supp(\ELTDHGen(\hk,(0)))$, and we define $\JtdZero(h)\cqq(h,\ELTDHDec(\td,h))$ for $h\in\ZNsqx$. For every $u\in[m]$ and $(h,e)=\ffk(u)$, we have $h=\Phihk(u-1)$, and by \cref{cor:eltdh-correctness} at coefficient vector $(0)$, we have $e=\ELTDHDec(\td,h)$, so $\ffk(u)=\JtdZero(\Phihk(u-1))$. Since $\{u-1\mid u\in[m]\}=[0..m-1]$, we obtain
	\[
		\abs{\Image(\ffk)}
		=\abs{\JtdZero\left(\Phihk([0..m-1])\right)}
		\le\abs{\Image(\Phihk)}
		\le r,
	\]
	where the last inequality is \cref{thm:eltdh-lossy-image}.
\end{proofE}

Under the predicate $0$, the first coordinate of $\ffk(u)$ is $g^{A(u)}$ for a bounded subset sum $A(u)$ of the exponents of the hash bases, and since $A(u)$ is smaller than the order of $g$, we count the image of $\ffk$ exactly through the values of $A(u)$.

\begin{lemmaE}[Exact Lossy Image]\label{cor:telf-exact-lossy-image}
	Let $m\ge2$, $r\in[m-1]$, and $\fk=(m,n,\secparam,\hk,\ek)\in\Supp(\TELFGenLos(m,r))$, and write $\gamma\cqq\floor{(r-1)/\wlam}$. If $\hk=\hkbot$ or $\gamma=0$, let $s_{1,k}\cqq0$ for $k\in[0..\wlam-1]$. Otherwise, let $s_{1,k}\in[\gamma]$ be the exponents with $g_{1,k}=g^{s_{1,k}}$ for the generator $g$ and the bases $g_{1,k}$ in $\hk$. For $u\in[m]$, define
	\[
		A(u)\cqq\sum_{j=1}^ns_{1,j-1}x_j(u).
	\]
	Then $\abs{\Image(\ffk)}=\abs{\{A(u)\mid u\in[m]\}}$.
\end{lemmaE}

\begin{proofE}
	We identify the digest of every input with $g^{A(u)}$ and then use that $A(u)$ is smaller than the order of $g$.

	We fix a trapdoor $\td$ with $(\ek,\td)\in\Supp(\ELTDHGen(\hk,(0)))$. As in the proof of \cref{thm:telf-lossy-image}, every $u\in[m]$ satisfies $\ffk(u)=\JtdZero(\Phihk(u-1))$ for $\JtdZero(h)\cqq(h,\ELTDHDec(\td,h))$, and $\JtdZero$ is injective because its first coordinate is $h$. Hence,
	\[
		\abs{\Image(\ffk)}=\abs{\{\Phihk(u-1)\mid u\in[m]\}}.
	\]

	If $\hk=\hkbot$ or $\gamma=0$, then every base $g_{1,k}$ equals $1$, so $\Phihk$ is constant, every $A(u)$ is zero, and both sides of the claim equal one. Otherwise, setup sampled a \DamgardJurik group and the bounded bases. Since $r<m\le2^n\le2^\secparam$, the capped parameter $\min\{r,2^\secparam\}$ of \cref{fig:eltdh-dcr} equals $r$, so $\gamma\ge1$ is the interval bound of setup and $g_{1,k}=g^{s_{1,k}}$ with $s_{1,k}\in[\gamma]$. By \cref{lem:dj-structure}, the order $M=p'q'$ of $g$ satisfies $M>2^{2\ell(\secparam)-4}\ge2^{2\secparam}\ge2^{2n}>m$. Since $\gamma\le r-1<M$, each $s_{1,k}$ is the unique exponent in $[\gamma]$ with $g_{1,k}=g^{s_{1,k}}$. For $u\in[m]$, we have $u-1<2^n$, so the bit of $u-1$ at position $k$ is $x_{k+1}(u)$ for $k\in[0..n-1]$ and zero for $k\in[n..\wlam-1]$, and at the all-zero tape,
	\[
		\Phihk(u-1)=\prod_{j=1}^ng_{1,j-1}^{x_j(u)}=g^{A(u)}.
	\]
	Since $0\le A(u)\le n\gamma\le\wlam\gamma\le r-1<M$, the map $a\mapsto g^a$ is injective on $[0..r-1]$. Hence, $\Phihk(u-1)=\Phihk(u'-1)$ holds exactly when $A(u)=A(u')$, and the two sets $\{\Phihk(u-1)\mid u\in[m]\}$ and $\{A(u)\mid u\in[m]\}$ have the same cardinality.
\end{proofE}

\subsection{Mode Indistinguishability}\label{subsec:telf-mode-indistinguishability}

By \cref{thm:telf-injective-inversion,thm:telf-lossy-image}, each generator has the property required of its mode, and it remains to show that the two distributions of public descriptions are indistinguishable. The two generators differ in the predicate of the encoding key and in the image parameter of the hash key, and we change these one at a time.

\begin{theoremE}[Mode Indistinguishability]\label{thm:telf-mode-indistinguishability}
	Suppose that $\DJGen(\cdot)$ satisfies $\DCR$ against nonuniform $\PPT$ adversaries and tunable $\DSEDL$. Then the generators in \cref{fig:small-exponent-telf} satisfy the mode-indistinguishability condition of \cref{def:trapdoor-extremely-lossy-function}.
\end{theoremE}
\ConferenceProofSketch{proof:telf-mode-indistinguishability}{
	We pass through the public description with an encoding key for the zero predicate under a hash key at the full image parameter $\Rlam$. By \cref{thm:eltdh-function-privacy} at a uniform security threshold, we replace the identity predicate by zero, and by \cref{thm:eltdh-mode-indistinguishability} at residue arity one, we lower the image parameter from $\Rlam$ to $r$. We bound each of these two transitions by $1/(2\delta(n))$ for every sufficiently large $m$ and add the two bounds. }{}{}
\begin{proofE}
	\label{proof:telf-mode-indistinguishability}
	We pass from injective to lossy generation through one intermediate hybrid. In the first transition, we replace the identity predicate by the zero predicate at the full image parameter $\Rlam$, using function privacy, and in the second, we lower the image parameter from $\Rlam$ to $r$ with the zero predicate, using \ELTDH{} mode indistinguishability.

	We fix a polynomial $t$ and a positive polynomial $\delta$, and we choose constants $C>0$ and $D\in\NN$ such that $t(z)\le C(z+1)^D$ and $2\delta(z)\le C(z+1)^D$ for every $z\in\NN$. We also fix a polynomial $p$ such that, for every $\secparam$ and every hash key $\hk$ in the support of $\ELTDHSetupLos(1^\secparam,1^1,\cdot)$, a circuit of size at most $p(\secparam)$ samples $(\ek,\td)\sample\ELTDHGen(\hk,(0))$ and assembles $(m,n,\secparam,\hk,\ek)$ from $\hk$, $\ek$, and hardwired $m$ and $n$ with $\groupsec(n)=\secparam$. We set
	\[
		\that(\secparam)\cqq C(\secparam+1)^D+p(\secparam),
		\qquad
		\dhat(\secparam)\cqq C(\secparam+1)^D.
	\]
	Since $\groupsec(n)\ge n$, every $n\in\NN$ satisfies $t(n)+p(\groupsec(n))\le\that(\groupsec(n))$ and $\dhat(\groupsec(n))\ge2\delta(n)$.

	For $m\ge2$ and $r\in[m-1]$, with $n=\clog{m}$ and $\secparam=\groupsec(n)$, each of the three hybrids samples a hash key and an encoding key and outputs the public description $(m,n,\secparam,\hk,\ek)$.

	\begin{proofbox}{Hybrids for mode indistinguishability of the trapdoor \ELF{}, at $m\ge2$ and $r\in[m-1]$. The highlighted line of $\Hybrid_i$ is the line in which it differs from $\Hybrid_{i-1}$.}{fig:telf-mode-hybrids}
			\begin{minipage}[t]{\telfhybwidth}
				\vspace{0pt}
				{\footnotesize\underline{$\Hybrid_0$}\codelinesrightalign{$\TELFGenInj(m)$}\par}
				\smallskip
				\begin{codelines}
					\item $\hk\sample\ELTDHSetupLos(1^\secparam,1^1,\Rlam)$
					\item $(\ek,\td)\sample\ELTDHGen(\hk,(1))$
					\item \pcreturn $(m,n,\secparam,\hk,\ek)$
				\end{codelines}
			\end{minipage}\hfill%
			\begin{minipage}[t]{\telfhybwidth}
				\vspace{0pt}
				{\footnotesize\underline{$\Hybrid_1$}\codelinesrightalign{zero predicate}\par}
				\smallskip
				\begin{codelines}
					\item $\hk\sample\ELTDHSetupLos(1^\secparam,1^1,\Rlam)$
					\item \label{line:telf-hyb-predicate}\hybchange{$(\ek,\td)\sample\ELTDHGen(\hk,(0))$}
					\item \pcreturn $(m,n,\secparam,\hk,\ek)$
				\end{codelines}
			\end{minipage}\telfhybbreak
			\begin{minipage}[t]{\telfhybwidth}
				\vspace{0pt}
				{\footnotesize\underline{$\Hybrid_2$}\codelinesrightalign{$\TELFGenLos(m,r)$}\par}
				\smallskip
				\begin{codelines}
					\item \label{line:telf-hyb-bases}\hybchange{$\hk\sample\ELTDHSetupLos(1^\secparam,1^1,r)$}
					\item $(\ek,\td)\sample\ELTDHGen(\hk,(0))$
					\item \pcreturn $(m,n,\secparam,\hk,\ek)$
				\end{codelines}
			\end{minipage}
			\end{proofbox}

	By the definitions in \cref{fig:small-exponent-telf}, $\Hybrid_0$ in \cref{fig:telf-mode-hybrids} is distributed as the public description output by $\TELFGenInj(m)$, and $\Hybrid_2$ is distributed as $\TELFGenLos(m,r)$. For a circuit $\advD$ of size at most $t(n)$, we write $p_i$ for the probability that $\advD$ outputs $1$ on a sample of $\Hybrid_i$, so that
	\[
		\abs{p_0-p_2}\le\abs{p_0-p_1}+\abs{p_1-p_2}.
	\]

	At \cref{line:telf-hyb-predicate}, we replace the identity predicate by zero at the full image parameter. We take the selector in the function-privacy condition of \cref{def:eltdh} at residue arity $d\equiv1$ that outputs the coefficient vectors $(1)$ and $(0)$ and the empty state, and we recall that $\Rfull(\secparam,1)=\Rlam$. By \cref{thm:eltdh-function-privacy} and \cref{lem:uniform-quantitative-threshold} with size bound $\that$ and target polynomial $\dhat$, there is $\secparam_1\in\NN$ such that, for every $\secparam\ge\secparam_1$, every circuit of size at most $\that(\secparam)$ distinguishes the two function-privacy ensembles at $\secparam$ with advantage at most $1/\dhat(\secparam)$. Our reduction $\advB_1$ receives $(\hk,\var{st},\ek)$ from one of these ensembles and outputs $\advD(m,n,\secparam,\hk,\ek)$. Its input to $\advD$ is distributed as $\Hybrid_0$ when $\ek$ is generated for $(1)$ and as $\Hybrid_1$ when $\ek$ is generated for $(0)$, and its size is at most $t(n)+p(\secparam)\le\that(\secparam)$. Hence, $\abs{p_0-p_1}\le1/\dhat(\secparam)$ whenever $\secparam\ge\secparam_1$.

	At \cref{line:telf-hyb-bases}, we lower the image parameter from $\Rlam$ to $r$ with the zero predicate. We apply \cref{thm:eltdh-mode-indistinguishability} at residue arity $d\equiv1$ with size bound $\that$ and target polynomial $\dhat$, and we obtain a polynomial $\Qhash$ and $\secparam_2\in\NN$ such that, for every $\secparam\ge\secparam_2$ and all $r_0,r_1\in[\Qhash(\secparam)..\Rlam]$, every circuit of size at most $\that(\secparam)$ distinguishes $\ELTDHSetupLos(1^\secparam,1^1,r_0)$ from $\ELTDHSetupLos(1^\secparam,1^1,r_1)$ with advantage at most $1/\dhat(\secparam)$. Our reduction $\advB_2$ receives a hash key $\hk$, samples $(\ek,\td)\sample\ELTDHGen(\hk,(0))$, and outputs $\advD(m,n,\secparam,\hk,\ek)$. Its input to $\advD$ is distributed as $\Hybrid_1$ when $\hk$ is sampled at the image parameter $\Rlam$ and as $\Hybrid_2$ when $\hk$ is sampled at $r$, and its size is at most $\that(\secparam)$. Since $\groupsec$ is polynomially bounded, there is a polynomial $\Qtelf$ with $\Qtelf(n)\ge\Qhash(\groupsec(n))$ for every $n\in\NN$. For $r\ge\Qtelf(n)$, we have $r\ge\Qhash(\secparam)$ and $r<m\le2^n\le2^\secparam<\Rlam$, so $\abs{p_1-p_2}\le1/\dhat(\secparam)$ whenever $\secparam\ge\secparam_2$.

	We take $\Qtelf$ as the polynomial of the mode-indistinguishability condition and consider every $m\ge2$ with $n=\clog{m}\ge\max\{\secparam_1,\secparam_2\}$, which holds for every sufficiently large $m$ and implies $\secparam=\groupsec(n)\ge\max\{\secparam_1,\secparam_2\}$, and every $r\in[\Qtelf(n)..m]$. If $r=m$, then $\TELFGenLos(m,m)$ outputs the public description of a run of $\TELFGenInj(m)$, so the two distributions are identical. If $r<m$, we add the bounds of the two transitions and use $\dhat(\secparam)\ge2\delta(n)$ to obtain
	\[
		\abs{p_0-p_2}\le\frac2{\dhat(\secparam)}\le\frac1{\delta(n)}.
	\]
	Hence, every circuit of size at most $t(n)$ distinguishes the two distributions of public descriptions with advantage at most $1/\delta(n)$, which is the condition of \cref{def:trapdoor-extremely-lossy-function}.
\end{proofE}

We have proved the four conditions of \cref{def:trapdoor-extremely-lossy-function} for the construction in \cref{lem:telf-efficiency,thm:telf-injective-inversion,thm:telf-lossy-image,thm:telf-mode-indistinguishability}.

\begin{theoremE}[Trapdoor Extremely Lossy Function]\label{thm:small-exponent-telf}
	Suppose that $\DJGen(\cdot)$ satisfies tunable $\DSEDL$ and $\DCR$ against nonuniform $\PPT$ adversaries. Then the algorithms in \cref{fig:small-exponent-telf} form a trapdoor extremely lossy function, which has perfect inversion and satisfies the lossy image bound at every lossy public description in the support.
\end{theoremE}
\ConferenceProofSketch{proof:telf}{
	Compact descriptions hold by \cref{lem:telf-efficiency}, the injective mode holds with probability one by \cref{thm:telf-injective-inversion}, the lossy mode holds with probability one by \cref{thm:telf-lossy-image}, and mode indistinguishability holds by \cref{thm:telf-mode-indistinguishability}. }{}{}
\begin{proofE}
	\label{proof:telf}
	We verify the conditions of \cref{def:trapdoor-extremely-lossy-function} for the interval $[\NTelf(m)]$. Compact descriptions follow from \cref{lem:telf-efficiency} and $\log\NTelf(m)=6\ell(\groupsec(\clog{m}))=\poly(\log m)$. The injective mode holds for every $(\fk,\td)\in\Supp(\TELFGenInj(m))$ by \cref{thm:telf-injective-inversion}, and the lossy mode holds for every $\fk\in\Supp(\TELFGenLos(m,r))$ by \cref{thm:telf-lossy-image}. Mode indistinguishability is \cref{thm:telf-mode-indistinguishability}.
\end{proofE}

\subsection{Enumerating the Lossy Image}\label{subsec:telf-image-enumeration}

For an efficiently enumerable image in the sense of~\cite{C:Zhandry16}, we compute the lossy image $\Image(\ffk)$ exactly from $\fk$ and $r$ in time $\poly(r,\log m)$.\ifConferenceVersion\ Since every partial product of the public bases in the first coordinate is $g^a$ with $a\in[0..r-1]$, we keep at most $r$ distinct products in a prefix dynamic program over the bits of $m-1$.\fi

\ifConferenceVersion\else
\phantomsection\label{app:telf-enumeration-algorithm}
\fi
\begin{relocated}{telf-enumeration-algorithm}
Since every output of $\ffk$ is a product of public bases selected by the bits of $u-1$, we enumerate the image exactly from the public description, without the secret exponents. Working downwards from the most significant bit of $m-1$, we keep the products for the processed prefixes of $u-1$ that are equal to the corresponding prefix of $m-1$ or smaller than it.

We state the procedure over an arbitrary finite abelian group $\bbA$ written multiplicatively, since its correctness depends only on the group operation. For $z,C\in\bbA$, we write $z\odot C$ for their product, and for $S\subseteq\bbA$, we write $S\odot C\cqq\{z\odot C\mid z\in S\}$. For a finite set $S$ of elements with canonical representations, $\Deduplicate(S)$ is the list of the distinct elements of $S$ in increasing order of the integers that their canonical representations encode, so that this list is uniquely determined by $S$.

\begin{proofbox}{Exact enumeration of the image of a lossy public description of the trapdoor \ELF{}, over $\bbA=\ZNsqx\times\ZNsqx$ with coordinatewise multiplication. Here $b_j=x_j(m)$ is the bit of $m-1$ at position $j$, $E$ holds the product for the processed prefix of $u-1$ equal to the corresponding prefix of $m-1$, and $L$ holds the products for smaller processed prefixes.}{fig:telf-image-enumeration}
		{\footnotesize\underline{$\Enumerate(\fk,r)$}\par}
		\smallskip
		\begin{codelines}
			\item \pcparse $\fk=(m,n,\secparam,\hk,\ek)$
			\item \pcparse $\hk=((N,g),1,\wlam,\expl,g_0,(g_{1,k})_k)$
			\item \pcparse $\ek=(N,1,\wlam,\expl,K_0,(K_{1,k})_k)$
			\item \pcfor $j\in[n]$
			\codelinesplus
			\item $C_j\gets(g_{1,j-1},K_{1,j-1})$; $b_j\gets x_j(m)$
			\codelinesminus
			\item $E\gets\{(1,1)\}$; $L\gets\varnothing$
			\item \pcfor $j=n,n-1,\dots,1$
			\codelinesplus
			\item \pcif $b_j=0$
			\codelinesplus
			\item $E'\gets E$; $L'\gets\Deduplicate(L\cup(L\odot C_j))$
			\codelinesminus
			\item \pcelse
			\codelinesplus
			\item $E'\gets E\odot C_j$
			\item $L'\gets\Deduplicate(L\cup(L\odot C_j)\cup E)$
			\codelinesminus
			\item $(E,L)\gets(E',L')$
            \codelinesplus
			\item \pcif $\abs E+\abs L>r+1$ \pcreturn $\varnothing$
			\codelinesminus
			\item \pcreturn $\Deduplicate\{(h,\DDLog_N(c))\mid(h,c)\in E\cup L\}$
		\end{codelines}
	\end{proofbox}

For a lossy public description, every product in $E\cup L$ is determined by a partial sum of the bounded exponents of the hash bases, so that $\abs E+\abs L\le r+1$ at every position. On every input, one iteration of \cref{fig:telf-image-enumeration} forms at most $2(r+1)$ products, because we return $\varnothing$ as soon as $\abs E+\abs L>r+1$.

\end{relocated}%

\begin{theoremE}[Strong Enumeration for the Trapdoor \ELF{}]\label{thm:telf-strong-enumeration}
	There is a deterministic algorithm $\Enumerate$ such that, for every $m\ge2$, $r\in[m]$, and $\fk\in\Supp(\TELFGenLos(m,r))$, $\Enumerate(\fk,r)$ returns $\Image(\ffk)$ in time $\poly(r,\log m)$. In particular, the construction in \cref{fig:small-exponent-telf} has a strongly efficiently enumerable image\ifConferenceVersion\else{} in the sense of \cref{def:efficiently-enumerable-telf}\fi.
\end{theoremE}
\ConferenceProofSketch{proof:telf-strong-enumeration}{
	In a prefix dynamic program over the bits of $m-1$, we keep the product of the pairs $(g_{1,j-1},K_{1,j-1})$ of public bases selected by the processed prefix of $m-1$, and a deduplicated set of the products for all smaller processed prefixes. After the last bit, these are the products for the inputs $u\in[m]$, and we obtain $\Image(\ffk)$ by applying $\DDLog_N$ to their second coordinates, by the definition of $\ELTDHEnc$. For a lossy public description, every partial product is $(g^a,g^{\sigma a})$ for a partial sum $a$ of the bounded exponents of \cref{cor:telf-exact-lossy-image}, so at most $r$ distinct products remain at each position. }{}{}
\begin{proofE}
	\label{proof:telf-strong-enumeration}
	We prove an invariant of the state sets of the algorithm $\Enumerate$ in \cref{fig:telf-image-enumeration} by induction over the processed positions, and we bound their sizes through the bounded subset sums of \cref{cor:telf-exact-lossy-image}.

	We write $\bbA\cqq\ZNsqx\times\ZNsqx$ and take $C_j\in\bbA$ and $b_j=x_j(m)$ as in \cref{fig:telf-image-enumeration}. For $u\in[2^n]$ and $j\in[n+1]$, we write $\pi_j(u)\cqq\floor{(u-1)/2^{j-1}}$ for the prefix of $u-1$ formed by the bits at positions $n,n-1,\dots,j$, and
	\[
		P_j(u)\cqq\prod_{i=j}^nC_i^{x_i(u)}\in\bbA
	\]
	for the corresponding product, so that $\pi_{n+1}(u)=0$ and $P_{n+1}(u)=(1,1)$. We show that after the iteration for $j$, and before the first iteration for $j=n+1$, the state sets satisfy
	\[
		E=\{P_j(m)\},
		\qquad
		L=\{P_j(u)\mid u\in[2^n],\ \pi_j(u)<\pi_j(m)\}.
	\]
	For $j=n+1$, we have $E=\{(1,1)\}$ and $L=\varnothing$. We suppose that the invariant holds after the iteration for $j+1$. Every $u\in[2^n]$ satisfies $\pi_j(u)=2\pi_{j+1}(u)+x_j(u)$ and $P_j(u)=P_{j+1}(u)\odot C_j^{x_j(u)}$, and $P_{j+1}(u)$ depends only on $\pi_{j+1}(u)$. Hence, $\pi_j(u)=\pi_j(m)$ holds exactly when $\pi_{j+1}(u)=\pi_{j+1}(m)$ and $x_j(u)=b_j$, and $\pi_j(u)<\pi_j(m)$ holds exactly when $\pi_{j+1}(u)<\pi_{j+1}(m)$, or when $\pi_{j+1}(u)=\pi_{j+1}(m)$ and $x_j(u)<b_j$. Every value of $\pi_{j+1}$ occurs with both values of the bit at position $j$, so the prefixes of the first kind contribute the products $L\cup(L\odot C_j)$. If $b_j=0$, the prefixes of the second kind do not occur, and the new state is $E'=E$ and $L'=L\cup(L\odot C_j)$. If $b_j=1$, the equal prefix has the product $P_{j+1}(m)\odot C_j$, and the only prefix of the second kind has the product $P_{j+1}(m)$, so the new state is $E'=E\odot C_j$ and $L'=L\cup(L\odot C_j)\cup E$. These are the two update rules of \cref{fig:telf-image-enumeration}, and $\Deduplicate$ preserves the set of elements, so the invariant holds after the iteration for $j$.

	After the iteration for $j=1$, we have $\pi_1(u)=u-1$, so $E\cup L=\{P_1(u)\mid u\in[m]\}$. For $u\in[m]$, the bits of $u-1$ at the positions $k\in[n..\wlam-1]$ are zero, so by the definitions of $\ELTDHHash$ and $\ELTDHEnc$ in \cref{fig:eltdh-dcr} at the all-zero tape, $P_1(u)=(h,c)$ for the digest $h=\Phihk(u-1)$ and the product $c$ formed in $\ELTDHEnc(\ek,(u-1);0)$, which returns $\DDLog_N(c)$. Hence, if the guard condition $\abs E+\abs L>r+1$ is false at every position, the last line of \cref{fig:telf-image-enumeration} returns the list of the distinct elements of $\{\ffk(u)\mid u\in[m]\}=\Image(\ffk)$.

	It remains to show that $\abs E+\abs L\le r+1$ after every iteration. If $r=m$, then $\abs E=1$ and $\abs L\le\pi_j(m)\le m-1$. If $r<m$, we write $\gamma\cqq\floor{(r-1)/\wlam}$ and take the exponents $s_{1,k}$ of \cref{cor:telf-exact-lossy-image}. If $\hk=\hkbot$ or $\gamma=0$, every base $g_{1,k}$ equals $1$, and the entries of the encoding key for the zero predicate are $K_{1,k}=g_{1,k}^\sigma=1$, so every product equals $(1,1)$ and $\abs E+\abs L\le2\le r+1$. Otherwise, $g_{1,k}=g^{s_{1,k}}$ and $K_{1,k}=g^{\sigma s_{1,k}}$, so $P_j(u)=(g^a,g^{\sigma a})$ for the partial sum
	\[
		a\cqq\sum_{i=j}^ns_{1,i-1}x_i(u)\in[0..n\gamma].
	\]
	Hence, $E\cup L$ has at most $n\gamma+1\le\wlam\gamma+1\le r$ elements, and $\abs E+\abs L\le\abs{E\cup L}+1\le r+1$.

	Therefore, for every supported public description, the guard condition is false at every position, and $\Enumerate(\fk,r)$ returns $\Image(\ffk)$. On every input, $\abs E+\abs L\le r+1$ holds before each iteration, so each iteration forms at most $2(r+1)$ products in $\bbA$ and deduplicates at most $2(r+1)$ elements. Over the $n$ iterations, the algorithm uses $\bigO(nr)$ multiplications in $\bbA$, sorts lists of at most $2(r+1)$ canonical representations, and ends with at most $r+1$ evaluations of $\DDLog_N$, each operation in time $\poly(\secparam)=\poly(\log m)$. Hence, it runs in time $\poly(r,\log m)$.
\end{proofE}

\ifConferenceVersion
The list returned by $\Enumerate(\fk,r)$ consists of the distinct elements of $\Image(\ffk)$. Hence, by choosing a uniform element of this list, we sample $\Image(\ffk)$ uniformly in time $\poly(r,\log m)$, as considered in~\cite{C:Zhandry16}.
\else
Uniform sampling from the lossy image, as considered in~\cite{C:Zhandry16}, follows from exact enumeration, since a uniform element of a list without repetitions is uniform on the set of its elements for every distribution of preimage sizes.

\begin{corollaryE}[Uniform Image Sampling]\label{cor:telf-image-sampling}
	Let $m\ge2$ and $r\in[m]$. For every $\fk\in\Supp(\TELFGenLos(m,r))$, we can sample uniformly from $\Image(\ffk)$ in time $\poly(r,\log m)$.
\end{corollaryE}
\begin{proofE}
	We run the algorithm of \cref{thm:telf-strong-enumeration} on $(\fk,r)$ and return a uniformly random element of the returned list. Since this list consists of the distinct elements of $\Image(\ffk)$, every element of $\Image(\ffk)$ is returned with probability $1/\abs{\Image(\ffk)}$, and the running time is $\poly(r,\log m)$.
\end{proofE}

\cite{EC:ABGMM26} constructs anamorphic-resistant encryption, in which every anamorphic triplet that is correct on average and semi-adaptively secure has an anamorphic message space of polynomial size, in two ways. The first construction combines a perfectly correct public-key encryption scheme, a universal hash family, fully unique non-interactive zero-knowledge arguments, which are known from indistinguishability obfuscation, and an \ELF{} with inputs of length $3\secparam$ and outputs of length polynomial in $\secparam$, and it uses the \ELF{} only through its two modes and mode indistinguishability, where injective generation must always return an injective function and the lossy image must be smaller than the image bound. The second construction avoids indistinguishability obfuscation. It combines public-key encryption and a lossy trapdoor function with a variant of the trapdoor \ELF{} of~\cite{EC:Zhandry19a}, and \cite{EC:ABGMM26} proves its anamorphic resistance by a direct reduction to exponential matrix $\DDH$ for that variant. Without its trapdoor, our trapdoor \ELF{} is an \ELF{} in the sense of the first construction.

\begin{corollaryE}[Always-Injective \ELF{}]\label{cor:telf-always-injective-elf}
	Suppose that $\DJGen(\cdot)$ satisfies $\DCR$ against nonuniform $\PPT$ adversaries and tunable $\DSEDL$. For $m\ge2$ and $r\in[2..m]$, let $\ELFGenAI(m,r)$ return the public description of $\TELFGenInj(m)$ if $r=m$, and the public description of $\TELFGenLos(m,r-1)$ if $r<m$, and let $\TELFEval$ evaluate it. Then the following hold.
	\begin{itemize}
		\item \textbf{Injective Mode.} Every output of $\ELFGenAI(m,m)$ describes an injective function.
		\item \textbf{Lossy Mode.} For every $r\in[2..m-1]$, every output $\fk$ of $\ELFGenAI(m,r)$ satisfies $\abs{\Image(\ffk)}<r$.
		\item \textbf{Mode Indistinguishability.} With evaluation by $\TELFEval$, the generator $\ELFGenAI$ satisfies the mode-indistinguishability condition of \cref{def:extremely-lossy-function}.
	\end{itemize}
\end{corollaryE}
\begin{proofE}
	The injective mode holds by \cref{thm:telf-injective-inversion}, and the lossy mode holds by \cref{thm:telf-lossy-image} at the image bound $r-1\in[m]$. For mode indistinguishability, we fix $t$ and $\delta$ and take the polynomial $\Qtelf$ that \cref{thm:telf-mode-indistinguishability} provides for them, and we set $Q\cqq\Qtelf+1$. For $r=m$, the two distributions are equal. For $r\in[Q(n_m)..m-1]$, the output of $\ELFGenAI(m,r)$ is that of $\TELFGenLos(m,r-1)$ with $r-1\ge\Qtelf(n_m)$, so \cref{thm:telf-mode-indistinguishability} bounds the advantage of every circuit of size $t(n_m)$ by $1/\delta(n_m)$ for every sufficiently large $m$.
\end{proofE}

By \cref{cor:telf-always-injective-elf} at $m=2^{3\secparam}$, where $\log\Nout(m)$ is polynomial in $\secparam$, the first construction of~\cite{EC:ABGMM26} is anamorphic-resistant and \INDCPA-secure with our \ELF{} in place of an \ELF{} from exponential $\DDH$, under tunable $\DSEDL$ and $\DCR$ together with the other building blocks of that construction. The same holds for the result of~\cite{EC:ABGMM26} that the first construction, applied to a public-key encryption scheme without anamorphic bits in the adaptive setting, returns such a scheme.

\begin{remark}[Partial Trapdoors]\label{rem:partial-trapdoors}
	The direct reduction for the second construction publishes one layer as $[\mat A]$ with $\mat A=\mat B\mat C$ for a uniform matrix $\mat B$, which serves as a partial trapdoor: with $\mat B$, the reduction decides membership in $[\Image\mat B]$, a superset of the image of that layer in both modes. Soundness of this test in the lossy mode uses the fact that the group of the layer has polynomial order, so every element that an adversary computes from $[\mat B\mat C]$ with $\mat C$ of rank one lies in a set of polynomial size. In our construction, the lossy image is small because the exponents are small in $\GN$, whose order is exponential. From the bases $g_{1,k}$ and the keys $K_{1,k}$, every party computes $h=\prod_kg_{1,k}^{y_k}$ and its consistent share $\DDLog_N(\prod_kK_{1,k}^{y_k})$ for arbitrary integers $y_k$. A public test that accepts these values for $y\in\bits^w$ and rejects them for large $y_k$ separates $g^z$ for $z$ in a short interval from $g^z$ for a large $z$, which tunable $\DSEDL$ assumes hard for observers of bounded size, and the trapdoor $\sigma$, which rejects them in the injective mode, determines the encoded predicate. Our proof of mode indistinguishability changes the predicate under $\DCR$, where $\sigma$ is unknown to the reduction, and the bases under tunable $\DSEDL$, where the size of the reduction is bounded independently of $\gamma$, so the reduction for the second construction does not carry over, and whether our construction admits a partial trapdoor is open.
\end{remark}
\fi

\subsection{Linear-Size Plain \texorpdfstring{\ELF{}}{ELF}}\label{subsec:plain-elf}

\ifConferenceVersion
With the same small-exponent subset products, we obtain a linear-size \ELF{} over every cyclic group sampler satisfying tunable $\DSEDL$, and its mode indistinguishability depends only on tunable $\DSEDL$. For $m\ge2$, we write $n\cqq\clog{m}$ and $x(u)=(x_1(u),\ldots,x_n(u))$ for the binary expansion of $u-1$. The public description consists of bases $g_1,\ldots,g_n$, and evaluation is
\[
	f(u)=\prod_{j=1}^ng_j^{x_j(u)}.
\]
In injective generation, we sample the bases uniformly. In the lossy mode, we set $\gamma\cqq\floor{(r-1)/n}$, and every base is the identity when $\gamma=0$, while for $\gamma\ge1$ we sample $s_j\sample[\gamma]$ and set $g_j=g^{s_j}$. Hence, the lossy image is indexed by a subset sum in $[0..n\gamma]$, and injectivity follows by a union bound over the at most $3^n-1$ nonzero relations in $\{-1,0,1\}^n$. With $\secparam=\ceil{cn}$ for a rational $c>\log_23$, the construction is an \ELF{} under tunable $\DSEDL$, with $n$ group elements in the public description, one group element in the output, and $\bigO(n)$ group operations per evaluation. For $\CGen=\DJGen$, whose group order exceeds $2^{2\secparam}$, every rational $c>(\log_23)/2$ suffices.
\else
\label{app:plain-elf}
\fi
\begin{relocated}{plain-elf-full}
With the same small-exponent subset products and without a trapdoor, we obtain an \ELF{} whose public description consists of $n$ group elements in addition to the group parameters. Its mode indistinguishability depends only on tunable $\DSEDL$, whereas in the trapdoor construction we use $\DCR$ to hide the predicate in the encoding key, and its injective and lossy modes hold unconditionally.

In the construction and its proofs, we use only uniform sampling in the group, exponentiation of the generator, and a lower bound on the group order, so we state the construction over an arbitrary cyclic group sampler $\CGen$ in the sense of \cref{def:group-sampler}. We fix such a sampler and a positive rational constant $c$, and for $m\ge2$, we set $n\cqq\clog{m}$ and, in this subsection, $\secparam\cqq\ceil{cn}$ in place of $\groupsec(n)$. The first two operations are part of \cref{def:group-sampler}, and the lower bound on the order determines the admissible range of $c$, because we use it to show that, except with negligible probability, the uniform exponents of the injective bases satisfy only the trivial relation with coefficients in $\{-1,0,1\}$. We count group elements, and each of them has a representation of at most $\grplen(\secparam)$ bits, so the public description is linear in $n$ as a count of elements and superlinear in bits.

To identify outputs with elements of an interval as in \cref{def:extremely-lossy-function}, we encode a group element by the length of its canonical representation, written in exactly $\ceil{\log_2(\grplen(\secparam)+1)}$ bits, followed by the representation padded with zeros to length $\grplen(\secparam)$. This encoding is injective, has the fixed length $L(\secparam)\cqq\grplen(\secparam)+\ceil{\log_2(\grplen(\secparam)+1)}$, and is computable and invertible in polynomial time. We interpret the resulting $L(\secparam)$-bit string as its nonnegative binary value plus one, which is an element of $[2^{L(\secparam)}]$, and we take $\Nout(m)\cqq2^{L(\ceil{c\clog{m}})}$, so that $\log \Nout(m)=\poly(\log m)$.

We replace every failure of $\CGen$ or of the base sampler by the fixed description $\fbot\cqq\left(1,m,n,\secparam,\bot,(1)_{j\in[n]}\right)$. Since its first entry is $1$, $\ELFEval(\fbot,u)=1\in[\Nout(m)]$ for every $u\in[m]$, and $\fbot$ can lie in the support of $\ELFGen(m,r)$. Both generators output $\fbot$ when $\CGen$ fails, and $\ELFGen(m,m)$ also outputs $\fbot$ when one of its $n$ base samples fails. The function $\ELFEval(\fbot,\cdot)$ satisfies every lossy image bound, and since $\fbot$ is output with negligible probability, its probability is a negligible term of the injectivity error, which \cref{def:extremely-lossy-function} requires to be negligible.

\begin{proofbox}{Linear-size \ELF{} over a cyclic group sampler $\CGen$, for $m\ge2$ and $r\in[m]$, with $n=\clog{m}$ and $\secparam=\ceil{cn}$. The first entry $\chi$ of a description equals $1$ exactly for $\fbot$.}{fig:plain-elf}
		\begin{minipage}[t]{0.54\linewidth}
			\vspace{0pt}
			{\footnotesize\underline{$\ELFGen(m,r)$}\par}
			\smallskip
			\begin{codelines}
				\item \pcassert $m\ge2\land r\in[m]$
				\item $n\gets\clog{m}$; $\secparam\gets\ceil{cn}$
				\item $\pp\sample\CGen(1^\secparam)$
				\item \pcif $\pp=\bot$ \pcreturn $f\gets \fbot$
				\item $g\gets g(\pp)$
				\item \pcif $r=m$
				\codelinesplus
				\item \pcfor $j\in[n]$
				\codelinesplus
					\item \label{line:plain-elf-uniform-base}$g_j\sample\GG(\pp)$
					\codelinesminus
					\item \pcif $\exists j\colon g_j=\bot$ \pcreturn $f\gets \fbot$
					\codelinesminus
				\item \pcelse
				\codelinesplus
				\item $\gamma\gets\floor{(r-1)/n}$
				\item \pcfor $j\in[n]$
				\codelinesplus
				\item \pcif $\gamma=0$\quad $g_j\gets1$
				\item \pcelse
				\codelinesplus
				\item \label{line:plain-elf-bounded-base}$s_j\sample[\gamma]$
				\item $g_j\gets g^{s_j}$
				\codelinesminus
				\codelinesminus
				\codelinesminus
				\item $f\gets(0,m,n,\secparam,\pp,(g_j)_{j\in[n]})$
				\item \pcreturn $f$
			\end{codelines}
		\end{minipage}\hfill%
		\begin{minipage}[t]{0.41\linewidth}
			\vspace{0pt}
			{\footnotesize\underline{$\ELFEval(f,u)$}\par}
			\smallskip
			\begin{codelines}
				\item \pcparse $f=(\chi,m,n,\secparam,\pp,(g_j)_{j\in[n]})$
				\item \pcassert $u\in[m]$
				\item \pcif $\chi=1$ \pcreturn $1$
				\item \pcreturn $\prod_{j=1}^ng_j^{x_j(u)}$
			\end{codelines}
		\end{minipage}
	\end{proofbox}

For mode indistinguishability, we identify each generator with one distribution of \cref{lem:tunable-scalar-to-vec-dsedl}, since a public description differs from a vector of bounded or uniform powers only in deterministic entries.

\begin{theoremE}[Generators as Vector \texorpdfstring{\DSEDL{}}{DSEDL} Samples]\label{thm:plain-mode-characterization}
	Let $m\ge2$ and $r\in[m-1]$, write $n=\clog{m}$, $\secparam=\ceil{cn}$, and $\gamma=\floor{(r-1)/n}$, and suppose that $1\le\gamma<\ord(\pp)$ for every $\pp\in\Supp(\CGen(1^\secparam))$. Let $\DvDSEDL_{0,\secparam,\gamma}$ and $\DvDSEDL_{1,\secparam,\gamma}$ be the distributions of \cref{lem:tunable-scalar-to-vec-dsedl} for the constant coordinate count $n$, and define the deterministic polynomial-time map $\psi$ by $\psi(\bot)\cqq\fbot$ and
	\[
		\psi\left(1^\secparam,\pp,\gamma,(G_j)_{j\in[n]}\right)\cqq\left(0,m,n,\secparam,\pp,(G_j)_{j\in[n]}\right).
	\]
	Then
	\[
		\begin{aligned}
			\left\{\begin{array}{@{}c|l@{}}f & f\sample\ELFGen(m,m)\end{array}\right\}
			&=\left\{\begin{array}{@{}c|l@{}}\psi(X) & X\sample\DvDSEDL_{1,\secparam,\gamma}\end{array}\right\},\\
			\left\{\begin{array}{@{}c|l@{}}f & f\sample\ELFGen(m,r)\end{array}\right\}
			&=\left\{\begin{array}{@{}c|l@{}}\psi(X) & X\sample\DvDSEDL_{0,\secparam,\gamma}\end{array}\right\},
		\end{aligned}
	\]
	and the polynomial-time map $(0,m,n,\secparam,\pp,(G_j)_{j\in[n]})\mapsto(1^\secparam,\pp,\gamma,(G_j)_{j\in[n]})$ inverts $\psi$ on the supports of both vector distributions without $\bot$.
\end{theoremE}

\begin{proofE}
	We compare the samplings of the generators and of the vector distributions step by step.

	Both generators and both vector distributions first sample $\pp\sample\CGen(1^\secparam)$, and on $\pp=\bot$ the generators output $\fbot=\psi(\bot)$, while the vector distributions output $\bot$. For $\pp\ne\bot$, $\ELFGen(m,m)$ samples $g_1,\dots,g_n\sample\GG(\pp)$ and outputs $\fbot$ if one of these samples fails and $(0,m,n,\secparam,\pp,(g_j)_{j\in[n]})$ otherwise, and $\DvDSEDL_{1,\secparam,\gamma}$ samples the same $n$ elements and outputs $\bot$ if one of them fails and $(1^\secparam,\pp,\gamma,(g_j)_{j\in[n]})$ otherwise, which implies the first identity. Since $r<m$ and $\floor{(r-1)/n}=\gamma\ge1$, $\ELFGen(m,r)$ samples $(s_j)_{j\in[n]}\sample[\gamma]^n$ and outputs $(0,m,n,\secparam,\pp,(g^{s_j})_{j\in[n]})$, and $\DvDSEDL_{0,\secparam,\gamma}$ samples the same exponents and outputs $(1^\secparam,\pp,\gamma,(g^{s_j})_{j\in[n]})$, which implies the second identity. Outside $\bot$, the map $\psi$ replaces the entries $1^\secparam$ and $\gamma$ by $(0,m,n)$, which depend only on $m$, and keeps $\secparam$, so the map in the statement inverts $\psi$.
\end{proofE}

From \cref{thm:plain-mode-characterization,lem:tunable-scalar-to-vec-dsedl}, a union bound over relations among uniform exponents, and a count of bounded subset sums, we obtain an \ELF{} for every rational $c>\log_23$, where we need $c>\log_23$ to compare the $3^n-1$ nonzero relations with the group order $\ord(\pp)\ge2^{\ceil{cn}}$.

\begin{theoremE}[Linear-Size \ELF{}]\label{thm:plain-elf}
	Let $c>\log_23$, and suppose that $\CGen(\cdot)$ satisfies tunable $\DSEDL$. Then the construction in \cref{fig:plain-elf} is an extremely lossy function. Excluding the group parameters $\pp$, its public description contains $n$ group elements, and evaluation returns one group element using $\bigO(n)$ group operations.
\end{theoremE}
\begin{proofE}
	We rule out nonzero relations with coefficients in $\{-1,0,1\}$ among the uniform exponents with a union bound, we count bounded subset sums for the lossy mode, and for mode indistinguishability we combine \cref{thm:plain-mode-characterization} with \cref{lem:tunable-scalar-to-vec-dsedl}.

	\subparhead{Efficiency and Output Encoding} The parameters satisfy $n=\Theta(\log m)$ and $\secparam=\Theta(n)$. In generation, we sample $n$ group elements or compute $g^{s_j}$ for $n$ exponents $s_j\le\gamma<m$, and in evaluation, we multiply the bases $g_j$ with $x_j(u)=1$, using at most $n-1$ group multiplications. By \cref{def:group-sampler}, these operations take time $\poly(\secparam)=\poly(\log m)$, and every group element has a representation of at most $\grplen(\secparam)$ bits, so under the padded encoding, the output lies in $[2^{L(\secparam)}]=[\Nout(m)]$ with $\log \Nout(m)=L(\secparam)=\poly(\log m)$.

	\subparhead{Injective Mode} We fix $\pp\in\Supp(\CGen(1^\secparam))$ and write $M\cqq\ord(\pp)$, so that $g=g(\pp)$ generates the cyclic group $\GG(\pp)$ of order $M$. Conditioned on the event that all $n$ base samples succeed, the bases are independent and uniform in $\GG(\pp)$ by \cref{def:group-sampler}, so $g_j=g^{a_j}$ for independent uniform $a_j\in\ZZ/M\ZZ$. We write $\Einj$ for the event that two distinct inputs $u,u'\in[m]$ have the same output. From such a collision, we obtain distinct $x(u),x(u')\in\bits^n$ and a nonzero vector $v\cqq x(u)-x(u')\in\{-1,0,1\}^n$ with
	\[
		\sum_{j=1}^nv_ja_j=0
	\]
	in $\ZZ/M\ZZ$. We fix a nonzero $v\in\{-1,0,1\}^n$ and choose $\jstar\in[n]$ with $v_{\jstar}\in\{-1,1\}$. Since $v_{\jstar}$ is a unit modulo $M$, after we condition on $(a_j)_{j\neq\jstar}$, exactly one value of $a_{\jstar}$ satisfies this equation, namely
	\[
		a_{\jstar}=-v_{\jstar}^{-1}\sum_{j\neq\jstar}v_ja_j.
	\]
	Since $a_{\jstar}$ is uniform in $\ZZ/M\ZZ$ and independent of $(a_j)_{j\neq\jstar}$, the equation holds with probability $1/M$. We take a union bound over the $3^n-1$ nonzero vectors $v$, use $M\ge2^\secparam\ge2^{cn}$ from \cref{def:group-sampler}, and obtain
	\[
		\begin{aligned}
			\Prob\left[\begin{array}{@{}c|l@{}}\Einj & (a_j)_{j\in[n]}\sample(\ZZ/M\ZZ)^n\end{array}\right]
			&\le\sum_{v\in\{-1,0,1\}^n\setminus\{0\}}\frac1M=\frac{3^n-1}{M}\\
			&<\frac{3^n}{2^{cn}}=2^{-(c-\log_23)n}.
		\end{aligned}
	\]
	This bound holds for every $\pp$, and it is negligible in $n$ because $c>\log_23$. Generation outputs $\fbot$ with probability at most
	\[
		\Prob\left[\begin{array}{@{}c|l@{}}\pp=\bot & \pp\sample\CGen(1^\secparam)\end{array}\right]+n\nuG(\secparam),
	\]
	which is negligible in $n$ because $cn\le\secparam\le cn+1$. Hence, $f$ is injective except with probability at most the sum of these two negligible terms.

	\subparhead{Lossy Mode} If $r=m$, the domain has $r$ points, and $\abs{\Image(f)}\le r$. We suppose that $r<m$. The fallback description $\fbot$ describes a constant function, and on every other description in the support, $\gamma=\floor{(r-1)/n}$. If $\gamma=0$, every base is $1$, and the function is constant. If $\gamma\ge1$, every output equals $g^{A(u)}$ for
	\[
		A(u)\cqq\sum_{j=1}^ns_jx_j(u)\in[0..n\gamma].
	\]
	In all cases, $\abs{\Image(f)}\le n\gamma+1=n\floor{(r-1)/n}+1\le r$.

	\subparhead{Mode Indistinguishability} We fix a polynomial $t$ and a positive polynomial $\delta$, and we choose constants $C>0$ and $D\in\NN$ such that $t(z)\le C(z+1)^D$ and $2\delta(z)\le C(z+1)^D$ for every $z\in\NN$. We write $\kappa(\secparam)\cqq\ceil{\secparam/c}$, and since $\secparam=\ceil{cn}\ge cn$, we have $n\le\secparam/c\le\kappa(\secparam)$. We fix a polynomial $p$ such that, for every $\secparam$, a circuit of size at most $p(\secparam)$ maps a challenge of \cref{lem:tunable-scalar-to-vec-dsedl} with $\kappa(\secparam)$ coordinates to the image under $\psi$ of its first $n$ coordinates, for hardwired $m$ and $n$, and we set
	\[
		\that(\secparam)\cqq C(\secparam/c+1)^D+p(\secparam),
		\qquad
		\dhat(\secparam)\cqq C(\secparam/c+1)^D.
	\]
	Both are polynomials because $c$ is a positive rational and $D\in\NN$, and from $n\le\secparam/c$, we obtain $t(n)+p(\secparam)\le\that(\secparam)$ and $\dhat(\secparam)\ge2\delta(n)$.

	We apply \cref{lem:tunable-scalar-to-vec-dsedl} with coordinate count $\kappa$, size bound $\that$, and target polynomial $\dhat$, and we obtain a polynomial $Q$ and $\secparam_0\in\NN$ such that, for every $\secparam\ge\secparam_0$ and every integer $\gamma$ with $Q(\secparam)\le\gamma<\ord(\pp)$ for all $\pp\in\Supp(\CGen(1^\secparam))$, every circuit of size at most $\that(\secparam)$ distinguishes $\DvDSEDL_{0,\secparam,\gamma}$ from $\DvDSEDL_{1,\secparam,\gamma}$ with $\kappa(\secparam)$ coordinates with advantage at most $1/\dhat(\secparam)$. We choose an integer-valued polynomial $\Qbar$ with $\Qbar(\secparam)\ge\max\{1,\ceil{Q(\secparam)}\}$ for every $\secparam$, and an integer-valued polynomial $\Qplain$ with $\Qplain(n)\ge n\Qbar(\ceil{cn})+1$ for every $n\in\NN$. Since $\kappa$ is polynomially bounded and $\nuG$ is negligible, $\epsunused(\secparam)\cqq\kappa(\secparam)\nuG(\secparam)$ is negligible, so there is $\secparam_1\in\NN$ with $\epsunused(\secparam)\le1/\dhat(\secparam)$ for every $\secparam\ge\secparam_1$.

	We fix $m\ge2$ with $\secparam=\ceil{cn}\ge\max\{\secparam_0,\secparam_1\}$, which holds for every sufficiently large $m$, and $r\in[\Qplain(n)..m]$ and a circuit $\advD$ of size at most $t(n)$, and we write $p_0$ and $p_1$ for the probabilities that $\advD$ outputs $1$ on a sample of $\ELFGen(m,m)$ and of $\ELFGen(m,r)$, respectively. If $r=m$, then $p_0=p_1$. We suppose that $r<m$ and set $\gamma\cqq\floor{(r-1)/n}$. From $r\ge\Qplain(n)\ge n\Qbar(\secparam)+1$, we obtain $\gamma\ge\Qbar(\secparam)\ge\max\{1,Q(\secparam)\}$. Since $c>\log_23>1$, \cref{def:group-sampler} implies $\gamma<r<m\le2^n<2^{cn}\le2^\secparam\le\ord(\pp)$ for every $\pp\in\Supp(\CGen(1^\secparam))$. Hence, the hypotheses of \cref{thm:plain-mode-characterization} hold.

	\begin{proofbox}{Hybrids for mode indistinguishability of the plain \ELF{}, at $r\in[m-1]$ with $\gamma=\floor{(r-1)/n}\ge1$. The highlighted line of $\Hybrid_1$ is the line in which it differs from $\Hybrid_0$.}{fig:plain-elf-mode-hybrids}
			\begin{minipage}[t]{0.47\linewidth}
				\vspace{0pt}
				{\footnotesize\underline{$\Hybrid_0$}\codelinesrightalign{$\ELFGen(m,m)$}\par}
				\smallskip
				\begin{codelines}
					\item $\pp\sample\CGen(1^\secparam)$
					\item \pcif $\pp=\bot$ \pcreturn $\fbot$
					\item $g\gets g(\pp)$
					\item $(g_j)_{j\in[n]}\sample\GG(\pp)^n$
					\item \pcif $\exists j\colon g_j=\bot$ \pcreturn $\fbot$
					\item \pcreturn $(0,m,n,\secparam,\pp,(g_j)_{j\in[n]})$
				\end{codelines}
			\end{minipage}\hfill%
			\begin{minipage}[t]{0.47\linewidth}
				\vspace{0pt}
				{\footnotesize\underline{$\Hybrid_1$}\codelinesrightalign{$\ELFGen(m,r)$}\par}
				\smallskip
				\begin{codelines}
					\item $\pp\sample\CGen(1^\secparam)$
					\item \pcif $\pp=\bot$ \pcreturn $\fbot$
					\item $g\gets g(\pp)$
					\item \label{line:plain-hyb-bases}\hybchange{$(s_j)_{j\in[n]}\sample[\gamma]^n$; $g_j\gets g^{s_j}$}
					\item \pcreturn $(0,m,n,\secparam,\pp,(g_j)_{j\in[n]})$
				\end{codelines}
			\end{minipage}
		\end{proofbox}

	In \cref{fig:plain-elf-mode-hybrids}, $\Hybrid_0$ and $\Hybrid_1$ are $\ELFGen(m,m)$ and $\ELFGen(m,r)$ for this $\gamma\ge1$, and at \cref{line:plain-hyb-bases}, we replace the uniform bases by bounded powers. Our reduction $\advB$ receives a challenge $X$ from $\DvDSEDL_{b,\secparam,\gamma}$ with $\kappa(\secparam)$ coordinates. If $X=\bot$, it outputs $\advD(\fbot)$, and otherwise, for $X=(1^\secparam,\pp,\gamma,(G_j)_{j\in[\kappa(\secparam)]})$, it outputs $\advD(0,m,n,\secparam,\pp,(G_j)_{j\in[n]})$. Its size is at most $t(n)+p(\secparam)\le\that(\secparam)$, and we write $q_b$ for the probability that it outputs $1$. For $b=0$, the first $n$ coordinates of $X$, with $\bot$ kept as $\bot$, are distributed as $\DvDSEDL_{0,\secparam,\gamma}$ with $n$ coordinates, because a bounded challenge equals $\bot$ exactly when $\CGen$ fails, so $q_0=p_1$ by \cref{thm:plain-mode-characterization}. For $b=1$, the first $n$ coordinates of $X$ are distributed as $\DvDSEDL_{1,\secparam,\gamma}$ with $n$ coordinates except on the event that $\CGen$ and the first $n$ coordinate samples succeed while one of the last $\kappa(\secparam)-n$ coordinate samples fails, whose probability is at most $\epsunused(\secparam)$, so $\abs{q_1-p_0}\le\epsunused(\secparam)$ by \cref{thm:plain-mode-characterization}. By the choice of $Q$, $\secparam_0$, and $\secparam_1$, we obtain
	\[
		\abs{p_0-p_1}\le\abs{q_0-q_1}+\epsunused(\secparam)\le\frac2{\dhat(\secparam)}\le\frac1{\delta(n)}.
	\]
	Hence, every circuit of size at most $t(n)$ distinguishes $\ELFGen(m,m)$ from $\ELFGen(m,r)$ with advantage at most $1/\delta(n)$, and with the efficiency and the injective and lossy modes established in this proof, the construction satisfies every condition of \cref{def:extremely-lossy-function}.
\end{proofE}

In the proof of \cref{thm:plain-elf}, we use $c>\log_23$ only through the lower bound $\ord(\pp)\ge2^\secparam$ of \cref{def:group-sampler}, in the union bound of the injective mode and in the inequality $2^n<\ord(\pp)$ of the mode transition. For $\CGen=\DJGen$, the order $M=p'q'$ exceeds $2^{2\ell(\secparam)-4}\ge2^{2\secparam}$ by \cref{lem:dj-structure}, so a smaller constant suffices, and since $\secparam=\ceil{cn}$, group elements are shorter for smaller $c$ at the same $m$.

\begin{corollaryE}[Linear-Size \ELF{} over \DamgardJurik Groups]\label{cor:plain-elf-dj}
	Let $c>(\log_23)/2$, and suppose that $\DJGen(\cdot)$ satisfies tunable $\DSEDL$. Then the construction in \cref{fig:plain-elf} at $\CGen=\DJGen$ is an extremely lossy function, with the description size, output size, and evaluation cost of \cref{thm:plain-elf}.
\end{corollaryE}

\begin{proofE}
	We repeat the proof of \cref{thm:plain-elf} with $\CGen=\DJGen$, which is a cyclic group sampler by \cref{lem:dj-structure}, and we replace the lower bound $\ord(\pp)\ge2^\secparam$ by $M=p'q'>2^{2\ell(\secparam)-4}\ge2^{2\secparam}\ge2^{2cn}$, which holds by \cref{lem:dj-structure} because $\ell(\secparam)\ge\secparam+2$. In the injective mode, the union bound over the $3^n-1$ nonzero vectors $v$ becomes
	\[
		\Prob\left[\begin{array}{@{}c|l@{}}\Einj & (a_j)_{j\in[n]}\sample(\ZZ/M\ZZ)^n\end{array}\right]
		\le\frac{3^n-1}{M}
		<\frac{3^n}{2^{2cn}}
		=2^{-(2c-\log_23)n},
	\]
	which is negligible because $2c>\log_23$. Since $2c>\log_23>1$, we have $2^n<2^{2cn}<M$ for every $n\ge1$, so $\gamma<r<m\le2^n<M$ in the mode transition. The remaining steps of the proof of \cref{thm:plain-elf} hold for every rational $c>0$, and $\fbot$ is the same for every sampler.
\end{proofE}

For $c<1$, the interval bound $\gamma=\floor{(r-1)/n}<2^n$ of the lossy mode can exceed the smallest prime factor $\min\{p',q'\}$ of the group order $M=p'q'$, since $p',q'<2^{\ell(\secparam)-1}$ and $\secparam=\ceil{cn}$. The generic-group bound of \cref{thm:generic-vector-dsedl} holds for every $\gamma$ smaller than the group order $q$, with the term $1/\gamma+2/\piq$ for each pair of forms, and hence also when $\gamma$ exceeds the smallest prime factor $\piq$ of $q$.

For the plain construction, we specialize \cref{fig:telf-image-enumeration} to one coordinate, since the factors of every output are the bases $g_j$ and, in the lossy mode, a partial product is determined by a bounded subset sum.

\begin{theoremE}[Strong Enumeration for the Plain \ELF{}]\label{thm:plain-elf-strong-enumeration}
	The construction in \cref{fig:plain-elf} has strongly efficiently enumerable images in the sense of \cref{def:efficiently-enum-img}. More precisely, there is a deterministic algorithm that, on input $(f,r)$ with $m\ge2$, $r\in[m]$, and $f\in\Supp(\ELFGen(m,r))$, returns $\{\ELFEval(f,u)\mid u\in[m]\}$ using $\bigO(nr)$ group multiplications and runs in time $\poly(r,\log m)$.
\end{theoremE}

\begin{proofE}
	We specialize the prefix dynamic program of \cref{fig:telf-image-enumeration} to one coordinate and omit its final $\DDLog_N$ step. If $f=\fbot$, then $\ELFEval(\fbot,\cdot)$ is constant, and we return the list $(1)$. Otherwise, we write $f=(0,m,n,\secparam,\pp,(g_j)_{j\in[n]})$, take $\bbA\cqq\GG(\pp)$ and $C_j\cqq g_j$, and run the updates of \cref{fig:telf-image-enumeration} from $E=\{1\}$ and $L=\varnothing$. By the prefix invariant in the proof of \cref{thm:telf-strong-enumeration}, which holds over every finite abelian group, the final deduplicated set is
	\[
		\left\{\prod_{j=1}^n g_j^{x_j(u)}\ \middle|\ u\in[m]\right\}
		=
		\{\ELFEval(f,u)\mid u\in[m]\}.
	\]

	If $r=m$, then $\abs E=1$ and at most $m-1$ smaller processed prefixes occur at any position, so $\abs E+\abs L\le r+1$. If $r<m$, we write $\gamma=\floor{(r-1)/n}$, and we set $s_j\cqq0$ for $j\in[n]$ when $\gamma=0$. When $\gamma\ge1$, the definition of $\ELFGen$ gives $g_j=g^{s_j}$ with $s_j\in[\gamma]$. Every partial product is then $g^a$ for a partial sum $a\in[0..n\gamma]$, so at most $n\gamma+1\le r$ elements remain after deduplication, and $\abs E+\abs L\le r+1$. As in \cref{fig:telf-image-enumeration}, we return $\varnothing$ when $\abs E+\abs L>r+1$, and by these bounds, this condition is false at every position for every $f$ in the support. Therefore, each iteration forms at most $2(r+1)$ products, and the algorithm uses $\bigO(nr)$ group multiplications and runs in time $\poly(r,\log m)$.
\end{proofE}

Uniform sampling from the lossy image, as considered in~\cite{C:Zhandry16}, follows from the exact image list, since a uniform position of a list without repetitions is a uniform element of the image for every distribution of preimage sizes.

\begin{corollaryE}[Uniform Sampling for the Plain \ELF{}]\label{cor:plain-elf-image-sampling}
	For every $m\ge2$, every $r\in[m]$, and every $f\in\Supp(\ELFGen(m,r))$, we can sample uniformly from $\{\ELFEval(f,u)\mid u\in[m]\}$ in time $\poly(r,\log m)$.
\end{corollaryE}

\begin{proofE}
	We enumerate the image with the algorithm of \cref{thm:plain-elf-strong-enumeration} and return a uniformly random element of the returned list. Since this list consists of the distinct elements of the image, each element of the image is returned with probability one over the size of the image, and the running time is $\poly(r,\log m)$.
\end{proofE}

\end{relocated}%

\section{All-But-Some and All-But-One Trapdoor \texorpdfstring{\ELF{}s}{ELFs}}\label{sec:branch-telf}
Following~\cite{EC:Zhandry19a}, we give trapdoor \ELF{}s branches: the functions on a selected set of polynomially many branches are lossy, the functions on every other branch are invertible, and the two generation modes are indistinguishable even given inversion queries on the nonlossy branches. If we evaluate several lossy functions in parallel, as in all-but-one constructions of lossy trapdoor functions~\cite{STOC:PeiWat08}, the lossy image size is the product of their image sizes, which is not polynomially bounded~\cite{EC:Zhandry19a}. We use one \DamgardJurik hash key of \cref{sec:hashing} for all branches and subtract a branch term from the encodings, so that on branch $b$ the input is multiplied by $P(\beta(b))$ for an integer $\beta(b)$ and a polynomial $P$ whose roots are the integers of the selected branches.

\ifConferenceVersion\else
\subsection{All-But-Some and All-But-One Syntax}
\label{subsec:abo-syntax}
We write $a\ge1$ for the branch length, so that a branch is an element of $\bits^a$. In the selective notion of~\cite{EC:Zhandry19a}, the adversary chooses the lossy branches before it receives the public description, and its second stage receives its first-stage state and an inversion oracle on every remaining branch.
\fi
\begin{relocated}{abs-telf-definition}
Our definition generalizes the all-but-one trapdoor \ELF{} of~\cite{EC:Zhandry19a}, which has one lossy branch, to a selected set $\Sstar$ of $s$ lossy branches: both generators receive $1^s$, the lossy generator receives $\Sstar$, the length bounds and the failure probabilities are polynomial and negligible in $n_m+a+s$, and the selective-security threshold depends on polynomial bounds $A$ and $S$ for the branch length and for $s$.

\begin{definition}[All-But-Some Trapdoor \ELF{}]\label{def:abs-telf}
	For every $m\ge2$, write $n_m\cqq\clog{m}$. An all-but-some trapdoor extremely lossy function consists of two probabilistic polynomial-time generation algorithms and two deterministic polynomial-time algorithms, $\ABSTELF=(\ABSGenInj,\allowbreak\ABSGenLos,\allowbreak\ABSEval,\allowbreak\ABSInv)$. Every public description $\fk$ determines a finite output space $\Yfk$ with an injective canonical encoding of length $\poly(n_m+a+s)$, and membership testing, encoding, and decoding run in time $\poly(n_m+a+s)$. The algorithms have the following syntax.
	\begin{itemize}
		\item $(\fk,\td)\sample\ABSGenInj(m,1^a,1^s)$. On input integers $m\ge2$ and $a,s\ge1$ with $s\le2^a$, the injective generator samples a public description $\fk$ and an inversion trapdoor $\td$ for the branches in $\bits^a$.
		\item $(\fk,\td)\sample\ABSGenLos(m,1^a,1^s,r,\Sstar)$. On input $m$, $a$, and $s$ as for $\ABSGenInj$, $r\in[m]$, and $\Sstar\subseteq\bits^a$ with $\abs{\Sstar}=s$, the lossy generator samples a public description $\fk$ and an inversion trapdoor $\td$ for the branches in $\bits^a\setminus\Sstar$.
		\item $y\gets\ABSEval(\fk,b,u)$. On input a public description $\fk$ and $(b,u)\in\bits^a\times[m]$, the evaluation algorithm returns the value $y\in\Yfk$ of branch $b$ at $u$. We write $f_{\fk,b}(u)\cqq\ABSEval(\fk,b,u)$.
		\item $u\gets\ABSInv(\fk,\td,b,y)$. On input a public description $\fk$, a trapdoor $\td$, a branch $b$, and a string $y$, the inversion algorithm returns an element of $[m]\cup\{\bot\}$, and it returns $\bot$ whenever $(b,y)\notin\bits^a\times\Yfk$.
	\end{itemize}
	We call $\ABSTELF$ an all-but-some trapdoor \ELF{} if it has the following properties.
	\begin{itemize}
		\item \textbf{Compact descriptions.} Public descriptions, trapdoors, branches, and outputs have length $\poly(n_m+a+s)$, and every algorithm runs in time polynomial in its input length.
		\item \textbf{Injective mode.} There is a negligible function $\nuinj\colon\NN\to[0,1]$ such that for all $m\ge2$ and $a,s\ge1$ with $s\le2^a$, with $z\cqq n_m+a+s$, the event $\Einj(\fk,\td)$ that
		\[
			\forall(b,u)\in\bits^a\times[m]\colon\ \ABSInv(\fk,\td,b,f_{\fk,b}(u))=u
		\]
		satisfies
		\[
			\Prob\left[
			\begin{array}{@{}c|l@{}}
				\Einj(\fk,\td) & (\fk,\td)\sample\ABSGenInj(m,1^a,1^s)
			\end{array}
			\right]
			\ge1-\nuinj(z).
		\]
		\item \textbf{Lossy mode.} There is a negligible function $\nulos\colon\NN\to[0,1]$ such that for all $m\ge2$ and $a,s\ge1$ with $s\le2^a$, every $r\in[m]$, and every $\Sstar\subseteq\bits^a$ with $\abs{\Sstar}=s$, with $z\cqq n_m+a+s$, the event $\Elos(\fk,\td)$ that
		\[
			\begin{aligned}
				&\forall b\in\Sstar\colon\ \abs{\Image(f_{\fk,b})}\le r\\
				\land\ &\forall(b,u)\in(\bits^a\setminus\Sstar)\times[m]\colon\ \ABSInv(\fk,\td,b,f_{\fk,b}(u))=u
			\end{aligned}
		\]
		satisfies
		\[
			\Prob\left[
			\begin{array}{@{}c|l@{}}
				\Elos(\fk,\td) & (\fk,\td)\sample\ABSGenLos(m,1^a,1^s,r,\Sstar)
			\end{array}
			\right]
			\ge1-\nulos(z).
		\]
		\item \textbf{Selective security.} For $c\in\bits$, integers $m\ge2$, $a,s\ge1$ with $s\le2^a$, $r\in[m]$, and a two-stage nonuniform probabilistic circuit $\advA=(\advA_1,\advA_2)$ whose second stage returns a bit, the experiment $\Game_{\ABSTELF,\advA}^c(m,a,s,r)$ with $n\cqq n_m$ runs as follows.
		\begin{enumerate}
			\item It samples $(\Sstar,\var{st})\sample\advA_1(1^n,m,1^a,1^s,r)$ and returns $0$ unless $\Sstar\subseteq\bits^a$ and $\abs{\Sstar}=s$.
			\item If $c=0$, it samples $(\fk,\td)\sample\ABSGenInj(m,1^a,1^s)$. If $c=1$, it samples $(\fk,\td)\sample\ABSGenLos(m,1^a,1^s,r,\Sstar)$.
			\item It returns $\advA_2^{\Oinv}(\var{st},\fk)$ for the total oracle
			\[
				\Oinv(b,y)\cqq
				\begin{cases}
					\bot,&(b,y)\notin(\bits^a\setminus\Sstar)\times\Yfk,\\
					\ABSInv(\fk,\td,b,y),&(b,y)\in(\bits^a\setminus\Sstar)\times\Yfk.
				\end{cases}
			\]
		\end{enumerate}
		Let $A$ and $S$ be positive integer-valued polynomials, the bounds on the branch length and on the selected-set size. For every polynomial $p$ and every positive polynomial $\delta$, there is a polynomial $Q_{\ABS}$ such that for every sufficiently large $m$, with $n\cqq n_m$, every $a\in[A(n)]$, every $s\in[S(n+a)]$ with $s\le2^a$, every $r\in[Q_{\ABS}(n+a)..m]$, and every two-stage circuit $\advA=(\advA_1,\advA_2)$ of size at most $p(n+a)$,
		\[
			p_c\cqq\Prob\left[\Game_{\ABSTELF,\advA}^c(m,a,s,r)=1\right]\ (c\in\bits),
			\quad
			\abs{p_0-p_1}\le\frac{1}{\delta(n+a)}.
		\]
	\end{itemize}
\end{definition}

\end{relocated}%
\ifConferenceVersion\else

For $s=1$, the selected set is one branch $\bstar$, and we obtain the all-but-one notion of~\cite{EC:Zhandry19a}, with signatures that omit the unary input $1^1$ and take $\bstar$ in place of $\{\bstar\}$.

\begin{definition}[All-But-One Trapdoor \ELF{}]\label{def:abo-telf}\label{rem:abo-gives-telf}
	An all-but-one trapdoor \ELF{} is a tuple of algorithms $\ABOTELF=(\ABOGenInj,\allowbreak\ABOGenLos,\allowbreak\ABOEval,\allowbreak\ABOInv)$ such that the tuple defined by $\ABSEval\cqq\ABOEval$, $\ABSInv\cqq\ABOInv$, and
	\[
		\begin{aligned}
			\ABSGenInj(m,1^a,1^1)&\cqq\ABOGenInj(m,1^a),\\
			\ABSGenLos(m,1^a,1^1,r,\{\bstar\})&\cqq\ABOGenLos(m,1^a,r,\bstar)
		\end{aligned}
	\]
	satisfies \cref{def:abs-telf} with every quantifier over $s$ restricted to $s=1$ and with $S\equiv1$. We write $Q_{\ABO}$ for the polynomial $Q_{\ABS}$ of its selective-security clause, and $\Game_{\ABOTELF,\advA}^c(m,a,r)$ for the experiment $\Game_{\ABSTELF,\advA}^c(m,a,1,r)$ in which $\advA_1$ receives $(1^n,m,1^a,r)$ and returns $\bstar$ in place of $\{\bstar\}$.
\end{definition}

For $a=1$ and $\bstar=0$, we keep only branch $0$ and obtain a trapdoor \ELF{} in the sense of \cref{def:trapdoor-extremely-lossy-function}: except with negligible probability, under $\ABOGenInj$, $\ABOInv$ inverts branch $0$ on every input, and under $\ABOGenLos$, branch $0$ has image size at most $r$, and selective security for adversaries that make no oracle query implies mode indistinguishability.
\fi

\subsection{Branch Polynomials}
\label{subsec:abo-branch-arithmetic}
On every nonselected branch $b$, we invert by dividing by the value of a branch polynomial at $\beta(b)$, so we need these values to be units modulo every honestly sampled modulus. For $b=(b_1,\dots,b_a)\in\bits^a$ and $\tau\cqq2^a$, we write
\[
	\beta(b)\cqq\sum_{k=1}^a2^{k-1}b_k .
\]
Hence, $\beta$ is a bijection from $\bits^a$ to $[0..\tau-1]$, and $\tau\notin[0..\tau-1]$. In the branch constructions, every \ELTDH{} input is a single residue $x\in\ZN$, which we use only as an input of $\ELTDHHash$ and $\ELTDHEnc$, so its bit decomposition appears only inside \cref{fig:eltdh-dcr}. Whenever we pass an integer such as $\tau$, $\beta(b)$, or a coefficient of a branch polynomial to $\ELTDHGen$, we pass its residue class modulo the modulus of $\hk$.

For a set $\Bstar=\{\beta^\star_1,\dots,\beta^\star_s\}\subseteq[0..\tau-1]$ of size $s$, we define the integer polynomials
\[
	P_{\Bstar}(T)\cqq\prod_{j=1}^s(\beta^\star_j-T),
	\qquad
	\Ptau(T)\cqq(\tau-T)^s,
\]
both of degree $s$ with leading coefficient $(-1)^s$. On the branches where we divide by them, both polynomials take values that are nonzero products of $s$ integers of absolute value at most $2^a$, and at the group parameter $\groupsec(n+as)$ both prime factors of the modulus exceed $2^{as}$.

\begin{lemmaE}[Invertible Branch Multipliers]\label{lem:branch-multipliers}
	Let $n,a,s\in\NN_{>0}$ with $s\le2^a$, let $\tau\cqq2^a$ and $\secparam\cqq\groupsec(n+as)$, and let $(N,g)\in\Supp(\DJGen(1^\secparam))$. Then $\gcd(\Ptau(\beta),N)=1$ for every $\beta\in[0..\tau-1]$, and $\gcd(P_{\Bstar}(\beta),N)=1$ for every $\Bstar\subseteq[0..\tau-1]$ with $\abs{\Bstar}=s$ and every $\beta\in[0..\tau-1]\setminus\Bstar$.
\end{lemmaE}
\ConferenceProofSketch{proof:branch-multipliers}{
	The factors of the two products are $\tau-\beta\in[1..2^a]$ and $\beta^\star_j-\beta$, which is nonzero for $\beta\notin\Bstar$ and has absolute value at most $2^a-1$, so both products are nonzero integers of absolute value at most $2^{as}$. By \cref{def:dj-group}, both prime factors of $N$ exceed $2^{\ell(\secparam)-1}$ with $\ell(\secparam)\ge\secparam+2$, and $\secparam=\groupsec(n+as)\ge n+as$ by the choice of $\groupsec$ in \cref{sec:elf}, so neither prime divides either product. }{}{}
\begin{proofE}
	\label{proof:branch-multipliers}
	Every factor of either product is a nonzero integer of absolute value at most $2^a$: we have $1\le\tau-\beta\le\tau=2^a$ because $0\le\beta\le\tau-1$, and $1\le\abs{\beta^\star_j-\beta}\le\tau-1$ because $\beta,\beta^\star_j\in[0..\tau-1]$ and $\beta\notin\Bstar$. Hence, $\Ptau(\beta)$ and $P_{\Bstar}(\beta)$ are nonzero integers of absolute value at most $2^{as}$. We write $N=pq$ as in \cref{def:dj-group}, so $p,q>2^{\ell(\secparam)-1}$ with $\ell(\secparam)\ge\secparam+2$, and $\secparam=\groupsec(n+as)\ge n+as$ by the choice of $\groupsec$ in \cref{sec:elf}. Hence, $p,q>2^{\secparam+1}>2^{as}$, so neither prime divides either product, and both greatest common divisors equal one.
\end{proofE}

\ifConferenceVersion
\subsection{All-But-Some Trapdoor \texorpdfstring{\ELF{}s}{ELFs}}
\label{subsec:abs-construction}\label{subsec:abo-construction}
In an all-but-some trapdoor \ELF{}, every branch of a selected set $\Sstar$ is lossy, and the size $s=\abs{\Sstar}$ can grow polynomially in $n+a$.
\PlaceRelocated{abs-telf-definition}
\else
\subsection{All-But-Some Construction}
\label{subsec:abs-construction}
\fi

The polynomial of the construction is $P\cqq\Ptau$ in the injective mode and $P\cqq P_{\beta(\Sstar)}$ in the lossy mode, where $\beta(\Sstar)=\{\beta(b)\mid b\in\Sstar\}$ has size $s$ because $\beta$ is injective. We write its coefficients as
\[
	P(T)=(-1)^sT^s+\sum_{i=0}^{s-1}c_iT^i,
	\qquad
	c_0,\dots,c_{s-1}\in\ZZ.
\]
We generate two independent bundles $q\in\bits$ of $s$ encoding keys each, where key $(q,i)$ encodes the coefficient $c_i$, and on branch $b$ we combine the encodings of bundle $q$ with the weights $1,\beta(b),\dots,\beta(b)^{s-1}$. For a public description with modulus $N$, we take $\Yfk\cqq\ZNsqx\times(\ZN)^2$ with $\secparam\cqq\groupsec(n+as)$. Following the padding convention of \cref{sec:elf}, we encode the group coordinate in exactly $4\ell(\secparam)$ bits and each residue coordinate in exactly $2\ell(\secparam)$ bits, and we read the concatenated $8\ell(\secparam)$-bit string as its nonnegative integer value plus one. The resulting encoding is an injective map from $\Yfk$ to $[2^{8\ell(\secparam)}]$, and membership testing, encoding, and decoding take time $\poly(\secparam)$, which is polynomial in $n+a+s$ because $n+as\le(n+a+s)^2$. We define the construction in \cref{fig:abs-telf}.

\begin{proofbox}{All-but-some trapdoor extremely lossy function with selected-set size $s$. The \ELTDH{} algorithms are those of \cref{fig:eltdh-dcr}, and $c_0,\dots,c_{s-1}$ are the coefficients of $T^0,\dots,T^{s-1}$ in $P$.}{fig:abs-telf}
		\begin{minipage}[t]{0.52\linewidth}
			\vspace{0pt}
			{\footnotesize\underline{$\ABSGenInj(m,1^a,1^s)$}\par}
			\smallskip
			\begin{codelines}
				\item \pcassert $m\ge2\land a\ge1\land1\le s\le2^a$
				\item $n\gets\clog{m}$; $\secparam\gets\groupsec(n+as)$; $\tau\gets2^a$
				\item $P\gets\Ptau$
				\item $R\gets2^{2\ell(\secparam)}$; $\hk\sample\ELTDHSetupLos(1^\secparam,1^1,R)$
				\item \pcfor $(q,i)\in\bits\times[0..s-1]$
				\codelinesplus
				\item $(\ek_{q,i},\td_{q,i})\sample\ELTDHGen(\hk,(c_i))$
				\codelinesminus
				\item $\fk\gets(m,n,a,s,\secparam,\hk,(\ek_{q,i})_{q,i})$
				\item \pcreturn $(\fk,(P,(\td_{q,i})_{q,i}))$
			\end{codelines}

			\medskip
			{\footnotesize\underline{$\ABSGenLos(m,1^a,1^s,r,\Sstar)$}\par}
			\smallskip
			\begin{codelines}
				\item \pcassert $m\ge2\land a\ge1\land1\le s\le2^a\land r\in[m]\land\Sstar\subseteq\bits^a\land\abs{\Sstar}=s$
				\item \pcif $r=m$ \pcreturn $\ABSGenInj(m,1^a,1^s)$
				\item $n\gets\clog{m}$; $\secparam\gets\groupsec(n+as)$
				\item $P\gets P_{\beta(\Sstar)}$
				\item $\hk\sample\ELTDHSetupLos(1^\secparam,1^1,r)$
				\item \pcfor $(q,i)\in\bits\times[0..s-1]$
				\codelinesplus
				\item $(\ek_{q,i},\td_{q,i})\sample\ELTDHGen(\hk,(c_i))$
				\codelinesminus
				\item $\fk\gets(m,n,a,s,\secparam,\hk,(\ek_{q,i})_{q,i})$
				\item \pcreturn $(\fk,(P,(\td_{q,i})_{q,i}))$
			\end{codelines}
		\end{minipage}\hfill
		\begin{minipage}[t]{0.45\linewidth}
			\vspace{0pt}
			{\footnotesize\underline{$\ABSEval(\fk,b,u)$}\par}
			\smallskip
			\begin{codelines}
				\item \pcparse $\fk=(m,n,a,s,\secparam,\hk,(\ek_{q,i})_{q,i})$
				\item \pcparse $\hk=((N,g),1,w,\expl,g_0,(g_{1,k})_k)$
				\item \pcassert $b\in\bits^a\land u\in[m]$
				\item $x\gets u-1\in\ZN$; $\beta\gets\beta(b)$
				\item $h\gets\ELTDHHash(\hk,(x);0)$
				\item \pcfor $q\in\bits$
				\codelinesplus
				\item $E_q\gets(-1)^s\beta^sx$
				\item \pcfor $i\in[0..s-1]$
				\codelinesplus
				\item $e_{q,i}\gets\ELTDHEnc(\ek_{q,i},(x);0)$; $E_q\gets E_q+\beta^ie_{q,i}$
				\codelinesminus
				\codelinesminus
				\item \pcreturn $(h,E_0,E_1)$
			\end{codelines}

			\medskip
			{\footnotesize\underline{$\ABSInv(\fk,\td,b,y)$}\par}
			\smallskip
			\begin{codelines}
				\item \pcparse $\fk=(m,n,a,s,\secparam,\hk,(\ek_{q,i})_{q,i})$; $\td=(P,(\td_{q,i})_{q,i})$
				\item \pcparse $\hk=((N,g),1,w,\expl,g_0,(g_{1,k})_k)$
				\item \pcif $b\notin\bits^a\lor y\notin\Yfk$ \pcreturn $\bot$
				\item \pcparse $y=(h,E_0,E_1)$; $\beta\gets\beta(b)$; $d\gets P(\beta)\bmod N$
				\item \pcif $\gcd(\rep(d),N)\ne1$ \pcreturn $\bot$
				\item $D\gets0$
				\item \pcfor $i\in[0..s-1]$
				\codelinesplus
				\item $D\gets D+\beta^i\ELTDHDec(\td_{0,i},h)$
				\codelinesminus
				\item $x\gets\rep(d^{-1}(E_0-D))$
				\item \pcif $x\ge m$ \pcreturn $\bot$
				\item $u\gets x+1$
				\item \pcif $\ABSEval(\fk,b,u)\ne y$ \pcreturn $\bot$
				\item \pcreturn $u$
			\end{codelines}
		\end{minipage}
	\end{proofbox}

From the decoded values of each bundle, we reconstruct the value of its polynomial at the branch. We fix a public description in \cref{fig:abs-telf} whose hash key $\hk\in\Supp(\ELTDHSetupLos(1^\secparam,1^1,r'))$, for some $r'$, has modulus $N$, a bundle $q\in\bits$, and integers $c_0,\dots,c_{s-1}$ with $(\ek_{q,i},\td_{q,i})\in\Supp(\ELTDHGen(\hk,(c_i)))$ for every $i\in[0..s-1]$, and we let $P$ be the polynomial with these coefficients and leading term $(-1)^sT^s$. For a branch $b$ with $\beta\cqq\beta(b)$ and an input $u\in[m]$ with $x=u-1$, let $(h,E_0,E_1)=\ABSEval(\fk,b,u)$. By \cref{cor:eltdh-correctness}, $e_{q,i}-\ELTDHDec(\td_{q,i},h)=c_ix$ for every $i$, and we weight these $s$ identities by $1,\beta,\dots,\beta^{s-1}$, add $(-1)^s\beta^sx$, and obtain
\begin{equation}\label{eq:abs-bundle-identity}
	E_q-\sum_{i=0}^{s-1}\beta^i\ELTDHDec(\td_{q,i},h)=P(\beta)x .
\end{equation}
Hence, if $P(\beta)$ is a unit modulo $N$, then we recover $x$ from bundle $q$, and if $P(\beta)=0$, then $E_q$ is a function of the digest $h$.

\ifConferenceVersion\else
In the security proof, we answer inversion queries through whichever bundle encodes a polynomial that is a unit at the queried branch, so we need the answer to be the same for every such bundle. For $q\in\bits$, we write $\ABSInvBundle{q}$ for the algorithm $\ABSInv$ of \cref{fig:abs-telf} in which we compute $D$ from the trapdoors $\td_{q,i}$ and use $E_q$ in place of $E_0$, so that $\ABSInvBundle{0}=\ABSInv$, and for a polynomial $P$, we write $\td^P\cqq(P,(\td_{q,i})_{q,i})$.

\begin{lemmaE}[Inversion Through One Bundle]\label{lem:abs-bundle-inversion}\label{lem:abo-inversion-public}
	Let $\fk$, $N$, $q$, $c_0,\dots,c_{s-1}$, and $P$ be defined as in \cref{eq:abs-bundle-identity}, let $b\in\bits^a$, and let $y$ be any string. If $\gcd(P(\beta(b)),N)\ne1$, then $\ABSInvBundle{q}(\fk,\td^P,b,y)=\bot$. If $\gcd(P(\beta(b)),N)=1$, then
	\begin{equation}\label{eq:abs-bundle-inversion}
		\ABSInvBundle{q}(\fk,\td^P,b,y)=
		\begin{cases}
			u, & u\in[m]\land\ABSEval(\fk,b,u)=y,\\
			\bot, & y\notin\Image(f_{\fk,b}),
		\end{cases}
	\end{equation}
	and the preimage $u$ in the first case is unique.
\end{lemmaE}
\begin{proofE}
	We decode bundle $q$ and divide by $P(\beta(b))$, and after the reevaluation test we return only a genuine preimage. The first claim is the test on $\gcd(\rep(d),N)$ in \cref{fig:abs-telf}. For the second, if $\ABSInvBundle{q}$ returns some $u\ne\bot$, then $\ABSEval(\fk,b,u)=y$ by its last test. Conversely, let $u\in[m]$ satisfy $\ABSEval(\fk,b,u)=y=(h,E_0,E_1)$, and let $x=u-1$. Every supported modulus, including the fallback modulus, satisfies $N>2^{2\ell(\secparam)-2}\ge2^\secparam\ge2^n\ge m$ by \cref{def:fallback-modulus}, so $x\in[0..m-1]$ is its own representative. By \cref{eq:abs-bundle-identity}, $E_q-D=P(\beta(b))x$, so $\ABSInvBundle{q}$ computes $\rep(d^{-1}(E_q-D))=u-1<m$, passes the reevaluation test, and returns $u$. Any two preimages $u,u'$ satisfy $P(\beta(b))(u-1)=P(\beta(b))(u'-1)$ in $\ZN$ by the same identity, and dividing by the unit $P(\beta(b))$, we obtain $u=u'$.
\end{proofE}
\fi

Functionality follows from \cref{eq:abs-bundle-identity} and \cref{lem:branch-multipliers}, and selective security follows from the function privacy and the mode indistinguishability of \ELTDH{}.

\begin{theoremE}[Selective All-But-Some Trapdoor \ELF{}]\label{thm:abs-telf}
	Suppose that $\DJGen$ satisfies $\DCR$ against nonuniform $\PPT$ adversaries and tunable $\DSEDL$. Then the construction in \cref{fig:abs-telf} is a selectively secure all-but-some trapdoor \ELF{} in the sense of \cref{def:abs-telf}. For all $m\ge2$ and $a,s\ge1$ with $s\le2^a$, with $n\cqq\clog{m}$ and $w\cqq2\ell(\groupsec(n+as))$, and excluding the modulus, the generator, and integer metadata, its public description contains $(2s+1)(w+1)$ group elements, its trapdoor contains $2s$ \ELTDH{} exponents, and its output contains one group element and two residues.
\end{theoremE}
\ConferenceProofSketch{proof:abs-telf}{
	By \cref{cor:eltdh-correctness}, \cref{eq:abs-bundle-identity} holds for every supported key, so by \cref{lem:branch-multipliers}, we invert every branch outside $\Sstar$, and for $b\in\Sstar$ we have $P_{\beta(\Sstar)}(\beta(b))=0$, so the output on $b$ is a function of its digest, which takes at most $r$ values by \cref{thm:eltdh-lossy-image}. For selective security, we change the $2s$ encoding keys one at a time from the coefficients of $\Ptau$ to those of $P_{\beta(\Sstar)}$ by \cref{thm:eltdh-function-privacy}, first in bundle $0$ and then in bundle $1$, and during each change we answer every inversion query with the other bundle, whose polynomial is a unit at every queried branch. We lose the negligible probability of the full-setup fallback event $\Efull$ once, when we switch the oracle between the two bundles, and we finally change the hash-key parameter from the full one to $r$ by \cref{thm:eltdh-mode-indistinguishability}, with a reduction that generates every encoding key and trapdoor itself. }{}{}
\begin{proofE}
	\label{proof:abs-telf}
	First, we derive functionality from \cref{eq:abs-bundle-identity} according to whether $P(\beta)$ is a unit or zero, and we check compactness for growing $s$. After that, we prove selective security with $2s$ function-privacy transitions, one change of the bundle that answers the oracle, and one hash-key mode transition.
	\ifConferenceVersion
	We begin with inversion through one bundle. For $q\in\bits$, we write $\ABSInvBundle{q}$ for the algorithm $\ABSInv$ of \cref{fig:abs-telf} in which we compute $D$ from the trapdoors $\td_{q,i}$ and use $E_q$ in place of $E_0$, so that $\ABSInvBundle{0}=\ABSInv$, and for a polynomial $P$, we write $\td^P\cqq(P,(\td_{q,i})_{q,i})$. Let $\fk$, $N$, $q$, $c_0,\dots,c_{s-1}$, and $P$ be as in \cref{eq:abs-bundle-identity}, and let $b\in\bits^a$ and $y$ be a string. If $\gcd(P(\beta(b)),N)\ne1$, then $\ABSInvBundle{q}$ returns $\bot$ on $(b,y)$ by its test on $\gcd(\rep(d),N)$. Otherwise, a returned $u\ne\bot$ satisfies $\ABSEval(\fk,b,u)=y$ by the last test. Conversely, let $u\in[m]$ satisfy $\ABSEval(\fk,b,u)=y=(h,E_0,E_1)$, and let $x=u-1$. Every supported modulus, including the fallback modulus, satisfies $N>2^{2\ell(\secparam)-2}\ge2^\secparam\ge2^n\ge m$ by \cref{def:fallback-modulus}, so $x\in[0..m-1]$ is its own representative. By \cref{eq:abs-bundle-identity}, $E_q-D=P(\beta(b))x$, so $\ABSInvBundle{q}$ computes $\rep(d^{-1}(E_q-D))=u-1$ and returns $u$, and any two preimages $u,u'$ satisfy $P(\beta(b))(u-1)=P(\beta(b))(u'-1)$ in $\ZN$, hence $u=u'$. Therefore, if $\gcd(P(\beta(b)),N)=1$, then
	\begin{equation}\label{eq:abs-bundle-inversion}
		\ABSInvBundle{q}(\fk,\td^P,b,y)=
		\begin{cases}
			u, & u\in[m]\land\ABSEval(\fk,b,u)=y,\\
			\bot, & y\notin\Image(f_{\fk,b}),
		\end{cases}
	\end{equation}
	with a unique preimage $u$ in the first case.

	\fi
	For the functionality, we fix $m\ge2$ and $a,s\ge1$ with $s\le2^a$, and we write $n\cqq\clog{m}$ and $\secparam\cqq\groupsec(n+as)$. In the injective mode, $P=\Ptau$, and setup samples $\hk$ at the full parameter $R=2^{2\ell(\secparam)}=\Rfull(\secparam,1)$. Outside the fallback event $\Efull$, the hash key has a modulus $N$ with $(N,g)\in\Supp(\DJGen(1^\secparam))$, so $\gcd(\Ptau(\beta(b)),N)=1$ for every $b\in\bits^a$ by \cref{lem:branch-multipliers}, and by \cref{eq:abs-bundle-inversion} with $q=0$, we have $\ABSInv(\fk,\td,b,f_{\fk,b}(u))=u$ for every $(b,u)\in\bits^a\times[m]$. If $r=m$, then $\ABSGenLos$ returns the output of $\ABSGenInj$, so the inversion guarantee is the injective one, and every branch has image size at most $m=r$. If $r<m$, then $P=P_{\beta(\Sstar)}$, and for $b\in\Sstar$, we have $P(\beta(b))=0$, so by \cref{eq:abs-bundle-identity}, both residue coordinates of $f_{\fk,b}(u)$ are functions of the digest $h=\ELTDHHash(\hk,(u-1);0)$, which is the first coordinate. Hence, $\abs{\Image(f_{\fk,b})}\le\abs{\Image(\Phihk)}\le r$ by \cref{thm:eltdh-lossy-image}, for every hash key in the support, the fallback key included. Outside the event that setup returns the fallback key, $\gcd(P_{\beta(\Sstar)}(\beta(b)),N)=1$ for every $b\in\bits^a\setminus\Sstar$ by \cref{lem:branch-multipliers} with $\Bstar=\beta(\Sstar)$, so by \cref{eq:abs-bundle-inversion} with $q=0$, we invert every input on these branches.

	In both modes, inversion can fail only if setup returns the fallback key. At the full parameter this is the event $\Efull$, of probability $\ffull(\secparam,1)$, and at a nonfull parameter it is the failure of $\DJGen(1^\secparam)$, on which the full setup also returns the fallback key, so its probability is at most $\ffull(\secparam,1)$. Since $n+a+s-1\le n+as\le\secparam$, both failure probabilities are at most
	\[
		\nuinj(k)=\nulos(k)\cqq\sup_{\secparam'\ge k-1}\ffull(\secparam',1)
	\]
	at $k=n+a+s$, and this function is negligible because $\ffull(\cdot,1)$ is negligible by the analysis in \cref{subsec:eltdh-construction} and the supremum over the tail of a negligible function is negligible.

	\subparhead{Compactness} Every root of $\Ptau$ or $P_{\beta(\Sstar)}$ is an integer of absolute value at most $2^a$, and each coefficient $c_i$ is, up to sign, a sum of $\binom{s}{i}\le2^s$ products of $s-i$ roots. Hence, $\abs{c_i}\le2^{s(a+1)}$, and every coefficient has $\bigO(s(a+1))$ bits. We compute all coefficients by multiplying the $s$ linear factors one at a time, with $\bigO(s^2)$ integer operations on integers of $\bigO(s(a+1))$ bits, so both generators run in polynomial time. Since $n+as\le(n+a+s)^2$, the parameters $\secparam=\groupsec(n+as)$ and $w=2\ell(\secparam)$ are polynomially bounded in $n+a+s$. By \cref{lem:eltdh-efficiency}, the hash key and each of the $2s$ encoding keys contain $w+1$ group elements and each trapdoor one exponent, and evaluation and inversion make $\bigO(s)$ calls to \ELTDH{} algorithms. Hence, the compactness clause of \cref{def:abs-telf} holds, with the element counts of the statement.

	\subparhead{Hybrids} We fix polynomials $A$, $S$, and $p$ and a positive polynomial $\delta$ as in \cref{def:abs-telf}. Since the conclusion for nondecreasing polynomial upper bounds of $p$, $\delta$, and $S$ implies the conclusion for $p$, $\delta$, and $S$, we assume that all three are nondecreasing. We fix $m$, with $n\cqq\clog{m}$, and $a\in[A(n)]$, $s\in[S(n+a)]$ with $s\le2^a$, $r\in[Q_{\ABS}(n+a)..m]$, where $Q_{\ABS}$ is a polynomial that depends only on $p$, $\delta$, and $S$, and an adversary $\advA=(\advA_1,\advA_2)$ of size at most $p(n+a)$, and we write $z\cqq n+a$ and $\secparam\cqq\groupsec(n+as)$. If $r=m$, then $\ABSGenLos$ returns the output of $\ABSGenInj$ and the two experiments are identical, so we suppose henceforth that $r<m$. Since $\advA_1$ returns $\Sstar$ before key generation, the coefficient lists
	\[
		\Ptau(T)=(-1)^sT^s+\sum_{i=0}^{s-1}c^\tau_iT^i,
		\qquad
		P_{\beta(\Sstar)}(T)=(-1)^sT^s+\sum_{i=0}^{s-1}c^\star_iT^i
	\]
	are functions of the output of $\advA_1$.

	We define the hybrids $\Hybrid_0,\dots,\Hybrid_{2s+1}$ and $\Hybrid_s'$. As in the experiments, each of them samples $(\Sstar,\var{st})\sample\advA_1(1^n,m,1^a,1^s,r)$ and returns $0$ unless $\Sstar$ is a valid selected set. Otherwise, it generates $\fk$ as in \cref{fig:abs-security-hybrids}, where $j$ ranges over $[s]$ and $\Hybrid_s'$ generates $\fk$ as $\Hybrid_s$ does, and it returns $\advA_2^{\calO}(\var{st},\fk)$. The oracle $\calO$ returns $\bot$ on every query $(b,y)$ with $b\in\Sstar$. On every other query, it returns $\ABSInvBundle{1}(\fk,\td^{\Ptau},b,y)$ in $\Hybrid_0,\dots,\Hybrid_s$, and it returns $\ABSInvBundle{0}(\fk,\td^{P_{\beta(\Sstar)}},b,y)$ in $\Hybrid_s'$ and in $\Hybrid_{s+1},\dots,\Hybrid_{2s+1}$. In $\Hybrid_0$, both bundles encode $\Ptau$, and $\ABSInvBundle{0}$ and $\ABSInvBundle{1}$ return the same answer for a common polynomial on every hash key, by the common test on $\gcd(\rep(d),N)$ and by \cref{eq:abs-bundle-inversion}. Hence, $\Hybrid_0$ is the experiment $\Game_{\ABSTELF,\advA}^0(m,a,s,r)$, and $\Hybrid_{2s+1}$, in which both bundles encode $P_{\beta(\Sstar)}$ under a hash key at parameter $r$ and the oracle is $\ABSInv$, is the experiment $\Game_{\ABSTELF,\advA}^1(m,a,s,r)$.

	\begin{proofbox}{Key generation in the hybrids for selective all-but-some security, for $j\in[s]$. The \ELTDH{} algorithms are those of \cref{fig:eltdh-dcr}, and the highlighted line is the line that changes from the preceding hybrid. The hybrid $\Hybrid_s'$ generates its keys as $\Hybrid_s$ does.}{fig:abs-security-hybrids}
	\begin{minipage}[t]{0.49\linewidth}
	\vspace{0pt}
	{\footnotesize\underline{$\Hybrid_0$}\codelinesrightalign{injective key}\par}
	\smallskip
	\begin{codelines}
	\item $n\gets\clog{m}$; $\secparam\gets\groupsec(n+as)$; $R\gets2^{2\ell(\secparam)}$
	\item $\hk\sample\allowbreak\ELTDHSetupLos(1^\secparam,1^1,R)$
	\item \pcfor $(q,i)\in\bits\times[0..s-1]$
	\codelinesplus
	\item $(\ek_{q,i},\td_{q,i})\sample\allowbreak\ELTDHGen(\hk,(c^\tau_i))$
	\codelinesminus
	\item \pcreturn $\fk\gets(m,n,a,s,\secparam,\hk,(\ek_{q,i})_{q,i})$
	\end{codelines}
	\end{minipage}\hfill%
	\begin{minipage}[t]{0.49\linewidth}
	\vspace{0pt}
	{\footnotesize\underline{$\Hybrid_j$}\codelinesrightalign{bundle $0$, key $j-1$}\par}
	\smallskip
	\begin{codelines}
	\item $n\gets\clog{m}$; $\secparam\gets\groupsec(n+as)$; $R\gets2^{2\ell(\secparam)}$
	\item $\hk\sample\allowbreak\ELTDHSetupLos(1^\secparam,1^1,R)$
	\item \pcfor $i\in[0..j-2]$
	\codelinesplus
	\item $(\ek_{0,i},\td_{0,i})\sample\allowbreak\ELTDHGen(\hk,(c^\star_i))$
	\codelinesminus
	\item \label{line:abs-hyb-bundle-zero}\hybchange{$(\ek_{0,j-1},\td_{0,j-1})\sample\allowbreak\ELTDHGen(\hk,(c^\star_{j-1}))$}
	\item \pcfor $i\in[j..s-1]$
	\codelinesplus
	\item $(\ek_{0,i},\td_{0,i})\sample\allowbreak\ELTDHGen(\hk,(c^\tau_i))$
	\codelinesminus
	\item \pcfor $i\in[0..s-1]$
	\codelinesplus
	\item $(\ek_{1,i},\td_{1,i})\sample\allowbreak\ELTDHGen(\hk,(c^\tau_i))$
	\codelinesminus
	\item \pcreturn $\fk\gets(m,n,a,s,\secparam,\hk,(\ek_{q,i})_{q,i})$
	\end{codelines}
	\end{minipage}

	\medskip

	\begin{minipage}[t]{0.49\linewidth}
	\vspace{0pt}
	{\footnotesize\underline{$\Hybrid_{s+j}$}\codelinesrightalign{bundle $1$, key $j-1$}\par}
	\smallskip
	\begin{codelines}
	\item $n\gets\clog{m}$; $\secparam\gets\groupsec(n+as)$; $R\gets2^{2\ell(\secparam)}$
	\item $\hk\sample\allowbreak\ELTDHSetupLos(1^\secparam,1^1,R)$
	\item \pcfor $i\in[0..s-1]$
	\codelinesplus
	\item $(\ek_{0,i},\td_{0,i})\sample\allowbreak\ELTDHGen(\hk,(c^\star_i))$
	\codelinesminus
	\item \pcfor $i\in[0..j-2]$
	\codelinesplus
	\item $(\ek_{1,i},\td_{1,i})\sample\allowbreak\ELTDHGen(\hk,(c^\star_i))$
	\codelinesminus
	\item \label{line:abs-hyb-bundle-one}\hybchange{$(\ek_{1,j-1},\td_{1,j-1})\sample\allowbreak\ELTDHGen(\hk,(c^\star_{j-1}))$}
	\item \pcfor $i\in[j..s-1]$
	\codelinesplus
	\item $(\ek_{1,i},\td_{1,i})\sample\allowbreak\ELTDHGen(\hk,(c^\tau_i))$
	\codelinesminus
	\item \pcreturn $\fk\gets(m,n,a,s,\secparam,\hk,(\ek_{q,i})_{q,i})$
	\end{codelines}
	\end{minipage}\hfill%
	\begin{minipage}[t]{0.49\linewidth}
	\vspace{0pt}
	{\footnotesize\underline{$\Hybrid_{2s+1}$}\codelinesrightalign{lossy key}\par}
	\smallskip
	\begin{codelines}
	\item $n\gets\clog{m}$; $\secparam\gets\groupsec(n+as)$
	\item \label{line:abs-hyb-bases}\hybchange{$\hk\sample\allowbreak\ELTDHSetupLos(1^\secparam,1^1,r)$}
	\item \pcfor $(q,i)\in\bits\times[0..s-1]$
	\codelinesplus
	\item $(\ek_{q,i},\td_{q,i})\sample\allowbreak\ELTDHGen(\hk,(c^\star_i))$
	\codelinesminus
	\item \pcreturn $\fk\gets(m,n,a,s,\secparam,\hk,(\ek_{q,i})_{q,i})$
	\end{codelines}
	\end{minipage}
	\end{proofbox}

	\subparhead{Key transitions} For $j\in[s]$, we change key $(0,j-1)$ at \cref{line:abs-hyb-bundle-zero} from $c^\tau_{j-1}$ to $c^\star_{j-1}$ in the transition from $\Hybrid_{j-1}$ to $\Hybrid_j$, and we bound this transition by a function-privacy reduction at residue arity one. Its selector receives $(1^\secparam,\hk)$, runs $(\Sstar,\var{st})\sample\advA_1(1^n,m,1^a,1^s,r)$, and, if $\Sstar$ is invalid, returns two equal predicates and a state on which the distinguisher returns $0$. Otherwise, the selector computes $(c^\tau_i)_i$ and $(c^\star_i)_i$, samples the keys $(0,i)$ with $i\ne j-1$ and all keys of bundle $1$ as in $\Hybrid_j$, with their trapdoors, and returns the predicates with coefficient vectors $(c^\tau_{j-1})$ and $(c^\star_{j-1})$ and the state $\var{st}'$ consisting of $\Sstar$, $\var{st}$, and these keys and trapdoors. The distinguisher receives $(\hk,\var{st}',\ek)$ for the state $\var{st}'$ of the selector, sets $\ek_{0,j-1}\cqq\ek$, assembles $\fk$, runs $\advA_2(\var{st},\fk)$, answers each permitted query with $\ABSInvBundle{1}$ and the trapdoors of bundle $1$, which encodes $\Ptau$, and returns the output bit of $\advA_2$. For every hash key, the fallback key included, the output of the distinguisher is distributed as the output of $\Hybrid_{j-1}$ for the first predicate and as the output of $\Hybrid_j$ for the second. With two bundles, we keep a complete bundle and its trapdoors while the other bundle is under the function-privacy challenge, whose trapdoor the reduction does not know. For $j\in[s]$, we change key $(1,j-1)$ at \cref{line:abs-hyb-bundle-one} in the transition from $\Hybrid_{s+j-1}$ to $\Hybrid_{s+j}$, with $\Hybrid_{s+0}\cqq\Hybrid_s'$, and we use the same reduction with the two bundles exchanged, where the distinguisher answers each permitted query with $\ABSInvBundle{0}$ and the trapdoors of bundle $0$, which encodes $P_{\beta(\Sstar)}$.

	From $\Hybrid_s$ to $\Hybrid_s'$, we change only the bundle that answers the oracle. Outside $\Efull$, the hash key of $\Hybrid_s$ has a modulus $N$ with $(N,g)\in\Supp(\DJGen(1^\secparam))$, so by \cref{lem:branch-multipliers}, $\Ptau(\beta(b))$ and $P_{\beta(\Sstar)}(\beta(b))$ are units modulo $N$ for every $b\in\bits^a\setminus\Sstar$, and by \cref{eq:abs-bundle-inversion}, both oracles return the unique preimage or $\bot$. Hence, the two hybrids differ only on $\Efull$, and
	\[
		\abs{\Prob\left[\Hybrid_s=1\right]-\Prob\left[\Hybrid_s'=1\right]}\le\Prob\left[\Efull\right]=\ffull(\secparam,1).
	\]

	\subparhead{Mode transition} At \cref{line:abs-hyb-bases}, both bundles encode $P_{\beta(\Sstar)}$, so our reduction receives a hash key $\hk$ from the distribution $\Dmode_{\secparam,r'}$ of \cref{def:eltdh} at residue arity one, runs $\advA_1$, samples every encoding key with its trapdoor for the coefficients $c^\star_i$, runs $\advA_2(\var{st},\fk)$, answers each permitted query with $\ABSInv$, and returns the output bit of $\advA_2$. For every hash key in either distribution, its output is distributed as the output of $\Hybrid_{2s}$ for $r'=\Rfull(\secparam,1)$ and as the output of $\Hybrid_{2s+1}$ for $r'=r$.

	\subparhead{Thresholds} Since $a,s\ge1$ and $\groupsec(v)\ge v$, we have $z\le n+as\le\secparam$ and $s\le\secparam$. Each reduction of this proof runs $\advA$, of size at most $p(z)\le p(\secparam)$, generates at most $2s$ further \ELTDH{} keys, and answers each query with $\bigO(s)$ \ELTDH{} decodings and one evaluation, so there is one polynomial $\phat$ such that every reduction, for every admissible choice of $m$, $a$, $s$, $r$, $\advA$, and the changed key, has size at most $\phat(\secparam)$. We write $\dhat(v)\cqq2(2v+1)\delta(v)$. There is $\secparam_{\var{fp}}$ such that for every $\secparam\ge\secparam_{\var{fp}}$ and every admissible choice with $\groupsec(n+as)=\secparam$, the function-privacy reduction for this choice has advantage at most $1/\dhat(\secparam)$. Otherwise, for infinitely many $\secparam$, we fix an admissible choice whose reduction has advantage greater than $1/\dhat(\secparam)$, and we let $\advS'$ and $\advD'$ be the nonuniform families that equal the selector and the distinguisher of this reduction at these $\secparam$ and that return two equal predicates and $0$ at every other $\secparam$. Since $\advS'$ is a nonuniform selector of size at most $\phat(\secparam)$, its two ensembles are computationally indistinguishable by \cref{thm:eltdh-function-privacy}, and by \cref{lem:uniform-quantitative-threshold} with size bound $\phat$ and target $\dhat$, the advantage of $\advD'$ is at most $1/\dhat(\secparam)$ for every sufficiently large $\secparam$, which contradicts the advantages of the fixed reductions.

	We apply \cref{thm:eltdh-mode-indistinguishability} at residue arity one with size bound $\phat$ and target $\dhat$, and we obtain a polynomial $\Qhash$ and $\secparam_{\var{mode}}$ such that $\Dmode_{\secparam,r_0}\indist_{\phat,1/\dhat}\Dmode_{\secparam,r_1}$ for every $\secparam\ge\secparam_{\var{mode}}$ and all $r_0,r_1\in[\Qhash(\secparam)..\Rfull(\secparam,1)]$. Since $r\le m\le2^n<2^\secparam<\Rfull(\secparam,1)$, the parameter $r$ lies in this interval once $r\ge\Qhash(\secparam)$. We take a polynomial upper bound $G$ on the increasing function $\groupsec$ and a nondecreasing polynomial upper bound $\Qhashbar$ on $\Qhash$. From $n,a\le z$ and $s\le S(z)$, we obtain $n+as\le z+zS(z)$ and $\secparam\le G(z+zS(z))$, and we choose a polynomial $Q_{\ABS}$ with $Q_{\ABS}(z)\ge\Qhashbar(G(z+zS(z)))$. Then $r\ge Q_{\ABS}(z)\ge\Qhash(\secparam)$ for every admissible $a$ and $s$. Finally, since $\ffull(\cdot,1)$ is negligible, there is $\secparam_{\var{full}}$ with $\ffull(\secparam,1)\le1/(2\delta(\secparam))$ for every $\secparam\ge\secparam_{\var{full}}$.

	\subparhead{Total} Since $\secparam\ge n=\clog{m}$, for every sufficiently large $m$, we have $\secparam\ge\max\{\secparam_{\var{fp}},\secparam_{\var{mode}},\secparam_{\var{full}}\}$. For such $m$, we sum the $2s$ key transitions, the oracle change, and the mode transition, and we use $s,z\le\secparam$ and that $\delta$ is nondecreasing:
	\[
		\begin{aligned}
			\abs{\Prob\left[\Hybrid_0=1\right]-\Prob\left[\Hybrid_{2s+1}=1\right]}
			&\le\frac{2s+1}{\dhat(\secparam)}+\ffull(\secparam,1)\\
			&\le\frac{2\secparam+1}{2(2\secparam+1)\delta(\secparam)}+\frac{1}{2\delta(\secparam)}
			\le\frac{1}{\delta(z)}.
		\end{aligned}
	\]
	Since $\Hybrid_0$ and $\Hybrid_{2s+1}$ are the two selective-security experiments, the selective-security clause of \cref{def:abs-telf} holds.
\end{proofE}

\ifConferenceVersion\else
\subsection{Enumerating Selected-Branch Images}
\label{subsec:abs-enumeration}
\fi
For $r<m$ and every $\bstar\in\Sstar$, the integer $\beta(\bstar)$ is a root of $P$, so by \cref{eq:abs-bundle-identity}, both residue coordinates on branch $\bstar$ are functions of the digest, and with the digest enumeration of \cref{thm:telf-strong-enumeration}, we obtain the exact selected-branch image, with an additional polynomial dependence on $s$.

\begin{corollaryE}[Strong Enumeration on Selected Branches]\label{cor:abs-strong-enumeration}
	There is a deterministic algorithm $\calE$ such that, for all $m\ge2$, $a,s\ge1$ with $s\le2^a$, $r\in[m]$, $\Sstar\subseteq\bits^a$ with $\abs{\Sstar}=s$, $(\fk,\td)\in\Supp(\ABSGenLos(m,1^a,1^s,r,\Sstar))$, and $\bstar\in\Sstar$, the output of $\calE(\fk,\bstar,r)$ is $\Image(\ABSEval(\fk,\bstar,\cdot))$ and its running time is $\poly(r,\log m+a+s)$.
\end{corollaryE}
\begin{proofE}
	We list one representative input for every digest and evaluate the branch on these representatives. If $r=m$, we evaluate $\ABSEval(\fk,\bstar,u)$ for all $u\in[m]$ and deduplicate the outputs, in time $m\cdot\poly(\log m+a+s)=\poly(r,\log m+a+s)$. We suppose henceforth that $r<m$. Then $P=P_{\beta(\Sstar)}$ and $P(\beta(\bstar))=0$, so by \cref{eq:abs-bundle-identity}, $\ABSEval(\fk,\bstar,u)$ is a function of the digest $\ELTDHHash(\hk,(u-1);0)$.

	We enumerate the distinct digests of the inputs $u-1$ for $u\in[m]$ with the prefix dynamic program of \cref{fig:telf-image-enumeration}, run on the hash bases of $\hk$, and we attach to each deduplicated state one processed prefix that produces it, keeping either prefix when two products coincide. By the bounded subset-sum argument in the proof of \cref{thm:telf-strong-enumeration}, at most $r$ digests occur at every position, so the state bound and the running time are unchanged, and we know an input for every final digest. We evaluate $\ABSEval(\fk,\bstar,\cdot)$ on these representatives and deduplicate. Every output on $\bstar$ arises from some digest, and equal digests have equal outputs, so the resulting list is the exact image, computed in time $\poly(r,\secparam)\cdot\bigO(s)=\poly(r,\log m+a+s)$ because $\secparam=\groupsec(n+as)$ and $n+as\le(n+a+s)^2$.
\end{proofE}

With the exact image list of \cref{cor:abs-strong-enumeration}, we can sample uniformly from every selected-branch image by choosing a uniform position in the deduplicated list, as for the list of \cref{thm:telf-strong-enumeration}, in time $\poly(r,\log m+a+s)$, which is $\poly(r,\log m+a)$ for $s=1$.

\ifConferenceVersion\else
\subsection{The All-But-One Case}
\label{subsec:abo-construction}\label{subsec:abo-security}
\fi
For $s=1$ and $\Sstar=\{\bstar\}$, we obtain the all-but-one trapdoor \ELF{} of~\cite{EC:Zhandry19a}, whose inversion oracle is available on every branch other than $\bstar$. Each bundle then consists of one encoding key, and in the proof of \cref{thm:abs-telf}, we answer inversion queries with the trapdoor of the second key while we change the first key under function privacy, and then we exchange the two keys.

\begin{corollaryE}[Selective All-But-One Trapdoor \ELF{}]\label{cor:abo-telf}\label{thm:abo-telf}\label{thm:abo-strong-enumeration}\label{thm:abo-functionality}
	Suppose that $\DJGen$ satisfies $\DCR$ against nonuniform $\PPT$ adversaries and tunable $\DSEDL$, and let $\ABOEval\cqq\ABSEval$, $\ABOInv\cqq\ABSInv$, and
	\[
		\begin{aligned}
			\ABOGenInj(m,1^a)&\cqq\ABSGenInj(m,1^a,1^1),\\
			\ABOGenLos(m,1^a,r,\bstar)&\cqq\ABSGenLos(m,1^a,1^1,r,\{\bstar\})
		\end{aligned}
	\]
	for the algorithms of \cref{fig:abs-telf}. Then $\ABOTELF$ is a selectively secure all-but-one trapdoor \ELF{}\ifConferenceVersion\else{} in the sense of \cref{def:abo-telf}\fi. For all $m\ge2$ and $a\ge1$, with $n\cqq\clog{m}$ and $w\cqq2\ell(\groupsec(n+a))$, and excluding the modulus, the generator, and integer metadata, its public description contains $3(w+1)$ group elements, its trapdoor contains two \ELTDH{} exponents, and its output contains one group element and two residues. There is a deterministic algorithm $\calE$ such that, for all $m\ge2$, $a\ge1$, $r\in[m]$, $\bstar\in\bits^a$, and $(\fk,\td)\in\Supp(\ABOGenLos(m,1^a,r,\bstar))$, the output of $\calE(\fk,\bstar,r)$ is $\Image(\ABOEval(\fk,\bstar,\cdot))$ and its running time is $\poly(r,\log m+a)$.
\end{corollaryE}
\begin{proofE}
	For $s=1$, the polynomials in \cref{fig:abs-telf} are $\Ptau(T)=\tau-T$ and $P_{\{\betastar\}}(T)=\betastar-T$ with $\betastar\cqq\beta(\bstar)$, so each bundle consists of one key encoding $\tau$ or $\betastar$, and $E_q=e_{q,0}-\beta(b)x$. The selected set $\{\bstar\}$ determines its unique branch, so the tuple \ifConferenceVersion obtained from $\ABOTELF$ by this identification \else that \cref{def:abo-telf} derives from $\ABOTELF$ \fi is the construction of \cref{fig:abs-telf} at $s=1$. Hence, the security, functionality, and size claims are \cref{thm:abs-telf} at $s=1$ and $S\equiv1$, with $Q_{\ABO}\cqq Q_{\ABS}$, and the enumeration claim is \cref{cor:abs-strong-enumeration} at $s=1$, since $\poly(r,\log m+a+1)=\poly(r,\log m+a)$.
\end{proofE}

\begin{aidisclosure}
	Commercial AI tools, including ChatGPT and Claude, were used to assist with grammar, proofreading, expository polishing, and \LaTeX{} formatting. These tools were not used in developing the scientific content, results, proofs, analyses, or conclusions, all of which were produced by the authors.
\end{aidisclosure}
\ifAnonymous\else
	\begin{acknowledgements}
        A.\ Datta and A.\ Dutta thank Nilabha Saha for discussions on the presentation of this work. A.\ Datta and D.\ Abram were supported by Input Output (iog.io) through their funding of the Edinburgh Blockchain Technology Lab. A.\ Dutta was supported by Cyberagentur under the project ``Multi-Party Computation in the Confidential Cloud (MPCC)''. L.\ Roy was supported by the European Research Council (ERC) under the European Union's Horizon 2020 research and innovation programme under grant agreement number 101124977 (DECRYPSIS), and by the Swiss National Science Foundation (SNSF) under the Starting grant ``Low-Entropy Cryptography for the Quantum Era''.
	\end{acknowledgements}
\fi

\ifConferenceVersion\else

	\linespread{1}
\fi
\PrintPaperBibliography

\ifConferenceVersion
	\BeginPaperAppendices

	\section{Additional Preliminaries}\label{app:supplementary-preliminaries}

\PlaceRelocated{notation}
\PlaceRelocated{elf-definition}
\PlaceRelocated{enumerability}
\PlaceRelocated{enumerability-telf}
\PlaceRelocated{group-samplers}
\PlaceRelocated{dj-groups}

\subsection{Offset Removal}\label{app:offset-removal}

For \cref{thm:eltdh-function-privacy}, we remove the predicate offsets from encoding keys whose bases are uniform in $\GN$, whose exponent is uniform in $[0..2^{\expl}-1]$, and whose offsets are chosen after the modulus and the bases. The function-privacy arguments of~\cite{C:DGIMMO19} and~\cite{C:RoySin21} remove such offsets for bases $r^N$ with $r$ uniform in $\ZNsqx$ and for exponents uniform in $[0..N-1]$. We adapt the offset removal of~\cite{C:BraGol10,C:DGIMMO19} to our distributions and reduce it to \cref{lem:interactive-vector-game} at one query.

\PlaceRelocated{offset-removal}

	\section{Image-Size Bounds and Generic-Group Analysis}\label{app:small-exponent-details}

We prove the birthday bounds behind the lower bound on the \ELF{} threshold in \cref{sec:limits-small-exponent} and the generic-group bounds for bounded powers summarized in \cref{subsec:generic-group-dsedl}.

\subsection{Birthday Collision Bounds}\label{app:collision-bounds}

\PlaceRelocated{collision-bounds}

\PlaceRelocated{elf-threshold-corollary}

\subsection{Generic-Group Bounds for Bounded Powers}\label{app:generic-group-evidence}

\PlaceRelocated{generic-evidence}

\PlaceRelocated{aux-generic-evidence}

	\section{Image Enumeration and the Plain \texorpdfstring{\ELF{}}{ELF}}\label{app:instantiations-enumeration}

We give the prefix dynamic program for \cref{thm:telf-strong-enumeration}, a counting bound for decoding from values in bounded sets, and the complete construction and proofs for the linear-size plain \ELF{} of \cref{subsec:plain-elf}.

\subsection{Enumerating the Trapdoor \texorpdfstring{\ELF{}}{ELF} Image}\label{app:telf-enumeration-algorithm}

\PlaceRelocated{telf-enumeration-algorithm}

In the trapdoor \ELF{} of \cref{subsec:telf-from-eltdh}, we recover the complete $n$-bit input from one decoded residue in $\ZN$. When every decoded value lies in a set of polynomial size, we need many decoded values, by a count of the possible tuples of decoded values.

\PlaceRelocated{packed-row-remark}

\subsection{Linear-Size Plain \texorpdfstring{\ELF{}}{ELF}}\label{app:plain-elf}

\PlaceRelocated{plain-elf-full}

	\section{Related Work}\label{sec:related-work}
	\PlaceRelocated{related-work}

	\PrintDeferredProofAppendix
\fi


\newcommand{\etalchar}[1]{$^{#1}$}
\begin{thebibliography}{OPWW15}

\bibitem[ABG{\etalchar{+}}26]{EC:ABGMM26}
Gennaro Avitabile, Vincenzo Botta, Emanuele Giunta, Marcin Mielniczuk, and
  Francesco Migliaro.
\newblock The malice of {ELFs}: Practical anamorphic-resistant encryption
  without random oracles.
\newblock In Joan Daemen and Emmanuel Thom{\'e}, editors, {\em Advances in
  Cryptology -- {EUROCRYPT}~2026, Part~V}, volume 16545 of {\em Lecture Notes
  in Computer Science}, pages 303--333, Rome, Italy, May~10--14, 2026.
  Springer, Cham, Switzerland.
\newblock \href {https://doi.org/10.1007/978-3-032-25330-9_11}
  {\path{doi:10.1007/978-3-032-25330-9_11}}.

\bibitem[ACK23]{PKC:AttCapKoh23}
Thomas Attema, Pedro Capit{\~a}o, and Lisa Kohl.
\newblock On homomorphic secret sharing from polynomial-modulus {LWE}.
\newblock In Alexandra Boldyreva and Vladimir Kolesnikov, editors, {\em
  PKC~2023: 26th International Conference on Theory and Practice of Public Key
  Cryptography, Part~II}, volume 13941 of {\em Lecture Notes in Computer
  Science}, pages 3--32, Atlanta, GA, USA, May~7--10, 2023. Springer, Cham,
  Switzerland.
\newblock \href {https://doi.org/10.1007/978-3-031-31371-4_1}
  {\path{doi:10.1007/978-3-031-31371-4_1}}.

\bibitem[ADOS22]{C:ADOS22}
Damiano Abram, Ivan Damg{\r a}rd, Claudio Orlandi, and Peter Scholl.
\newblock An algebraic framework for silent preprocessing with trustless setup
  and active security.
\newblock In Yevgeniy Dodis and Thomas Shrimpton, editors, {\em Advances in
  Cryptology -- {CRYPTO}~2022, Part~IV}, volume 13510 of {\em Lecture Notes in
  Computer Science}, pages 421--452, Santa Barbara, CA, USA, August~15--18,
  2022. Springer, Cham, Switzerland.
\newblock \href {https://doi.org/10.1007/978-3-031-15985-5_15}
  {\path{doi:10.1007/978-3-031-15985-5_15}}.

\bibitem[AFLT12]{EC:AFLT12}
Michel Abdalla, Pierre-Alain Fouque, Vadim Lyubashevsky, and Mehdi Tibouchi.
\newblock Tightly-secure signatures from lossy identification schemes.
\newblock In David Pointcheval and Thomas Johansson, editors, {\em Advances in
  Cryptology -- {EUROCRYPT}~2012}, volume 7237 of {\em Lecture Notes in
  Computer Science}, pages 572--590, Cambridge, UK, April~15--19, 2012.
  Springer Berlin Heidelberg, Germany.
\newblock \href {https://doi.org/10.1007/978-3-642-29011-4_34}
  {\path{doi:10.1007/978-3-642-29011-4_34}}.

\bibitem[AFLT16]{JC:AFLT16}
Michel Abdalla, Pierre-Alain Fouque, Vadim Lyubashevsky, and Mehdi Tibouchi.
\newblock Tightly secure signatures from lossy identification schemes.
\newblock {\em Journal of Cryptology}, 29(3):597--631, July 2016.
\newblock \href {https://doi.org/10.1007/s00145-015-9203-7}
  {\path{doi:10.1007/s00145-015-9203-7}}.

\bibitem[AMR25a]{TCC:AbrMalRoy25}
Damiano Abram, Giulio Malavolta, and Lawrence Roy.
\newblock Slightly sublinear trapdoor hash functions and {PIR} from low-noise
  {LPN}.
\newblock In Benny Applebaum and Huijia~(Rachel) Lin, editors, {\em TCC~2025:
  23rd Theory of Cryptography Conference, Part~I}, volume 16268 of {\em Lecture
  Notes in Computer Science}, pages 280--300, Aarhus, Denmark, December~1--5,
  2025. Springer, Cham, Switzerland.
\newblock \href {https://doi.org/10.1007/978-3-032-12287-2_10}
  {\path{doi:10.1007/978-3-032-12287-2_10}}.

\bibitem[AMR25b]{STOC:AbrMalRoy25}
Damiano Abram, Giulio Malavolta, and Lawrence Roy.
\newblock Succinct oblivious tensor evaluation and applications:
  Adaptively-secure laconic function evaluation and trapdoor hashing for all
  circuits.
\newblock In Michal Kouck\'{y} and Nikhil Bansal, editors, {\em 57th Annual
  {ACM} Symposium on Theory of Computing}, pages 1875--1886, Prague, Czechia,
  June~23--27, 2025. {ACM} Press.
\newblock \href {https://doi.org/10.1145/3717823.3718164}
  {\path{doi:10.1145/3717823.3718164}}.

\bibitem[BCD{\etalchar{+}}25]{EC:BCDJMS25}
Pedro Branco, Arka~Rai Choudhuri, Nico D{\"o}ttling, Abhishek Jain, Giulio
  Malavolta, and Akshayaram Srinivasan.
\newblock Black-box non-interactive zero knowledge from vector trapdoor hash.
\newblock In Serge Fehr and Pierre-Alain Fouque, editors, {\em Advances in
  Cryptology -- {EUROCRYPT}~2025, Part~IV}, volume 15604 of {\em Lecture Notes
  in Computer Science}, pages 64--92, Madrid, Spain, May~4--8, 2025. Springer,
  Cham, Switzerland.
\newblock \href {https://doi.org/10.1007/978-3-031-91134-7_3}
  {\path{doi:10.1007/978-3-031-91134-7_3}}.

\bibitem[BCG{\etalchar{+}}17]{CCS:BCGIO17}
Elette Boyle, Geoffroy Couteau, Niv Gilboa, Yuval Ishai, and Michele Orr{\`u}.
\newblock Homomorphic secret sharing: Optimizations and applications.
\newblock In Bhavani~M. Thuraisingham, David Evans, Tal Malkin, and Dongyan Xu,
  editors, {\em ACM CCS 2017: 24th Conference on Computer and Communications
  Security}, pages 2105--2122, Dallas, TX, USA, October~31~--~November~2, 2017.
  {ACM} Press.
\newblock \href {https://doi.org/10.1145/3133956.3134107}
  {\path{doi:10.1145/3133956.3134107}}.

\bibitem[BDHL25]{AC:BDHL25}
Rishabh Bhadauria, Nico D{\"o}ttling, Carmit Hazay, and Chuanwei Lin.
\newblock Laconic cryptography with preprocessing.
\newblock In Goichiro Hanaoka and Bo-Yin Yang, editors, {\em Advances in
  Cryptology -- {ASIACRYPT}~2025, Part~V}, volume 16249 of {\em Lecture Notes
  in Computer Science}, pages 228--261, Melbourne, VIC, Australia,
  December~8--12, 2025. Springer, Singapore, Singapore.
\newblock \href {https://doi.org/10.1007/978-981-95-5116-3_8}
  {\path{doi:10.1007/978-981-95-5116-3_8}}.

\bibitem[BDSZ24]{PKC:BDSZ24}
Pedro Branco, Nico D{\"o}ttling, Akshayaram Srinivasan, and Riccardo Zanotto.
\newblock Rate-1 fully local somewhere extractable hashing from {DDH}.
\newblock In Qiang Tang and Vanessa Teague, editors, {\em PKC~2024: 27th
  International Conference on Theory and Practice of Public Key Cryptography,
  Part~III}, volume 14603 of {\em Lecture Notes in Computer Science}, pages
  356--386, Sydney, NSW, Australia, April~15--17, 2024. Springer, Cham,
  Switzerland.
\newblock \href {https://doi.org/10.1007/978-3-031-57725-3_12}
  {\path{doi:10.1007/978-3-031-57725-3_12}}.

\bibitem[Ber80]{Ber1980}
Geoffrey~C. Berresford.
\newblock {The Uniformity Assumption in the Birthday Problem}.
\newblock {\em Mathematics Magazine}, 53(5):286--288, 1980.
\newblock \href {https://doi.org/10.1080/0025570x.1980.11976872}
  {\path{doi:10.1080/0025570x.1980.11976872}}.

\bibitem[BFO08]{C:BolFehONe08}
Alexandra Boldyreva, Serge Fehr, and Adam O'Neill.
\newblock On notions of security for deterministic encryption, and efficient
  constructions without random oracles.
\newblock In David Wagner, editor, {\em Advances in Cryptology --
  {CRYPTO}~2008}, volume 5157 of {\em Lecture Notes in Computer Science}, pages
  335--359, Santa Barbara, CA, USA, August~17--21, 2008. Springer Berlin
  Heidelberg, Germany.
\newblock \href {https://doi.org/10.1007/978-3-540-85174-5_19}
  {\path{doi:10.1007/978-3-540-85174-5_19}}.

\bibitem[BG10]{C:BraGol10}
Zvika Brakerski and Shafi Goldwasser.
\newblock Circular and leakage resilient public-key encryption under subgroup
  indistinguishability - (or: Quadratic residuosity strikes back).
\newblock In Tal Rabin, editor, {\em Advances in Cryptology -- {CRYPTO}~2010},
  volume 6223 of {\em Lecture Notes in Computer Science}, pages 1--20, Santa
  Barbara, CA, USA, August~15--19, 2010. Springer Berlin Heidelberg, Germany.
\newblock \href {https://doi.org/10.1007/978-3-642-14623-7_1}
  {\path{doi:10.1007/978-3-642-14623-7_1}}.

\bibitem[BG25]{EC:BitGar25}
Nir Bitansky and Rachit Garg.
\newblock Succinct randomized encodings from laconic function evaluation,
  faster and simpler.
\newblock In Serge Fehr and Pierre-Alain Fouque, editors, {\em Advances in
  Cryptology -- {EUROCRYPT}~2025, Part~VII}, volume 15607 of {\em Lecture Notes
  in Computer Science}, pages 406--436, Madrid, Spain, May~4--8, 2025.
  Springer, Cham, Switzerland.
\newblock \href {https://doi.org/10.1007/978-3-031-91098-2_15}
  {\path{doi:10.1007/978-3-031-91098-2_15}}.

\bibitem[BGI15]{EC:BoyGilIsh15}
Elette Boyle, Niv Gilboa, and Yuval Ishai.
\newblock Function secret sharing.
\newblock In Elisabeth Oswald and Marc Fischlin, editors, {\em Advances in
  Cryptology -- {EUROCRYPT}~2015, Part~II}, volume 9057 of {\em Lecture Notes
  in Computer Science}, pages 337--367, Sofia, Bulgaria, April~26--30, 2015.
  Springer Berlin Heidelberg, Germany.
\newblock \href {https://doi.org/10.1007/978-3-662-46803-6_12}
  {\path{doi:10.1007/978-3-662-46803-6_12}}.

\bibitem[BGI16a]{C:BoyGilIsh16}
Elette Boyle, Niv Gilboa, and Yuval Ishai.
\newblock Breaking the circuit size barrier for secure computation under {DDH}.
\newblock In Matthew Robshaw and Jonathan Katz, editors, {\em Advances in
  Cryptology -- {CRYPTO}~2016, Part~I}, volume 9814 of {\em Lecture Notes in
  Computer Science}, pages 509--539, Santa Barbara, CA, USA, August~14--18,
  2016. Springer Berlin Heidelberg, Germany.
\newblock \href {https://doi.org/10.1007/978-3-662-53018-4_19}
  {\path{doi:10.1007/978-3-662-53018-4_19}}.

\bibitem[BGI16b]{CCS:BoyGilIsh16}
Elette Boyle, Niv Gilboa, and Yuval Ishai.
\newblock Function secret sharing: Improvements and extensions.
\newblock In Edgar~R. Weippl, Stefan Katzenbeisser, Christopher Kruegel,
  Andrew~C. Myers, and Shai Halevi, editors, {\em ACM CCS 2016: 23rd Conference
  on Computer and Communications Security}, pages 1292--1303, Vienna, Austria,
  October~24--28, 2016. {ACM} Press.
\newblock \href {https://doi.org/10.1145/2976749.2978429}
  {\path{doi:10.1145/2976749.2978429}}.

\bibitem[BGI{\etalchar{+}}18]{ITCS:BGILT18}
Elette Boyle, Niv Gilboa, Yuval Ishai, Huijia Lin, and Stefano Tessaro.
\newblock Foundations of homomorphic secret sharing.
\newblock In Anna~R. Karlin, editor, {\em ITCS 2018: 9th Innovations in
  Theoretical Computer Science Conference}, volume~94, pages 21:1--21:21,
  Cambridge, MA, USA, January~11--14, 2018. Leibniz International Proceedings
  in Informatics (LIPIcs).
\newblock \href {https://doi.org/10.4230/LIPIcs.ITCS.2018.21}
  {\path{doi:10.4230/LIPIcs.ITCS.2018.21}}.

\bibitem[BHK13]{C:BelHoaKee13}
Mihir Bellare, Viet~Tung Hoang, and Sriram Keelveedhi.
\newblock Instantiating random oracles via {UCEs}.
\newblock In Ran Canetti and Juan~A. Garay, editors, {\em Advances in
  Cryptology -- {CRYPTO}~2013, Part~II}, volume 8043 of {\em Lecture Notes in
  Computer Science}, pages 398--415, Santa Barbara, CA, USA, August~18--22,
  2013. Springer Berlin Heidelberg, Germany.
\newblock \href {https://doi.org/10.1007/978-3-642-40084-1_23}
  {\path{doi:10.1007/978-3-642-40084-1_23}}.

\bibitem[BK04]{EC:BelKoh04}
Mihir Bellare and Tadayoshi Kohno.
\newblock Hash function balance and its impact on birthday attacks.
\newblock In Christian Cachin and Jan Camenisch, editors, {\em Advances in
  Cryptology -- {EUROCRYPT}~2004}, volume 3027 of {\em Lecture Notes in
  Computer Science}, pages 401--418, Interlaken, Switzerland, May~2--6, 2004.
  Springer Berlin Heidelberg, Germany.
\newblock \href {https://doi.org/10.1007/978-3-540-24676-3_24}
  {\path{doi:10.1007/978-3-540-24676-3_24}}.

\bibitem[BKPW12]{EC:BKPW12}
Mihir Bellare, Eike Kiltz, Chris Peikert, and Brent Waters.
\newblock Identity-based (lossy) trapdoor functions and applications.
\newblock In David Pointcheval and Thomas Johansson, editors, {\em Advances in
  Cryptology -- {EUROCRYPT}~2012}, volume 7237 of {\em Lecture Notes in
  Computer Science}, pages 228--245, Cambridge, UK, April~15--19, 2012.
  Springer Berlin Heidelberg, Germany.
\newblock \href {https://doi.org/10.1007/978-3-642-29011-4_15}
  {\path{doi:10.1007/978-3-642-29011-4_15}}.

\bibitem[BKS19]{EC:BoyKohSch19}
Elette Boyle, Lisa Kohl, and Peter Scholl.
\newblock Homomorphic secret sharing from lattices without {FHE}.
\newblock In Yuval Ishai and Vincent Rijmen, editors, {\em Advances in
  Cryptology -- {EUROCRYPT}~2019, Part~II}, volume 11477 of {\em Lecture Notes
  in Computer Science}, pages 3--33, Darmstadt, Germany, May~19--23, 2019.
  Springer, Cham, Switzerland.
\newblock \href {https://doi.org/10.1007/978-3-030-17656-3_1}
  {\path{doi:10.1007/978-3-030-17656-3_1}}.

\bibitem[BL17]{C:BoyLi17}
Xavier Boyen and Qinyi Li.
\newblock All-but-many lossy trapdoor functions from lattices and applications.
\newblock In Jonathan Katz and Hovav Shacham, editors, {\em Advances in
  Cryptology -- {CRYPTO}~2017, Part~III}, volume 10403 of {\em Lecture Notes in
  Computer Science}, pages 298--331, Santa Barbara, CA, USA, August~20--24,
  2017. Springer, Cham, Switzerland.
\newblock \href {https://doi.org/10.1007/978-3-319-63697-9_11}
  {\path{doi:10.1007/978-3-319-63697-9_11}}.

\bibitem[BW10]{ACNS:BoyWat10}
Xavier Boyen and Brent Waters.
\newblock Shrinking the keys of discrete-log-type lossy trapdoor functions.
\newblock In Jianying Zhou and Moti Yung, editors, {\em ACNS 2010: 8th
  International Conference on Applied Cryptography and Network Security},
  volume 6123 of {\em Lecture Notes in Computer Science}, pages 35--52,
  Beijing, China, June~22--25, 2010. Springer Berlin Heidelberg, Germany.
\newblock \href {https://doi.org/10.1007/978-3-642-13708-2_3}
  {\path{doi:10.1007/978-3-642-13708-2_3}}.

\bibitem[CDG{\etalchar{+}}17]{C:CDGGMP17}
Chongwon Cho, Nico D{\"o}ttling, Sanjam Garg, Divya Gupta, Peihan Miao, and
  Antigoni Polychroniadou.
\newblock Laconic oblivious transfer and its applications.
\newblock In Jonathan Katz and Hovav Shacham, editors, {\em Advances in
  Cryptology -- {CRYPTO}~2017, Part~II}, volume 10402 of {\em Lecture Notes in
  Computer Science}, pages 33--65, Santa Barbara, CA, USA, August~20--24, 2017.
  Springer, Cham, Switzerland.
\newblock \href {https://doi.org/10.1007/978-3-319-63715-0_2}
  {\path{doi:10.1007/978-3-319-63715-0_2}}.

\bibitem[CDG18]{C:CorDodGuo18}
Sandro Coretti, Yevgeniy Dodis, and Siyao Guo.
\newblock Non-uniform bounds in the random-permutation, ideal-cipher, and
  generic-group models.
\newblock In Hovav Shacham and Alexandra Boldyreva, editors, {\em Advances in
  Cryptology -- {CRYPTO}~2018, Part~I}, volume 10991 of {\em Lecture Notes in
  Computer Science}, pages 693--721, Santa Barbara, CA, USA, August~19--23,
  2018. Springer, Cham, Switzerland.
\newblock \href {https://doi.org/10.1007/978-3-319-96884-1_23}
  {\path{doi:10.1007/978-3-319-96884-1_23}}.

\bibitem[CGH98]{STOC:CanGolHal98}
Ran Canetti, Oded Goldreich, and Shai Halevi.
\newblock The random oracle methodology, revisited (preliminary version).
\newblock In {\em 30th Annual {ACM} Symposium on Theory of Computing}, pages
  209--218, Dallas, TX, USA, May~23--26, 1998. {ACM} Press.
\newblock \href {https://doi.org/10.1145/276698.276741}
  {\path{doi:10.1145/276698.276741}}.

\bibitem[CHP25]{EC:CouHegPu25}
Geoffroy Couteau, Aditya Hegde, and Sihang Pu.
\newblock Enhanced trapdoor hashing from {DDH} and {DCR}.
\newblock In Serge Fehr and Pierre-Alain Fouque, editors, {\em Advances in
  Cryptology -- {EUROCRYPT}~2025, Part~VI}, volume 15606 of {\em Lecture Notes
  in Computer Science}, pages 123--151, Madrid, Spain, May~4--8, 2025.
  Springer, Cham, Switzerland.
\newblock \href {https://doi.org/10.1007/978-3-031-91095-1_5}
  {\path{doi:10.1007/978-3-031-91095-1_5}}.

\bibitem[CHW26]{EC:CorHenWu26}
Henry {Corrigan-Gibbs}, Alexandra Henzinger, and David~J. Wu.
\newblock The structured generic-group model.
\newblock In Joan Daemen and Emmanuel Thom{\'e}, editors, {\em Advances in
  Cryptology -- {EUROCRYPT}~2026, Part~V}, volume 16545 of {\em Lecture Notes
  in Computer Science}, pages 95--124, Rome, Italy, May~10--14, 2026. Springer,
  Cham, Switzerland.
\newblock \href {https://doi.org/10.1007/978-3-032-25330-9_4}
  {\path{doi:10.1007/978-3-032-25330-9_4}}.

\bibitem[CK18]{EC:CorKog18}
Henry {Corrigan-Gibbs} and Dmitry Kogan.
\newblock The discrete-logarithm problem with preprocessing.
\newblock In Jesper~Buus Nielsen and Vincent Rijmen, editors, {\em Advances in
  Cryptology -- {EUROCRYPT}~2018, Part~II}, volume 10821 of {\em Lecture Notes
  in Computer Science}, pages 415--447, Tel Aviv, Israel, April~29~--~May~3,
  2018. Springer, Cham, Switzerland.
\newblock \href {https://doi.org/10.1007/978-3-319-78375-8_14}
  {\path{doi:10.1007/978-3-319-78375-8_14}}.

\bibitem[CS06]{PKC:ComSem06}
An~Commeine and Igor Semaev.
\newblock An algorithm to solve the discrete logarithm problem with the number
  field sieve.
\newblock In Moti Yung, Yevgeniy Dodis, Aggelos Kiayias, and Tal Malkin,
  editors, {\em PKC~2006: 9th International Conference on Theory and Practice
  of Public Key Cryptography}, volume 3958 of {\em Lecture Notes in Computer
  Science}, pages 174--190, New York, NY, USA, April~24--26, 2006. Springer
  Berlin Heidelberg, Germany.
\newblock \href {https://doi.org/10.1007/11745853_12}
  {\path{doi:10.1007/11745853_12}}.

\bibitem[DGGM19]{FOCS:DGGM19}
Nico D{\"o}ttling, Sanjam Garg, Vipul Goyal, and Giulio Malavolta.
\newblock Laconic conditional disclosure of secrets and applications.
\newblock In David Zuckerman, editor, {\em 60th Annual Symposium on Foundations
  of Computer Science}, pages 661--685, Baltimore, MD, USA, November~9--12,
  2019. {IEEE} Computer Society Press.
\newblock \href {https://doi.org/10.1109/FOCS.2019.00046}
  {\path{doi:10.1109/FOCS.2019.00046}}.

\bibitem[DGI{\etalchar{+}}19]{C:DGIMMO19}
Nico D{\"o}ttling, Sanjam Garg, Yuval Ishai, Giulio Malavolta, Tamer Mour, and
  Rafail Ostrovsky.
\newblock Trapdoor hash functions and their applications.
\newblock In Alexandra Boldyreva and Daniele Micciancio, editors, {\em Advances
  in Cryptology -- {CRYPTO}~2019, Part~III}, volume 11694 of {\em Lecture Notes
  in Computer Science}, pages 3--32, Santa Barbara, CA, USA, August~18--22,
  2019. Springer, Cham, Switzerland.
\newblock \href {https://doi.org/10.1007/978-3-030-26954-8_1}
  {\path{doi:10.1007/978-3-030-26954-8_1}}.

\bibitem[DGKV22]{FOCS:DGKV22}
Lalita Devadas, Rishab Goyal, Yael Kalai, and Vinod Vaikuntanathan.
\newblock Rate-1 non-interactive arguments for batch-{NP} and applications.
\newblock In {\em 63rd Annual Symposium on Foundations of Computer Science},
  pages 1057--1068, Denver, CO, USA, October~31~--~November~3, 2022. {IEEE}
  Computer Society Press.
\newblock \href {https://doi.org/10.1109/FOCS54457.2022.00103}
  {\path{doi:10.1109/FOCS54457.2022.00103}}.

\bibitem[DIJL23]{C:DIJL23}
Quang Dao, Yuval Ishai, Aayush Jain, and Huijia Lin.
\newblock Multi-party homomorphic secret sharing and sublinear {MPC} from
  sparse {LPN}.
\newblock In Helena Handschuh and Anna Lysyanskaya, editors, {\em Advances in
  Cryptology -- {CRYPTO}~2023, Part~II}, volume 14082 of {\em Lecture Notes in
  Computer Science}, pages 315--348, Santa Barbara, CA, USA, August~20--24,
  2023. Springer, Cham, Switzerland.
\newblock \href {https://doi.org/10.1007/978-3-031-38545-2_11}
  {\path{doi:10.1007/978-3-031-38545-2_11}}.

\bibitem[DJ01]{PKC:DamJur01}
Ivan Damg{\aa}rd and Mats Jurik.
\newblock A generalisation, a simplification and some applications of
  {Paillier}'s probabilistic public-key system.
\newblock In Kwangjo Kim, editor, {\em PKC~2001: 4th International Workshop on
  Theory and Practice in Public Key Cryptography}, volume 1992 of {\em Lecture
  Notes in Computer Science}, pages 119--136, Cheju Island, South Korea,
  February~13--15, 2001. Springer Berlin Heidelberg, Germany.
\newblock \href {https://doi.org/10.1007/3-540-44586-2_9}
  {\path{doi:10.1007/3-540-44586-2_9}}.

\bibitem[DJ24]{C:DaoJai24}
Quang Dao and Aayush Jain.
\newblock Lossy cryptography from code-based assumptions.
\newblock In Leonid Reyzin and Douglas Stebila, editors, {\em Advances in
  Cryptology -- {CRYPTO}~2024, Part~III}, volume 14922 of {\em Lecture Notes in
  Computer Science}, pages 34--75, Santa Barbara, CA, USA, August~18--22, 2024.
  Springer, Cham, Switzerland.
\newblock \href {https://doi.org/10.1007/978-3-031-68382-4_2}
  {\path{doi:10.1007/978-3-031-68382-4_2}}.

\bibitem[DKK18]{C:DinKelKle18}
Itai Dinur, Nathan Keller, and Ohad Klein.
\newblock An optimal distributed discrete log protocol with applications to
  homomorphic secret sharing.
\newblock In Hovav Shacham and Alexandra Boldyreva, editors, {\em Advances in
  Cryptology -- {CRYPTO}~2018, Part~III}, volume 10993 of {\em Lecture Notes in
  Computer Science}, pages 213--242, Santa Barbara, CA, USA, August~19--23,
  2018. Springer, Cham, Switzerland.
\newblock \href {https://doi.org/10.1007/978-3-319-96878-0_8}
  {\path{doi:10.1007/978-3-319-96878-0_8}}.

\bibitem[DM13]{EC:DotMul13}
Nico D{\"o}ttling and J{\"o}rn {M{\"u}ller-Quade}.
\newblock Lossy codes and a new variant of the learning-with-errors problem.
\newblock In Thomas Johansson and Phong~Q. Nguyen, editors, {\em Advances in
  Cryptology -- {EUROCRYPT}~2013}, volume 7881 of {\em Lecture Notes in
  Computer Science}, pages 18--34, Athens, Greece, May~26--30, 2013. Springer
  Berlin Heidelberg, Germany.
\newblock \href {https://doi.org/10.1007/978-3-642-38348-9_2}
  {\path{doi:10.1007/978-3-642-38348-9_2}}.

\bibitem[Fel68]{Fel1968a}
William Feller.
\newblock {\em {An Introduction to Probability Theory and Its Applications}},
  volume~1 of {\em Wiley Series in Probability and Statistics}.
\newblock John Wiley \& Sons, New York, 3 edition, 1968.

\bibitem[FGK{\etalchar{+}}10]{PKC:FGKRS10}
David~Mandell Freeman, Oded Goldreich, Eike Kiltz, Alon Rosen, and Gil Segev.
\newblock More constructions of lossy and correlation-secure trapdoor
  functions.
\newblock In Phong~Q. Nguyen and David Pointcheval, editors, {\em PKC~2010:
  13th International Conference on Theory and Practice of Public Key
  Cryptography}, volume 6056 of {\em Lecture Notes in Computer Science}, pages
  279--295, Paris, France, May~26--28, 2010. Springer Berlin Heidelberg,
  Germany.
\newblock \href {https://doi.org/10.1007/978-3-642-13013-7_17}
  {\path{doi:10.1007/978-3-642-13013-7_17}}.

\bibitem[FR23]{TCC:FisRoh23}
Marc Fischlin and Felix Rohrbach.
\newblock Searching for {ELFs} in the cryptographic forest.
\newblock In Guy~N. Rothblum and Hoeteck Wee, editors, {\em TCC~2023: 21st
  Theory of Cryptography Conference, Part~III}, volume 14371 of {\em Lecture
  Notes in Computer Science}, pages 207--236, Taipei, Taiwan,
  November~29~--~December~2, 2023. Springer, Cham, Switzerland.
\newblock \href {https://doi.org/10.1007/978-3-031-48621-0_8}
  {\path{doi:10.1007/978-3-031-48621-0_8}}.

\bibitem[GI14]{EC:GilIsh14}
Niv Gilboa and Yuval Ishai.
\newblock Distributed point functions and their applications.
\newblock In Phong~Q. Nguyen and Elisabeth Oswald, editors, {\em Advances in
  Cryptology -- {EUROCRYPT}~2014}, volume 8441 of {\em Lecture Notes in
  Computer Science}, pages 640--658, Copenhagen, Denmark, May~11--15, 2014.
  Springer Berlin Heidelberg, Germany.
\newblock \href {https://doi.org/10.1007/978-3-642-55220-5_35}
  {\path{doi:10.1007/978-3-642-55220-5_35}}.

\bibitem[GOR11]{TCC:GoyONeRao11}
Vipul Goyal, Adam O'Neill, and Vanishree Rao.
\newblock Correlated-input secure hash functions.
\newblock In Yuval Ishai, editor, {\em TCC~2011: 8th Theory of Cryptography
  Conference}, volume 6597 of {\em Lecture Notes in Computer Science}, pages
  182--200, Providence, RI, USA, March~28--30, 2011. Springer Berlin
  Heidelberg, Germany.
\newblock \href {https://doi.org/10.1007/978-3-642-19571-6_12}
  {\path{doi:10.1007/978-3-642-19571-6_12}}.

\bibitem[GW11]{STOC:GenWic11}
Craig Gentry and Daniel Wichs.
\newblock Separating succinct non-interactive arguments from all falsifiable
  assumptions.
\newblock In Lance Fortnow and Salil~P. Vadhan, editors, {\em 43rd Annual {ACM}
  Symposium on Theory of Computing}, pages 99--108, San Jose, CA, USA,
  June~6--8, 2011. {ACM} Press.
\newblock \href {https://doi.org/10.1145/1993636.1993651}
  {\path{doi:10.1145/1993636.1993651}}.

\bibitem[HLOV11]{AC:HLOV11}
Brett Hemenway, Beno{\^i}t Libert, Rafail Ostrovsky, and Damien Vergnaud.
\newblock Lossy encryption: Constructions from general assumptions and
  efficient selective opening chosen ciphertext security.
\newblock In Dong~Hoon Lee and Xiaoyun Wang, editors, {\em Advances in
  Cryptology -- {ASIACRYPT}~2011}, volume 7073 of {\em Lecture Notes in
  Computer Science}, pages 70--88, Seoul, South Korea, December~4--8, 2011.
  Springer Berlin Heidelberg, Germany.
\newblock \href {https://doi.org/10.1007/978-3-642-25385-0_4}
  {\path{doi:10.1007/978-3-642-25385-0_4}}.

\bibitem[HO13]{AC:HemOst13}
Brett Hemenway and Rafail Ostrovsky.
\newblock Building lossy trapdoor functions from lossy encryption.
\newblock In Kazue Sako and Palash Sarkar, editors, {\em Advances in Cryptology
  -- {ASIACRYPT}~2013, Part~II}, volume 8270 of {\em Lecture Notes in Computer
  Science}, pages 241--260, Bengalore, India, December~1--5, 2013. Springer
  Berlin Heidelberg, Germany.
\newblock \href {https://doi.org/10.1007/978-3-642-42045-0_13}
  {\path{doi:10.1007/978-3-642-42045-0_13}}.

\bibitem[Hoe63]{Hoe1963}
Wassily Hoeffding.
\newblock {Probability Inequalities for Sums of Bounded Random Variables}.
\newblock {\em Journal of the American Statistical Association},
  58(301):13--30, 1963.
\newblock \href {https://doi.org/10.1080/01621459.1963.10500830}
  {\path{doi:10.1080/01621459.1963.10500830}}.

\bibitem[Hof12]{EC:Hofheinz12}
Dennis Hofheinz.
\newblock All-but-many lossy trapdoor functions.
\newblock In David Pointcheval and Thomas Johansson, editors, {\em Advances in
  Cryptology -- {EUROCRYPT}~2012}, volume 7237 of {\em Lecture Notes in
  Computer Science}, pages 209--227, Cambridge, UK, April~15--19, 2012.
  Springer Berlin Heidelberg, Germany.
\newblock \href {https://doi.org/10.1007/978-3-642-29011-4_14}
  {\path{doi:10.1007/978-3-642-29011-4_14}}.

\bibitem[HW15]{ITCS:HubWic15}
Pavel Hubacek and Daniel Wichs.
\newblock On the communication complexity of secure function evaluation with
  long output.
\newblock In Tim Roughgarden, editor, {\em ITCS 2015: 6th Conference on
  Innovations in Theoretical Computer Science}, pages 163--172, Rehovot,
  Israel, January~11--13, 2015. Association for Computing Machinery.
\newblock \href {https://doi.org/10.1145/2688073.2688105}
  {\path{doi:10.1145/2688073.2688105}}.

\bibitem[IKNP03]{C:IKNP03}
Yuval Ishai, Joe Kilian, Kobbi Nissim, and Erez Petrank.
\newblock Extending oblivious transfers efficiently.
\newblock In Dan Boneh, editor, {\em Advances in Cryptology -- {CRYPTO}~2003},
  volume 2729 of {\em Lecture Notes in Computer Science}, pages 145--161, Santa
  Barbara, CA, USA, August~17--21, 2003. Springer Berlin Heidelberg, Germany.
\newblock \href {https://doi.org/10.1007/978-3-540-45146-4_9}
  {\path{doi:10.1007/978-3-540-45146-4_9}}.

\bibitem[IL25]{FOCS:IlaLom25}
Rahul Ilango and Alex Lombardi.
\newblock Cryptography meets worst-case complexity: Optimal security and more
  from {iO} and worst-case assumptions.
\newblock In {\em 66th Annual Symposium on Foundations of Computer Science},
  pages 585--614, Sydney, Australia, December~14--17, 2025. {IEEE} Computer
  Society Press.
\newblock \href {https://doi.org/10.1109/FOCS63196.2025.00032}
  {\path{doi:10.1109/FOCS63196.2025.00032}}.

\bibitem[ILM21]{PKC:IshLaiMal21}
Yuval Ishai, Russell W.~F. Lai, and Giulio Malavolta.
\newblock A geometric approach to homomorphic secret sharing.
\newblock In Juan Garay, editor, {\em PKC~2021: 24th International Conference
  on Theory and Practice of Public Key Cryptography, Part~II}, volume 12711 of
  {\em Lecture Notes in Computer Science}, pages 92--119, Virtual Event,
  May~10--13, 2021. Springer, Cham, Switzerland.
\newblock \href {https://doi.org/10.1007/978-3-030-75248-4_4}
  {\path{doi:10.1007/978-3-030-75248-4_4}}.

\bibitem[KK04]{PKC:KosKur04}
Takeshi Koshiba and Kaoru Kurosawa.
\newblock Short exponent {Diffie}-{Hellman} problems.
\newblock In Feng Bao, Robert Deng, and Jianying Zhou, editors, {\em PKC~2004:
  7th International Workshop on Theory and Practice in Public Key
  Cryptography}, volume 2947 of {\em Lecture Notes in Computer Science}, pages
  173--186, Singapore, March~1--4, 2004. Springer Berlin Heidelberg, Germany.
\newblock \href {https://doi.org/10.1007/978-3-540-24632-9_13}
  {\path{doi:10.1007/978-3-540-24632-9_13}}.

\bibitem[LMS18]{AC:LaiMalSch18}
Russell W.~F. Lai, Giulio Malavolta, and Dominique Schr{\"o}der.
\newblock Homomorphic secret sharing for low degree polynomials.
\newblock In Thomas Peyrin and Steven Galbraith, editors, {\em Advances in
  Cryptology -- {ASIACRYPT}~2018, Part~III}, volume 11274 of {\em Lecture Notes
  in Computer Science}, pages 279--309, Brisbane, Queensland, Australia,
  December~2--6, 2018. Springer, Cham, Switzerland.
\newblock \href {https://doi.org/10.1007/978-3-030-03332-3_11}
  {\path{doi:10.1007/978-3-030-03332-3_11}}.

\bibitem[LSSS17]{C:LSSS17}
Beno{\^i}t Libert, Amin Sakzad, Damien Stehl{\'e}, and Ron Steinfeld.
\newblock All-but-many lossy trapdoor functions and selective opening
  chosen-ciphertext security from {LWE}.
\newblock In Jonathan Katz and Hovav Shacham, editors, {\em Advances in
  Cryptology -- {CRYPTO}~2017, Part~III}, volume 10403 of {\em Lecture Notes in
  Computer Science}, pages 332--364, Santa Barbara, CA, USA, August~20--24,
  2017. Springer, Cham, Switzerland.
\newblock \href {https://doi.org/10.1007/978-3-319-63697-9_12}
  {\path{doi:10.1007/978-3-319-63697-9_12}}.

\bibitem[Mac29]{Mac1729}
Colin MacLaurin.
\newblock {{IV}. A Second Letter from Mr. Colin {McLaurin}, Professor of
  Mathematicks in the University of Edinburgh and {F. R. S.} to Martin Folkes,
  Esq.; Concerning the Roots of Equations, with the Demonstration of Other
  Rules in Algebra; Being the Continuation of the Letter Published in the
  {Philosophical Transactions}, No. 394}.
\newblock {\em Philosophical Transactions of the Royal Society of London},
  36(408):59--96, 1729.
\newblock \href {https://doi.org/10.1098/rstl.1729.0011}
  {\path{doi:10.1098/rstl.1729.0011}}.

\bibitem[Mau05]{IMA:Maurer05}
Ueli~M. Maurer.
\newblock Abstract models of computation in cryptography (invited paper).
\newblock In Nigel~P. Smart, editor, {\em 10th IMA International Conference on
  Cryptography and Coding}, volume 3796 of {\em Lecture Notes in Computer
  Science}, pages 1--12, Cirencester, UK, December~19--21, 2005. Springer
  Berlin Heidelberg, Germany.
\newblock \href {https://doi.org/10.1007/11586821_1}
  {\path{doi:10.1007/11586821_1}}.

\bibitem[MMN06]{EPRINT:MirMitNis06}
Ilya Mironov, Anton Mityagin, and Kobbi Nissim.
\newblock Hard instances of the constrained discrete logarithm problem.
\newblock Cryptology ePrint Archive, Report 2006/253, 2006.
\newblock URL: \url{https://eprint.iacr.org/2006/253}.

\bibitem[Mun77]{Mun1977}
A.~G. Munford.
\newblock {A Note on the Uniformity Assumption in the Birthday Problem}.
\newblock {\em The American Statistician}, 31(3):119, 1977.
\newblock \href {https://doi.org/10.1080/00031305.1977.10479214}
  {\path{doi:10.1080/00031305.1977.10479214}}.

\bibitem[MY10]{PKC:MolYil10}
Petros Mol and Scott Yilek.
\newblock Chosen-ciphertext security from slightly lossy trapdoor functions.
\newblock In Phong~Q. Nguyen and David Pointcheval, editors, {\em PKC~2010:
  13th International Conference on Theory and Practice of Public Key
  Cryptography}, volume 6056 of {\em Lecture Notes in Computer Science}, pages
  296--311, Paris, France, May~26--28, 2010. Springer Berlin Heidelberg,
  Germany.
\newblock \href {https://doi.org/10.1007/978-3-642-13013-7_18}
  {\path{doi:10.1007/978-3-642-13013-7_18}}.

\bibitem[Nao03]{C:Naor03}
Moni Naor.
\newblock On cryptographic assumptions and challenges (invited talk).
\newblock In Dan Boneh, editor, {\em Advances in Cryptology -- {CRYPTO}~2003},
  volume 2729 of {\em Lecture Notes in Computer Science}, pages 96--109, Santa
  Barbara, CA, USA, August~17--21, 2003. Springer Berlin Heidelberg, Germany.
\newblock \href {https://doi.org/10.1007/978-3-540-45146-4_6}
  {\path{doi:10.1007/978-3-540-45146-4_6}}.

\bibitem[Nec94]{Nechaev94}
Vassiliy~Ilyich Nechaev.
\newblock Complexity of a determinate algorithm for the discrete logarithm.
\newblock {\em Mathematical Notes}, 55(2):165--172, February 1994.
\newblock \href {https://doi.org/10.1007/BF02113297}
  {\path{doi:10.1007/BF02113297}}.

\bibitem[OPWW15]{AC:OPWW15}
Tatsuaki Okamoto, Krzysztof Pietrzak, Brent Waters, and Daniel Wichs.
\newblock New realizations of somewhere statistically binding hashing and
  positional accumulators.
\newblock In Tetsu Iwata and Jung~Hee Cheon, editors, {\em Advances in
  Cryptology -- {ASIACRYPT}~2015, Part~I}, volume 9452 of {\em Lecture Notes in
  Computer Science}, pages 121--145, Auckland, New Zealand,
  November~30~--~December~3, 2015. Springer Berlin Heidelberg, Germany.
\newblock \href {https://doi.org/10.1007/978-3-662-48797-6_6}
  {\path{doi:10.1007/978-3-662-48797-6_6}}.

\bibitem[OSY21]{EC:OrlSchYak21}
Claudio Orlandi, Peter Scholl, and Sophia Yakoubov.
\newblock The rise of paillier: Homomorphic secret sharing and public-key
  silent {OT}.
\newblock In Anne Canteaut and Fran\c{c}ois-Xavier Standaert, editors, {\em
  Advances in Cryptology -- {EUROCRYPT}~2021, Part~I}, volume 12696 of {\em
  Lecture Notes in Computer Science}, pages 678--708, Zagreb, Croatia,
  October~17--21, 2021. Springer, Cham, Switzerland.
\newblock \href {https://doi.org/10.1007/978-3-030-77870-5_24}
  {\path{doi:10.1007/978-3-030-77870-5_24}}.

\bibitem[Pai99]{EC:Paillier99}
Pascal Paillier.
\newblock Public-key cryptosystems based on composite degree residuosity
  classes.
\newblock In Jacques Stern, editor, {\em Advances in Cryptology --
  {EUROCRYPT}'99}, volume 1592 of {\em Lecture Notes in Computer Science},
  pages 223--238, Prague, Czech Republic, May~2--6, 1999. Springer Berlin
  Heidelberg, Germany.
\newblock \href {https://doi.org/10.1007/3-540-48910-X_16}
  {\path{doi:10.1007/3-540-48910-X_16}}.

\bibitem[Pol78]{Pol1978}
John~M. Pollard.
\newblock {Monte Carlo Methods for Index Computation (Mod P)}.
\newblock {\em Mathematics of Computation}, 32(143):918--924, July 1978.
\newblock \href {https://doi.org/10.2307/2006496} {\path{doi:10.2307/2006496}}.

\bibitem[Pol00]{JC:Pollard00}
John~M. Pollard.
\newblock Kangaroos, monopoly and discrete logarithms.
\newblock {\em Journal of Cryptology}, 13(4):437--447, September 2000.
\newblock \href {https://doi.org/10.1007/s001450010010}
  {\path{doi:10.1007/s001450010010}}.

\bibitem[PRS12]{TCC:PieRosSeg12}
Krzysztof Pietrzak, Alon Rosen, and Gil Segev.
\newblock Lossy functions do not amplify well.
\newblock In Ronald Cramer, editor, {\em TCC~2012: 9th Theory of Cryptography
  Conference}, volume 7194 of {\em Lecture Notes in Computer Science}, pages
  458--475, Taormina, Sicily, Italy, March~19--21, 2012. Springer Berlin
  Heidelberg, Germany.
\newblock \href {https://doi.org/10.1007/978-3-642-28914-9_26}
  {\path{doi:10.1007/978-3-642-28914-9_26}}.

\bibitem[PS98]{C:PatSun98}
Sarvar Patel and Ganapathy~S. Sundaram.
\newblock An efficient discrete log pseudo random generator.
\newblock In Hugo Krawczyk, editor, {\em Advances in Cryptology --
  {CRYPTO}'98}, volume 1462 of {\em Lecture Notes in Computer Science}, pages
  304--317, Santa Barbara, CA, USA, August~23--27, 1998. Springer Berlin
  Heidelberg, Germany.
\newblock \href {https://doi.org/10.1007/BFb0055737}
  {\path{doi:10.1007/BFb0055737}}.

\bibitem[PW08]{STOC:PeiWat08}
Chris Peikert and Brent Waters.
\newblock Lossy trapdoor functions and their applications.
\newblock In Richard~E. Ladner and Cynthia Dwork, editors, {\em 40th Annual
  {ACM} Symposium on Theory of Computing}, pages 187--196, Victoria, BC,
  Canada, May~17--20, 2008. {ACM} Press.
\newblock \href {https://doi.org/10.1145/1374376.1374406}
  {\path{doi:10.1145/1374376.1374406}}.

\bibitem[QWW18]{FOCS:QuaWeeWic18}
Willy Quach, Hoeteck Wee, and Daniel Wichs.
\newblock Laconic function evaluation and applications.
\newblock In Mikkel Thorup, editor, {\em 59th Annual Symposium on Foundations
  of Computer Science}, pages 859--870, Paris, France, October~7--9, 2018.
  {IEEE} Computer Society Press.
\newblock \href {https://doi.org/10.1109/FOCS.2018.00086}
  {\path{doi:10.1109/FOCS.2018.00086}}.

\bibitem[QWW21]{C:QuaWatWic21}
Willy Quach, Brent Waters, and Daniel Wichs.
\newblock Targeted lossy functions and applications.
\newblock In Tal Malkin and Chris Peikert, editors, {\em Advances in Cryptology
  -- {CRYPTO}~2021, Part~IV}, volume 12828 of {\em Lecture Notes in Computer
  Science}, pages 424--453, Virtual Event, August~16--20, 2021. Springer, Cham,
  Switzerland.
\newblock \href {https://doi.org/10.1007/978-3-030-84259-8_15}
  {\path{doi:10.1007/978-3-030-84259-8_15}}.

\bibitem[RS21]{C:RoySin21}
Lawrence Roy and Jaspal Singh.
\newblock Large message homomorphic secret sharing from {DCR} and applications.
\newblock In Tal Malkin and Chris Peikert, editors, {\em Advances in Cryptology
  -- {CRYPTO}~2021, Part~III}, volume 12827 of {\em Lecture Notes in Computer
  Science}, pages 687--717, Virtual Event, August~16--20, 2021. Springer, Cham,
  Switzerland.
\newblock \href {https://doi.org/10.1007/978-3-030-84252-9_23}
  {\path{doi:10.1007/978-3-030-84252-9_23}}.

\bibitem[Sha71]{Sha1971}
Daniel Shanks.
\newblock {Class Number, a Theory of Factorization, and Genera}.
\newblock In {\em 1969~{N}umber Theory Institute}, volume~20 of {\em
  Proceedings of Symposia in Pure Mathematics}, pages 415--440. American
  Mathematical Society, 1971.
\newblock \href {https://doi.org/10.1090/pspum/020/0316385}
  {\path{doi:10.1090/pspum/020/0316385}}.

\bibitem[Sho97]{EC:Shoup97}
Victor Shoup.
\newblock Lower bounds for discrete logarithms and related problems.
\newblock In Walter Fumy, editor, {\em Advances in Cryptology --
  {EUROCRYPT}'97}, volume 1233 of {\em Lecture Notes in Computer Science},
  pages 256--266, Konstanz, Germany, May~11--15, 1997. Springer Berlin
  Heidelberg, Germany.
\newblock \href {https://doi.org/10.1007/3-540-69053-0_18}
  {\path{doi:10.1007/3-540-69053-0_18}}.

\bibitem[vOW96]{EC:vOoWie96}
Paul~C. van Oorschot and Michael~J. Wiener.
\newblock On {Diffie}-{Hellman} key agreement with short exponents.
\newblock In Ueli~M. Maurer, editor, {\em Advances in Cryptology --
  {EUROCRYPT}'96}, volume 1070 of {\em Lecture Notes in Computer Science},
  pages 332--343, Saragossa, Spain, May~12--16, 1996. Springer Berlin
  Heidelberg, Germany.
\newblock \href {https://doi.org/10.1007/3-540-68339-9_29}
  {\path{doi:10.1007/3-540-68339-9_29}}.

\bibitem[vzGS13]{VS2013}
Joachim von~zur Gathen and Igor~E. Shparlinski.
\newblock {Generating Safe Primes}.
\newblock {\em Journal of Mathematical Cryptology}, 7(4):333--365, 2013.
\newblock \href {https://doi.org/10.1515/jmc-2013-5011}
  {\path{doi:10.1515/jmc-2013-5011}}.

\bibitem[Web95]{EC:Weber95}
Damian Weber.
\newblock An implementation of the general number field sieve to compute
  discrete logarithms mod p.
\newblock In Louis~C. Guillou and Jean-Jacques Quisquater, editors, {\em
  Advances in Cryptology -- {EUROCRYPT}'95}, volume 921 of {\em Lecture Notes
  in Computer Science}, pages 95--105, Saint-Malo, France, May~21--25, 1995.
  Springer Berlin Heidelberg, Germany.
\newblock \href {https://doi.org/10.1007/3-540-49264-X_8}
  {\path{doi:10.1007/3-540-49264-X_8}}.

\bibitem[Wee12]{EC:Wee12}
Hoeteck Wee.
\newblock Dual projective hashing and its applications - lossy trapdoor
  functions and more.
\newblock In David Pointcheval and Thomas Johansson, editors, {\em Advances in
  Cryptology -- {EUROCRYPT}~2012}, volume 7237 of {\em Lecture Notes in
  Computer Science}, pages 246--262, Cambridge, UK, April~15--19, 2012.
  Springer Berlin Heidelberg, Germany.
\newblock \href {https://doi.org/10.1007/978-3-642-29011-4_16}
  {\path{doi:10.1007/978-3-642-29011-4_16}}.

\bibitem[YYHK16]{C:YYHK16}
Takashi Yamakawa, Shota Yamada, Goichiro Hanaoka, and Noboru Kunihiro.
\newblock Adversary-dependent lossy trapdoor function from hardness of
  factoring semi-smooth {RSA} subgroup moduli.
\newblock In Matthew Robshaw and Jonathan Katz, editors, {\em Advances in
  Cryptology -- {CRYPTO}~2016, Part~II}, volume 9815 of {\em Lecture Notes in
  Computer Science}, pages 3--32, Santa Barbara, CA, USA, August~14--18, 2016.
  Springer Berlin Heidelberg, Germany.
\newblock \href {https://doi.org/10.1007/978-3-662-53008-5_1}
  {\path{doi:10.1007/978-3-662-53008-5_1}}.

\bibitem[Zha16]{C:Zhandry16}
Mark Zhandry.
\newblock The magic of {ELFs}.
\newblock In Matthew Robshaw and Jonathan Katz, editors, {\em Advances in
  Cryptology -- {CRYPTO}~2016, Part~I}, volume 9814 of {\em Lecture Notes in
  Computer Science}, pages 479--508, Santa Barbara, CA, USA, August~14--18,
  2016. Springer Berlin Heidelberg, Germany.
\newblock \href {https://doi.org/10.1007/978-3-662-53018-4_18}
  {\path{doi:10.1007/978-3-662-53018-4_18}}.

\bibitem[Zha19]{EC:Zhandry19a}
Mark Zhandry.
\newblock On {ELFs}, deterministic encryption, and correlated-input security.
\newblock In Yuval Ishai and Vincent Rijmen, editors, {\em Advances in
  Cryptology -- {EUROCRYPT}~2019, Part~III}, volume 11478 of {\em Lecture Notes
  in Computer Science}, pages 3--32, Darmstadt, Germany, May~19--23, 2019.
  Springer, Cham, Switzerland.
\newblock \href {https://doi.org/10.1007/978-3-030-17659-4_1}
  {\path{doi:10.1007/978-3-030-17659-4_1}}.

\end{thebibliography}
\end{document}